\documentclass[]{lmcs}
\pdfoutput=1

\usepackage{lastpage}
\lmcsdoi{21}{3}{18}
\lmcsheading{}{\pageref{LastPage}}{}{}%
{Jan.~02,~2023}{Aug.~12,~2025}{}

\pdfoutput=1
\usepackage{macros}

\usepackage[leftmargin=1em,rightmargin=1ex,vskip=1ex]{quoting}
\usepackage{lscape}

\title[Coinductive Streams in Monoidal Categories]{Coinductive Streams in Monoidal Categories}

\address{Tallinn University of Technology, Ehitajate tee 5, 12616 Tallinn, Estonia}
\address{Università di Pisa, Lungarno Antonio Pancinotti 43, 56126 Pisa, Italy}
\address{University of Oxford, Oxford OX1 2JD, United Kingdom}
\email{elena.di.lavore@cs.ox.ac.uk, mario.roman.garcia@cs.ox.ac.uk}
\address{Quantinuum, Oxford OX1 2JD, United Kingdom}
\email{giovanni.defelice@cambridgequantum.com}

\author[E. Di Lavore]{Elena Di Lavore\lmcsorcid{0000-0002-7783-5079}}[a,b,c]
\author[G. de Felice]{Giovanni de Felice\lmcsorcid{0009-0001-5222-6824}}[c,d]
\author[M. Rom\'an]{Mario Rom\'an\lmcsorcid{0000-0003-3158-1226}}[a,c]
 
\thanks{\noindent The present is an extended version of ``Monoidal Streams for Dataflow Programming'', by the same authors. Elena Di Lavore and Mario Rom\'an were supported by the ESF funded Estonian IT Academy research measure (project 2014-2020.4.05.19-0001). Mario Román was supported by the Air Force Office of Scientific Research (AFOSR) award number FA9550-21-1-0038. Elena Di Lavore and Mario Román were supported by the Advanced Research + Invention Agency (ARIA) Safeguarded AI Programme.}

\makeatletter
\def\@copyrightspace{\relax}
\makeatother

\begin{document}

\begin{abstract}
  We introduce monoidal streams.
  Monoidal streams are a generalization of causal stream functions, which can be defined in cartesian monoidal categories, to arbitrary symmetric monoidal categories.
  In the same way that streams provide semantics to dataflow programming with pure functions, monoidal streams provide semantics to dataflow programming with arbitrary processes.
  Monoidal streams also form a feedback monoidal category.
  In the same way that the coinductive stream calculus reasons about signal flow graphs, coinductive string diagrams reason about feedback monoidal categories, with semantics in monoidal streams. We highlight the probabilistic case, characterizing probabilistic streams as causal stochastic processes.
\end{abstract}

\keywords{Monoidal stream, Stream, Monoidal category, Dataflow programming, Feedback, Signal flow graph, Coalgebra, Stochastic process.}
\maketitle

\newpage
\setcounter{tocdepth}{1}
\vspace*{-1.5\baselineskip}
\tableofcontents

\section{Introduction}\label{section:introduction}

\paragraph{Dataflow languages}
Dataflow (or \emph{stream-based}) programming languages, such as \textsc{Lucid}~\cite{wadge1985lucid,halbwachs1991lustre},
follow a paradigm in which every declaration represents an infinite list of values:
a \emph{stream}~\cite{benveniste93,uustalu05}.
The following program in a \textsc{Lucid}-like language (\Cref{diagram:dataflowfibonacci}) computes the Fibonacci sequence using the $\Fby$ (``followed by'') and $\Wait$ (``wait'') operators: $\Fby$ constructs a stream from a first element and a stream; $\Wait$ delays a stream by one unit of time.
\begin{figure}[H]
  \centering $\fib = 0\ \Fby\ (\fib + (1\ \Fby\ \Wait(\fib)))$
\caption{\textit{The Fibonacci sequence is 0 followed
by the Fibonacci sequence plus the Fibonacci sequence preceded by a 1.}}\label{diagram:dataflowfibonacci}
\end{figure}
The control structure of dataflow programs is inspired by \emph{signal flow graphs}~\cite{benveniste93,shannon42,mason53}.
Signal flow graphs are diagrammatic specifications of processes with feedback loops, widely used in control system engineering.
In a dataflow program, feedback loops represent how the current value of a stream may depend on its previous values.
For instance, the previous program (\Cref{diagram:dataflowfibonacci}) corresponds to the signal flow graph in \Cref{figure:fibonacci}, left.
\begin{figure}[h!]
  \begin{minipage}{0.3\linewidth}
 \tikzset{every picture/.style={line width=0.85pt}} %
\begin{tikzpicture}[x=0.75pt,y=0.75pt,yscale=-1,xscale=1]
\draw   (50,100) -- (90,100) -- (90,120) -- (50,120) -- cycle ;
\draw   (85,125) -- (105,125) -- (105,145) -- (85,145) -- cycle ;
\draw    (120,95) .. controls (120,135.2) and (116.29,135.8) .. (105,135) ;
\draw    (90,70) .. controls (89.8,55.8) and (120.2,56.6) .. (120,70) ;
\draw    (70,120) .. controls (72,132.09) and (72,134.66) .. (85,135) ;
\draw    (130,62) -- (130,90) ;
\draw [shift={(130,60)}, rotate = 90] [color={rgb, 255:red, 0; green, 0; blue, 0 }  ][line width=0.75]    (10.93,-3.29) .. controls (6.95,-1.4) and (3.31,-0.3) .. (0,0) .. controls (3.31,0.3) and (6.95,1.4) .. (10.93,3.29)   ;
\draw    (105.1,59.1) .. controls (104.9,44.9) and (130.2,46.6) .. (130,60) ;
\draw  [fill={rgb, 255:red, 0; green, 0; blue, 0 }  ,fill opacity=1 ] (102.2,59.1) .. controls (102.2,57.5) and (103.5,56.2) .. (105.1,56.2) .. controls (106.7,56.2) and (108,57.5) .. (108,59.1) .. controls (108,60.7) and (106.7,62) .. (105.1,62) .. controls (103.5,62) and (102.2,60.7) .. (102.2,59.1) -- cycle ;
\draw   (45,70) -- (65,70) -- (65,90) -- (45,90) -- cycle ;
\draw   (55,155) -- (95,155) -- (95,175) -- (55,175) -- cycle ;
\draw   (45,125) -- (65,125) -- (65,145) -- (45,145) -- cycle ;
\draw    (65,195) -- (65,210) ;
\draw   (70,70) -- (110,70) -- (110,90) -- (70,90) -- cycle ;
\draw    (55,90) .. controls (54.57,99.23) and (65.14,93.51) .. (65,100) ;
\draw    (120,70) -- (120,100) ;
\draw    (90,90) .. controls (89.57,99.23) and (80.14,93.51) .. (80,100) ;
\draw    (95,145) .. controls (94.57,154.23) and (85.14,148.51) .. (85,155) ;
\draw    (55,145) .. controls (54.57,154.23) and (65.14,148.51) .. (65,155) ;
\draw    (130,90) .. controls (130.25,158.17) and (105.14,170.09) .. (105,195) ;
\draw    (75,175) -- (75,183.87) ;
\draw    (65,195) .. controls (64.8,180.8) and (85.2,180.8) .. (85,195) ;
\draw  [fill={rgb, 255:red, 0; green, 0; blue, 0 }  ,fill opacity=1 ] (72.1,183.87) .. controls (72.1,182.27) and (73.4,180.97) .. (75,180.97) .. controls (76.6,180.97) and (77.9,182.27) .. (77.9,183.87) .. controls (77.9,185.47) and (76.6,186.77) .. (75,186.77) .. controls (73.4,186.77) and (72.1,185.47) .. (72.1,183.87) -- cycle ;
\draw    (85,195) .. controls (85,204.2) and (105,205) .. (105,195) ;
\draw (55,80) node    {$1$};
\draw (95,135) node    {$+$};
\draw (70,110) node    {$fby$};
\draw (55,135) node    {$0$};
\draw (75,165) node    {$fby$};
\draw (90,80) node    {$wait$};
\end{tikzpicture}
  \end{minipage} \begin{minipage}{0.5\linewidth}
    \begin{gather*} \mathsf{fib} \defn\fbk( \COPY \comp \\
      \partial (1 \times \WAIT) \times \im \comp \\
      \partial(\fby) \times \im \comp \\
      \partial(+) \comp \\
      0 \times \im \comp \\
      \fby \comp \\
      \COPY ) \end{gather*}
  \end{minipage}
  \caption{Fibonacci: signal flow graph and morphism.}\label{figure:fibonacci}\label{figure:finalfibonaccigraph}
\end{figure}

\paragraph{Monoidal categories}
Any theory of processes that \emph{compose sequentially and in parallel}, satisfying reasonable axioms, forms a \emph{monoidal category}.
Examples include functions~\cite{lambek1986a}, probabilistic channels~\cite{cho2019, fritz2020}, partial maps~\cite{robinson88,cockett02}, database queries \cite{Bonchi18}, linear resource theories~\cite{coeckeFS16} and quantum processes~\cite{abramsky2009categorical}.
Signal flow graphs are the graphical syntax for \emph{\feedbackMonoidalCategories{}}~\cite{katis02,bonchi14,bonchi15,feedbackspans2020,kaye22}: they are the \emph{string diagrams} for any of these theories extended with \emph{feedback}.
For instance, the previous program (\Cref{diagram:dataflowfibonacci}) corresponds to the formal morphism in \Cref{figure:fibonacci}.

Yet, semantics of dataflow languages have been mostly restricted to theories of pure functions~\cite{benveniste93,uustalu2008comonadic,cousot19,delpeuch19,oliveira84}: what are called \emph{cartesian} monoidal categories.
We claim that this restriction is actually inessential;
dataflow programs may take semantics in non-cartesian monoidal categories, exactly as their signal flow graphs do.

The present work provides this missing semantics:
we construct \emph{monoidal streams} over a symmetric monoidal category, which form a \emph{\feedbackMonoidalCategory{}}.
Monoidal streams model the values of a \emph{monoidal dataflow language}, in the same way that streams model the values of a classical dataflow language.
This opens the door to stochastic, effectful, or quantum dataflow languages.
In particular, we give semantics and string diagrams for a \emph{stochastic dataflow programming language}, where the following code can be run (\Cref{diagram:dataflowwalk}).

\begin{figure}[H]
  \centering  $\walk = 0\ \Fby\ (\uniform{(-1,1)} + \walk)$
  \caption{\textit{A stochastic dataflow program. A random walk is 0 followed by the random walk plus a stochastic stream of steps to the left (-1) or to the right (1), sampled uniformly.}}
\label{diagram:dataflowwalk}
\end{figure}

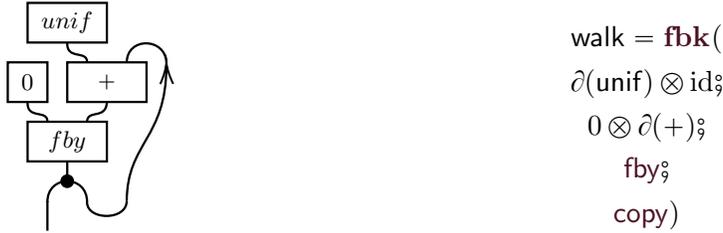
\begin{figure}[h!]
\begin{minipage}{0.3\linewidth}
\tikzset{every picture/.style={line width=0.85pt}}
\begin{tikzpicture}[x=0.75pt,y=0.75pt,yscale=-1,xscale=1]
\draw   (190,110) -- (230,110) -- (230,130) -- (190,130) -- cycle ;
\draw   (170,140) -- (210,140) -- (210,160) -- (170,160) -- cycle ;
\draw   (160,110) -- (180,110) -- (180,130) -- (160,130) -- cycle ;
\draw    (190,160) -- (190,170) ;
\draw    (220,110) .. controls (220.2,97) and (239.4,97.4) .. (240,110) ;
\draw   (170,80) -- (210,80) -- (210,100) -- (170,100) -- cycle ;
\draw    (170,130) .. controls (169.57,139.23) and (180.14,133.51) .. (180,140) ;
\draw    (210,130) .. controls (209.57,139.23) and (200.14,133.51) .. (200,140) ;
\draw    (180,180) -- (180,195) ;
\draw    (180,180.65) .. controls (179.8,166.45) and (200.2,166.45) .. (200,180.65) ;
\draw  [fill={rgb, 255:red, 0; green, 0; blue, 0 }  ,fill opacity=1 ] (187.1,170) .. controls (187.1,168.4) and (188.39,167.1) .. (190,167.1) .. controls (191.6,167.1) and (192.9,168.4) .. (192.9,170) .. controls (192.9,171.6) and (191.6,172.9) .. (190,172.9) .. controls (188.39,172.9) and (187.1,171.6) .. (187.1,170) -- cycle ;
\draw    (200,180.65) .. controls (200,189.85) and (220,190) .. (220,180) ;
\draw    (190,100) .. controls (189.57,109.23) and (200.14,103.51) .. (200,110) ;
\draw    (220,180) .. controls (220.2,150.25) and (239.99,148.64) .. (240.01,111.71) ;
\draw [shift={(240,110)}, rotate = 89.11] [color={rgb, 255:red, 0; green, 0; blue, 0 }  ][line width=0.75]    (10.93,-3.29) .. controls (6.95,-1.4) and (3.31,-0.3) .. (0,0) .. controls (3.31,0.3) and (6.95,1.4) .. (10.93,3.29)   ;

\draw (210,120) node    {$+$};
\draw (170,120) node    {$0$};
\draw (190,150) node    {$fby$};
\draw (190,90) node    {$unif$};
\end{tikzpicture}
\end{minipage} \begin{minipage}{0.5\linewidth}
  \begin{gather*} \mathsf{walk} \defn
    \fbk( \\
    \partial(\mathsf{unif}) \tensor  \im \comp \\
    0 \tensor \partial(+) \comp \\
    \fby \comp \\
    \COPY )
  \end{gather*}
\end{minipage}
\caption{Random walk: signal flow graph and morphism.}
\label{string:walk}
\end{figure}

\subsection{Contributions}

Our main novel contribution is the definition of a \feedbackMonoidalCategory{} of \emph{monoidal streams}
over a \symmetricMonoidalCategory{} ($\STREAM$,~\Cref{def:monoidalstream,th:monoidalstreamsfeedback}).
Monoidal streams form a \finalCoalgebra{}; for sufficiently well-behaved monoidal categories (\Cref{def:productive}), we give an explicit construction of this coalgebra (\Cref{def:observationalsequence}).

In \cartesianCategories{}, the \emph{causal functions} of Uustalu and Vene~\cite{uustalu05} (see also \cite{jacobs:causalfunctions})
are a particular case of our monoidal streams (\Cref{th:cartesianstreams,th:nelist}).
In the category of stochastic functions, our construction captures the notion of
\emph{controlled stochastic process}~\cite{fleming1975,ross1996stochastic} (\Cref{th:stochasticprocesses}).

In order to arrive to this definition, we unify the previous literature:
we characterize the cartesian ``intensional stateful sequences'' of Sprunger and Katsumata~\cite{katsumata19}
with a final coalgebra (\Cref{th:intensionalcoalgebra}), and then ``dinatural stateful sequences''
in terms of the ``feedback monoidal categories'' of Katis, Sabadini and Walters~\cite{katis02} (\Cref{th:ext-stateful-sequences}).
We justify observational equivalence with a refined fixpoint equation that employs \emph{coends} (\Cref{theorem:observationalfinalcoalgebra}).
We strictly generalize ``stateful sequences'' from the cartesian to the monoidal case.

\subsection{Related work}\label{sec:related-work}
The theory of monoidal streams that we present in this manuscript merges multiple research threads: feedback and traced monoidal categories, morphism sequences, coalgebra, and categorical dataflow programming. 
We briefly summarize here these different threads.

\paragraph{Stream programming.}
The theory of stream programming originates in the '70s with the work of Kahn~\cite{kahn1974semantics,kahn1976coroutines}, following the development of dataflow languages~\cite{fosseen1972representation,dennis1972schemas,dennis1974dataflow,weng1975steam}.
During this period, the interest in dataflow languages is widespread, leading to languages like \Lucid{} and \Lustre{}~\cite{kosinski1973data,ashcroft1977lucid,wadge1985lucid,halbwachs1991lustre,ackerman1979data}.

\paragraph{Coalgebraic streams.}
Uustalu and Vene~\cite{uustalu05} provide elegant \emph{comonadic} semantics for a (cartesian) \Lucid-like programming language.
We shall prove that their exact technique cannot be possibly extended to arbitrary monoidal categories (\Cref{th:nelist}).
However, we recover their semantics as a particular case of our monoidal streams (\Cref{th:cartesianstreams}).

\paragraph{Feedback.}
Feedback monoidal categories are a weakening of \emph{traced monoidal categories}.
The construction of the free such categories is originally due to Katis, Sabadini and Walters~\cite{katis02}.
Feedback monoidal categories and their free construction have been repurposed and rediscovered multiple times in the literature~\cite{sabadini95,hoshino14,bonchi19,gay03}.
Di Lavore et al.~\cite{feedbackspans2020} summarize these uses and introduce \emph{delayed feedback}.

\paragraph{Cartesian stateful sequences.}
Sprunger and Katsumata constructed the category of \emph{stateful sequences} in the cartesian case~\cite{katsumata19}.
The present work extends their ideas to the monoidal case.

Sprunger and Katsumata first introduced the idea of a \emph{delayed trace operator} with a \emph{time-shifting functor}.
Its type differs slightly from our \emph{feedback operator} (\Cref{def:feedback}) which follows Katis, Sabadini and Walters~\cite{katis02}: it adds a ``point'', and a forward operator is applied to the other side of the morphism.
However, its purpose and essence is similar: we take care that both operators are interderivable in the cartesian case (\Cref{rem:pointed-delayed-feedback}).

Sprunger and Katsumata also introduced for the first time the idea of equality by truncation of infinite sequences, and called it ``dinatural equivalence''.
We shall define a similar \emph{truncation up to a future} and prove it is actually the \emph{observational equivalence} of a final coalgebra.
Happily, this still particularizes to dinatural equality in the sense of Sprunger and Katsumata when a cartesian structure is present (\Cref{th:cartesianstreams}).
\emph{Sliding}, which appeared first as a consequence of truncation in the work of Sprunger and Katsumata (``Shim lemma'' in \cite{katsumata19}) can be now, for the first time, the fundamental notion of \emph{dinatural equivalence} (\Cref{def:extensionalequality}).

Finally, there is an unfortunate clash of terminology: we use $\St(\bullet)$ for the free category with feedback in the sense of Katis, Sabadini and Walters~\cite{sabadini95,hoshino14,bonchi19}.
\Cref{th:ext-stateful-sequences} relates this $\St(\bullet)$ to the work of Sprunger and Katsumata, who use $\mathrm{St}(\bullet)$ for the cartesian analogue of our $\STREAM$ construction.

\paragraph{Monoidal stateful sequences.}
Our work is based and significantly expands an unpublished work by Rom{\'a}n~\cite{roman2020}, who first exteneded Sprunger and Katsumata's~\cite{katsumata19} definition to the symmetric monoidal case, using \emph{coends} to justify \emph{dinatural equality}.
Shortly after, Carette, de Visme and Perdrix~\cite{carette21} rederived this construction and applied it to the case of completely positive maps between Hilbert spaces, using (a priori) a slightly different notion of equality.
We synthetise this previous work on monoidal sequences, we justify it for the first time using coalgebra and we particularize it to some cases of interest.

\paragraph{General and dependent streams.}
Our work concerns \emph{synchronous} streams: those where, at each point in time $t = 0,1,\dots$,
the stream process takes exactly one input and produces exactly one output.
This condition is important in certain contexts like, for instance, real-time embedded systems; %
but it is not always present.
The study of asynchronous stream transformers and their universal properties is considerably different~\cite{abadi15},
and we refer the reader to the recent work of Garner~\cite{garner2021stream} for a discussion on \emph{non-synchronous} streams.
Finally, when we are concerned with \emph{dependent streams} indexed by time steps,
a possible approach, when our base category is a topos, is to use the \emph{topos of trees}~\cite{birkedal11}.

\paragraph{Categorical dataflow programming.}
Category theory is a common tool of choice for dataflow programming~\cite{rutten00,gay03,mamouras20}.
In particular, profunctors and coends are used by Hildebrandt, Panangaden and Winskel~\cite{hildebrandt1998relational}
to generalise a model of non-deterministic dataflow,
which has been the main focus~\cite{panangaden1988computations,lynch1989proof,lee09} outside cartesian categories.

\subsection{Synopsis}

This manuscript contains three main definitions in terms of universal properties (\emph{intensional, dinatural and observational streams}, \Cref{def:intensionalmonoidalstream,def:extensionalmonoidalstreams,def:observationalmonoidalstream}); and three explicit constructions for them (\emph{intensional, dinatural and observational sequences}, \Cref{def:intensionalsequence,def:extensional-sequence,def:observationalsequence}).
Each definition is refined into the next one: each construction is a quotienting of the previous one.

\Cref{section-prelude-coalgebra,sec:dinaturality} contain expository material on coalgebra and dinaturality. \Cref{sec:int-stateful-sequences} presents intensional monoidal streams. \Cref{section:extensional} introduces dinatural monoidal streams in terms of feedback monoidal categories. \Cref{section:observational} introduces the definitive \emph{observational} equivalence and defines \emph{monoidal streams}. \Cref{sec:monoidal-streams} constructs the feedback monoidal category of monoidal streams. \Cref{section:classicalstreams,section:stochastic-streams} present two examples: cartesian and stochastic streams. %

\section{Coalgebra \& Monoidal Categories}

\subsection{Coalgebra}
\label{section-prelude-coalgebra}
In this preparatory section, we introduce some background material on coalgebra~\cite{rutten00,jacobs2005coalgebras,adamek2005introduction}.
Coalgebra is the category-theoretic study of stateful systems and infinite data-structures, such as streams.
These structures arise as \emph{final coalgebras}: universal solutions to certain functor equations.

Let us fix an endofunctor $F \colon \catC \to \catC$ through the section.

\begin{defi}
  A \emph{coalgebra} $(Y,\beta)$ is an object $Y \in \catC$, together with a
  morphism $\beta \colon Y \to FY$. A \emph{coalgebra
  morphism} $g \colon (Y,\beta) \to (Y',\beta')$ is a morphism
  $g \colon Y \to Y'$ such that $g \comp \beta' = \beta \comp Fg$.
  \[\begin{tikzcd}
    Y \rar{g}\dar[swap]{\beta} & Y' \dar{\beta'} \\
   FY \rar{Fg} & FY'
\end{tikzcd}\]
\end{defi}

Coalgebras for an endofunctor form a category with coalgebra morphisms between them.
A \emph{final coalgebra} is a final object in this category.
As such, final coalgebras are unique up to isomorphism when they exist.

\begin{defi}\defining{linkfinalcoalgebra}{}
A \emph{final coalgebra} is a coalgebra $(Z,\gamma)$ such that for any other coalgebra $(Y,\beta)$ there exists a unique coalgebra morphism $g \colon (Y,\beta) \to (Z,\gamma)$.
\end{defi}

Our interest in final coalgebras derives from the fact that they are canonical fixpoints of an endofunctor.
Specifically, Lambek's theorem (\Cref{th:lambektheorem}) states that whenever the final coalgebra exists, it is a fixpoint.

\begin{defi}
  A \emph{fixpoint} is a coalgebra $(Y,\beta)$ such that $\beta \colon Y \to FY$ is an
  isomorphism. A \emph{fixpoint morphism} is a coalgebra morphism between fixpoints:
  fixpoints and fixpoint morphisms form a full subcategory of the category of coalgebras.
  A \emph{final fixpoint} is a final object in this category.
\end{defi}

\begin{thm}[Lambek, \cite{lambek68}]\label{th:lambektheorem}
  Final coalgebras are fixpoints.
  As a consequence, when they exist, they are final fixpoints.
\end{thm}

The last question before continuing is how to explicitly construct a final coalgebra.
This is answered by Adamek's theorem (\Cref{th:adamek}).
The reader may be familiar with Kleene's theorem for constructing fixpoints~\cite{stoltenberg}:
the least fixpoint of a monotone function $f \colon X \to X$ in a directed-complete partial order $(X,\leq)$ is
the supremum of the chain $\bot \leq f(\bot) \leq f(f(\bot)) \leq \dots$,
where $\bot$ is the least element of the partial order, whenever this supremum is preserved by $f$.
This same result can be categorified into a fixpoint theorem for constructing final coalgebras:
the directed-complete poset becomes a category with $\omega$-chain limits;
the monotone function becomes an endofunctor;
and the least element becomes the final object.

\begin{thm}[Adamek, \cite{adamek74}]\label{th:adamek}
  Let $\catD$ be a category with a final object $1$ and $\omega^{op}$-shaped
  limits. Let $F \colon \catD \to \catD$ be an endofunctor. We write $L \defn \lim\nolimits_{n} F^{n}1$ for
  the limit of the following $\omega$-chain, which is called the \emph{terminal sequence} of $F$.
  \[1 \overset{!}\longleftarrow F1 \overset{F!}\longleftarrow FF1 \overset{FF!}\longleftarrow FFF1 \overset{FFF!}\longleftarrow \dots \]
  Assume that $F$ preserves this limit, meaning that the canonical morphism $FL \to L$ is an isomorphism.
  Then, $L$ is the final $F$-coalgebra.
\end{thm}

\subsection{Monoidal categories}
In order to fix notation, we include background material on monoidal categories.
Monoidal categories have a sound and complete syntax in terms of string diagrams, which is the one we will mostly use during this text \cite{joyal1991geometry}.
However, when proving that some category is indeed monoidal, we may prefer the classical definition of monoidal categories.

\begin{defiC}[\cite{maclane78}]
  A \defining{linkmonoidalcategory}{\textbf{monoidal category}},
  \[(\catC, \otimes, I, \alpha, \lambda, \rho),\] is a category $\catC$
  equipped with a functor $\tensor \colon \catC \times \catC \to \catC$,
  a unit $\sI \in \catC$, and three natural isomorphisms: the associator $\alpha_{\sA,\sB,\sC} \colon (\sA \tensor \sB) \tensor \sC \cong \sA \tensor (\sB \tensor \sC)$, the left unitor $\lambda_{\sA} \colon \sI \tensor \sA \cong \sA$ and
  the right unitor $\rho_{\sA} \colon \sA \tensor \sI \cong \sA$;
  such that $\alpha_{\sA,\sI,\sB} \comp (\im_{\sA} \tensor \lambda_{\sB}) = \rho_{\sA} \tensor \im_{\sB}$ and
  $(\alpha_{A,B,C} \tensor \im) \comp \alpha_{A,B \tensor C, D} \comp (\im_{A} \tensor \alpha_{B,C,D}) = \alpha_{A\tensor B,C,D} \comp \alpha_{A,B,C \tensor D}$.  A monoidal category is \emph{strict} if $\alpha$, $\lambda$ and $\rho$ are identities.
\end{defiC}

\begin{defi}[Monoidal functor, \cite{maclane78}]\defining{linkmonoidalfunctor}{}
  Let \[(\catC,\tensor,\sI,\alpha^{\catC},\lambda^{\catC},\rho^{\catC})\mbox{ and } (\catD,\boxtimes,\sJ,\alpha^{\catD},\lambda^{\catD},\rho^{\catD})\] be \hyperlink{linkmonoidalcategory}{monoidal categories}.
  A \defining{linkmonoidalfunctor}{\emph{monoidal functor}} (sometimes called \emph{strong monoidal functor}) is a triple
  $(F,\varepsilon,\mu)$ consisting of a functor $F \colon \catC \to \catD$ and two natural
  isomorphisms $\varepsilon \colon \sJ \cong F(\sI)$ and $\mu \colon F(\sA \tensor \sB) \cong F(\sA) \boxtimes F(\sB)$;
  such that the associators satisfy the \emph{pentagon equation} 
  \[\begin{aligned}
      \alpha^{\catD}_{FA,FB,FC} \comp (\im_{FA} \tensor \mu_{B,C}) \comp \mu_{A,B \tensor C}
      =  (\mu_{A,B} \tensor \im_{FC}) \comp \mu_{A \tensor B,C} \comp F(\alpha^{\catC}_{A,B,C}),\end{aligned}\]
  and the unitors satisfy $(\varepsilon \tensor \im_{FA}) \comp \mu_{I,A} \comp F(\lambda^{\catC}_{A}) = \lambda^{\catD}_{FA}$
  and $(\im_{FA} \tensor \varepsilon) \comp \mu_{A,I} \comp F(\rho^{\catC}_{FA}) = \rho^{\catD}_{FA}$.
  
  A monoidal functor is a \emph{monoidal equivalence} if it is moreover an equivalence of categories.  Two monoidal categories are monoidally equivalent if there exists a monoidal equivalence between them.
\end{defi}

During most of the paper, we omit all associators and unitors from monoidal categories, implicitly using the \emph{coherence theorem} for monoidal categories (\Cref{remark:usingcoherence}).

\begin{thm}[Coherence theorem, \cite{maclane78}]%
  \label{theorem:coherence}
  Every monoidal category is monoidally equivalent to a strict monoidal category.
\end{thm}

\begin{rem}
  \label{remark:usingcoherence}
Let us comment further on how we use the coherence theorem. Each time we have a
morphism $f \colon A \to B$ in a monoidal category, we have a corresponding
morphism $A \to B$ in its strictification. This morphism can be lifted to the
original category to uniquely produce, say, a morphism $(\lambda_{A} \comp f
\comp \lambda_{B}^{-1}) \colon I \otimes A \to I \otimes B$. Each time the
source and the target are clearly determined, we simply write $f$ again for this
new morphism.
\end{rem}

Instead of dealing with arbitrary monoidal categories, we mostly work with these that represent theories of processes: \emph{symmetric monoidal categories}, those where resources can be rerouted freely thanks to an self-invertible \emph{symmetry} morphism $\sigma_{A,B} \colon A \otimes B \to B \otimes A$.

\begin{defi}[Symmetric monoidal category, \cite{maclane78}]
  A \defining{linksymmetricmonoidalcategory}{\emph{symmetric monoidal category}}
  $(\catC, \otimes, I, \alpha, \lambda, \rho, \sigma)$ is a monoidal category
  $(\catC, \otimes, I, \alpha, \lambda, \rho)$ equipped with a symmetry
  $\sigma_{A,B} \colon A \otimes B \to B \otimes A$, which satisfies the hexagon
  equation
  \[\alpha_{A,B,C} \comp \sigma_{A,B \tensor C} \comp \alpha_{B,C,A} = (\sigma_{A,B} \tensor \im) \comp \alpha_{B,A,C} \comp (\im \tensor \sigma_{A,C})\]
  and additionally satisifes $\sigma_{A,B} \comp \sigma_{B,A} = \im$.
\end{defi}

\begin{rem}[Notation]
  String diagrams will be drawn top to bottom. We write $\tid{a}$ for the
  morphism $a$ tensored with some identities when these can be deduced from the
  context; we write $\sigma$ for any symmetry, and only omit it when it can be
  deduced from the context. For instance, let $f \colon A \to B$, let $h \colon
  B \to D$ and let $g \colon C \tensor D \to E$. We write $\tid{(f \comp h)}
  \comp \sigma \comp g \colon A \tensor C \to E$ for the morphism $(f \tensor
  \im) \comp \sigma \comp (\im \tensor h) \comp g$, which could have been also
  written as $(f \tensor \im) \comp (h \tensor \im) \comp \sigma \comp g$.
\end{rem}

\begin{defiC}[\cite{maclane78}]
  A \defining{linksymmetricmonoidalfunctor}{symmetric monoidal functor}
  between two \hyperlink{linksymmetricmonoidalcategory}{symmetric monoidal categories} $(\catC, \sigma^{\catC})$
  and $(\catD, \sigma^{\catD})$ is a monoidal functor $F \colon \catC \to \catD$ such that $\sigma^{\catD} \comp \mu = \mu \comp F(\sigma^{\catC})$.
\end{defiC}

\begin{defi}
  A \defining{linkcartesianmonoidalcategory}{\emph{cartesian monoidal category}} is
  a monoidal category whose tensor is the categorical product and whose unit is
  a terminal object.
\end{defi}
\subsection{Size concerns, limits and colimits}

We call $\Set$ the category of sets and functions below a certain Grothendieck universe~\cite{maclane06:universe}.
We do take colimits (and coends) over this category without creating size issues: we can be sure of their existence in our metatheoretic category of sets.

\begin{prop}[from~{\cite[Theorem 13]{adamek19}}]
  Terminal coalgebras exist in $\Set_{\leq \lambda}$.
  More precisely, the category of sets below a certain regular uncountable cardinal $\lambda$ is algebraically complete and cocomplete; meaning that every $\Set$-endofunctor has a terminal coalgebra and an initial algebra.
\end{prop}

A connected category is an inhabited category whose underlying graph is connected.
A connected limit is the limit of a diagram \(\fun{D} \colon \cat{J} \to \cat{C}\) whose domain category \(\cat{J}\) is connected.
The definition of intensional sequences (\Cref{sec:intensional-sequences}) relies on connected limits, and their characterisation on the following folklore result.

\begin{thm}[Coproducts commute with connected limits]
  \defining{linkconnectedlimits}
  Let $I$ be a set, understood as a discrete category, and let $\catA$ be a
  connected category with $F \colon I \times \catA \to \Set$ a functor. The
  canonical morphism
  \[ \sum_{i \in I} \lim_{a \in A} F(i,a) \to \lim_{a \in A}\sum_{i \in I} F(i,a)\]
  is an isomorphism.

  In particular, let $F_{n} \colon I \to \Set$ be a family of functors indexed
  by the natural numbers with a family of natural transformations
  $\alpha_{n} \colon F_{n+1} \to F_{n}$. The canonical morphism
  \[\sum_{i \in I} \lim_{n \in \naturals} F_{n}(i) \to \lim_{n \in \naturals}\sum_{i \in I} F_{n}(i) \]
  is an isomorphism.
\end{thm}
\begin{proof}
  Note that there are no morphisms between any two indices $i, j \in I$.
  Once some $i \in I$ is chosen in any factor of the connected limit, it
  forces any other factor to also choose $i \in I$. This makes the local
  choice of $i \in I$ be equivalent to the global choice of $i \in I$.
\end{proof}

\section{Intensional Monoidal Streams}
\label{sec:int-stateful-sequences}

This section introduces a preliminary definition of \emph{monoidal stream} in terms of a fixpoint equation (in \Cref{eq:intensionalstreamshort}).
In later sections, we refine both this definition and its characterization into the definitive notion of \emph{monoidal stream}.

We fix a small \symmetricMonoidalCategory{} $(\catC,\tensor,I)$.

\subsection{The fixpoint equation}\label{section:fixpoint}

Classically, type-variant streams have a neat coinductive definition~\cite{jacobs2005coalgebras,rutten00} that says:
\begin{quoting}\defining{linkstreamfun}{}
\emph{``A stream of type \(\stream{A} = \streamExpr{A}\) is an
  element of \(A_{0}\) together with a stream of type
  \(\tail{\stream{A}} = (A_{1}, A_{2}, \ldots)\)''}.
\end{quoting}
Formally, the set of streams is the final fixpoint of the equation
\[\streamFun(A_{0},A_{1},\ldots) \cong A_{0} \times \streamFun(A_{1},A_{2},\ldots);\]
and this fixpoint is computed to be \(\streamFun(\stream{A}) = \prod_{n \in \naturals}^{\infty} A_{n}\). In the same vein, we want to introduce not only streams but \emph{stream processes}
over a fixed theory of processes.
\begin{quoting}
\emph{``A stream process from
  \(\stream{X} = \streamExpr{X}\) to \(\stream{Y} = \streamExpr{Y}\)
  is a process from \(X_{0}\) to \(Y_{0}\) communicating along a channel $M$
  with a stream process from \(\tail{\stream{X}} = (X_{1}, X_{2}, \ldots)\) to
  \(\tail{\stream{Y}} = (Y_{1}, Y_{2}, \ldots)\).''}
\end{quoting}
Streams are recovered as stream processes on an empty input,
so we take this more general slogan as our preliminary definition of \emph{monoidal stream} (in~\Cref{def:intensionalmonoidalstream}).
Formally, they are the final fixpoint of the equation in \Cref{eq:intensionalstreamshort}.

\begin{defi}\label{def:intensionalmonoidalstream}\defining{linkintensionalmonoidalstream}{}
  The set of \emph{intensional monoidal streams} $\fun{T} \colon \NcatC^{op} \times \NcatC \to \Set$, depending on inputs and outputs, %
  is the final fixpoint of the equation in \Cref{eq:intensionalstreamshort}.
\end{defi}

\begin{figure}[!h]
  \centering \defining{linkistream}{}
  $\displaystyle \fun{T}(\stream{X},\stream{Y}) \cong \sum_{M \in \catC}
    \hom{}(X_{0}, M \tensor Y_{0}) \times \fun{T}(\act{M}{\tail{\stream{X}}}, \tail{\stream{Y}}).$
  \caption{Fixpoint equation for intensional streams.}
\label{eq:intensionalstreamshort}
\end{figure}
\begin{rem}[Notation]
  Let  $\stream{X} \in \NcatC$ be a sequence of objects $\streamExpr{X}$.
  We write $\defining{linktail}{\tail{\stream{X}}}$ for its \emph{tail} $(X_{1},X_{2},\dots)$.
  Given $M \in \catC$, we write $\act{M}{\stream{X}}$ for the sequence $(M \tensor X_{0},X_{1},X_{2},\dots)$;
  As a consequence, we write $\act{M}{\tail{\stream{X}}}$ for $(M \otimes X_{1},X_{2},X_{3},\dots)$.

\end{rem}

\begin{rem}[Initial fixpoint]
  There exists an obvious fixpoint for the equation in
  \Cref{eq:intensionalstreamshort}: the constant empty set.
  This solution is the \emph{initial} fixpoint, a minimal solution.
  The \emph{final} fixpoint will be realized by the set of \emph{intensional sequences}.
\end{rem}

\subsection{Intensional sequences}\label{sec:intensional-sequences}
We now construct the set of \intensionalStreams{} explicitly (\Cref{th:intensionalcoalgebra}).
For this purpose, we generalize the ``stateful morphism sequences'' of Sprunger and Katsumata~\cite{katsumata19} from cartesian to arbitrary \symmetricMonoidalCategories{} (\Cref{def:intensionalsequence}). We derive a novel characterization of these ``sequences'' as the desired final fixpoint (\Cref{th:intensionalcoalgebra}).

In the work of Sprunger and Katsumata, a stateful sequence is a sequence of morphisms $f_{n} \colon M_{n-1} \times X_{n} \to M_{n} \times Y_{n}$ in a \cartesianMonoidalCategory{}.
These morphisms represent a process at each point in time $n=0,1,2,\dots$.
At each step \(n\), the process takes an input \(X_{n}\) and,
together with the stored memory \(M_{n-1}\), produces some output \(Y_{n}\) and
writes to a new memory \(M_{n}\). The memory is initially empty, with
\(M_{-1} \defn \mathbf{1}\) being the final object by convention.
We extend this definition to any \symmetricMonoidalCategory{}.
\begin{defi}\label{def:intensionalsequence}\defining{linkintensionalstatefulsequence}{}
  Let $\stream{X}$ and $\stream{Y}$ be two sequences of objects
  representing inputs and outputs, respectively. An
  \defining{linkstreamtransducer}{\emph{intensional sequence}} is a
  sequence of objects $\streamExpr{M}$ together with a sequence of morphisms
  \[\intseq{f_{n} \colon M_{n-1} \tensor X_{n} \to M_{n} \tensor Y_{n}},\]
  where, by convention, $M_{-1} \defn I$ is the unit of the monoidal category.
  In other words, the set of intensional sequences is
  \[\iSeq(\stream{X},\stream{Y}) \coloneqq
    \sum_{M \in [\naturals,\catC]} \prod_{n = 0}^{\infty} \hom{}(M_{n-1} \tensor X_{n} , M_{n} \tensor Y_{n}).\]
\end{defi}

We now prove that \intensionalSequences{} are the final fixpoint of the equation in \Cref{eq:intensionalstreamshort}.
The following~\Cref{th:intensionalcoalgebra} serves two purposes:
it gives an explicit final solution to this fixpoint equation and
it gives a novel universal property to \intensionalSequences{}.

\begin{thm}\label{th:intensionalcoalgebra}\label{corollary:intensionalstreams}
\IntensionalSequences{} are the explicit construction of \intensionalStreams{}, $\iStream{} \cong \iSeq{}$.
  In other words, they are a fixpoint of the equation in \Cref{eq:intensionalstreamshort},
and they are the final such one.
\end{thm}
\begin{proof}
  Let us first see that our candidate, Expression \ref{eq:limit-goal}, is indeed a fixpoint of Equation~\ref{eq:intensionalstreamshort}.
  \begin{equation}\label{eq:limit-goal}
    \sum_{M_n \in \NcatC}
    \prod_{i=0}^{\infty} \hom{}(M_{i-1} \tensor X_i; M_i \tensor Y_i).
  \end{equation}
  For this, we note that cartesian products distribute over coproducts: the following Expression~\ref{eq:fixpoint-expr-intensional} is isomorphic to Expression~\ref{eq:limit-goal}. This proves it is a fixpoint.
  \begin{equation}\label{eq:fixpoint-expr-intensional}
    \sum_{M_{0}} \hom{}(X_{0},M_{0} \tensor Y_{0}) \times \sum_{M \in \NcatC} \prod_{n = 1}^{\infty} \hom{}(M_{n-1} \tensor X_{n}, M_{n} \tensor Y_{n}).
  \end{equation}
  If Expression~\ref{eq:limit-goal} were to coincide with Expression \ref{eq:limit-intensional}, our desired result would follow by Adamek's theorem (\Cref{th:adamek}).
  Adamek's theorem    states that, if the following limit exists and is a fixpoint, then it is indeed the final fixpoint. We already know that it must exist: categories of functors over sets, such as $[\NcatC^{op}\times\NcatC,\Set]$, have all limits.
  \begin{equation}\label{eq:limit-intensional}
    \lim_{n \in \naturals} \sum_{M_{0},\dots,M_{n}} \prod^{n}_{t=0} \hom{}(M_{t-1} \tensor X_{t}, M_{t} \tensor Y_{t})
  \end{equation}
  Therefore, it suffices to prove that Expression~\ref{eq:limit-goal} and Expression~\ref{eq:limit-intensional} are isomorphic.
  
  Let us prove this isomorphism. For convenience, we call the main factor of this limit $\Omega_i(M_i,M_{i-1}) = \hom{}(M_{i-1} \tensor X_i; M_i \tensor Y_i)$. Let us first note that the limit in \Cref{eq:limit-goal} can be simplified, using that (\emph{i}) limits commute with limits and that (\emph{ii}) \hyperlink{linkconnectedlimits}{connected limits commute with coproducts}.
  \[
    \sum_{M_n \in \NcatC}
    \prod_{i=0}^{\infty} \Omega_i(M_i,M_{i-1})
    \overset{\emph{(i)}}{\cong}
    \sum_{M_n \in \NcatC}
    \lim_{k \in \mathbb{N}} \prod_{i=0}^{k} \Omega_i(M_i,M_{i-1})
    \overset{\emph{(ii)}}{\cong}
    \lim_{k \in \mathbb{N}} 
    \sum_{M_n \in \NcatC}
    \prod_{i=0}^{k} \Omega_i(M_i,M_{i-1}).
  \]
  Again for convenience, we now call $\Psi_k(M_0,\dots,M_k) = \prod_{i=0}^{k} \Omega_i(M_{i-1},M_i)$. We will use that the coproduct over mute variables can be read as a product, $\sum_{a \in A} B \cong A \times B$, to simplify each one of the factors of the limit. For each $k \in \mathbb{N}$, we have that
  \[\sum_{M_n \in \NcatC} \Psi_k(M_0,\dots,M_k)
    \cong
    \left( \prod_{j = k_n}^{\infty} \mathbb{C}_{\mathrm{obj}}  \right) \times \sum_{M_0,...,M_k} \Psi_k(M_0,\dots,M_k).
  \]
  As a result, we can compute the following limit.
  \[\lim_k \sum_{M_n \in \NcatC} \Psi_k(M_0,\dots,M_k)
    \cong
    \lim_k 
    \left( \prod_{j = k_n}^{\infty} \mathbb{C}_{\mathrm{obj}} \right) \times \sum_{M_0,...,M_k} \Psi_k(M_0,\dots,M_k).
  \]
  Under this isomorphism, the limit diagram gets transformed into a more tractable diagram that we will be able to compute. The original limit diagram consisted only of projections, 
  $$\pi_{k} \colon \sum_{M_n \in \NcatC} \Psi_{k+1}(M_0,\dots,M_{k+1}) \to \sum_{M_n \in \NcatC} \Psi_k(M_0,\dots,M_k).$$
  The diagram under the isomorphism contains objects $\mathbb{C}_{\textrm{obj}}^{(k,j)}$ for each $j > k$, and objects $\sum_{M_n \in \NcatC} \Psi_k(M_0,\dots,M_k)$ for each $k \in \mathbb{N}$; let us use $\Sigma\Psi_n$ as a shorthand for this term. The diagram contains three families of arrows: the projections between sums,
  $\alpha_k \colon \Sigma\Psi_{k+1} \to \Sigma\Psi_k$
  the projections out of sums, $\beta_k \colon \Sigma\Psi_{k+1} \to \mathbb{C}^{(k,k+1)},$
  and projections between objet choices, $\gamma_{k,j} \colon \mathbb{C}^{(k+1,j)} \to \mathbb{C}^{(k,j)}$. See the right hand side of the following \Cref{limit:simplified}.

  \begin{figure}[ht]
  \[
    \lim_{k,j}\left(
      \begin{tikzcd}[column sep=1ex,row sep=1ex]
          & \vdots \dlar & \vdots\dlar & \vdots\dlar & \vdots \dar \dlar \\
        \cdots & 
        \mathbb{C}^{(3,6)}\dlar & \mathbb{C}^{(3,5)}\dlar & \mathbb{C}^{(3,4)}\dlar & \Sigma\Psi_3 \dlar \dar \\
        \cdots & 
        \mathbb{C}^{(2,5)}\dlar & \mathbb{C}^{(2,4)}\dlar & \mathbb{C}^{(2,3)}\dlar & \Sigma\Psi_2 \dlar \dar \\
        \cdots & 
        \mathbb{C}^{(1,4)}\dlar & \mathbb{C}^{(1,3)}\dlar & \mathbb{C}^{(1,2)}\dlar & \Sigma\Psi_1 \dlar \dar \\
        \cdots & 
        \mathbb{C}^{(0,3)} & \mathbb{C}^{(0,2)} & \mathbb{C}^{(0,1)} & \Sigma\Psi_0 \\
      \end{tikzcd}
    \right) \cong
    \lim_{k,j}\left(
      \begin{tikzcd}[column sep=1ex,row sep=1ex]
        \vdots \dar \\
        \Sigma\Psi_3 \dar \\
        \Sigma\Psi_2 \dar \\
        \Sigma\Psi_1 \dar \\
        \Sigma\Psi_0 \\
      \end{tikzcd}
    \right)
  \]
  \caption{Simplification of a limit with an initial subdiagram.}
  \label{limit:simplified}
  \end{figure}

  Finally, we can use the fact that the limit of a diagram coincides with the limit of any of its dominating chains \cite[Theorem IX.3.1]{maclane78}. In other words, the chain of projections on the left hand side of \Cref{limit:simplified} is initial, and we can conclude that
  $$
  \lim_k \sum_{M_n \in \NcatC} \Psi_k(M_0,\dots,M_k)\ \cong\ 
  \lim_k \sum_{M_0,...,M_k} \Psi_k(M_0,\dots,M_k).
  $$
  Putting all these isomorphisms together, we have proven that the two limits in \Cref{eq:limit-intensional} and \Cref{eq:limit-goal} are isomorphic. This concludes the proof.
\end{proof}
\section{Dinaturality and Coends}
\subsection{Interlude: Dinaturality}
\label{sec:dinaturality}

During the rest of this text, we deal with two different definitions of what it means for two processes to be equal: \emph{dinatural} and \emph{observational equivalence}, apart from pure \emph{intensional equality}.
Fortunately, when working with functors of the form $P \colon \catC^{op} \times \catC \to \Set$, the so-called \emph{endoprofunctors}, we already have a canonical notion of equivalence.

Endoprofunctors $P \colon \catC^{op} \times \catC \to \Set$ can be thought as indexing families
of processes $P(M,N)$ by the types of an input channel $M$ and an output channel $N$. A process
$p \in P(M,N)$ writes to a channel of type $N$ and then reads from a channel of type $M$.

Now, assume we also have a transformation $r \colon N \to M$ translating from output to input types.
Then, we can \emph{plug the output to the input}:
the process $p$ writes with type $N$, then $r$ translates from $N$ to $M$, and then $p$ uses this same output as its input $M$.
This composite process can be given two sligthly different descriptions; the process could
\begin{itemize}[label=$\triangleright$]
  \item translate \emph{after writing}, $P(M,r)(p) \in P(M,M)$, or
  \item translate \emph{before reading}, $P(r,N)(p) \in P(N,N)$.
\end{itemize}
These two processes have different types.
However, with the output plugged to the input, it does not really matter when to apply the translation.
These two descriptions represent the same process: they are \emph{dinaturally equivalent}.

\begin{defi}[Dinatural equivalence] \label{def:dinaturality}\defining{linkdinaturality}
  For any functor $P \colon \catC^{op} \times \catC \to \Set$, consider the set
  \[S_{P} \defn \sum_{M \in \catC} P(M,M).\]
  \emph{Dinatural equivalence}, $(\sim)$, on the set $S_{P}$
  is the smallest equivalence relation satisfying
  $P(M,r)(p) \sim P(r,N)(p)$ for each $p \in P(M,N)$ and each $r \in \hom{}(N,M)$.
\end{defi}

Coproducts quotiented by \dinaturalEquivalence{} construct a particular form of colimit called a \emph{coend}.
Under the process interpretation of profunctors, taking a coend means \emph{plugging an output to an input}
of the same type.

\subsection{Coend calculus}

\emph{Coend calculus} is the name given to the a branch of category theory that
describes the behaviour of certain colimits called \emph{coends}.
MacLane \cite{maclane78} and Loregian \cite{loregian2021} give introductory presentations of the coend calculus.

\begin{defi}[Coend, \cite{maclane78,loregian2021}]
  Let $P \colon \catC^{op} \times \catC \to \Set$ be a functor.
  Its \emph{coend} is the coproduct of $P(M,M)$ indexed by $M \in \catC$, quotiented by \dinaturalEquivalence{}.
  \[\coend{M\in\catC}P(M,M) \defn \left(\sum_{M \in \catC} P(M,M) \bigg/ \sim \right).\]
  That is, the coend is the colimit of the diagram with a \emph{cospan} $P(M,M) \gets P(M,N) \to P(N,N)$
  for each $f \colon N \to M$.
\end{defi}

\begin{prop}[Yoneda reduction]
  \defining{linkcoyoneda}{} Let $\catC$ be any category and
  let $F \colon \catC \to \Set$ be a functor; the following isomorphism holds
  for any given object $A \in \catC$.
\label{prop:yonedareduction}
\[\coend{X \in \catC} \idProf(X,A) \times FX \cong FA. \]
Following the analogy with classical analysis, the $\idProf$ profunctor works as
a Dirac delta.
\end{prop}

\begin{prop}[Fubini rule]
  \defining{linkfubini}{} Coends commute between them; that is,
  there exists a natural isomorphism
  \label{prop:fubinirule}
  \[\begin{aligned}
    & \coend{X_{1} \in \catC} \coend{X_{2} \in \catC} P(X_{1},X_{2},X_{1},X_{2})
    \cong & 
    & \coend{X_{2} \in \catC} \coend{X_{1} \in \catC}  P(X_{1},X_{2},X_{1},X_{2}).
  \end{aligned}\]
  In fact, they are both isomorphic to the coend over the product category,
  \[ \coend{(X_{1}, X_{2}) \in \catC \times \catC} P(X_{1},X_{2},X_{1},X_{2}). \]
  Following the analogy with classical analysis, coends follow the Fubini rule
  for integrals.
\end{prop}

\subsection{Towards dinatural memory channels}
\label{sec:towardsextensional}
Let us go back to intensional monoidal streams.
Consider a family of processes $f_{n} \colon M_{n-1} \tensor X_{n} \to Y_{n} \tensor N_{n}$ reading from memories of type $M_{n}$ but writing to memories of type $N_{n}$.
Assume we also have processes $r_{n} \colon N_{n} \to M_{n}$ translating from output to input memory.
Then, we can consider the process that does \(f_{n}\), translates from memory \(N_{n}\) to memory \(M_{n}\) and then does \(f_{n+1}\).
This process is described by two different \intensionalSequences{},
  \begin{itemize}
    \item $\intseq{f_{n} ; (r_{n} \otimes \im) \colon M_{n-1} 
      \tensor X_{n} \to M_{n} \tensor Y_{n}}$, and
    \item $\intseq{(r_{n-1} \otimes \im) ; f_{n} \colon N_{n-1} 
      \tensor X_{n} \to N_{n} \tensor Y_{n}}$.
  \end{itemize}
These two \intensionalSequences{} have different types for the memory channels.
However, in some sense, they represent \emph{the same process description}.
If we do not care about what exactly it is that we save to memory, we should consider two such processes to be equal (as in \Cref{diagram:dataflowwalk2}, where ``the same process'' can keep two different values in memory).
Indeed, dinaturality in the memory channels $M_{n}$ is the smallest equivalence relation $(\sim)$ satisfying
\[(f_{n} ; (r_{n} \otimes \im))_{n \in \naturals} \sim ((r_{n-1} \otimes \im); f_{n})_{n \in \naturals}.\]
This is precisely the quotienting that we
perform in order to define \emph{dinatural sequences}.

\begin{defi}\label{def:extensionalequality}\defining{linkextensionalequality}{}\label{def:extensional-sequence}
\emph{Dinatural equivalence} of \intensionalSequences{}, $(\sim)$, is dinatural equivalence in the memory channels $M_{n}$.
A \defining{linkextensionalstatefulsequence}{\emph{dinatural sequence}} from \(\stream{X}\) to \(\stream{Y}\) is an equivalence class
  \[\extseq{f_{n} \colon M_{n-1} \tensor X \to M_{n} \tensor Y}\]
of \intensionalSequences{} under dinatural equivalence.

In other words, the set of dinatural sequences is the set of intensional sequences substituting the coproduct by a coend,
\[\eSeq(\stream{X},\stream{Y}) \defn \coend{M \in [\naturals,\catC]} \prod^{\infty}_{i=0} \hom{}(X_{i} \tensor M_{i-1}, Y_{i} \tensor M_{i}).\]
\end{defi}

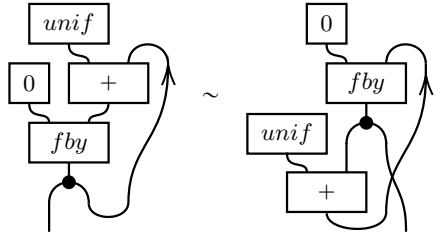
\begin{figure}[h!]
\tikzset{every picture/.style={line width=0.85pt}} %
\begin{tikzpicture}[x=0.75pt,y=0.75pt,yscale=-1,xscale=1]
\draw   (190,110) -- (230,110) -- (230,130) -- (190,130) -- cycle ;
\draw   (170,140) -- (210,140) -- (210,160) -- (170,160) -- cycle ;
\draw   (160,110) -- (180,110) -- (180,130) -- (160,130) -- cycle ;
\draw    (190,160) -- (190,170) ;
\draw    (220,110) .. controls (220.2,97) and (239.4,97.4) .. (240,110) ;
\draw   (170,80) -- (210,80) -- (210,100) -- (170,100) -- cycle ;
\draw    (170,130) .. controls (169.57,139.23) and (180.14,133.51) .. (180,140) ;
\draw    (210,130) .. controls (209.57,139.23) and (200.14,133.51) .. (200,140) ;
\draw    (180,180) -- (180,195) ;
\draw    (180,180.65) .. controls (179.8,166.45) and (200.2,166.45) .. (200,180.65) ;
\draw  [fill={rgb, 255:red, 0; green, 0; blue, 0 }  ,fill opacity=1 ] (187.1,170) .. controls (187.1,168.4) and (188.39,167.1) .. (190,167.1) .. controls (191.6,167.1) and (192.9,168.4) .. (192.9,170) .. controls (192.9,171.6) and (191.6,172.9) .. (190,172.9) .. controls (188.39,172.9) and (187.1,171.6) .. (187.1,170) -- cycle ;
\draw    (200,180.65) .. controls (200,189.85) and (220,190) .. (220,180) ;
\draw    (190,100) .. controls (189.57,109.23) and (200.14,103.51) .. (200,110) ;
\draw    (220,180) .. controls (220.2,150.25) and (239.99,148.64) .. (240.01,111.71) ;
\draw [shift={(240,110)}, rotate = 89.11] [color={rgb, 255:red, 0; green, 0; blue, 0 }  ][line width=0.75]    (10.93,-3.29) .. controls (6.95,-1.4) and (3.31,-0.3) .. (0,0) .. controls (3.31,0.3) and (6.95,1.4) .. (10.93,3.29)   ;
\draw   (299.99,165) -- (339.99,165) -- (339.99,185) -- (299.99,185) -- cycle ;
\draw   (319.97,110) -- (359.97,110) -- (359.97,130) -- (319.97,130) -- cycle ;
\draw   (309.97,80) -- (329.97,80) -- (329.97,100) -- (309.97,100) -- cycle ;
\draw   (279.99,135.56) -- (319.99,135.56) -- (319.99,155.56) -- (279.99,155.56) -- cycle ;
\draw    (319.97,100) .. controls (319.55,109.23) and (330.12,103.51) .. (329.97,110) ;
\draw    (329.97,150) -- (329.99,165) ;
\draw    (319.99,185.65) .. controls (319.99,195.63) and (350.32,196.97) .. (349.99,185) ;
\draw    (349.99,185) .. controls (350.18,155.25) and (369.98,153.64) .. (370,116.71) ;
\draw [shift={(369.99,115)}, rotate = 89.11] [color={rgb, 255:red, 0; green, 0; blue, 0 }  ][line width=0.75]    (10.93,-3.29) .. controls (6.95,-1.4) and (3.31,-0.3) .. (0,0) .. controls (3.31,0.3) and (6.95,1.4) .. (10.93,3.29)   ;
\draw    (369.99,110) -- (369.99,115) ;
\draw    (339.98,130) -- (339.98,140) ;
\draw    (329.99,150.65) .. controls (329.79,136.45) and (350.19,136.45) .. (349.99,150.65) ;
\draw  [fill={rgb, 255:red, 0; green, 0; blue, 0 }  ,fill opacity=1 ] (337.08,140) .. controls (337.08,138.4) and (338.38,137.1) .. (339.98,137.1) .. controls (341.58,137.1) and (342.88,138.4) .. (342.88,140) .. controls (342.88,141.6) and (341.58,142.9) .. (339.98,142.9) .. controls (338.38,142.9) and (337.08,141.6) .. (337.08,140) -- cycle ;
\draw    (299.99,155) .. controls (299.56,164.23) and (310.13,158.51) .. (309.99,165) ;
\draw    (349.99,110) .. controls (350.19,97) and (369.39,97.4) .. (369.99,110) ;
\draw    (349.99,150) .. controls (349.99,171.97) and (360.32,171.63) .. (359.99,195) ;
\draw (210,120) node    {$+$};
\draw (170,120) node    {$0$};
\draw (190,150) node    {$fby$};
\draw (190,90) node    {$unif$};
\draw (319.98,175) node    {$+$};
\draw (319.97,90) node    {$0$};
\draw (339.97,120) node    {$fby$};
\draw (299.99,145.56) node    {$unif$};
\draw (251,123.4) node [anchor=north west][inner sep=0.75pt]    {$\ \sim $};
\end{tikzpicture}
\caption{Dinaturally equivalent walks keeping different memories: the \emph{current} position vs. the \emph{next} position.}
\label{diagram:dataflowwalk2}
\end{figure}

 \newpage
\section{Dinatural Monoidal Streams}
\label{section:extensional}

In this section, we introduce \emph{dinatural monoidal streams} in terms of a universal property: \emph{dinatural streams are the morphisms of the free delayed-feedback monoidal category} (\Cref{th:extensionalfreefeedback}).

Feedback monoidal categories come from the work of Katis, Sabadini and Walters \cite{katis02}.
They are instrumental to our goal of describing and composing signal flow graphs: they axiomatize a graphical calculus that extends the well-known string diagrams for monoidal categories with \emph{feedback loops}~\cite{katis02,feedbackspans2020}.
Constructing the \emph{free feedback monoidal category}~(\Cref{def:stconstruction}) will lead to the main result of this section: dinatural sequences are the explicit construction of dinatural streams (\Cref{th:extensionalfreefeedback}).

We finish the section by exploring how dinatural equivalence may not be enough to capture true observational equality of processes (\Cref{example:observe}).

\subsection{Feedback monoidal categories}

\emph{Feedback monoidal categories} are \symmetricMonoidalCategories{} with a ``feedback'' operation that connects outputs to inputs.
They have a natural axiomatisation (\Cref{def:feedback}) that has been rediscovered independently multiple times, with only slight variations~\cite{bloom93,katis02,katis99,bonchi19}.
Feedback is weaker than \emph{trace} while still satisfying a normalisation property that allows a straightforward construction of the free \feedbackMonoidalCategory{} (\Cref{th:free-feedback}, cf.~\cite{feedbackspans2020}). Traced monoidal categories are better known---yet more restrictive: it is the failure of their yanking axiom what will allow stateful morphisms.
We present a novel definition that generalizes the previous ones by allowing the feedback operator to be guarded by a \monoidalEndofunctor{}.

\begin{defi}\label{def:feedback}\defining{linkcategorywithfeedback}{}
  A \emph{feedback monoidal category} is a symmetric monoidal category $(\catC,\otimes,I)$ endowed with a symmetric strong monoidal endofunctor $\fun{F} \colon \catC \to \catC$ and an operation
  \[\defining{linkfbkop}{\ensuremath{\mathbf{fbk}}}_{S} \colon
     \idProf(\fun{F}(S) \tensor X, S \tensor Y) \to \idProf(X,Y)\]
  for all \(S\), \(X\) and \(Y\) objects of \(\catC\);
  this operation needs to satisfy the following axioms.
  \begin{enumerate}[ label={(A\arabic*).}, ref={\textbf{(A\arabic*)}}, start=1 ]
    \item\label{axiom:tight} Tightening: \(u \dcomp \fbk[S](f) \dcomp v = \fbk[S]((\id{\fun{F}S} \tensor u) \dcomp f \dcomp(\id{S} \tensor v))\).
    \item\label{axiom:vanish} Vanishing: \(\fbk[\monoidalunit](f) = f\).
    \item\label{axiom:join} Joining: \(\fbk[T](\fbk[S](f)) = \fbk[S \tensor T](f)\)
    \item\label{axiom:strength} Strength: \(\fbk[S](f) \tensor g = \fbk[S](f \tensor g)\).
    \item\label{axiom:slide} Sliding: \(\fbk[S]((\fun{F}h \tensor \id{X}) \dcomp f) = \fbk[T](f \dcomp (h \tensor \id{Y}))\).
  \end{enumerate}
\end{defi}

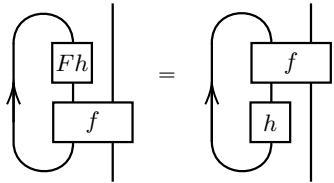
\begin{figure}[H]
\tikzset{every picture/.style={line width=0.85pt}} %
\begin{tikzpicture}[x=0.75pt,y=0.75pt,yscale=-1,xscale=1]
\draw    (40,60) -- (40,70) ;
\draw   (30,70) -- (70,70) -- (70,90) -- (30,90) -- cycle ;
\draw    (10,40.35) .. controls (10,20.2) and (40,20.2) .. (40,40.35) ;
\draw   (29.5,40) -- (49.5,40) -- (49.5,60) -- (29.5,60) -- cycle ;
\draw    (10,90) .. controls (10,109) and (40.4,110.2) .. (40,90) ;
\draw    (10,62) -- (10,90) ;
\draw [shift={(10,60)}, rotate = 90] [color={rgb, 255:red, 0; green, 0; blue, 0 }  ][line width=0.75]    (10.93,-3.29) .. controls (6.95,-1.4) and (3.31,-0.3) .. (0,0) .. controls (3.31,0.3) and (6.95,1.4) .. (10.93,3.29)   ;
\draw    (60,20) -- (60,70) ;
\draw    (60,90) -- (60,110) ;
\draw    (140,60) -- (140,70) ;
\draw   (130,40) -- (170,40) -- (170,60) -- (130,60) -- cycle ;
\draw    (110,40.35) .. controls (109.6,19.8) and (140.4,19.8) .. (140,40.35) ;
\draw   (129.5,70) -- (149.5,70) -- (149.5,90) -- (129.5,90) -- cycle ;
\draw    (110,90) .. controls (110.4,109.4) and (140.4,109.8) .. (140,90) ;
\draw    (110,40) -- (110,90) ;
\draw    (160,60) -- (160,110) ;
\draw    (160,20) -- (160,40) ;
\draw    (10,40) -- (10,80) ;
\draw    (110,62) -- (110,90) ;
\draw [shift={(110,60)}, rotate = 90] [color={rgb, 255:red, 0; green, 0; blue, 0 }  ][line width=0.75]    (10.93,-3.29) .. controls (6.95,-1.4) and (3.31,-0.3) .. (0,0) .. controls (3.31,0.3) and (6.95,1.4) .. (10.93,3.29)   ;
\draw (50,80) node    {$f$};
\draw (39.5,50) node    {$Fh$};
\draw (81,53.4) node [anchor=north west][inner sep=0.75pt]    {$=$};
\draw (150,50) node    {$f$};
\draw (139.5,80) node    {$h$};
\end{tikzpicture}
\caption{The sliding axiom (A5).}\label{diagram:sliding}
\end{figure}

\begin{defi}\label{def:feedbackfunctor}
  A \defining{linkfeedbackfunctor}{\emph{feedback functor}}
  between two \feedbackMonoidalCategories{}
  $(\catC,\fun{F}^{\catC},\fbk^{\catC})$ and $(\catD,\fun{F}^{\catD},\fbk^{\catD})$ is a symmetric \monoidalFunctor{} $G \colon \catC \to \catD$ such that $G \comp \fun{F}^{\catD} = \fun{F}^{\catC} \comp G$ and
  \[
    G(\fbk[S]^\catC(\fm)) = \fbk[GS]^{\catD}(\mu_{\fun{F}S,A} \comp Gf \comp \mu^{-1}_{S,B})
  \]
  for each $\fm \colon \fun{F}S \tensor X \to S \tensor Y$, where \(\mu_{A,B} \colon G(\sA) \tensor G(\sB) \to G(\sA \tensor \sB )\) is the structure morphism of the monoidal functor \(G\).
\end{defi}

A traced monoidal category is, in fact, a feedback monoidal category with an extra axiom, the \emph{yanking axiom}, that makes the feedback loop ``instantaneous'' (\Cref{fig:yanking-wait}, left).
This requirement imposes that the internal state does not change between iterations and the process has already reached equilibrium: the yanking axiom implies memoryless processes.
These conditions are too strong for our purposes of modelling dataflow computations, so we choose feedback monoidal categories instead.
For a thorough discussion about the conceptual difference between trace and feedback, and the necessity to drop the yanking axiom for modelling state, see~\cite[Sections~3.2 and~3.3]{2023canonicalalgebra}.

\begin{rem}[Wait or trace]\label{remark:wait}\defining{linkmorphwait}
  In a \feedbackMonoidalCategory{} $(\catC,\fbk)$,
  we construct the morphism $\mathsf{wait}_{X} \colon X \to FX$ as a feedback loop over the symmetry,
  $\mathsf{wait}_{X} = \fbk(\sigma_{F(X),X})$ (\Cref{fig:yanking-wait}, right).
  A \emph{traced monoidal category}~\cite{joyal96} is a \feedbackMonoidalCategory{} guarded by the identity functor
  such that $\mathsf{wait}_{X} = \im_{X}$.
  \begin{figure}[h!]
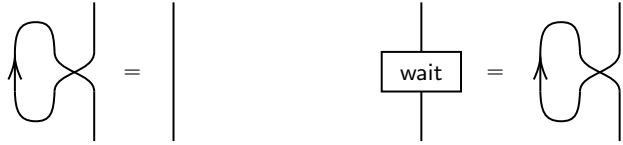

    \[\waitTraceFig\]
    \caption{The yanking axiom of traced monoidal categories (left) and the morphism $\mathsf{wait}$ in feedback monoidal categories (right).\label{fig:yanking-wait}}
  \end{figure}

\end{rem}

The ``state construction'', \(\St(\bullet)\), realizes the \emph{free} \feedbackMonoidalCategory{}.
As it happens with feedback monoidal categories, this construction has appeared in the literature in slightly different forms. It has been used for describing a ``memoryful geometry of interaction''~\cite{hoshino14}, ``stateful resource calculi''~\cite{bonchi19}, and ``processes with feedback''~\cite{sabadini95,katis02}.
The idea in all of these cases is the same: we allow the morphisms of a monoidal category to depend on a ``state space'' $S$, possibly guarded by a functor.
Adding a state space is equivalent to freely adding feedback~\cite{katis02,feedbackspans2020}.

\begin{defi}
  A \emph{stateful morphism} is a pair $(S,f)$ consisting of a ``state space'' $S \in \catC$ and a morphism $f \colon \fun{F}S \tensor X \to S \tensor Y$.
  We say that two stateful morphisms are \emph{sliding equivalent} if they are related by the smallest equivalence relation satisfying
  $\reprCoend{(\fun{F}r \tensor \im) \dcomp h}{S} \sim \reprCoend{h \dcomp (r \tensor \im)}{T}$
  for each $h \colon X \tensor \fun{F}T \to S \tensor Y$ and each $r \colon S \to T$.

  In other words, sliding equivalence is \emph{dinaturality in $S$}.
\end{defi}

\begin{defi}[St($\bullet$) construction,~\cite{katis02,feedbackspans2020}]\label{def:stconstruction}
  We write \(\defining{linkSt}{\ensuremath{\mathsf{St}}}_{\fun{F}}(\catC)\) for the \symmetricMonoidalCategory{} that has the same objects as \(\catC\) and whose morphisms from $X$ to $Y$ are stateful morphisms $f \colon \fun{F}S \tensor X \to S \tensor Y$ \emph{up to sliding}.
  \[\hom{\St_{\fun{F}}(\catC)}(X,Y) \coloneqq \coend{S \in \catC} \idProf[\catC](\fun{F}S \tensor X, S \tensor Y).\]
\end{defi}

\begin{thm}[see~\cite{katis02}]\label{th:free-feedback}\label{ax:th:free-feedback}
  \(\St_{\fun{F}}(\catC)\) is the free feedback monoidal category over \((\catC,\fun{F})\).
\end{thm}
\begin{proof}[Proof sketch]
  Let $(\catD,\fun{F}^{\catD},\fbk^{\catD})$ be any other symmetric monoidal category with an endofunctor,
  and let $H \colon \catC \to \catD$ be such that $\fun{F} \comp H = H \comp \fun{F}^{\catD}$.
  We will prove that it can be extended uniquely to a feedback functor $\tilde{H} \colon \St_{\fun{F}}(\catC) \to \catD$.

  It can be proven that any expression involving feedback can be reduced applying the feedback axioms to an expression of the form $\fbk(f)$ for some $f \colon \fun{F}S \tensor X \to S \tensor Y$. After this, the definition of $\tilde{H}$ in this morphism is forced to be $\tilde{H}(\fbk{f}) = \fbk^{\catD}(\catD)$. This reduction is uniquely determined, up to sliding, and the morphisms of the $\St(\bullet)$ construction are precisely morphisms $f \colon \fun{F}S \tensor X \to S \tensor Y$ quotiented by sliding equivalence.  This is the core of the proof by Katis, Sabadini, and Walters \cite{katis02}.
\end{proof}

\subsection{Dinatural monoidal streams}

Monoidal streams should be, in some sense, the minimal way of adding feedback to a theory of processes.
The output of this feedback, however, should be delayed by one unit of time: the category $\NcatC$ is naturally equipped with a \emph{delay} endofunctor that shifts a sequence by one.
Dinatural monoidal streams form the free delayed-feedback category.

\begin{defi}[Delay functor]\label{def:delay-fun-cn}
  Let \(\defining{linkdelay}{\delay} \colon \sequencesCat{\catC} \to \sequencesCat{\catC}\) be the endofunctor defined on objects \(\stream{X} = \streamExpr{X}\), as \(\delay[\stream{X}] \defn \streamExprExp{\monoidalunit, X_0, X_1}\); and on morphisms \(\stream{f} = \streamExpr{f}\) as \(\delay[\stream{f}] \defn \streamExprExp{\id{\monoidalunit}, f_0, f_1}\).
\end{defi}

\begin{defi}\label{def:extensionalmonoidalstreams}
  \defining{linkestream}{}\defining{linkextensionalmonoidalstream}{}
  The set of \emph{dinatural monoidal streams} from \(\stream{X}\) to \(\stream{Y}\) is $\eStream(\stream{X},\stream{Y})$, the hom-set of the free feedback monoidal category over \((\NcatC, \delay)\).
\end{defi}

We characterize now dinatural streams in terms of dinatural sequences and the $\St(\bullet)$-construction.

\begin{thm}\label{th:ext-stateful-sequences}\label{th:extensionalfreefeedback}
  \ExtensionalSequences{} are the explicit construction of \extensionalStreams{}, $\eStream(\stream{X},\stream{Y}) \cong \St_{\delay}(\NcatC)(\stream{X},\stream{Y})$, naturally on $\stream{X}$ and $\stream{Y}$.
\end{thm}
\begin{proof}
  Dinatural sequences coincide, by \Cref{def:stconstruction}, with the state construction, \linebreak[5] $\eSeq(\stream{X},\stream{Y}) = \St_{\delay}(\NcatC)(\stream{X},\stream{Y})$. This is still different from our definition of dinatural monoidal streams as morphisms of the free feedback monoidal category; but it follows by \Cref{th:free-feedback} that they both categories are equivalent, inducing a natural bijection on the hom-sets.
\end{proof}

As a consequence, the calculus of signal flow graphs given by the syntax of feedback monoidal categories is sound and complete for dinatural equivalence over $\NcatC$.

\subsection{Towards observational processes}

\ExtensionalSequences{} were an improvement over \intensionalSequences{} because they allowed us to equate process descriptions that were \emph{essentially the same}.
However, we could still have two processes that are ``observationally the same'' without them being described in the same way.

\begin{rem}[Followed by]
  \defining{linkmorphfollowedby}{}
  As we saw in the Introduction, \emph{``followed by''} is a crucial operation in dataflow programming.
  Any sequence can be decomposed as $\stream{X} \cong \act{X_{0}}{\delay(\tail{\stream{X}})}$.\footnote{This can also be seen as the isomorphism making ``sequences'' a final coalgebra. That is, the first slogan we saw in~\Cref{section:fixpoint}.}
  We call \emph{``followed by''} the coherence isomorphism in $\NcatC$ that witnesses this decomposition.
  \[\morphfby_{\stream{X}} \colon \act{X_{0}}{\delay{(\tail{\stream{X}})}} \to \stream{X}\]
  In the case of constant sequences $\stream{X} = (X,X,\dots)$, we have that $\tail{\stream{X}} = \stream{X}$; which means that ``followed by'' has type
  $\morphfby_{\stream{X}} \colon \act{X}{\delay\stream{X}} \to \stream{X}$.
\end{rem}

\begin{exa}\label{example:observe}
  Consider the \extensionalStatefulSequence{}, in any \cartesianMonoidalCategory{}, that saves the first input to memory without ever producing an output.
  \emph{Observationally}, this is no different from simply discarding the first input, $(\blackComonoidUnit)_{X} \colon X \to 1$. However, in principle, we cannot show that these are \emph{dinaturally equal}, that is, $\fbk(\morphfby_{\stream{X}}) \neq (\blackComonoidUnit)_{X}$.

  More generally, discarding the result of any stochastic or deterministic signal flow graph is, \emph{observationally}, the same as doing nothing (\Cref{figure:silentwalk}, consequence of \Cref{prop:semicartesian-sequences}).
\begin{figure}
\tikzset{every picture/.style={line width=0.85pt}} %
\begin{tikzpicture}[x=0.75pt,y=0.75pt,yscale=-1,xscale=1]
\draw   (299.99,165) -- (339.99,165) -- (339.99,185) -- (299.99,185) -- cycle ;
\draw   (319.97,110) -- (359.97,110) -- (359.97,130) -- (319.97,130) -- cycle ;
\draw   (309.97,80) -- (329.97,80) -- (329.97,100) -- (309.97,100) -- cycle ;
\draw   (279.99,135.56) -- (319.99,135.56) -- (319.99,155.56) -- (279.99,155.56) -- cycle ;
\draw    (319.97,100) .. controls (319.55,109.23) and (330.12,103.51) .. (329.97,110) ;
\draw    (329.97,150) -- (329.99,165) ;
\draw    (319.99,185.65) .. controls (319.99,195.63) and (350.32,196.97) .. (349.99,185) ;
\draw    (349.99,185) .. controls (350.18,155.25) and (369.98,153.64) .. (370,116.71) ;
\draw [shift={(369.99,115)}, rotate = 89.11] [color={rgb, 255:red, 0; green, 0; blue, 0 }  ][line width=0.75]    (10.93,-3.29) .. controls (6.95,-1.4) and (3.31,-0.3) .. (0,0) .. controls (3.31,0.3) and (6.95,1.4) .. (10.93,3.29)   ;
\draw    (369.99,110) -- (369.99,115) ;
\draw    (339.98,130) -- (339.98,140) ;
\draw    (329.99,150.65) .. controls (329.79,136.45) and (350.19,136.45) .. (349.99,150.65) ;
\draw  [fill={rgb, 255:red, 0; green, 0; blue, 0 }  ,fill opacity=1 ] (337.08,140) .. controls (337.08,138.4) and (338.38,137.1) .. (339.98,137.1) .. controls (341.58,137.1) and (342.88,138.4) .. (342.88,140) .. controls (342.88,141.6) and (341.58,142.9) .. (339.98,142.9) .. controls (338.38,142.9) and (337.08,141.6) .. (337.08,140) -- cycle ;
\draw    (299.99,155) .. controls (299.56,164.23) and (310.13,158.51) .. (309.99,165) ;
\draw    (349.99,110) .. controls (350.19,97) and (369.39,97.4) .. (369.99,110) ;
\draw    (349.99,150) .. controls (349.99,171.97) and (360.32,171.63) .. (359.99,195) ;
\draw  [fill={rgb, 255:red, 0; green, 0; blue, 0 }  ,fill opacity=1 ] (357.09,195) .. controls (357.09,193.4) and (358.39,192.1) .. (359.99,192.1) .. controls (361.59,192.1) and (362.89,193.4) .. (362.89,195) .. controls (362.89,196.6) and (361.59,197.9) .. (359.99,197.9) .. controls (358.39,197.9) and (357.09,196.6) .. (357.09,195) -- cycle ;
\draw  [dash pattern={on 4.5pt off 4.5pt}] (410,80) -- (470,80) -- (470,200) -- (410,200) -- cycle ;
\draw (319.98,175) node    {$+$};
\draw (319.97,90) node    {$0$};
\draw (339.97,120) node    {$fby$};
\draw (299.99,145.56) node    {$unif$};
\draw (381,132.4) node [anchor=north west][inner sep=0.75pt]    {$\approx $};
\end{tikzpicture}
\caption{\emph{Observationally}, a silent process does nothing.}
\label{figure:silentwalk}
\end{figure}
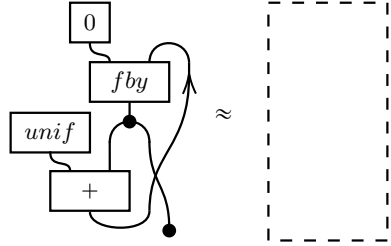
\end{exa}
\section{Observational Monoidal Streams}\label{section:observational}

In this section, we introduce our definitive \emph{monoidal streams}: \emph{\observationalStreams{}} (\Cref{def:observationalmonoidalstream}).
Their explicit construction is given by \observationalSequences{}: \extensionalSequences{} quotiented by \observationalEquivalence{}.

Intuitively, two processes are observationally equal if they are ``equal up to stage \(n\)'', for any \(n \in \naturals\).
We show that, in sufficiently well-behaved monoidal categories (which we call \emph{productive}, \Cref{def:productive}), the set of observational sequences given some inputs and outputs is the final coalgebra of a fixpoint equation (\Cref{eq:observationalstreamshort}).
The name ``observational equivalence'' is commonly used to denote equality on the final coalgebra: \Cref{theorem:observationalfinalcoalgebra} justifies our use of the term.

\subsection{Observational streams}

We saw in~\Cref{sec:int-stateful-sequences} that we can define \intensionalSequences{} as a solution to a fixpoint equation.
We now consider the same equation, just substituting the coproduct for a coend.

\begin{defi}
  For a monoidal category \(\cat{C}\), the endofunctor $(\hom{} \andThen\ \bullet) \colon [\NcatC^{op} \times \NcatC, \Set] \to [\NcatC^{op} \times \NcatC, \Set]$ is defined by
  \[(\hom{} \andThen\ Q)(\stream{X};\stream{Y}) \coloneqq \coend{M \in \catC} \hom{}(X_{0}, M \tensor Y_{0}) \times Q(M \cdot \tail{\stream{X}}, \tail{\stream{Y}}).\]
\end{defi}

\begin{defi}[Observational streams]\label{def:observationalmonoidalstream}\defining{linkobservationalmonoidalstream}{}
  The set of \emph{observational monoidal streams}, depending on inputs and outputs, %
  is the functor $\fun{Q} \colon \NcatC^{op} \times \NcatC \to \Set$ given by %
  the final fixpoint of the equation in \Cref{eq:observationalstreamshort}.
\end{defi}
\vspace{-1em}
\begin{figure}[!h]
  \centering
  $\displaystyle \fun{Q}(\stream{X},\stream{Y}) \cong \coend{M \in \catC}
    \hom{}(X_{0}, M \tensor Y_{0}) \times \fun{Q}(\act{M}{\tail{\stream{X}}}, \tail{\stream{Y}})$.
  \caption{Fixpoint equation for observational streams.}\label{eq:observationalstreamshort}
\end{figure}

The explicit construction of this final fixpoint will be given by \observationalSequences{} (\Cref{theorem:observationalfinalcoalgebra}).

\subsection{Observational sequences}

We said that observational equivalence is ``equality up to stage $n$'', so our first step will be to define what it means to truncate an \extensionalSequence{} at any given $n \in \naturals$.

\begin{defi}[$n$-Stage process]\label{def:n-stages-process}
  An \emph{n-stage process} from inputs $\stream{X} = \streamExpr{X}$ to outputs
  $\stream{Y} = \streamExpr{Y}$ is an element of the set
  \[\defining{linkstage}{\ensuremath{\fun{Stage}_{n}}}(\stream{X},\stream{Y}) \defn \StageExpr{n}{X}{Y}.\]
\end{defi}

\begin{rem}\label{remark:stagenotation}
  In other words, $n$-stage processes are $n$-tuples $(f_{i} \colon M_{i-1} \tensor X_{i} \to M_{i} \tensor Y_{i})^{n}_{i=0}$, for some choice of $M_{i}$ \emph{up to dinaturality}, that we write as
  \[\bra{f_{0}|f_{1}|\dots|f_{n}} \in \Stage{n}(\stream{X},\stream{Y}).\]
  In this notation, \emph{dinaturality} means that morphisms can \emph{slide past the bars}.
  That is, for any $r_{i} \colon N_{i} \to M_{i}$ and any tuple, dinaturality says that
  \[\bra{f_{0} \comp (r_{0} \tensor \im)|f_{1} \comp (r_{1} \tensor \im) | \dots | f_{n} \comp (r_{n} \tensor \im)} = 
  \bra{f_{0} | (r_{0} \tensor \im) \comp f_{1} | \dots | (r_{n-1} \tensor \im) \comp f_{n}}.\]
  Note that the last $r_{n}$ is removed by dinaturality.
\end{rem}

\begin{defi}[Truncation]
  Let $\braket{f_{n} \colon M_{n-1} \tensor X_{n} \to M_{n} \tensor Y_{n}} \in \eSeq(\stream{X},\stream{Y})$ be an \extensionalSequence{}.
  Its \emph{$k$th-truncation} is the element $\bra{f_{0}|\dots|f_{k}} \in \Stage{k}(\stream{X},\stream{Y})$.
  Truncation is well-defined under dinatural equivalence (see \Cref{remark:stagenotation}).

  For $k \leq n$, the \emph{$k$th-truncation} of an n-stage process given by $\bra{f_{0}|f_{1}|\dots|f_{n}} \in \Stage{n}(\stream{X},\stream{Y})$
  is $\bra{f_{0}|\dots|f_{k}} \in \Stage{k}(\stream{X},\stream{Y})$. This induces functions $\pi_{n,k} \colon \Stage{n}(\stream{X},\stream{Y}) \to \Stage{k}(\stream{X},\stream{Y})$, with
  the property that $\pi_{n,m} \comp \pi_{m,k} = \pi_{n,k}$.
\end{defi}

\begin{defi}[Observational equivalence]
  \defining{linkobservationallyequal}{}\label{def:observationallyequal}
  Two dinatural stateful sequences
  \[ \extseq{f}, \extseq{g} \in \coend{M \in [\naturals,\catC]} \prod^{\infty}_{i=0} \hom{}(M_{i-1} \tensor X_{i}, Y_{i} \tensor M_{i})\]
  are \emph{observationally equivalent} when all their n-stage truncations are equal. That is, $\bra{f_{0}|\dots | f_{n}} = \bra{g_{0}|\dots | g_{n}}$, for each $n \in \naturals$.
  We write this as $f \obsEqRel g$.
\end{defi}

\begin{rem}
  Formally, this is to say that the sequences \(\extseq{f}\) and \(\extseq{g}\) have the same image on the limit
  \[\lim\nolimits_{n} \Stage{n}(\stream{X},\stream{Y}),\]
  over the chain $\pi_{n,k} \colon \Stage{n}(\stream{X},\stream{Y}) \to \Stage{k}(\stream{X},\stream{Y})$.
\end{rem}

\begin{defi}\label{def:observationalsequence}
  An \defining{linkobservationalsequence}{\emph{observational sequence}} from \(\stream{X}\) to \(\stream{Y}\) is an
  equivalence class
  \[[\extseq{f_{n} \colon M_{n-1} \tensor X_{n} \to M_{n} \tensor Y_{n}}]_{\approx}\]
  of \extensionalSequences{} under \observationalEquality{}.
  In other words, the set of observational sequences is
  \[\oSeq(\stream{X},\stream{Y}) \cong \left(\coend{M \in [\naturals,\catC]} \prod^{\infty}_{i=0} \hom{}(M_{i-1} \tensor X_{i}, M_{i} \tensor Y_{i})\right)
  \bigg/\approx.\]
\end{defi}

Discarding the output of a process is, \emph{observationally}, the same as discarding its input (cf.~\Cref{figure:silentwalk}).

\begin{prop}\label{prop:semicartesian-sequences} In any semicartesian category (any symmetric
  monoidal category whose unit is terminal), observational sequences with no
  output coincide with the discard sequence.
\end{prop}
\begin{proof}
  Let \(\coUnit_{X}\) indicate the unique morphism to the terminal object in the base category \(\cat{C}\).
  For a stream of objects \(\stream{X}\), the discard sequence is \(\extseq{\coUnit_{X_{n}}}\).
  Consider an observational sequence \(\extseq{f_{n}} \colon \stream{X} \to \stream{I}\) with no output.
  We show by induction that \(\bra{f_{0} | \cdots | f_{n} \dcomp \coUnit_{M_{n}}} = \bra{\coUnit_{X_{0}} | \cdots | \coUnit_{X_{n}}}\) without sliding on the last memory \(M_{n}\).
  For \(n=0\), \(f_{0} \dcomp \coUnit_{M_{0}} = \coUnit_{X_{0}}\) by semicartesianity in \(\cat{C}\), then \(\bra{f_{0} \dcomp \coUnit_{M_{0}}} = \bra{\coUnit_{X_{0}}}\).
  Suppose that \(\bra{f_{0} | \cdots | f_{n} \dcomp \coUnit_{M_{n}}} = \bra{\coUnit_{X_{0}} | \cdots | \coUnit_{X_{n}}}\) without sliding on the last memory \(M_{n}\), then, 
  \begin{align*}
    & \bra{f_{0} | \cdots | f_{n} | f_{n+1} \dcomp \coUnit_{M_{n+1}} }  &&\\
    =\ & \bra{f_{0} | \cdots | f_{n} | \coUnit_{M_{n}} \tensor \coUnit_{X_{n+1}}} && \text{(by semicartesianity)}\\
    =\ & \bra{f_{0} | \cdots | f_{n} \dcomp \coUnit_{M_{n}} | \coUnit_{X_{n+1}}} && \text{(by dinaturality)}\\
    =\ & \bra{\coUnit_{X_{0}} | \cdots | \coUnit_{X_{n}} | \coUnit_{X_{n+1}}} && \text{(by induction).}
  \end{align*}
  Thus, we can rewrite the truncations of \(\extseq{f_{n}}\), by dinaturality, as follows
  \[
     \bra{f_{0} | \cdots | f_{n}}
     = \bra{f_{0} | \cdots | f_{n} \dcomp \coUnit_{M_{n}}} 
     = \bra{\coUnit_{X_{0}} | \cdots | \coUnit_{X_{n}}}. 
  \]
  This proves that \(\extseq{f_{n}} = \extseq{\coUnit_{X_{n}}}\) as observational sequences.
\end{proof}
\subsection{Productive categories}
\label{section:productive}

The interaction between dinatural and observational equivalence is of particular interest in some well-behaved categories that we call \emph{productive categories}.
In productive categories, observational sequences are the final fixpoint of an equation (\Cref{theorem:observationalfinalcoalgebra}), analogous to that of \Cref{sec:int-stateful-sequences}.

An important property of programs is \emph{termination}: a terminating program always halts in a finite amount of time.
However, some programs (such as servers, drivers) are not actually intended to terminate but to produce infinite output streams.
A more appropriate notion in these cases is that of \emph{productivity}:
a program that outputs an infinite stream of data is productive if each individual component of the stream is produced in finite time.
To quip, \emph{``a productive stream is a terminating first component followed by a productive stream''}.

The first component of our streams is only defined \emph{up to some future}.
It is an equivalence class $\alpha \in \Stage{1}(\stream{X},\stream{Y})$,
with representatives $\alpha_{i} \colon X_{0} \to M_{i} \tensor Y_{0}$ that do not necessarily share the same memory.
But, if it does terminate, there is a process $\alpha_{0} \colon X_{0} \to M_{0} \tensor Y_{0}$ in our theory representing the process just until $Y_{0}$ is output.

\begin{defi}[Terminating component]
  A single stage (0-stage) process $\alpha \in \Stage{0}(\stream{X},\stream{Y})$
  is \emph{terminating relative to $\catC$} if there exist $M_0$ and $\alpha_{0}
  \colon X_{0} \to M_{0} \tensor Y_{0}$ such that each one of its
  representatives, $\bra{\alpha'} = \alpha$, can be written as $\alpha' =
  \alpha_{0} \comp (s' \tensor \im)$ for some $s' \colon M_{0} \to
  M^{(i)}$.

The morphisms (e.g.~$s'$ or $s''$) represent what is unique to each
representative (e.g.~$\alpha'$, or $\alpha''$, respectively), and so we ask
that, for any $u\colon M_{0} \tensor A \to U \tensor B$ and $v \colon M_{0}
\tensor A \to V \tensor B$, the equality $\bra{\alpha' \tensor \im_{A} \comp u
\tensor \im_{Y_{0}}} = \bra{\alpha'' \tensor \im_{A} \comp v \tensor
\im_{Y_{0}}}$ implies $\bra{s' \tensor \im_{A} \comp u} = \bra{s'' \tensor
\im_{A} \comp v}$.
\end{defi}

\begin{defi}[Productive category]
  \defining{linkproductive}{}\label{def:productive}
  A symmetric monoidal category $(\catC,\tensor,\monoidalunit)$ is \emph{productive} when
every 0-stage process is terminating relative to $\catC$.
\end{defi}

\begin{rem}
  Cartesian monoidal categories are \productive{} (\Cref{prop:cartesianproductive}).
  \emph{Markov categories} \cite{fritz2020} with \emph{conditionals} and \emph{ranges} are \productive{} (\Cref{appendix:productivemarkov}).
  Free symmetric monoidal categories and compact closed categories are always productive.
\end{rem}

\subsection{Observational streams in productive categories}
  We conclude this section proving the following \Cref{theorem:observationalfinalcoalgebra}: whenever the category is \productive{}, \observationalSequences{} are explicit construction of \observationalStreams{}. The proof shows that, from the mild assumption of producitvity, we can extract all of the relevant mathematical structure to make the fixpoint equation work as expected.

  \begin{defi}[Truncating coherently]\label{def:truncatingcoherently}
    \defining{linktruncatingcoherently}{}
    Let $f_{n}^{k} \colon  M_{n-1} \tensor X_{n}  \to M_{n} \tensor Y_{n}$ be a family of families of morphisms of increasing length, indexed by $k \in \naturals$ and $n \leq k$.
    We say that this family \emph{truncates coherently} if
    $\bra{f_{0}^{p}|\dots|f_{n}^{p}} = \bra{f_{0}^{q}|\dots|f_{n}^{q}}$
    for each $p, q \in \naturals$ and each $n \leq \min\{p,q\}$.
  \end{defi}

  \begin{lem}[Factoring a family of processes]\label{lemma:familyfactoring}
    In a \productive{} category, let
    $(\bra{f_{0}^{k}|\dots|f_{k}^{k}})_{k \in \naturals}$ be a sequence
    of sequences that \truncatesCoherently{}. Then, there exists another sequence
    $h_{i} \colon N_{i-1} \tensor X_{i}  \to N_{i} \tensor Y_{i}$, together with morphisms $s^{k}_{i} \colon N_i \to M_i$ satisfying $s^{k}_{i-1}f^{k}_{i} = h_{i}s^{k}_{i}$ and such that, for each $k \in \naturals$ and each $n \leq k$,
    $\bra{f_{0}^{k}|\dots|f_{n}^{k}} = \bra{h_{0}|\dots|h_{n}}$.
    Moreover, this family $h_{i}$ is such that
    $\bra{h_{0}\dots h_{n}s^{p}_{n}u} = \bra{h_{0}\dots h_{n}s^{q}_{n}v}$
    implies $\bra{s^{p}_{n}u} = \bra{s^{q}_{n}v}$.
  \end{lem}
  \begin{proof}
    We construct the family by induction. In the case $n=0$, we use that
    the family \truncatesCoherently{} to have that $\bra{f^{p}_{0}} = \bra{f^{q}_{0}}$
    and thus, by \productivity{}, create an $h_{0}$ with $f^{k}_{0} = h_{0}s^{k}_{0}$
    such that $\bra{h_{0}s^{p}_{0}u} = \bra{h_{0}s^{q}_{0}v}$ implies $\bra{s^{p}_{0}u} = \bra{s^{q}_{0}v}$.
  
    In the general case, assume we already have constructed $h_{0},\dots,h_{n-1}$
    with $s^{k}_{i-1}f^{k}_{i} = h_{i}s^{k}_{i}$ such that, for each $k \in \naturals$ and
    $\bra{f_{0}^{k}|\dots|f_{n-1}^{k}} = \bra{h_{0}|\dots|h_{n-1}}$. Moreover,
    $\bra{h_{0}\dots h_{n-1}s^{p}_{n-1}u} = \bra{h_{0}\dots h_{n-1}s^{q}_{n-1}v}$
    implies $\bra{s^{p}_{n-1}u} = \bra{s^{q}_{n-1}v}$.
  
    In this case, we use the fact that composition ``along a bar'' is dinatural:
    $\bra{f_{0}^p|\dots|f_{n}^{p}} = \bra{f_{0}^{q}|\dots|f_{n}^{q}}$ implies that
    $\bra{\tid{f_{0}^p\dots f_{n}^{p}}} = \bra{\tid{f_{0}^q\dots f_{n}^{q}}}$. This can be then
    rewritten as
    $$\bra{h_{0}\dots h_{n-1}s^{p}_{n-1}f_{n}^{p}} = \bra{h_{0}\dots h_{n-1}s^{q}_{n-1}f_{n}^{q}},$$
    which in turn implies
    $\bra{s^{p}_{n-1}f_{n}^{p}} = \bra{s^{q}_{n-1}f_{n}^{q}}$. By \productivity{},
    there exists $h_{n}$ with $s_{n-1}^{k}f^{k}_{n} = h_{n}s^{k}_{n}$ such that
    $\bra{f_{0}^{k}|\dots|f_{n}^{k}} = \bra{h_{0}|\dots|h_{n}}$.
  
    Finally, assume that
    $\bra{h_{0}\dots h_{n-1}h_{n}s_{n}^{p}u} = \bra{h_{0} \dots  h_{n-1}h_{n}s_{n}^{q}v}$. Thus, we
    have $$\bra{h_{0}\dots  h_{n-1}s_{n-1}^{p}f_{n}^{p}u} = \bra{h_{0} \dots  h_{n-1}s_{n-1}^{q}f_{n}^{q}v}$$
    and $\bra{s_{n-1}^{p}f_{n}^{p}u} = \bra{s_{n-1}^{q}f_{n}^{q}v}$. This can be rewritten as
    $\bra{h_{n}s_{n}^{p}u} = \bra{h_{n}s_{n}^{q}v}$, which in turn implies $\bra{s_{n}^{p}u} = \bra{s_{n}^{q}v}$.
    We have shown that the $h_{n}$ that we constructed satisfies the desired property.
  \end{proof}

  \begin{lem}\label{lemma:observationalstatefulisterminalsequence}
    In a \productive{} category, the set of observational sequences is isomorphic to the limit of the terminal
    sequence of the endofunctor $(\hom{} \odot\, \bullet)$ via the canonical
    map between them.
  \end{lem}
  \begin{proof}
    We start by noting that observational equivalence of sequences is, by
    definition, the same thing as being equal under the canonical map to the
    limit of the terminal sequence
    \[\lim_{n} \int^{M_{0},\dots,M_{n}} \prod_{i=0}^{n} 
      \hom{}(M_{i-1}\tensor X_{i},  M_{i} \tensor Y_{i}).\]
    We will show that this canonical map is surjective. This means that the domain quotiented by equality under the map is isomorphic to the codomain, q.e.d.
  
    Indeed, given any family $f_{n}^{k}$ that \truncatesCoherently{}, we can
    apply \Cref{lemma:familyfactoring} to find a sequence $h_{i}$ such
    that $\bra{f_{0}^{k}|\dots|f_{n}^{k}} = \bra{h_{0}|\dots|h_{n}}$. This means
    that it is the image of the stateful sequence $h_{i}$.
  \end{proof}

  \begin{lem}[Factoring two processes]\label{lemma:multiplefactoring}
    In a \productive{} category, let $\bra{f_{0}} = \bra{g_{0}}$. Then there
    exists $h_{0}$ with $f_{0} = h_{0}s_{0}$ and $g_{0} = h_{0}t_{0}$ such that
    \[\bra{f_{0}|\dots|f_{n}} = \bra{g_{0}|\dots|g_{n}}\]
    implies the existence of a family $h_{i}$ together with $s_{i}$ and $t_{i}$ such that
    $s_{i-1}f_{i} = h_{i}s_{i}$ and $t_{i-1}g_{i} = h_{i}t_{i}$; and moreover, such that
    \[\bra{h_{0}\dots h_{n}s_{n}u} = \bra{h_{0}\dots h_{n}t_{n}v} \mbox{ implies } \bra{s_{n}u} = \bra{t_{n}v}.\]
  \end{lem}
  \begin{proof}
    By \productivity{}, we can find such a factorization $f_{0} = h_{0}s_{0}$ and $g_{0} = h_{0}t_{0}$.
  
    Assume now that we have a family of morphisms such that $\bra{f_{0}|\dots|f_{n}} = \bra{g_{0}|\dots|g_{n}}$.
    We proceed by induction on $n$, the size of the family. The case $n = 0$ follows from
    the definition of \productive{} category.
  
    In the general case, we will construct the relevant $h_{n}$. The assumption $\bra{f_{0}|\dots|f_{n}} = \bra{g_{0}|\dots|g_{n}}$
    implies, in particular, that
    $\bra{f_{0}|\dots|f_{n-1}} = \bra{g_{0}|\dots|g_{n-1}}$.
    Thus, by induction
    hypothesis, there exist $h_{1},\dots,h_{n-1}$ together with
    $s_{i-1}f_{i} = h_{i}s_{i}$ and $t_{i-1}g_{i} = h_{i}t_{i}$, such that
    \[\bra{h_{0}\dots h_{n}s_{n-1}u} = \bra{h_{0}\dots h_{n}t_{n-1}v} \mbox{ implies } \bra{s_{n-1}u} = \bra{t_{n-1}v}.\]
    We know that $\bra{f_{0}\dots f_{n}} = \bra{g_{0} \dots g_{n}}$ and thus,
    \[\bra{h_{0}\dots h_{n-1}s_{n-1}f_{n}} = \bra{h_{0} \dots h_{n-1}t_{n-1}g_{n}},\]
    which, by induction hypothesis, implies $\bra{s_{n-1}f_{n}} = \bra{t_{n-1}g_{n}}$.
    By \productivity{}, there exists $h_{n}$ with $s_{n-1}f_{n} = h_{n}s_{n}$ and
    $t_{n-1}g_{n} = h_{n}t_{n}$ such that
    $\bra{h_{n}s_{n}u} = \bra{h_{n}t_{n}v} \mbox{ implies } \bra{s_{n}u} = \bra{t_{n}v}$.
  
    Finally, assume that $\bra{h_{0}\dots h_{n-1}h_{n}s_{n}u} = \bra{h_{0} \dots  h_{n-1}h_{n}t_{n}v}$. Thus, we have $$\bra{h_{0}\dots  h_{n-1}s_{n-1}f_{n}u} = \bra{h_{0} \dots  h_{n-1}t_{n-1}g_{n}v}$$ and $\bra{s_{n-1}f_{n}u} = \bra{t_{n-1}g_{n}v}$. This can be rewritten as $\bra{h_{n}s_{n}u} = \bra{h_{n}t_{n}v}$, which in turn implies $\bra{s_{n}u} = \bra{t_{n}v}$.
    We have shown that the $h_{n}$ that we constructed satisfies the desired property.
  \end{proof}

  \begin{lem}[Removing the first step]\label{lemma:removestep}
    In a \productive{} category, let $\bra{f_{0}} = \bra{g_{0}}$. Then there
    exists $h$ with $f_{0} = hs$ and $g_{0} = ht$ such that
    $\bra{f_{0}|\dots|f_{n}} = \bra{g_{0}|\dots|g_{n}}$ implies
    $\bra{sf_{1}|\dots|f_{n}} = \bra{tg_{1}|\dots|g_{n}}$.
  \end{lem}
  \begin{proof}
    By \Cref{lemma:multiplefactoring}, we obtain a factorization $f_{0} = hs$ and
    $g_{0} = ht$. Moreover, each time that we have
    $\bra{f_{0}|\dots|f_{n}} = \bra{g_{0}|\dots|g_{n}}$, we can obtain a family
    $h_{i}$ together with $s_{i}$ and $t_{i}$ such that
    $s_{i-1}f_{i} = h_{i}s_{i}$ and $t_{i-1}f_{i} = h_{i}t_{i}$. Using the fact
    that $\bra{h_{n}s_{n}} = \bra{h_{n}t_{n}}$, we have that
    $\bra{h_{1}|\dots|h_{n}s_{n}} = \bra{h_{1}|\dots|h_{n}t_{n}}$, which can be
    rewritten using dinaturality as
    $\bra{s_{0}f_{1}|\dots|f_{n}} = \bra{t_{0}g_{1}|\dots|g_{n}}$.
  \end{proof}
  
  \begin{lem}\label{lemma:coalgebraexists}
    In a \productive{} category, the final coalgebra of the endofunctor $(\hom{} \odot\, \bullet)$ does exist and it is given by the limit of the terminal sequence
    \[L \coloneqq \lim_{n} \int^{M_{0},\dots,M_{n}} 
      \prod_{i=0}^{n} \hom{}(M_{i-1} \tensor X_{i}, M_{i} \tensor Y_{i}).\]
  \end{lem}
  \begin{proof}
    We will apply \Cref{th:adamek}.
    The endofunctor $(\hom{} \odot\, \bullet)$ acts on the category $[(\NcatC)^{op} \times \NcatC, \Set]$, which, being a presheaf category, has all small limits. We will show that there is an isomorphism $\hom{} \odot\ L \cong L$ given by the canonical morphism between them.
  
    First, note that the set $L(\stream{X};\stream{Y})$ is, explicitly,
    \[\lim_{n} \int^{M_{1},\dots,M_{n}} \prod_{i=1}^{n} \hom{}(M_{i-1} \tensor X_{i}, M_{i}  \tensor Y_{i}).\]
    A generic element from this set is a \emph{sequence of sequences of increasing length}.
    Moreover, the sequences must \emph{truncate coherently} (\Cref{def:truncatingcoherently}).
  
    Secondly, note that the set
    $(\hom{} \odot\ L)(\stream{X};\stream{Y})$ is, explicitly,
    \begin{equation*}
      \int^{M_{0}}\left( \hom{}(X_{0},  M_{0} \tensor Y_{0}) \times \lim_{n} \int^{M_{1},\dots,M_{n}} \prod_{i=1}^{n} \hom{}(M_{i-1} \tensor X_{i}, M_{i} \tensor Y_{i})\right).
    \end{equation*}
    A generic element from this set is of the form
    \[\bra{f|(\bra{f_{1}^{k}|\dots|f_{k}^{k}})_{k \in \naturals}},\]
    that is, a pair consisting on a first morphism
    $f \colon X_{0} \to Y_{0} \tensor M_{0}$ and a family of sequences
    $(\bra{f_{1}^{k}|\dots|f_{k}^{k}})$, quotiented by dinaturality of $M_{0}$ and
    truncating coherently. The canonical map to $L(\stream{X};\stream{Y})$ maps
    this generic element to the family of sequences
    $(\bra{f_{0}|f_{1}^{k}|\dots|f_{k}^{k}})_{k \in \naturals}$, which truncates
    coherently because the previous family did and we are precomposing with $f_{0}$,
    which is dinatural.
  
    Thirdly, this map is injective. Imagine a pair of elements
    $\bra{f_{0}|(\bra{f_{1}^{k}|\dots|f_{k}^{k}})_{k \in \naturals}}$ and
    $\bra{g_{0}|(\bra{g_{1}^{k}|\dots|g_{k}^{k}})_{k \in \naturals}}$ that have the
    same image, meaning that, for each $k \in \naturals$,
    \[\bra{f_{0}|f_{1}^{k}|\dots|f_{k}^{k}} = \bra{g_{0}|g_{1}^{k}|\dots|g_{k}^{k}}.\]
    By \Cref{lemma:removestep}, we can find $h$ with $f_{0} = hs$ and $g_{0} = ht$
    such that, for each $k \in \naturals$,
    $\bra{sf_{1}^{k}|\dots|f_{k}^{k}} = \bra{tg_{1}^{k}|\dots|g_{k}^{k}}$. Thus,
    \[\begin{aligned}
      \bra{f_{0}|(\bra{f_{1}^{k}|\dots|f_{k}^{k}})_{k \in \naturals}} =
      \bra{h|(\bra{sf_{1}^{k}|\dots|f_{k}^{k}})_{k \in \naturals}} = \\
      \bra{h|(\bra{tg_{1}^{k}|\dots|g_{k}^{k}})_{k \in \naturals}} =
      \bra{g_{0}|(\bra{tg_{1}^{k}|\dots|g_{k}^{k}})_{k \in \naturals}}.
    \end{aligned}\]
  
  Finally, this map is also surjective. From \Cref{lemma:familyfactoring}, it follows
  that any family that \truncatesCoherently{} can be equivalently written as $\bra{h_{0}|\dots|h_{n}}_{n \in \naturals}$,
  which is the image of the element $\bra{h_{0}|\bra{h_{1}|\dots|h_{n}}_{n \in \naturals}}$.
  \end{proof}

\begin{thm}\label{theorem:observationalfinalcoalgebra}
  \ObservationalSequences{} are the explicit construction of \observationalStreams{} when the category is \productive{}.
More precisely, in a \productive{} category, the final fixpoint of the equation in \Cref{eq:observationalstreamshort} is given by the set of \observationalSequences{}, $\oSeq$.

In a \productive{} category, the final coalgebra of the endofunctor
$(\hom{} \odot\, \bullet)$ exists and it is given by the set of stateful
sequences quotiented by observational equivalence.
\[\left(\int^{M \in [\naturals,\catC]} \prod^{\infty}_{i=0} 
  \hom{}(M_{i-1}\tensor X_{i}, M_{i} \tensor Y_{i})\right)
\bigg/\approx\]
\end{thm}
\begin{proof}[Proof sketch]
  The terminal sequence for this final coalgebra is given by $\Stage{n}(\stream{X},\stream{Y})$.
  In productive categories, we have proven that the limit $\lim_{n} \Stage{n}(\stream{X},\stream{Y})$ is a fixpoint of the equation in \Cref{eq:observationalstreamshort} (\Cref{lemma:coalgebraexists}). The hypothesis of productivity is used to make a particular coend commute with a limit, so that Adamek's theorem can be applied (\Cref{th:adamek}).
  By \Cref{lemma:coalgebraexists}, we know that the final coalgebra exists and is given by the limit of the terminal sequence.
  By \Cref{lemma:observationalstatefulisterminalsequence}, we know that it is isomorphic to the set of stateful sequences quotiented by observational equivalence.
\end{proof}

 \newpage
\section{The Category of Monoidal Streams}\label{sec:monoidal-streams}

We are ready to construct $\STREAM$: the \feedbackMonoidalCategory{} of monoidal streams.
Let us recast the definitive notion of monoidal stream (\Cref{def:observationalmonoidalstream}) coinductively.

\begin{defi}%
\label{def:monoidalstream}\defining{linkmonoidalstream}
A \emph{monoidal stream} $f \in \STREAM(\stream{X},\stream{Y})$ is a triple consisting of
  \begin{itemize}
    \item $M(f) \in \obj{\catC}$, the \emph{memory},
    \item $\defining{linknow}{\ensuremath{\now}}(f) \in \hom{}(X_{0}, M(f) \tensor Y_{0})$, the \emph{first action},
    \item $\defining{linklater}{\ensuremath{\later}}(f) \in \STREAM(\act{M(f)}{\tail{\stream{X}}},\tail{\stream{Y}})$, the \emph{rest of the action},
  \end{itemize}
  quotiented by dinaturality in $M(f)$.
Explicitly, monoidal streams are quotiented by the equivalence relation $f \sim g$ generated by
\begin{itemize}
  \item the existence of $r \colon M(g) \to M(f)$,
  \item such that $\now(f) = \now(g) \comp r$,
  \item and such that $\act{r}{\later(f)} \sim \later(g)$.
\end{itemize}
Here, $\act{r}{\later(f)} \in \STREAM(\act{M(g)}{\tail{\stream{X}}},\tail{\stream{Y}})$ is obtained by precomposition of the first action of $\later(f)$ with $r$.
\end{defi}

\begin{rem}
  In terms of string diagrams (\Cref{strings:monoidal-stream}, left), a monoidal stream is defined by two parts connected by some wires and separated by some space representing the context between them. The first part represents the first action which is a morphism in the base monoidal category; the second part represents the rest of the action, which is again a stream. This is, the two parts of the diagram are in two different categories. However, we can still slide morphisms along the wires connecting both parts; this sliding encodes dinaturality (\Cref{strings:monoidal-stream}, right).
  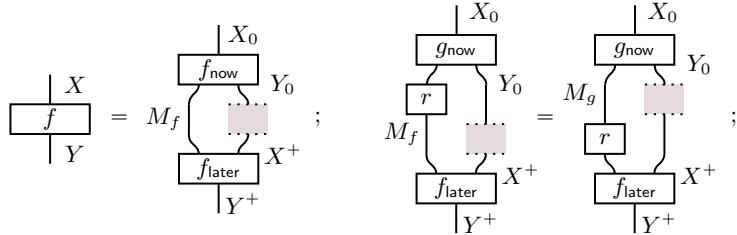
\begin{figure}[ht]

\tikzset{every picture/.style={line width=0.75pt}} %

\begin{tikzpicture}[x=0.75pt,y=0.75pt,yscale=-1,xscale=1]
\draw  [line width=0.75]  (5,55) -- (45,55) -- (45,70) -- (5,70) -- cycle ;
\draw [line width=0.75]    (25,40) -- (25,55) ;
\draw [line width=0.75]    (25,70) -- (25,85) ;
\draw  [color={rgb, 255:red, 0; green, 0; blue, 0 }  ,draw opacity=1 ][line width=0.75]  (89.82,30) -- (129.82,30) -- (129.82,45) -- (89.82,45) -- cycle ;
\draw [color={rgb, 255:red, 0; green, 0; blue, 0 }  ,draw opacity=1 ][line width=0.75]    (110,15) -- (110,30) ;
\draw [color={rgb, 255:red, 0; green, 0; blue, 0 }  ,draw opacity=1 ][line width=0.75]    (95,55) -- (95,70) ;
\draw [color={rgb, 255:red, 0; green, 0; blue, 0 }  ,draw opacity=1 ]   (100,45) .. controls (100.03,49.67) and (95.37,49.33) .. (95,55) ;
\draw [color={rgb, 255:red, 0; green, 0; blue, 0 }  ,draw opacity=1 ][line width=0.75]  [dash pattern={on 0.84pt off 2.51pt}]  (115,55) -- (135,55) ;
\draw [color={rgb, 255:red, 0; green, 0; blue, 0 }  ,draw opacity=1 ][line width=0.75]  [dash pattern={on 0.84pt off 2.51pt}]  (115,70) -- (135,70) ;
\draw  [line width=0.75]  (90,80) -- (130,80) -- (130,95) -- (90,95) -- cycle ;
\draw [line width=0.75]    (110,95) -- (110,110) ;
\draw    (125,70) .. controls (124.7,75.67) and (119.7,75) .. (120,80) ;
\draw [color={rgb, 255:red, 0; green, 0; blue, 0 }  ,draw opacity=1 ]   (120,45) .. controls (119.7,50.67) and (124.7,50) .. (125,55) ;
\draw [color={rgb, 255:red, 0; green, 0; blue, 0 }  ,draw opacity=1 ]   (95,70) .. controls (94.7,75.67) and (99.7,75) .. (100,80) ;
\draw  [draw opacity=0][fill={rgb, 255:red, 75; green, 17; blue, 33 }  ,fill opacity=0.15 ] (115,55) -- (135,55) -- (135,70) -- (115,70) -- cycle ;
\draw  [draw opacity=0][line width=0.75]  (45,50) -- (75,50) -- (75,75) -- (45,75) -- cycle ;
\draw  [draw opacity=0] (70,45) -- (95,45) -- (95,80) -- (70,80) -- cycle ;
\draw  [draw opacity=0] (110.18,10) -- (135,10) -- (135,30) -- (110.18,30) -- cycle ;
\draw  [draw opacity=0] (130,35) -- (154.82,35) -- (154.82,55) -- (130,55) -- cycle ;
\draw  [draw opacity=0] (130.18,70) -- (155,70) -- (155,90) -- (130.18,90) -- cycle ;
\draw  [draw opacity=0] (110,95) -- (134.82,95) -- (134.82,115) -- (110,115) -- cycle ;
\draw  [draw opacity=0] (25,35) -- (49.82,35) -- (49.82,55) -- (25,55) -- cycle ;
\draw  [draw opacity=0] (25.18,70) -- (50,70) -- (50,90) -- (25.18,90) -- cycle ;
\draw  [draw opacity=0] (25.18,70) -- (50,70) -- (50,90) -- (25.18,90) -- cycle ;
\draw  [draw opacity=0][line width=0.75]  (140,55) -- (170,55) -- (170,70) -- (140,70) -- cycle ;
\draw  [color={rgb, 255:red, 0; green, 0; blue, 0 }  ,draw opacity=1 ][line width=0.75]  (209.82,20) -- (249.82,20) -- (249.82,35) -- (209.82,35) -- cycle ;
\draw [color={rgb, 255:red, 0; green, 0; blue, 0 }  ,draw opacity=1 ][line width=0.75]    (230,5) -- (230,20) ;
\draw [color={rgb, 255:red, 0; green, 0; blue, 0 }  ,draw opacity=1 ][line width=0.75]    (214.82,60) -- (214.82,80) ;
\draw [color={rgb, 255:red, 0; green, 0; blue, 0 }  ,draw opacity=1 ]   (220,35) .. controls (220.03,39.67) and (215.37,39.33) .. (215,45) ;
\draw [color={rgb, 255:red, 0; green, 0; blue, 0 }  ,draw opacity=1 ][line width=0.75]  [dash pattern={on 0.84pt off 2.51pt}]  (235,65) -- (255,65) ;
\draw [color={rgb, 255:red, 0; green, 0; blue, 0 }  ,draw opacity=1 ][line width=0.75]  [dash pattern={on 0.84pt off 2.51pt}]  (235,80) -- (255,80) ;
\draw  [line width=0.75]  (209.82,90) -- (249.82,90) -- (249.82,105) -- (209.82,105) -- cycle ;
\draw [line width=0.75]    (229.82,105) -- (229.82,120) ;
\draw    (244.82,80) .. controls (244.52,85.67) and (239.52,85) .. (239.82,90) ;
\draw [color={rgb, 255:red, 0; green, 0; blue, 0 }  ,draw opacity=1 ]   (240,35) .. controls (239.7,40.67) and (244.7,40) .. (245,45) ;
\draw [color={rgb, 255:red, 0; green, 0; blue, 0 }  ,draw opacity=1 ]   (214.82,80) .. controls (214.52,85.67) and (219.52,85) .. (219.82,90) ;
\draw  [draw opacity=0][fill={rgb, 255:red, 75; green, 17; blue, 33 }  ,fill opacity=0.15 ] (234.82,65) -- (254.82,65) -- (254.82,80) -- (234.82,80) -- cycle ;
\draw  [draw opacity=0] (230,0) -- (254.82,0) -- (254.82,20) -- (230,20) -- cycle ;
\draw  [draw opacity=0] (245,35) -- (269.82,35) -- (269.82,55) -- (245,55) -- cycle ;
\draw  [draw opacity=0] (250,80) -- (274.82,80) -- (274.82,100) -- (250,100) -- cycle ;
\draw  [line width=0.75]  (204.82,45) -- (224.82,45) -- (224.82,60) -- (204.82,60) -- cycle ;
\draw [color={rgb, 255:red, 0; green, 0; blue, 0 }  ,draw opacity=1 ][line width=0.75]    (244.82,45) -- (244.82,65) ;
\draw  [color={rgb, 255:red, 0; green, 0; blue, 0 }  ,draw opacity=1 ][line width=0.75]  (299.82,20) -- (339.82,20) -- (339.82,35) -- (299.82,35) -- cycle ;
\draw [color={rgb, 255:red, 0; green, 0; blue, 0 }  ,draw opacity=1 ][line width=0.75]    (320,5) -- (320,20) ;
\draw [color={rgb, 255:red, 0; green, 0; blue, 0 }  ,draw opacity=1 ][line width=0.75]    (304.82,45) -- (304.82,65) ;
\draw [color={rgb, 255:red, 0; green, 0; blue, 0 }  ,draw opacity=1 ]   (310,35) .. controls (310.03,39.67) and (305.37,39.33) .. (305,45) ;
\draw [color={rgb, 255:red, 0; green, 0; blue, 0 }  ,draw opacity=1 ][line width=0.75]  [dash pattern={on 0.84pt off 2.51pt}]  (325,45) -- (345,45) ;
\draw [color={rgb, 255:red, 0; green, 0; blue, 0 }  ,draw opacity=1 ][line width=0.75]  [dash pattern={on 0.84pt off 2.51pt}]  (325,60) -- (345,60) ;
\draw  [line width=0.75]  (299.82,90) -- (339.82,90) -- (339.82,105) -- (299.82,105) -- cycle ;
\draw [line width=0.75]    (319.82,105) -- (319.82,120) ;
\draw    (334.82,80) .. controls (334.52,85.67) and (329.52,85) .. (329.82,90) ;
\draw [color={rgb, 255:red, 0; green, 0; blue, 0 }  ,draw opacity=1 ]   (330,35) .. controls (329.7,40.67) and (334.7,40) .. (335,45) ;
\draw [color={rgb, 255:red, 0; green, 0; blue, 0 }  ,draw opacity=1 ]   (304.82,80) .. controls (304.52,85.67) and (309.52,85) .. (309.82,90) ;
\draw  [draw opacity=0][fill={rgb, 255:red, 75; green, 17; blue, 33 }  ,fill opacity=0.15 ] (324.82,45) -- (344.82,45) -- (344.82,60) -- (324.82,60) -- cycle ;
\draw  [draw opacity=0] (340,25) -- (364.82,25) -- (364.82,45) -- (340,45) -- cycle ;
\draw  [draw opacity=0] (340,80) -- (364.82,80) -- (364.82,100) -- (340,100) -- cycle ;
\draw  [line width=0.75]  (294.82,65) -- (314.82,65) -- (314.82,80) -- (294.82,80) -- cycle ;
\draw [color={rgb, 255:red, 0; green, 0; blue, 0 }  ,draw opacity=1 ][line width=0.75]    (334.82,60) -- (334.82,80) ;
\draw  [draw opacity=0][line width=0.75]  (259.82,50) -- (289.82,50) -- (289.82,75) -- (259.82,75) -- cycle ;
\draw  [draw opacity=0][line width=0.75]  (354.82,55) -- (384.82,55) -- (384.82,70) -- (354.82,70) -- cycle ;
\draw  [draw opacity=0] (229.82,105) -- (254.64,105) -- (254.64,125) -- (229.82,125) -- cycle ;
\draw  [draw opacity=0] (319.82,0) -- (344.64,0) -- (344.64,20) -- (319.82,20) -- cycle ;
\draw  [draw opacity=0] (319.82,105) -- (344.64,105) -- (344.64,125) -- (319.82,125) -- cycle ;
\draw  [draw opacity=0] (190,60) -- (214.82,60) -- (214.82,80) -- (190,80) -- cycle ;
\draw  [draw opacity=0] (279.82,40) -- (304.64,40) -- (304.64,60) -- (279.82,60) -- cycle ;

\draw (25,62.5) node  [font=\footnotesize]  {${\textstyle f}$};
\draw (110,37.5) node  [font=\footnotesize,color={rgb, 255:red, 0; green, 0; blue, 0 }  ,opacity=1 ]  {$f_{\mathsf{now}}$};
\draw (110,87.5) node  [font=\footnotesize,color={rgb, 255:red, 0; green, 0; blue, 0 }  ,opacity=1 ]  {$f_{\mathsf{later}}$};
\draw (60,62.5) node  [font=\footnotesize]  {$=$};
\draw (82.5,62.5) node  [font=\footnotesize,color={rgb, 255:red, 0; green, 0; blue, 0 }  ,opacity=1 ]  {$M_{f}$};
\draw (122.59,20) node  [font=\footnotesize,color={rgb, 255:red, 0; green, 0; blue, 0 }  ,opacity=1 ]  {$X_{0}$};
\draw (142.41,45) node  [font=\footnotesize,color={rgb, 255:red, 0; green, 0; blue, 0 }  ,opacity=1 ]  {$Y_{0}$};
\draw (142.59,80) node  [font=\footnotesize,color={rgb, 255:red, 0; green, 0; blue, 0 }  ,opacity=1 ]  {$X^{+}$};
\draw (122.41,105) node  [font=\footnotesize,color={rgb, 255:red, 0; green, 0; blue, 0 }  ,opacity=1 ]  {$Y^{+}$};
\draw (37.59,45) node  [font=\footnotesize,color={rgb, 255:red, 0; green, 0; blue, 0 }  ,opacity=1 ]  {$X$};
\draw (37.59,80) node  [font=\footnotesize,color={rgb, 255:red, 0; green, 0; blue, 0 }  ,opacity=1 ]  {$Y$};
\draw (160,62.5) node  [font=\footnotesize]  {$;$};
\draw (229.82,27.5) node  [font=\footnotesize,color={rgb, 255:red, 0; green, 0; blue, 0 }  ,opacity=1 ]  {$g_{\mathsf{now}}$};
\draw (244.14,8.5) node  [font=\footnotesize,color={rgb, 255:red, 0; green, 0; blue, 0 }  ,opacity=1 ]  {$X_{0}$};
\draw (257.41,45) node  [font=\footnotesize,color={rgb, 255:red, 0; green, 0; blue, 0 }  ,opacity=1 ]  {$Y_{0}$};
\draw (262.41,90) node  [font=\footnotesize,color={rgb, 255:red, 0; green, 0; blue, 0 }  ,opacity=1 ]  {$X^{+}$};
\draw (242.23,115) node  [font=\footnotesize,color={rgb, 255:red, 0; green, 0; blue, 0 }  ,opacity=1 ]  {$Y^{+}$};
\draw (214.82,52.5) node  [font=\footnotesize,color={rgb, 255:red, 0; green, 0; blue, 0 }  ,opacity=1 ]  {$r$};
\draw (229.82,97.5) node  [font=\footnotesize,color={rgb, 255:red, 0; green, 0; blue, 0 }  ,opacity=1 ]  {$f_{\mathsf{later}}$};
\draw (319.82,27.5) node  [font=\footnotesize,color={rgb, 255:red, 0; green, 0; blue, 0 }  ,opacity=1 ]  {$g_{\mathsf{now}}$};
\draw (352.41,35) node  [font=\footnotesize,color={rgb, 255:red, 0; green, 0; blue, 0 }  ,opacity=1 ]  {$Y_{0}$};
\draw (352.41,90) node  [font=\footnotesize,color={rgb, 255:red, 0; green, 0; blue, 0 }  ,opacity=1 ]  {$X^{+}$};
\draw (304.82,72.5) node  [font=\footnotesize,color={rgb, 255:red, 0; green, 0; blue, 0 }  ,opacity=1 ]  {$r$};
\draw (319.82,97.5) node  [font=\footnotesize,color={rgb, 255:red, 0; green, 0; blue, 0 }  ,opacity=1 ]  {$f_{\mathsf{later}}$};
\draw (274.82,62.5) node  [font=\footnotesize]  {$=$};
\draw (369.82,62.5) node  [font=\footnotesize]  {$;$};
\draw (334.14,8.5) node  [font=\footnotesize,color={rgb, 255:red, 0; green, 0; blue, 0 }  ,opacity=1 ]  {$X_{0}$};
\draw (332.23,115) node  [font=\footnotesize,color={rgb, 255:red, 0; green, 0; blue, 0 }  ,opacity=1 ]  {$Y^{+}$};
\draw (202.41,70) node  [font=\footnotesize,color={rgb, 255:red, 0; green, 0; blue, 0 }  ,opacity=1 ]  {$M_{f}$};
\draw (292.23,50) node  [font=\footnotesize,color={rgb, 255:red, 0; green, 0; blue, 0 }  ,opacity=1 ]  {$M_{g}$};

\end{tikzpicture}

     \caption{String diagrammatic definition of a monoidal stream.}
    \label{strings:monoidal-stream}
  \end{figure}
\end{rem}

\begin{rem}
  This definition of monoidal streams is the coinductive definition of the functor
  \[\STREAM \colon \NcatC^{op} \times \NcatC \to \Set.\]
  In principle, arbitrary final coalgebras do not need to exist. Moreover, it is usually difficult to explicitly construct such coalgebras~\cite{adamek74}. However, in \productive{} categories, this coalgebra does exist and is constructed by \observationalSequences{}.

  From now on, we reason \emph{coinductively}~\cite{kozen17}, a style particularly suited for all the following definitions. Coinduction is the dual to induction. It uses coalgebra to formalize the reasoning with non-well-founded structures \cite{park1981concurrency,honsell1981modelli,turi1998foundations,aczel1989final}, and it has been employed for many decades now for reasoning with infinite data in computer science \cite{rutten00,hermida98}. In practice, we can appeal to the \emph{coinduction hypothesis} after providing the first step. When constructing, we employ the universal property of the final coalgebra to uniquely extend any coalgebra morphism; when proving, we use the fact that two elements of the final coalgebra are equal when they are indistinguishable, or \emph{bisimilar} \cite{bloom1989remark,sangiorgi2009origins}.

  Coinductive reasoning will not only occur at the level of expressions but at the level of string diagrams. 
  Each string diagram of monoidal streams can be unfold repeatedly using the notation of \Cref{strings:monoidal-stream}.
  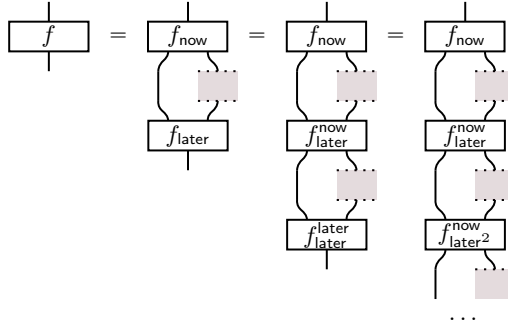
\begin{figure}[ht]

\tikzset{every picture/.style={line width=0.75pt}} %

\begin{tikzpicture}[x=0.75pt,y=0.75pt,yscale=-1,xscale=1]
\draw  [line width=0.75]  (5,15) -- (45,15) -- (45,30) -- (5,30) -- cycle ;
\draw [line width=0.75]    (25,5) -- (25,15) ;
\draw [line width=0.75]    (25,30) -- (25,40) ;
\draw  [color={rgb, 255:red, 0; green, 0; blue, 0 }  ,draw opacity=1 ][line width=0.75]  (75,15) -- (115,15) -- (115,30) -- (75,30) -- cycle ;
\draw [color={rgb, 255:red, 0; green, 0; blue, 0 }  ,draw opacity=1 ][line width=0.75]    (95.18,5) -- (95.18,15) ;
\draw [color={rgb, 255:red, 0; green, 0; blue, 0 }  ,draw opacity=1 ][line width=0.75]    (80.18,40) -- (80.18,55) ;
\draw [color={rgb, 255:red, 0; green, 0; blue, 0 }  ,draw opacity=1 ]   (85.18,30) .. controls (85.22,34.67) and (80.55,34.33) .. (80.18,40) ;
\draw [color={rgb, 255:red, 0; green, 0; blue, 0 }  ,draw opacity=1 ][line width=0.75]  [dash pattern={on 0.84pt off 2.51pt}]  (100.18,40) -- (120.18,40) ;
\draw [color={rgb, 255:red, 0; green, 0; blue, 0 }  ,draw opacity=1 ][line width=0.75]  [dash pattern={on 0.84pt off 2.51pt}]  (100.18,55) -- (120.18,55) ;
\draw  [line width=0.75]  (75.18,65) -- (115.18,65) -- (115.18,80) -- (75.18,80) -- cycle ;
\draw [line width=0.75]    (95.18,80) -- (95.18,90) ;
\draw    (110.18,55) .. controls (109.88,60.67) and (104.88,60) .. (105.18,65) ;
\draw [color={rgb, 255:red, 0; green, 0; blue, 0 }  ,draw opacity=1 ]   (105.18,30) .. controls (104.88,35.67) and (109.88,35) .. (110.18,40) ;
\draw [color={rgb, 255:red, 0; green, 0; blue, 0 }  ,draw opacity=1 ]   (80.18,55) .. controls (79.88,60.67) and (84.88,60) .. (85.18,65) ;
\draw  [draw opacity=0][fill={rgb, 255:red, 75; green, 17; blue, 33 }  ,fill opacity=0.15 ] (100.18,40) -- (120.18,40) -- (120.18,55) -- (100.18,55) -- cycle ;
\draw  [draw opacity=0][line width=0.75]  (45,10) -- (75,10) -- (75,35) -- (45,35) -- cycle ;
\draw  [draw opacity=0] (95.18,-5) -- (120,-5) -- (120,15) -- (95.18,15) -- cycle ;
\draw  [draw opacity=0] (95.36,80) -- (120.18,80) -- (120.18,100) -- (95.36,100) -- cycle ;
\draw  [draw opacity=0][line width=0.75]  (35,115) -- (65,115) -- (65,130) -- (35,130) -- cycle ;
\draw  [draw opacity=0][line width=0.75]  (115,10) -- (145,10) -- (145,35) -- (115,35) -- cycle ;
\draw  [color={rgb, 255:red, 0; green, 0; blue, 0 }  ,draw opacity=1 ][line width=0.75]  (145,15) -- (185,15) -- (185,30) -- (145,30) -- cycle ;
\draw [color={rgb, 255:red, 0; green, 0; blue, 0 }  ,draw opacity=1 ][line width=0.75]    (165.18,5) -- (165.18,15) ;
\draw [color={rgb, 255:red, 0; green, 0; blue, 0 }  ,draw opacity=1 ][line width=0.75]    (150.18,40) -- (150.18,55) ;
\draw [color={rgb, 255:red, 0; green, 0; blue, 0 }  ,draw opacity=1 ]   (155.18,30) .. controls (155.22,34.67) and (150.55,34.33) .. (150.18,40) ;
\draw [color={rgb, 255:red, 0; green, 0; blue, 0 }  ,draw opacity=1 ][line width=0.75]  [dash pattern={on 0.84pt off 2.51pt}]  (170.18,40) -- (190.18,40) ;
\draw [color={rgb, 255:red, 0; green, 0; blue, 0 }  ,draw opacity=1 ][line width=0.75]  [dash pattern={on 0.84pt off 2.51pt}]  (170.18,55) -- (190.18,55) ;
\draw  [line width=0.75]  (145.18,65) -- (185.18,65) -- (185.18,80) -- (145.18,80) -- cycle ;
\draw    (180.18,55) .. controls (179.88,60.67) and (174.88,60) .. (175.18,65) ;
\draw [color={rgb, 255:red, 0; green, 0; blue, 0 }  ,draw opacity=1 ]   (175.18,30) .. controls (174.88,35.67) and (179.88,35) .. (180.18,40) ;
\draw [color={rgb, 255:red, 0; green, 0; blue, 0 }  ,draw opacity=1 ]   (150.18,55) .. controls (149.88,60.67) and (154.88,60) .. (155.18,65) ;
\draw  [draw opacity=0][fill={rgb, 255:red, 75; green, 17; blue, 33 }  ,fill opacity=0.15 ] (170.18,40) -- (190.18,40) -- (190.18,55) -- (170.18,55) -- cycle ;
\draw [color={rgb, 255:red, 0; green, 0; blue, 0 }  ,draw opacity=1 ][line width=0.75]    (150.18,90) -- (150.18,105) ;
\draw [color={rgb, 255:red, 0; green, 0; blue, 0 }  ,draw opacity=1 ]   (155.18,80) .. controls (155.22,84.67) and (150.55,84.33) .. (150.18,90) ;
\draw [color={rgb, 255:red, 0; green, 0; blue, 0 }  ,draw opacity=1 ][line width=0.75]  [dash pattern={on 0.84pt off 2.51pt}]  (170.18,90) -- (190.18,90) ;
\draw [color={rgb, 255:red, 0; green, 0; blue, 0 }  ,draw opacity=1 ][line width=0.75]  [dash pattern={on 0.84pt off 2.51pt}]  (170.18,105) -- (190.18,105) ;
\draw    (180.18,105) .. controls (179.88,110.67) and (174.88,110) .. (175.18,115) ;
\draw [color={rgb, 255:red, 0; green, 0; blue, 0 }  ,draw opacity=1 ]   (175.18,80) .. controls (174.88,85.67) and (179.88,85) .. (180.18,90) ;
\draw [color={rgb, 255:red, 0; green, 0; blue, 0 }  ,draw opacity=1 ]   (150.18,105) .. controls (149.88,110.67) and (154.88,110) .. (155.18,115) ;
\draw  [draw opacity=0][fill={rgb, 255:red, 75; green, 17; blue, 33 }  ,fill opacity=0.15 ] (170.18,90) -- (190.18,90) -- (190.18,105) -- (170.18,105) -- cycle ;
\draw  [line width=0.75]  (145.18,115) -- (185.18,115) -- (185.18,130) -- (145.18,130) -- cycle ;
\draw [color={rgb, 255:red, 0; green, 0; blue, 0 }  ,draw opacity=1 ][line width=0.75]    (165.18,130) -- (165.18,140) ;
\draw  [draw opacity=0][line width=0.75]  (185,10) -- (215,10) -- (215,35) -- (185,35) -- cycle ;
\draw  [color={rgb, 255:red, 0; green, 0; blue, 0 }  ,draw opacity=1 ][line width=0.75]  (214.82,15) -- (254.82,15) -- (254.82,30) -- (214.82,30) -- cycle ;
\draw [color={rgb, 255:red, 0; green, 0; blue, 0 }  ,draw opacity=1 ][line width=0.75]    (235,5) -- (235,15) ;
\draw [color={rgb, 255:red, 0; green, 0; blue, 0 }  ,draw opacity=1 ][line width=0.75]    (220,40) -- (220,55) ;
\draw [color={rgb, 255:red, 0; green, 0; blue, 0 }  ,draw opacity=1 ]   (225,30) .. controls (225.03,34.67) and (220.37,34.33) .. (220,40) ;
\draw [color={rgb, 255:red, 0; green, 0; blue, 0 }  ,draw opacity=1 ][line width=0.75]  [dash pattern={on 0.84pt off 2.51pt}]  (240,40) -- (260,40) ;
\draw [color={rgb, 255:red, 0; green, 0; blue, 0 }  ,draw opacity=1 ][line width=0.75]  [dash pattern={on 0.84pt off 2.51pt}]  (240,55) -- (260,55) ;
\draw  [line width=0.75]  (215,65) -- (255,65) -- (255,80) -- (215,80) -- cycle ;
\draw    (250,55) .. controls (249.7,60.67) and (244.7,60) .. (245,65) ;
\draw [color={rgb, 255:red, 0; green, 0; blue, 0 }  ,draw opacity=1 ]   (245,30) .. controls (244.7,35.67) and (249.7,35) .. (250,40) ;
\draw [color={rgb, 255:red, 0; green, 0; blue, 0 }  ,draw opacity=1 ]   (220,55) .. controls (219.7,60.67) and (224.7,60) .. (225,65) ;
\draw  [draw opacity=0][fill={rgb, 255:red, 75; green, 17; blue, 33 }  ,fill opacity=0.15 ] (240,40) -- (260,40) -- (260,55) -- (240,55) -- cycle ;
\draw [color={rgb, 255:red, 0; green, 0; blue, 0 }  ,draw opacity=1 ][line width=0.75]    (220,90) -- (220,105) ;
\draw [color={rgb, 255:red, 0; green, 0; blue, 0 }  ,draw opacity=1 ]   (225,80) .. controls (225.03,84.67) and (220.37,84.33) .. (220,90) ;
\draw [color={rgb, 255:red, 0; green, 0; blue, 0 }  ,draw opacity=1 ][line width=0.75]  [dash pattern={on 0.84pt off 2.51pt}]  (240,90) -- (260,90) ;
\draw [color={rgb, 255:red, 0; green, 0; blue, 0 }  ,draw opacity=1 ][line width=0.75]  [dash pattern={on 0.84pt off 2.51pt}]  (240,105) -- (260,105) ;
\draw    (250,105) .. controls (249.7,110.67) and (244.7,110) .. (245,115) ;
\draw [color={rgb, 255:red, 0; green, 0; blue, 0 }  ,draw opacity=1 ]   (245,80) .. controls (244.7,85.67) and (249.7,85) .. (250,90) ;
\draw [color={rgb, 255:red, 0; green, 0; blue, 0 }  ,draw opacity=1 ]   (220,105) .. controls (219.7,110.67) and (224.7,110) .. (225,115) ;
\draw  [draw opacity=0][fill={rgb, 255:red, 75; green, 17; blue, 33 }  ,fill opacity=0.15 ] (240,90) -- (260,90) -- (260,105) -- (240,105) -- cycle ;
\draw  [line width=0.75]  (215,115) -- (255,115) -- (255,130) -- (215,130) -- cycle ;
\draw [color={rgb, 255:red, 0; green, 0; blue, 0 }  ,draw opacity=1 ][line width=0.75]    (220,140) -- (220,155) ;
\draw [color={rgb, 255:red, 0; green, 0; blue, 0 }  ,draw opacity=1 ]   (225,130) .. controls (225.03,134.67) and (220.37,134.33) .. (220,140) ;
\draw [color={rgb, 255:red, 0; green, 0; blue, 0 }  ,draw opacity=1 ][line width=0.75]  [dash pattern={on 0.84pt off 2.51pt}]  (240,140) -- (260,140) ;
\draw [color={rgb, 255:red, 0; green, 0; blue, 0 }  ,draw opacity=1 ]   (245,130) .. controls (244.7,135.67) and (249.7,135) .. (250,140) ;
\draw  [draw opacity=0][fill={rgb, 255:red, 75; green, 17; blue, 33 }  ,fill opacity=0.15 ] (240,140) -- (260,140) -- (260,155) -- (240,155) -- cycle ;

\draw (25,22.5) node  [font=\footnotesize]  {${\textstyle f}$};
\draw (95.18,22.5) node  [font=\footnotesize,color={rgb, 255:red, 0; green, 0; blue, 0 }  ,opacity=1 ]  {$f_{\mathsf{now}}$};
\draw (95.18,72.5) node  [font=\footnotesize,color={rgb, 255:red, 0; green, 0; blue, 0 }  ,opacity=1 ]  {$f_{\mathsf{later}}$};
\draw (60,22.5) node  [font=\footnotesize]  {$=$};
\draw (130,22.5) node  [font=\footnotesize]  {$=$};
\draw (165.18,22.5) node  [font=\footnotesize,color={rgb, 255:red, 0; green, 0; blue, 0 }  ,opacity=1 ]  {$f_{\mathsf{now}}$};
\draw (165.18,122.5) node  [font=\footnotesize,color={rgb, 255:red, 0; green, 0; blue, 0 }  ,opacity=1 ]  {$f_{\mathsf{later}}^{\mathsf{later}}$};
\draw (165.18,72.5) node  [font=\footnotesize,color={rgb, 255:red, 0; green, 0; blue, 0 }  ,opacity=1 ]  {$f_{\mathsf{later}}^{\mathsf{now}}$};
\draw (200,22.5) node  [font=\footnotesize]  {$=$};
\draw (235,22.5) node  [font=\footnotesize,color={rgb, 255:red, 0; green, 0; blue, 0 }  ,opacity=1 ]  {$f_{\mathsf{now}}$};
\draw (235,122.5) node  [font=\footnotesize,color={rgb, 255:red, 0; green, 0; blue, 0 }  ,opacity=1 ]  {$f_{\mathsf{later}^{2}}^{\mathsf{now}}$};
\draw (235,72.5) node  [font=\footnotesize,color={rgb, 255:red, 0; green, 0; blue, 0 }  ,opacity=1 ]  {$f_{\mathsf{later}}^{\mathsf{now}}$};
\draw (235,165) node  [font=\footnotesize,color={rgb, 255:red, 0; green, 0; blue, 0 }  ,opacity=1 ]  {$\dotsc $};

\end{tikzpicture}
     \caption{Coinductive unfolding of a string diagram.}
    \label{strings:unfold}
  \end{figure}

  Moreover, every time we introduce a new morphism, we can define it in terms of a string diagram as in \Cref{strings:monoidal-stream}: the second part of this string diagram can depend, coinductively, on the same definition we are writing. 
\end{rem}

\subsection{The Symmetric Monoidal Category of Streams}

The definitions for the operations of sequential and parallel composition are described in two steps.
We first define an operation that takes into account an extra \emph{memory channel} (see, for instance, \Cref{strings:monoidal-stream-sequential}); we use this extra generality to strengthen the \emph{coinduction hypothesis}.
We then define the desired operation as a particular case of this coinductively defined one.

\begin{defi}[Sequential composition]
  \label{def:sequentialstream}
  Given two streams $f \in \STREAM(\act{\sA}{\stream{X}},\stream{Y})$ and
  $g \in \STREAM(\act{\sB}{\stream{Y}},\stream{Z})$, we compute
  $(\Ncomp{f}{\sA}{g}{\sB}) \in\STREAM(\act{(\sA \tensor \sB)}{\stream{X}},\stream{Z})$,
  their \emph{sequential composition with memories $\sA$ and $\sB$}, as
  \begin{itemize}
    \item $M(\Ncomp{f}{\sA}{g}{\sB}) = M(f) \tensor M(g)$,
    \item $\now(\Ncomp{f}{\sA}{g}{\sB}) = \sigma \dcomp (\now(\fm) \tensor \im) \dcomp \sigma \dcomp (\now(\gm) \tensor \im)$,
    \item $\later(\Ncomp{f}{\sA}{g}{\sB}) = \Ncomp{\later(f)}{M(f)}{\later(g)}{M(g)}$.
  \end{itemize}
  We write $(f \dcomp g)$ for $(\Ncomp{f}{\monoidalunit}{g}{\monoidalunit}) \in \STREAM(\stream{X},\stream{Z})$; the \emph{sequential composition} of $f  \in \STREAM(\stream{X},\stream{Y})$ and $g \in \STREAM(\stream{Y},\stream{Z})$.
\end{defi}

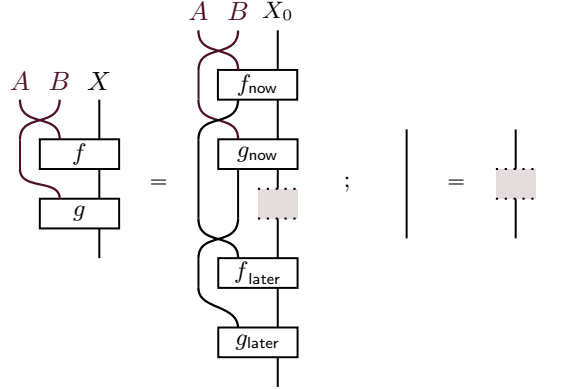
\begin{figure}[ht]

\tikzset{every picture/.style={line width=0.85pt}} %

\begin{tikzpicture}[x=0.75pt,y=0.75pt,yscale=-1,xscale=1]
\draw  [draw opacity=0][fill={rgb, 255:red, 75; green, 17; blue, 33 }  ,fill opacity=0.15 ] (170,105) -- (190,105) -- (190,120) -- (170,120) -- cycle ;
\draw  [line width=0.75]  (60,80) -- (100,80) -- (100,95) -- (60,95) -- cycle ;
\draw [line width=0.75]    (90,60) -- (90,80) ;
\draw [line width=0.75]    (90,95) -- (90,110) ;
\draw  [color={rgb, 255:red, 0; green, 0; blue, 0 }  ,draw opacity=1 ][line width=0.75]  (149.82,45) -- (189.82,45) -- (189.82,60) -- (149.82,60) -- cycle ;
\draw [color={rgb, 255:red, 0; green, 0; blue, 0 }  ,draw opacity=1 ][line width=0.75]    (180,25) -- (180,45) ;
\draw [color={rgb, 255:red, 0; green, 0; blue, 0 }  ,draw opacity=1 ][line width=0.75]  [dash pattern={on 0.84pt off 2.51pt}]  (170,120) -- (190,120) ;
\draw  [draw opacity=0][line width=0.75]  (105,95) -- (135,95) -- (135,110) -- (105,110) -- cycle ;
\draw  [line width=0.75]  (60,110) -- (100,110) -- (100,125) -- (60,125) -- cycle ;
\draw [line width=0.75]    (90,125) -- (90,140) ;
\draw [color={rgb, 255:red, 75; green, 17; blue, 33 }  ,draw opacity=1 ]   (50,60) .. controls (49.82,74.69) and (70.15,66.24) .. (70,80) ;
\draw [color={rgb, 255:red, 75; green, 17; blue, 33 }  ,draw opacity=1 ]   (70,60) .. controls (69.82,74.69) and (50.15,66.24) .. (50,80) ;
\draw [color={rgb, 255:red, 75; green, 17; blue, 33 }  ,draw opacity=1 ][line width=0.75]    (50,80) -- (50,95) ;
\draw [color={rgb, 255:red, 75; green, 17; blue, 33 }  ,draw opacity=1 ]   (50,95) .. controls (49.82,106.02) and (70.15,99.68) .. (70,110) ;
\draw [color={rgb, 255:red, 75; green, 17; blue, 33 }  ,draw opacity=1 ]   (140,25) .. controls (139.96,28.27) and (140.93,30.39) .. (142.47,31.86) .. controls (147.85,37.02) and (160.12,34.3) .. (160,44.99) ;
\draw [color={rgb, 255:red, 75; green, 17; blue, 33 }  ,draw opacity=1 ]   (160,25.09) .. controls (159.82,39.74) and (140.2,31.32) .. (140.05,45.04) ;
\draw  [color={rgb, 255:red, 0; green, 0; blue, 0 }  ,draw opacity=1 ][line width=0.75]  (150,80) -- (190,80) -- (190,95) -- (150,95) -- cycle ;
\draw [color={rgb, 255:red, 75; green, 17; blue, 33 }  ,draw opacity=1 ]   (140,60) .. controls (139.82,74.69) and (160.15,66.24) .. (160,80) ;
\draw [color={rgb, 255:red, 0; green, 0; blue, 0 }  ,draw opacity=1 ]   (160,60) .. controls (159.82,74.69) and (140.15,66.24) .. (140,80) ;
\draw [color={rgb, 255:red, 0; green, 0; blue, 0 }  ,draw opacity=1 ][line width=0.75]    (180,60) -- (180,80) ;
\draw [color={rgb, 255:red, 75; green, 17; blue, 33 }  ,draw opacity=1 ][line width=0.75]    (140.05,45.04) -- (140.05,60) ;
\draw [color={rgb, 255:red, 0; green, 0; blue, 0 }  ,draw opacity=1 ][line width=0.75]    (160,95) -- (160,120) ;
\draw  [draw opacity=0] (180.18,160) -- (205,160) -- (205,180) -- (180.18,180) -- cycle ;
\draw  [line width=0.75]  (150,140) -- (190,140) -- (190,155) -- (150,155) -- cycle ;
\draw [line width=0.75]    (180,120) -- (180,140) ;
\draw [line width=0.75]    (180,155) -- (180,175) ;
\draw  [line width=0.75]  (150,175) -- (190,175) -- (190,190) -- (150,190) -- cycle ;
\draw [line width=0.75]    (180,190) -- (180,205) ;
\draw [color={rgb, 255:red, 0; green, 0; blue, 0 }  ,draw opacity=1 ]   (140,120) .. controls (139.82,134.69) and (160.15,126.24) .. (160,140) ;
\draw [color={rgb, 255:red, 0; green, 0; blue, 0 }  ,draw opacity=1 ]   (160,120) .. controls (159.82,134.69) and (140.15,126.24) .. (140,140) ;
\draw [color={rgb, 255:red, 0; green, 0; blue, 0 }  ,draw opacity=1 ][line width=0.75]    (140,140) -- (140,155) ;
\draw [color={rgb, 255:red, 0; green, 0; blue, 0 }  ,draw opacity=1 ]   (140,155) .. controls (139.82,166.02) and (159.43,163.2) .. (160,175) ;
\draw [color={rgb, 255:red, 0; green, 0; blue, 0 }  ,draw opacity=1 ][line width=0.75]    (140,80) -- (140,120) ;
\draw [line width=0.75]    (180,95) -- (180,105) ;
\draw [color={rgb, 255:red, 0; green, 0; blue, 0 }  ,draw opacity=1 ][line width=0.75]  [dash pattern={on 0.84pt off 2.51pt}]  (170,105) -- (190,105) ;
\draw  [draw opacity=0][line width=0.75]  (200,95) -- (230,95) -- (230,110) -- (200,110) -- cycle ;
\draw [line width=0.75]    (245,75) -- (245,130) ;
\draw  [draw opacity=0][line width=0.75]  (255,95) -- (285,95) -- (285,110) -- (255,110) -- cycle ;
\draw  [draw opacity=0][fill={rgb, 255:red, 75; green, 17; blue, 33 }  ,fill opacity=0.15 ] (290,95) -- (310,95) -- (310,110) -- (290,110) -- cycle ;
\draw [color={rgb, 255:red, 0; green, 0; blue, 0 }  ,draw opacity=1 ][line width=0.75]  [dash pattern={on 0.84pt off 2.51pt}]  (290,110) -- (310,110) ;
\draw [line width=0.75]    (300,110) -- (300,130) ;
\draw [line width=0.75]    (300,75) -- (300,95) ;
\draw [color={rgb, 255:red, 0; green, 0; blue, 0 }  ,draw opacity=1 ][line width=0.75]  [dash pattern={on 0.84pt off 2.51pt}]  (290,95) -- (310,95) ;
\draw  [draw opacity=0][line width=0.75]  (315,95) -- (345,95) -- (345,110) -- (315,110) -- cycle ;

\draw (80,87.5) node  [font=\footnotesize]  {${\textstyle f}$};
\draw (120,102.5) node  [font=\footnotesize]  {$=$};
\draw (80,117.5) node  [font=\footnotesize]  {$g$};
\draw (170,147.5) node  [font=\footnotesize]  {${\textstyle f}_{\mathsf{later}}$};
\draw (170,182.5) node  [font=\footnotesize]  {$g_{\mathsf{later}}$};
\draw (169.82,52.5) node  [font=\footnotesize,color={rgb, 255:red, 0; green, 0; blue, 0 }  ,opacity=1 ]  {$f_{\mathsf{now}}$};
\draw (170,87.5) node  [font=\footnotesize,color={rgb, 255:red, 0; green, 0; blue, 0 }  ,opacity=1 ]  {$g_{\mathsf{now}}$};
\draw (50,56.6) node [anchor=south] [inner sep=0.75pt]  [font=\small,color={rgb, 255:red, 75; green, 17; blue, 33 }  ,opacity=1 ]  {$A$};
\draw (70,56.6) node [anchor=south] [inner sep=0.75pt]  [font=\small,color={rgb, 255:red, 75; green, 17; blue, 33 }  ,opacity=1 ]  {$B$};
\draw (90,56.6) node [anchor=south] [inner sep=0.75pt]  [font=\small,color={rgb, 255:red, 0; green, 0; blue, 0 }  ,opacity=1 ]  {$X$};
\draw (140,21.6) node [anchor=south] [inner sep=0.75pt]  [font=\small,color={rgb, 255:red, 75; green, 17; blue, 33 }  ,opacity=1 ]  {$A$};
\draw (160,21.6) node [anchor=south] [inner sep=0.75pt]  [font=\small,color={rgb, 255:red, 75; green, 17; blue, 33 }  ,opacity=1 ]  {$B$};
\draw (180,21.6) node [anchor=south] [inner sep=0.75pt]  [font=\footnotesize,color={rgb, 255:red, 0; green, 0; blue, 0 }  ,opacity=1 ]  {$X_{0}$};
\draw (215,102.5) node  [font=\footnotesize]  {$;$};
\draw (270,102.5) node  [font=\footnotesize]  {$=$};
\draw (330,102.5) node  [font=\footnotesize]  {$;$};

\end{tikzpicture}
   \caption{Sequential composition and identity of monoidal streams.}
  \label{strings:monoidal-stream-sequential}
\end{figure}

\begin{defi} The \emph{identity}
  $\im_{\stream{X}} \in \STREAM(\stream{X},\stream{X})$ is defined by $M(\im_{\stream{X}}) = \monoidalunit$,
  $\now(\im_{\stream{X}}) = \im_{X_{0}}$, and $\later(\im_{\stream{X}}) = \im_{\tail{\stream{X}}}$.
\end{defi}

\begin{lem}
  \label{lemma:sequentialcompositionassociativememories}
  Sequential composition of streams with memories (\Cref{def:sequentialstream}) is associative.
  Given three streams $f \in \STREAM(\sA \cdot \stream{X},\stream{Y})$, $g \in \STREAM(\sB \cdot \stream{Y},\stream{Z})$, and $h \in \STREAM(\sB \cdot \stream{Z},\stream{W})$;
  we can compose them in two different ways,
  \begin{itemize}
    \item $(f^{\sA} \comp g^{\sB}) \comp h^{\sC} \in \STREAM((\sA \tensor \sB) \tensor \sC \cdot \stream{X},\stream{W})$, or
    \item $f^{\sA} \comp (g^{\sB} \comp h^{\sC}) \in \STREAM(\sA \tensor (\sB \tensor \sC) \cdot \stream{X},\stream{W})$.
  \end{itemize}
  We claim that
  \[((f^{\sA} \comp g^{\sB}) \comp h^{\sC}) = 
  \alpha_{\sA,\sB,\sC} \cdot (f^{\sA} \comp (g^{\sB} \comp h^{\sC})).\]
\end{lem}
\begin{proof}
  First, we note that both sides of the equation represent streams with different memories.
  \begin{itemize}
    \item $M((f^{\sA} \comp g^{\sB}) \comp h^{\sC}) = 
      (M(f) \tensor M(g)) \tensor M(h)$,
    \item $M(f^{\sA} \comp (g^{\sB} \comp h^{\sC})) = M(f) \tensor (M(g) \tensor M(h))$.
  \end{itemize}
We will prove they are related by dinaturality over the associator $\alpha$.
We know that $\now((f^{\sA} \comp g^{\sB}) \comp h^{\sC}) = \now(f^{\sA} \comp (g^{\sB} \comp h^{\sC}))$ by string diagrams (see~\Cref{strings:monoidal-stream-assoc}).
Then, by coinduction, we know that
\[\begin{gathered}
  (\later(f)^{M(f)} ⨾ \later(g)^{M(g)}) ⨾ \later(h)^{M(h)} = \\
  \act{α}{(\later(f)^{M(f)} ⨾ (\later(g)^{M(g)} ⨾ \later(h)^{M(h)}))},
\end{gathered}\]
that is, $\later((f^{\sA} ⨾ g^{\sB}) ⨾ h^{\sC}) = 
α \cdot \later(f^{\sA} ⨾ (g^{\sB} ⨾ h^{\sC}))$.
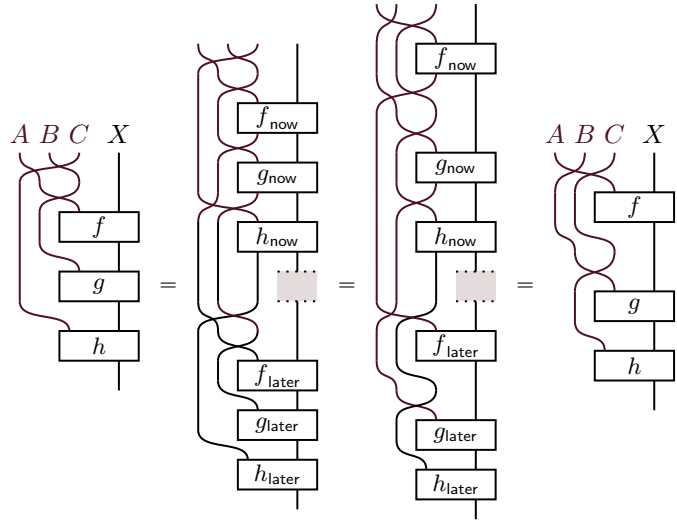
\begin{figure}[ht]

\tikzset{every picture/.style={line width=0.75pt}} %

\begin{tikzpicture}[x=0.75pt,y=0.75pt,yscale=-1,xscale=1]
\draw  [draw opacity=0][fill={rgb, 255:red, 75; green, 17; blue, 33 }  ,fill opacity=0.15 ] (170,140) -- (190,140) -- (190,155) -- (170,155) -- cycle ;
\draw  [line width=0.75]  (60,110) -- (100,110) -- (100,125) -- (60,125) -- cycle ;
\draw [line width=0.75]    (90,80) -- (90,110) ;
\draw [line width=0.75]    (90,125) -- (90,140) ;
\draw [color={rgb, 255:red, 0; green, 0; blue, 0 }  ,draw opacity=1 ][line width=0.75]  [dash pattern={on 0.84pt off 2.51pt}]  (170,155) -- (190,155) ;
\draw  [draw opacity=0][line width=0.75]  (100,140) -- (130,140) -- (130,155) -- (100,155) -- cycle ;
\draw  [line width=0.75]  (60,140) -- (100,140) -- (100,155) -- (60,155) -- cycle ;
\draw [line width=0.75]    (90,155) -- (90,170) ;
\draw [color={rgb, 255:red, 75; green, 17; blue, 33 }  ,draw opacity=1 ]   (70,80) .. controls (69.82,91.02) and (40.15,84.68) .. (40,95) ;
\draw [color={rgb, 255:red, 75; green, 17; blue, 33 }  ,draw opacity=1 ]   (40,80) .. controls (39.82,91.02) and (50.15,84.68) .. (50,95) ;
\draw [color={rgb, 255:red, 75; green, 17; blue, 33 }  ,draw opacity=1 ][line width=0.75]    (50,110) -- (50,125) ;
\draw [color={rgb, 255:red, 75; green, 17; blue, 33 }  ,draw opacity=1 ]   (50,125) .. controls (49.82,136.02) and (70.15,129.68) .. (70,140) ;
\draw [color={rgb, 255:red, 0; green, 0; blue, 0 }  ,draw opacity=1 ][line width=0.75]  [dash pattern={on 0.84pt off 2.51pt}]  (170,140) -- (190,140) ;
\draw [color={rgb, 255:red, 75; green, 17; blue, 33 }  ,draw opacity=1 ]   (55,80) .. controls (54.82,91.02) and (70.15,84.68) .. (70,95) ;
\draw [color={rgb, 255:red, 75; green, 17; blue, 33 }  ,draw opacity=1 ]   (50,95) .. controls (49.82,106.02) and (70.15,99.68) .. (70,110) ;
\draw [color={rgb, 255:red, 75; green, 17; blue, 33 }  ,draw opacity=1 ]   (70,95) .. controls (69.82,106.02) and (50.15,99.68) .. (50,110) ;
\draw  [line width=0.75]  (60,170) -- (100,170) -- (100,185) -- (60,185) -- cycle ;
\draw [line width=0.75]    (90,185) -- (90,200) ;
\draw [color={rgb, 255:red, 75; green, 17; blue, 33 }  ,draw opacity=1 ][line width=0.75]    (40,95) -- (40,155) ;
\draw [color={rgb, 255:red, 75; green, 17; blue, 33 }  ,draw opacity=1 ]   (40,155) .. controls (39.82,166.02) and (65.15,159.68) .. (65,170) ;
\draw  [line width=0.75]  (150,55) -- (190,55) -- (190,70) -- (150,70) -- cycle ;
\draw [line width=0.75]    (180,25) -- (180,55) ;
\draw [line width=0.75]    (180,70) -- (180,85) ;
\draw  [line width=0.75]  (150,85) -- (190,85) -- (190,100) -- (150,100) -- cycle ;
\draw [line width=0.75]    (180,100) -- (180,115) ;
\draw [color={rgb, 255:red, 75; green, 17; blue, 33 }  ,draw opacity=1 ]   (160,25) .. controls (159.82,36.02) and (130.15,29.68) .. (130,40) ;
\draw [color={rgb, 255:red, 75; green, 17; blue, 33 }  ,draw opacity=1 ]   (130,25) .. controls (129.82,36.02) and (140.15,29.68) .. (140,40) ;
\draw [color={rgb, 255:red, 75; green, 17; blue, 33 }  ,draw opacity=1 ][line width=0.75]    (140,55) -- (140,70) ;
\draw [color={rgb, 255:red, 75; green, 17; blue, 33 }  ,draw opacity=1 ]   (145,25) .. controls (144.82,36.02) and (160.15,29.68) .. (160,40) ;
\draw [color={rgb, 255:red, 75; green, 17; blue, 33 }  ,draw opacity=1 ]   (140,40) .. controls (139.82,51.02) and (160.15,44.68) .. (160,55) ;
\draw [color={rgb, 255:red, 75; green, 17; blue, 33 }  ,draw opacity=1 ]   (160,40) .. controls (159.82,51.02) and (140.15,44.68) .. (140,55) ;
\draw  [line width=0.75]  (150,115) -- (190,115) -- (190,130) -- (150,130) -- cycle ;
\draw [line width=0.75]    (180,130) -- (180,140) ;
\draw [color={rgb, 255:red, 75; green, 17; blue, 33 }  ,draw opacity=1 ][line width=0.75]    (130,40) -- (130,100) ;
\draw  [line width=0.75]  (150,185) -- (190,185) -- (190,200) -- (150,200) -- cycle ;
\draw [line width=0.75]    (180,155) -- (180,185) ;
\draw [line width=0.75]    (180,200) -- (180,210) ;
\draw  [line width=0.75]  (150,210) -- (190,210) -- (190,225) -- (150,225) -- cycle ;
\draw [line width=0.75]    (180,225) -- (180,235) ;
\draw [color={rgb, 255:red, 0; green, 0; blue, 0 }  ,draw opacity=1 ]   (160,155) .. controls (159.82,166.02) and (130.15,159.68) .. (130,170) ;
\draw [color={rgb, 255:red, 0; green, 0; blue, 0 }  ,draw opacity=1 ]   (130,155) .. controls (129.82,166.02) and (140.15,159.68) .. (140,170) ;
\draw [color={rgb, 255:red, 0; green, 0; blue, 0 }  ,draw opacity=1 ][line width=0.75]    (140,185) -- (140,195) ;
\draw [color={rgb, 255:red, 0; green, 0; blue, 0 }  ,draw opacity=1 ]   (140,195) .. controls (139.82,206.02) and (160.15,199.68) .. (160,210) ;
\draw [color={rgb, 255:red, 75; green, 17; blue, 33 }  ,draw opacity=1 ]   (140,155) .. controls (139.82,166.02) and (160.15,159.68) .. (160,170) ;
\draw [color={rgb, 255:red, 0; green, 0; blue, 0 }  ,draw opacity=1 ]   (140,170) .. controls (139.82,181.02) and (160.15,174.68) .. (160,185) ;
\draw [color={rgb, 255:red, 0; green, 0; blue, 0 }  ,draw opacity=1 ]   (160,170) .. controls (159.82,181.02) and (140.15,174.68) .. (140,185) ;
\draw  [line width=0.75]  (150,235) -- (190,235) -- (190,250) -- (150,250) -- cycle ;
\draw [line width=0.75]    (180,250) -- (180,260) ;
\draw [color={rgb, 255:red, 0; green, 0; blue, 0 }  ,draw opacity=1 ][line width=0.75]    (130,170) -- (130,220) ;
\draw [color={rgb, 255:red, 0; green, 0; blue, 0 }  ,draw opacity=1 ]   (130,220) .. controls (129.82,231.02) and (155.15,224.68) .. (155,235) ;
\draw [line width=0.75]    (160,130) -- (160,155) ;
\draw [color={rgb, 255:red, 75; green, 17; blue, 33 }  ,draw opacity=1 ][line width=0.75]    (140,85) -- (140,100) ;
\draw [line width=0.75]    (140,115) -- (140,155) ;
\draw [color={rgb, 255:red, 75; green, 17; blue, 33 }  ,draw opacity=1 ]   (140,70) .. controls (139.82,81.02) and (160.15,74.68) .. (160,85) ;
\draw [color={rgb, 255:red, 75; green, 17; blue, 33 }  ,draw opacity=1 ]   (160,70) .. controls (159.82,81.02) and (140.15,74.68) .. (140,85) ;
\draw [color={rgb, 255:red, 75; green, 17; blue, 33 }  ,draw opacity=1 ]   (130,100) .. controls (129.82,111.02) and (160.15,104.68) .. (160,115) ;
\draw [color={rgb, 255:red, 75; green, 17; blue, 33 }  ,draw opacity=1 ]   (160,100) .. controls (159.82,111.02) and (140.15,104.68) .. (140,115) ;
\draw [color={rgb, 255:red, 75; green, 17; blue, 33 }  ,draw opacity=1 ]   (140,100) .. controls (139.82,111.02) and (130.15,104.68) .. (130,115) ;
\draw [line width=0.75]    (130,115) -- (130,155) ;
\draw  [draw opacity=0][line width=0.75]  (190,140) -- (220,140) -- (220,155) -- (190,155) -- cycle ;
\draw  [draw opacity=0][line width=0.75]  (20,55) -- (60,55) -- (60,70) -- (20,70) -- cycle ;
\draw  [line width=0.75]  (240.01,170) -- (280.01,170) -- (280.01,185) -- (240.01,185) -- cycle ;
\draw [line width=0.75]    (270.01,155) -- (270.01,170) ;
\draw [line width=0.75]    (270.01,185) -- (270.01,215) ;
\draw  [line width=0.75]  (240.01,215) -- (280.01,215) -- (280.01,230) -- (240.01,230) -- cycle ;
\draw [line width=0.75]    (270.01,230) -- (270,240) ;
\draw [color={rgb, 255:red, 0; green, 0; blue, 0 }  ,draw opacity=1 ]   (250.01,155) .. controls (249.82,166.02) and (230.16,159.68) .. (230.01,170) ;
\draw [color={rgb, 255:red, 0; green, 0; blue, 0 }  ,draw opacity=1 ][line width=0.75]    (230.01,170) -- (230.01,185) ;
\draw [color={rgb, 255:red, 75; green, 17; blue, 33 }  ,draw opacity=1 ]   (230.01,200) .. controls (229.82,211.02) and (250.16,204.68) .. (250.01,215) ;
\draw [color={rgb, 255:red, 75; green, 17; blue, 33 }  ,draw opacity=1 ]   (230.01,155) .. controls (229.82,166.02) and (220.15,159.68) .. (220,170) ;
\draw  [line width=0.75]  (240.01,240) -- (280.01,240) -- (280.01,255) -- (240.01,255) -- cycle ;
\draw [line width=0.75]    (270.01,255) -- (270.01,265) ;
\draw [color={rgb, 255:red, 75; green, 17; blue, 33 }  ,draw opacity=1 ][line width=0.75]    (220,170) -- (220,185) ;
\draw [color={rgb, 255:red, 0; green, 0; blue, 0 }  ,draw opacity=1 ]   (230.01,225) .. controls (229.82,236.02) and (245.16,229.68) .. (245.01,240) ;
\draw [color={rgb, 255:red, 75; green, 17; blue, 33 }  ,draw opacity=1 ]   (220.01,155) .. controls (219.82,166.02) and (250.16,159.68) .. (250.01,170) ;
\draw [color={rgb, 255:red, 75; green, 17; blue, 33 }  ,draw opacity=1 ]   (220,185) .. controls (219.82,196.02) and (230.16,189.68) .. (230.01,200) ;
\draw [color={rgb, 255:red, 0; green, 0; blue, 0 }  ,draw opacity=1 ]   (230.01,185) .. controls (229.82,196.02) and (250.16,189.68) .. (250.01,200) ;
\draw [color={rgb, 255:red, 0; green, 0; blue, 0 }  ,draw opacity=1 ]   (250.01,200) .. controls (249.82,211.02) and (230.16,204.68) .. (230.01,215) ;
\draw [color={rgb, 255:red, 0; green, 0; blue, 0 }  ,draw opacity=1 ][line width=0.75]    (230.01,215) -- (230.01,225) ;
\draw  [draw opacity=0][fill={rgb, 255:red, 75; green, 17; blue, 33 }  ,fill opacity=0.15 ] (260,140) -- (280,140) -- (280,155) -- (260,155) -- cycle ;
\draw [color={rgb, 255:red, 0; green, 0; blue, 0 }  ,draw opacity=1 ][line width=0.75]  [dash pattern={on 0.84pt off 2.51pt}]  (260,155) -- (280,155) ;
\draw [color={rgb, 255:red, 0; green, 0; blue, 0 }  ,draw opacity=1 ][line width=0.75]  [dash pattern={on 0.84pt off 2.51pt}]  (260,140) -- (280,140) ;
\draw  [draw opacity=0][line width=0.75]  (280,140) -- (310,140) -- (310,155) -- (280,155) -- cycle ;
\draw  [line width=0.75]  (240,25) -- (280,25) -- (280,40) -- (240,40) -- cycle ;
\draw [line width=0.75]    (270,5) -- (270,25) ;
\draw [line width=0.75]    (270,40) -- (270,80) ;
\draw  [line width=0.75]  (240.01,80) -- (280.01,80) -- (280.01,95) -- (240.01,95) -- cycle ;
\draw [line width=0.75]    (270.01,95) -- (270,115) ;
\draw [color={rgb, 255:red, 75; green, 17; blue, 33 }  ,draw opacity=1 ]   (250,5) .. controls (249.82,19.69) and (230.15,11.24) .. (230,25) ;
\draw [color={rgb, 255:red, 75; green, 17; blue, 33 }  ,draw opacity=1 ][line width=0.75]    (230,25) -- (230,40) ;
\draw [color={rgb, 255:red, 75; green, 17; blue, 33 }  ,draw opacity=1 ]   (230.01,60) .. controls (229.82,74.69) and (250.16,66.24) .. (250.01,80) ;
\draw [color={rgb, 255:red, 75; green, 17; blue, 33 }  ,draw opacity=1 ]   (230,5) .. controls (229.82,19.69) and (220.15,11.24) .. (220,25) ;
\draw  [line width=0.75]  (240.01,115) -- (280.01,115) -- (280.01,130) -- (240.01,130) -- cycle ;
\draw [line width=0.75]    (270.01,130) -- (270,140) ;
\draw [color={rgb, 255:red, 75; green, 17; blue, 33 }  ,draw opacity=1 ][line width=0.75]    (220,25) -- (220,40) ;
\draw [color={rgb, 255:red, 75; green, 17; blue, 33 }  ,draw opacity=1 ]   (230.01,95) .. controls (229.82,109.69) and (250.15,101.24) .. (250,115) ;
\draw [color={rgb, 255:red, 75; green, 17; blue, 33 }  ,draw opacity=1 ]   (220,5) .. controls (219.82,19.69) and (250.15,11.24) .. (250,25) ;
\draw [color={rgb, 255:red, 75; green, 17; blue, 33 }  ,draw opacity=1 ]   (220,40) .. controls (219.82,54.69) and (230.16,46.24) .. (230.01,60) ;
\draw [color={rgb, 255:red, 75; green, 17; blue, 33 }  ,draw opacity=1 ]   (230.01,40) .. controls (229.82,54.69) and (250.16,46.24) .. (250.01,60) ;
\draw [color={rgb, 255:red, 75; green, 17; blue, 33 }  ,draw opacity=1 ]   (250.01,60) .. controls (249.82,74.69) and (230.16,66.24) .. (230.01,80) ;
\draw [color={rgb, 255:red, 75; green, 17; blue, 33 }  ,draw opacity=1 ][line width=0.75]    (230.01,80) -- (230.01,95) ;
\draw [color={rgb, 255:red, 75; green, 17; blue, 33 }  ,draw opacity=1 ]   (250,40) .. controls (249.82,54.69) and (220.15,46.24) .. (220,60) ;
\draw [color={rgb, 255:red, 75; green, 17; blue, 33 }  ,draw opacity=1 ][line width=0.75]    (220,60) -- (220.01,155) ;
\draw [color={rgb, 255:red, 75; green, 17; blue, 33 }  ,draw opacity=1 ]   (250,95) .. controls (249.82,109.69) and (230.15,101.24) .. (230,115) ;
\draw [color={rgb, 255:red, 75; green, 17; blue, 33 }  ,draw opacity=1 ][line width=0.75]    (230,115) -- (230,155) ;
\draw [line width=0.75]    (250,130) -- (250.01,155) ;
\draw  [draw opacity=0][line width=0.75]  (280,140) -- (310,140) -- (310,155) -- (280,155) -- cycle ;
\draw  [line width=0.75]  (330,100) -- (370,100) -- (370,115) -- (330,115) -- cycle ;
\draw [line width=0.75]    (360,80) -- (360,100) ;
\draw [line width=0.75]    (360,115) -- (360,150) ;
\draw  [line width=0.75]  (330,150) -- (370,150) -- (370,165) -- (330,165) -- cycle ;
\draw [line width=0.75]    (360,165) -- (360,180) ;
\draw [color={rgb, 255:red, 75; green, 17; blue, 33 }  ,draw opacity=1 ]   (340,80) .. controls (339.82,94.69) and (320.15,86.24) .. (320,100) ;
\draw [color={rgb, 255:red, 75; green, 17; blue, 33 }  ,draw opacity=1 ][line width=0.75]    (320,100) -- (320,115) ;
\draw [color={rgb, 255:red, 75; green, 17; blue, 33 }  ,draw opacity=1 ]   (320,130) .. controls (319.82,144.69) and (340.15,136.24) .. (340,150) ;
\draw [color={rgb, 255:red, 75; green, 17; blue, 33 }  ,draw opacity=1 ]   (325,80) .. controls (324.82,94.69) and (310.15,86.24) .. (310,100) ;
\draw  [line width=0.75]  (330,180) -- (370,180) -- (370,195) -- (330,195) -- cycle ;
\draw [line width=0.75]    (360,195) -- (360,210) ;
\draw [color={rgb, 255:red, 75; green, 17; blue, 33 }  ,draw opacity=1 ][line width=0.75]    (310,100) -- (310,115) ;
\draw [color={rgb, 255:red, 75; green, 17; blue, 33 }  ,draw opacity=1 ]   (320,165) .. controls (319.82,176.02) and (335.15,169.68) .. (335,180) ;
\draw  [draw opacity=0][line width=0.75]  (295,60) -- (380,60) -- (380,80) -- (295,80) -- cycle ;
\draw [color={rgb, 255:red, 75; green, 17; blue, 33 }  ,draw opacity=1 ]   (310,80) .. controls (309.82,94.69) and (340.15,86.24) .. (340,100) ;
\draw [color={rgb, 255:red, 75; green, 17; blue, 33 }  ,draw opacity=1 ]   (310,115) .. controls (309.82,126.02) and (320.15,119.68) .. (320,130) ;
\draw [color={rgb, 255:red, 75; green, 17; blue, 33 }  ,draw opacity=1 ]   (320,115) .. controls (319.82,126.02) and (340.15,119.68) .. (340,130) ;
\draw [color={rgb, 255:red, 75; green, 17; blue, 33 }  ,draw opacity=1 ]   (340,130) .. controls (339.82,144.69) and (320.15,136.24) .. (320,150) ;
\draw [color={rgb, 255:red, 75; green, 17; blue, 33 }  ,draw opacity=1 ][line width=0.75]    (320,150) -- (320,165) ;

\draw (80,117.5) node  [font=\footnotesize]  {${\textstyle f}$};
\draw (115,147.5) node  [font=\footnotesize]  {$=$};
\draw (80,147.5) node  [font=\footnotesize]  {$g$};
\draw (170,192.5) node  [font=\footnotesize]  {${\textstyle f}_{\mathsf{later}}$};
\draw (40,76.6) node [anchor=south] [inner sep=0.75pt]  [font=\small,color={rgb, 255:red, 75; green, 17; blue, 33 }  ,opacity=1 ]  {$A$};
\draw (55,76.6) node [anchor=south] [inner sep=0.75pt]  [font=\small,color={rgb, 255:red, 75; green, 17; blue, 33 }  ,opacity=1 ]  {$B$};
\draw (90,76.6) node [anchor=south] [inner sep=0.75pt]  [font=\small,color={rgb, 255:red, 0; green, 0; blue, 0 }  ,opacity=1 ]  {$X$};
\draw (70,76.6) node [anchor=south] [inner sep=0.75pt]  [font=\small,color={rgb, 255:red, 75; green, 17; blue, 33 }  ,opacity=1 ]  {$C$};
\draw (80,177.5) node  [font=\footnotesize]  {$h$};
\draw (205,147.5) node  [font=\footnotesize]  {$=$};
\draw (170,217.5) node  [font=\footnotesize]  {$g_{\mathsf{later}}$};
\draw (170,242.5) node  [font=\footnotesize]  {$h_{\mathsf{later}}$};
\draw (260.01,177.5) node  [font=\footnotesize]  {${\textstyle f}_{\mathsf{later}}$};
\draw (260.01,222.5) node  [font=\footnotesize]  {$g_{\mathsf{later}}$};
\draw (260.01,247.5) node  [font=\footnotesize]  {$h_{\mathsf{later}}$};
\draw (295,147.5) node  [font=\footnotesize]  {$=$};
\draw (350,107.5) node  [font=\footnotesize]  {${\textstyle f}$};
\draw (350,157.5) node  [font=\footnotesize]  {$g$};
\draw (310,76.6) node [anchor=south] [inner sep=0.75pt]  [font=\small,color={rgb, 255:red, 75; green, 17; blue, 33 }  ,opacity=1 ]  {$A$};
\draw (325,76.6) node [anchor=south] [inner sep=0.75pt]  [font=\small,color={rgb, 255:red, 75; green, 17; blue, 33 }  ,opacity=1 ]  {$B$};
\draw (360,76.6) node [anchor=south] [inner sep=0.75pt]  [font=\small,color={rgb, 255:red, 0; green, 0; blue, 0 }  ,opacity=1 ]  {$X$};
\draw (340,76.6) node [anchor=south] [inner sep=0.75pt]  [font=\small,color={rgb, 255:red, 75; green, 17; blue, 33 }  ,opacity=1 ]  {$C$};
\draw (350,187.5) node  [font=\footnotesize]  {$h$};
\draw (170,62.5) node  [font=\footnotesize]  {${\textstyle f}_{\mathsf{now}}$};
\draw (170,92.5) node  [font=\footnotesize]  {$g_{\mathsf{now}}$};
\draw (170,122.5) node  [font=\footnotesize]  {$h_{\mathsf{now}}$};
\draw (260,32.5) node  [font=\footnotesize]  {${\textstyle f}_{\mathsf{now}}$};
\draw (260.01,87.5) node  [font=\footnotesize]  {$g_{\mathsf{now}}$};
\draw (260.01,122.5) node  [font=\footnotesize]  {$h_{\mathsf{now}}$};

\end{tikzpicture}
   \caption{Associativity of sequential composition.}
  \label{strings:monoidal-stream-assoc}
\end{figure}
\end{proof}

\begin{lem}
  \label{lemma:sequentialcompositionassociative}
  Sequential composition of streams (\Cref{def:sequentialstream}) is associative.
  Given three streams
  $f \in \STREAM(\stream{X},\stream{Y})$,
  $g \in \STREAM(\stream{Y},\stream{Z})$,
  and $h \in \STREAM(\stream{Z},\stream{W})$;
  we claim that $$((f ⨾ g) ⨾ h) = (f ⨾ (g ⨾ h)).$$
\end{lem}
\begin{proof}
  Direct consequence of~\Cref{lemma:sequentialcompositionassociativememories},
  after considering the appropriate coherence morphisms.
\end{proof}

\begin{defi}[Parallel composition]\label{def:parallelstream}
  Given two streams $f \in \STREAM(\act{\sA}{\stream{X}}, \stream{Y})$
  and $g \in \STREAM(\act{\sB}{\stream{X}'},\stream{Y}')$, we compute $(\Ntensor{f}{\sA}{g}{\sB}) \in \STREAM(\act{(\sA \tensor \sB)}{(\stream{X} \tensor \stream{X'})},\stream{Y}  \tensor \stream{Y'})$, their \emph{parallel composition
  with memories $\sA$ and $\sB$}, as
  \begin{itemize}
    \item $M(\Ntensor{f}{\sA}{g}{\sB}) = M(f) ⊗ M(g)$,
    \item $\now(\Ntensor{f}{\sA}{g}{\sB}) = 
      \sigma ⨾ (\now(f) ⊗ \now(g)) ⨾ \sigma$,
    \item $\later(\Ntensor{f}{\sA}{g}{\sB}) = 
      \Ntensor{\later(f)}{M(f)}{\later(g)}{M(g)}$.
  \end{itemize}
  We write $(f ⊗ g)$ for 
  $(\Ntensor{f}{\monoidalunit}{g}{\monoidalunit}) \in 
  \STREAM(\stream{X} ⊗ \stream{X'},\stream{Y}  \tensor \stream{Y'})$; 
  we call it the \emph{parallel composition} of $f \in \STREAM(\stream{X},\stream{Y})$ and $g \in \STREAM(\stream{X}',\stream{Y}')$.
  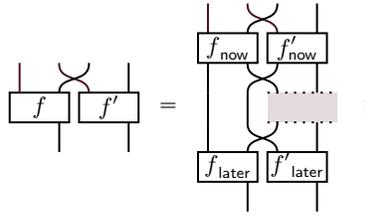
\begin{figure}[ht]

\tikzset{every picture/.style={line width=0.75pt}} %

\begin{tikzpicture}[x=0.75pt,y=0.75pt,yscale=-1,xscale=1]
\draw  [line width=0.75]  (15,55) -- (45,55) -- (45,70) -- (15,70) -- cycle ;
\draw [color={rgb, 255:red, 75; green, 17; blue, 33 }  ,draw opacity=1 ]   (40,40) .. controls (39.15,50.13) and (54.75,45.73) .. (55,55) ;
\draw  [line width=0.75]  (50,55) -- (80,55) -- (80,70) -- (50,70) -- cycle ;
\draw [line width=0.75]    (75,40) -- (75,55) ;
\draw [line width=0.75]    (75,70) -- (75,85) ;
\draw [line width=0.75]    (40,70) -- (40,85) ;
\draw [color={rgb, 255:red, 0; green, 0; blue, 0 }  ,draw opacity=1 ]   (55,40) .. controls (55.35,50.73) and (39.75,45.73) .. (40,55) ;
\draw [color={rgb, 255:red, 75; green, 17; blue, 33 }  ,draw opacity=1 ][line width=0.75]    (20,40) -- (20,55) ;
\draw  [draw opacity=0][line width=0.75]  (80,55) -- (110,55) -- (110,70) -- (80,70) -- cycle ;
\draw [color={rgb, 255:red, 75; green, 17; blue, 33 }  ,draw opacity=1 ]   (135,10) .. controls (134.15,20.13) and (149.75,15.73) .. (150,25) ;
\draw [line width=0.75]    (170,10) -- (170,25) ;
\draw [color={rgb, 255:red, 0; green, 0; blue, 0 }  ,draw opacity=1 ]   (150,10) .. controls (150.35,20.73) and (134.75,15.73) .. (135,25) ;
\draw [color={rgb, 255:red, 75; green, 17; blue, 33 }  ,draw opacity=1 ][line width=0.75]    (115,10) -- (115,25) ;
\draw  [line width=0.75]  (110,25) -- (140,25) -- (140,40) -- (110,40) -- cycle ;
\draw  [line width=0.75]  (145,25) -- (175,25) -- (175,40) -- (145,40) -- cycle ;
\draw [color={rgb, 255:red, 0; green, 0; blue, 0 }  ,draw opacity=1 ]   (135.03,40) .. controls (134.18,50.13) and (149.78,45.73) .. (150.03,55) ;
\draw [color={rgb, 255:red, 0; green, 0; blue, 0 }  ,draw opacity=1 ]   (150.03,40) .. controls (150.38,50.73) and (134.78,45.73) .. (135.03,55) ;
\draw [line width=0.75]    (170,40) -- (170,55) ;
\draw [line width=0.75]    (115,40) -- (115,55) ;
\draw  [draw opacity=0][fill={rgb, 255:red, 75; green, 17; blue, 33 }  ,fill opacity=0.15 ] (145,55) -- (180,55) -- (180,70) -- (145,70) -- cycle ;
\draw [color={rgb, 255:red, 0; green, 0; blue, 0 }  ,draw opacity=1 ][line width=0.75]  [dash pattern={on 0.84pt off 2.51pt}]  (145,70) -- (180,70) ;
\draw [color={rgb, 255:red, 0; green, 0; blue, 0 }  ,draw opacity=1 ][line width=0.75]  [dash pattern={on 0.84pt off 2.51pt}]  (145,55) -- (180,55) ;
\draw  [line width=0.75]  (110,85) -- (140,85) -- (140,100) -- (110,100) -- cycle ;
\draw [color={rgb, 255:red, 0; green, 0; blue, 0 }  ,draw opacity=1 ]   (135,70) .. controls (134.15,80.13) and (149.75,75.73) .. (150,85) ;
\draw  [line width=0.75]  (145,85) -- (175,85) -- (175,100) -- (145,100) -- cycle ;
\draw [line width=0.75]    (170,70) -- (170,85) ;
\draw [line width=0.75]    (170,100) -- (170,115) ;
\draw [line width=0.75]    (135,100) -- (135,115) ;
\draw [color={rgb, 255:red, 0; green, 0; blue, 0 }  ,draw opacity=1 ]   (150,70) .. controls (150.35,80.73) and (134.75,75.73) .. (135,85) ;
\draw [color={rgb, 255:red, 0; green, 0; blue, 0 }  ,draw opacity=1 ][line width=0.75]    (115,70) -- (115,85) ;
\draw [line width=0.75]    (135,55) -- (135,70) ;
\draw [line width=0.75]    (115,55) -- (115,70) ;
\draw  [draw opacity=0][line width=0.75]  (180,55) -- (210,55) -- (210,70) -- (180,70) -- cycle ;

\draw (30,62.5) node  [font=\footnotesize]  {${\textstyle f}$};
\draw (65,62.5) node  [font=\footnotesize]  {${\textstyle f'}$};
\draw (95,62.5) node  [font=\footnotesize]  {$=$};
\draw (125,32.5) node  [font=\footnotesize]  {${\textstyle f}_{\mathsf{now}}$};
\draw (160,32.5) node  [font=\footnotesize]  {$f'_{\mathsf{now}}$};
\draw (125,92.5) node  [font=\footnotesize]  {${\textstyle f}_{\mathsf{later}}$};
\draw (160,91.5) node  [font=\footnotesize]  {${\textstyle f'}_{\mathsf{later}}$};
\draw (195,62.5) node  [font=\footnotesize]  {$;$};

\end{tikzpicture}
     \caption{Parallel composition.}
    \label{strings:parallel}
  \end{figure}
\end{defi}

\begin{lem}\label{lemma:parfunassociativememories}
  Parallel composition of streams with memories is functorial with regards to sequential composition of streams with memories.
  Given four streams
  $f \in \STREAM(\sA \cdot \stream{X},\stream{Y})$,
  $f' \in \STREAM(\sA' \cdot \stream{X}',\stream{Y}')$,
  $g \in \STREAM(\sB \cdot \stream{Y},\stream{Z})$, and
  $g' \in \STREAM(\sB' \cdot \stream{Y}',\stream{Z}')$,
  we can compose them in two different ways,
  \begin{itemize}
    \item $\Ncomp{(\Ntensor{f}{A}{f'}{A'})}{(A \tensor A')}{(\Ntensor{g}{B}{g'}{B'})}{(B \tensor B')}$, and
    \item $\Ntensor{(\Ncomp{f}{A}{g}{B})}{(A \tensor B)}{(\Ncomp{f'}{A'}{g'}{B'})}{(A' \tensor B')}$,
  \end{itemize}
  having slightly different types, respectively,
  \begin{itemize}
    \item $\STREAM(\act{(A \tensor A') \tensor (B \tensor B')}{\stream{X} \tensor \stream{X}'},\stream{Z} \tensor \stream{Z}')$, and
    \item $\STREAM(\act{(A \tensor B) \tensor (A' \tensor B')}{\stream{X} \tensor \stream{X}'},\stream{Z} \tensor \stream{Z}')$.
  \end{itemize}
  We claim that
  \[\begin{gathered}
      \Ncomp{(\Ntensor{f}{A}{f'}{A'})}{(A \tensor A')}{(\Ntensor{g}{B}{g'}{B'})}{(B \tensor B')} =
      \act{\sigma_{A',B}}{\Ntensor{(\Ncomp{f}{A}{g}{B})}{(A \tensor B)}{(\Ncomp{f'}{A'}{g'}{B'})}{(A' \tensor B')}}.
    \end{gathered}\]
\end{lem}
\begin{proof}
  First, we note that both sides of the equation (which, from now on, we call $LHS$ and $RHS$, respectively) represent strams with different memories.
  \begin{itemize}
    \item $M(LHS) = (M(f) \tensor M(f')) \tensor (M(g) \tensor M(g'))$,
    \item $M(RHS) = (M(f) \tensor M(g)) \tensor (M(f') \tensor M(g'))$.
  \end{itemize}
  We will prove they are related by dinaturality over the symmetry $\sigma$.
  We know that $\now(LHS) ⨾ \sigma = \now(RHS)$ by string diagrams (see \Cref{strings:functorpar}). Then, by coinduction, we know that $\later(LHS) = \act{\sigma}{\later(RHS)}$.
\end{proof}

\begin{figure}
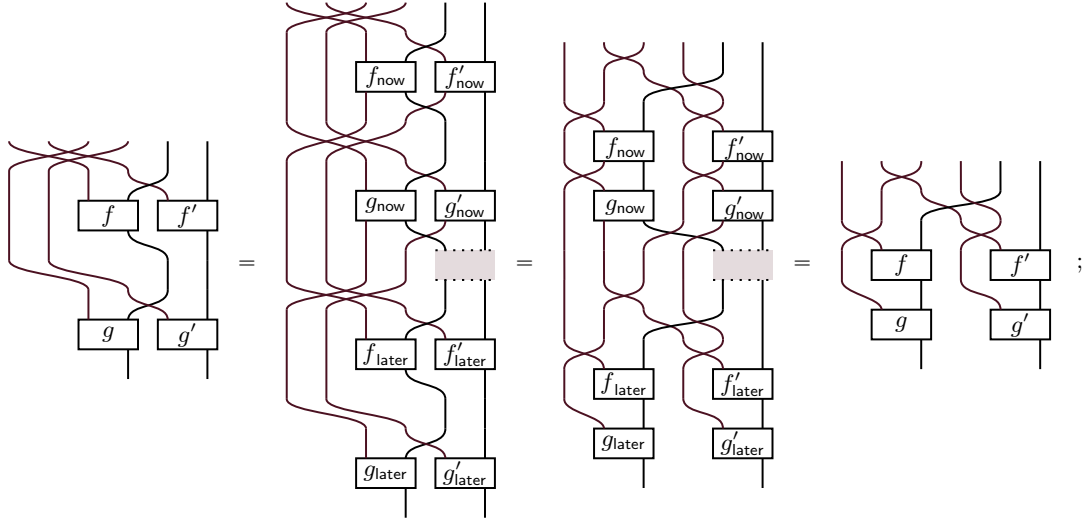


\tikzset{every picture/.style={line width=0.75pt}} %


   \caption{Functoriality of parallel composition.}
  \label{strings:functorpar}
\end{figure}

\begin{lem}\label{lemma:parfun}
  Parallel composition of streams is functorial with respect to sequential composition of streams.
  Given four streams
   $f \in \STREAM(\stream{X},\stream{Y})$,
   $f' \in \STREAM(\stream{X}',\stream{Y}')$,
   $g \in \STREAM(\stream{Y},\stream{Z})$, and
   $g' \in \STREAM(\stream{Y},\stream{Z})$;
  we claim that $(f ⊗ f') ⨾ (g ⊗ g') = (f ⨾ g) ⊗ (f' ⨾ g')$.
\end{lem}
\begin{proof}
  Direct consequence of~\Cref{lemma:parfunassociativememories},
  after considering the appropriate coherence morphisms.
\end{proof}

\begin{defi}[Memoryless and constant streams]
  \label{def:inclusion}
  Each sequence \(\stream{f} = \streamExpr{f}\), with $f_{n} \colon X_{n} \to Y_{n}$, induces a stream $[\stream{f}] \in \STREAM(\stream{X}, \stream{Y})$ defined by $M([\stream{f}]) \defn \monoidalunit$, $\now([\stream{f}]) \defn f_{0}$, and $\later([\stream{f}]) \defn [\lift{\tail{\stream{f}}}]$.
  Streams of this form are called \emph{memoryless}, i.e. their memories are given by the monoidal unit.

  Moreover, each morphism $f_0 \colon X \to Y$ induces a \emph{constant} memoryless stream that we call
  $f \in \STREAM(X,Y)$, defined by $M(\lift{f}) \defn \monoidalunit$, $\now(\lift{f}) \defn f_{0}$, and $\later(\lift{f}) \defn \lift{f}$.
\end{defi}

\begin{thm}[see~\cite{roman2020}]
  \label{th:category}\label{th:monoidalstreamscategory}
  Monoidal streams over a \productive{} symmetric monoidal category $(\catC, \tensor, \monoidalunit)$ form a \symmetricMonoidalCategory{} $\STREAM$ with a symmetric monoidal identity-on-objects functor from $\NcatC$.
\end{thm}
\begin{proof}
  Sequential composition of streams (\Cref{def:sequentialstream}) is associative (\Cref{lemma:sequentialcompositionassociative}) and unital with respect to identities.
  Parallel composition is bifunctorial with respect to sequential composition (\Cref{lemma:parfun}); this determines a bifunctor, which is the tensor of the monoidal category.
  The coherence morphisms and the symmetry can be included from sets, so they still satisfy the pentagon and triangle equations.
\end{proof}

The construction of monoidal streams preserves semicartesianity.
This is another way of saying that silent processes are equivalent to doing nothing (\Cref{figure:silentwalk}).

\begin{prop}\label{prop:semicartesian-streams}
  Monoidal streams over a \productive{} semicartesian category \(\cat{C}\) form a semicartesian category. 
\end{prop}
\begin{proof}
  We will show by coinduction that the counit stream is natural, i.e.~\(f \dcomp [\coUnit_{Y_{n}}] = \act{\coUnit_{M}}{[\coUnit_{X_{n}}]}\) for any stream with memory $f \in \STREAM(\act{M}{\stream{X}},\stream{Y})$.
  We will prove that they are related by dinaturality over the discard map, $\coUnit_{M_{0}}$.
  We can see that the first actions are related by semicartesianity,
  \[
      \now(f \dcomp [\coUnit_{Y_{n}}]) \dcomp \coUnit_{M_{0}} %
     = \now(f) \dcomp (\coUnit_{M_{0}} \tensor \coUnit_{Y_{0}}) %
    = \coUnit_{M} \tensor \coUnit_{X_{0}} %
    = \now(\act{\coUnit_{M}}{[\coUnit_{X_{n}}]}).\] %
  We can see that the rest of the actions are related by coinduction.
  \[\later(f \dcomp [\coUnit_{Y_{n}}])
  = \later(f) \dcomp [\coUnit_{Y_{n+1}}] 
  = \act{\coUnit_{M_{0}}}{[\coUnit_{X_{n+1}}]}
  = \act{\coUnit_{M_{0}}}{\later([\coUnit_{X_{n}}])}.\]
  This shows \(f \dcomp [\coUnit_{Y_{n}}] = \act{\coUnit_{M}}{[\coUnit_{X_{n}}]}\), by definition of equality in monoidal streams (\Cref{def:monoidalstream}).
\end{proof}

\subsection{Delayed feedback for streams}

Monoidal streams form a \emph{delayed feedback} monoidal category.
Given some stream in $\STREAM(\delay \stream{S} ⊗ \stream{X}, \stream{S} ⊗ \stream{Y})$, we can create a new stream in $\STREAM(\stream{X}, \stream{Y})$ that passes the output in $\stream{S}$ as a memory channel that gets used as the input in $\delay\stream{S}$.
As a consequence, the category of monoidal streams has a graphical calculus given by that of \categoriesWithFeedback{}. This graphical calculus is complete for dinatural equivalence (as we saw in \Cref{th:ext-stateful-sequences}).

\begin{defi}[Delay functor]\label{def:delay-fun-stream}
  The functor from \Cref{def:delay-fun-cn} can be lifted to a monoidal functor $\mathbf{\delay} \colon \STREAM \to \STREAM$ that acts on objects in the same way.
  It acts on morphisms by sending a stream $f \in \STREAM(\stream{X},\stream{Y})$ to the stream given by $M(\delay f) = \monoidalunit$, $\now(\delay f) = \im_{\monoidalunit}$ and $\later(\delay f) = f$.
\end{defi}

\begin{defi}[Feedback operation]
  Given any morphism of the form 
  $\fm \in \STREAM(
    \act{N}{\delay \stream{S} ⊗ \stream{X}}, 
    \stream{S} ⊗ \stream{Y})$,
  we define $\fbk(f^{N}) \in \STREAM(\act{N}{\stream{X}}, \stream{Y})$ as
  \begin{itemize}
    \item $M(\fbk(f^{N})) = M(f) \tensor S_{0}$,
    \item $\now(\fbk(f^{N})) = \now(f)$ and
    \item $\later(\fbk(f^{N})) = \fbk(\later(f)^{M(f) \tensor S_{0}})$.
  \end{itemize}
  We write \(\fbk(f) \in \STREAM(\stream{X}, \stream{Y})\) for \(\fbk(f^{\monoidalunit})\), the feedback of $f \in \STREAM(\delay \stream{S} \tensor \stream{X}, \stream{S} \tensor \stream{Y})$ tensor.
\end{defi}

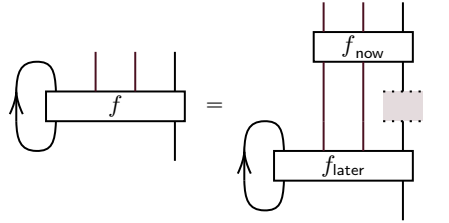
\begin{figure}[ht]

\tikzset{every picture/.style={line width=0.75pt}} %

\begin{tikzpicture}[x=0.75pt,y=0.75pt,yscale=-1,xscale=1]
\draw  [line width=0.75]  (65,75) -- (135,75) -- (135,90) -- (65,90) -- cycle ;
\draw [line width=0.75]    (130,55) -- (130,75) ;
\draw [color={rgb, 255:red, 0; green, 0; blue, 0 }  ,draw opacity=1 ]   (50,75) .. controls (50.25,65.45) and (52.08,60.22) .. (60,60) .. controls (67.92,59.78) and (69.58,64.12) .. (70,75) ;
\draw [color={rgb, 255:red, 75; green, 17; blue, 33 }  ,draw opacity=1 ][line width=0.75]    (90,55) -- (90,75) ;
\draw [color={rgb, 255:red, 75; green, 17; blue, 33 }  ,draw opacity=1 ][line width=0.75]    (110,55) -- (110,75) ;
\draw [color={rgb, 255:red, 0; green, 0; blue, 0 }  ,draw opacity=1 ]   (50,90) .. controls (50.25,99.55) and (52.08,104.78) .. (60,105) .. controls (67.92,105.22) and (69.58,100.88) .. (70,90) ;
\draw [line width=0.75]    (130,90) -- (130,110) ;
\draw  [draw opacity=0][line width=0.75]  (135,75) -- (165,75) -- (165,90) -- (135,90) -- cycle ;
\draw  [draw opacity=0][fill={rgb, 255:red, 75; green, 17; blue, 33 }  ,fill opacity=0.15 ] (235,75) -- (255,75) -- (255,90) -- (235,90) -- cycle ;
\draw [color={rgb, 255:red, 0; green, 0; blue, 0 }  ,draw opacity=1 ][line width=0.75]  [dash pattern={on 0.84pt off 2.51pt}]  (235,90) -- (255,90) ;
\draw [color={rgb, 255:red, 0; green, 0; blue, 0 }  ,draw opacity=1 ][line width=0.75]  [dash pattern={on 0.84pt off 2.51pt}]  (235,75) -- (255,75) ;
\draw  [line width=0.75]  (200,45) -- (250,45) -- (250,60) -- (200,60) -- cycle ;
\draw [color={rgb, 255:red, 75; green, 17; blue, 33 }  ,draw opacity=1 ][line width=0.75]    (205,30) -- (205,45) ;
\draw [color={rgb, 255:red, 75; green, 17; blue, 33 }  ,draw opacity=1 ][line width=0.75]    (225,30) -- (225,45) ;
\draw [line width=0.75]    (245,30) -- (245,45) ;
\draw  [line width=0.75]  (180,105.01) -- (250,105.01) -- (250,120.01) -- (180,120.01) -- cycle ;
\draw [line width=0.75]    (245,90.01) -- (245,105.01) ;
\draw [color={rgb, 255:red, 0; green, 0; blue, 0 }  ,draw opacity=1 ]   (165,105.01) .. controls (165.25,95.46) and (167.08,90.22) .. (175,90.01) .. controls (182.92,89.79) and (184.58,94.12) .. (185,105.01) ;
\draw [color={rgb, 255:red, 75; green, 17; blue, 33 }  ,draw opacity=1 ][line width=0.75]    (205,90.01) -- (205,105.01) ;
\draw [color={rgb, 255:red, 75; green, 17; blue, 33 }  ,draw opacity=1 ][line width=0.75]    (225,90.01) -- (225,105.01) ;
\draw [line width=0.75]    (165,105.01) -- (165,120.01) ;
\draw [shift={(165,105.51)}, rotate = 90] [color={rgb, 255:red, 0; green, 0; blue, 0 }  ][line width=0.75]    (10.93,-3.29) .. controls (6.95,-1.4) and (3.31,-0.3) .. (0,0) .. controls (3.31,0.3) and (6.95,1.4) .. (10.93,3.29)   ;
\draw [color={rgb, 255:red, 0; green, 0; blue, 0 }  ,draw opacity=1 ]   (165,120.01) .. controls (165.25,129.56) and (167.08,134.79) .. (175,135.01) .. controls (182.92,135.22) and (184.58,130.89) .. (185,120.01) ;
\draw [line width=0.75]    (245,120.01) -- (245,140.01) ;
\draw  [draw opacity=0][line width=0.75]  (250,105.01) -- (280,105.01) -- (280,120.01) -- (250,120.01) -- cycle ;
\draw [color={rgb, 255:red, 75; green, 17; blue, 33 }  ,draw opacity=1 ][line width=0.75]    (205,60) -- (205,90) ;
\draw [color={rgb, 255:red, 75; green, 17; blue, 33 }  ,draw opacity=1 ][line width=0.75]    (225,60) -- (225,90) ;
\draw [line width=0.75]    (50,75) -- (50,90) ;
\draw [shift={(50,75.5)}, rotate = 90] [color={rgb, 255:red, 0; green, 0; blue, 0 }  ][line width=0.75]    (10.93,-3.29) .. controls (6.95,-1.4) and (3.31,-0.3) .. (0,0) .. controls (3.31,0.3) and (6.95,1.4) .. (10.93,3.29)   ;
\draw [line width=0.75]    (245,60) -- (245,75) ;
\draw  [draw opacity=0][line width=0.75]  (255,75) -- (285,75) -- (285,90) -- (255,90) -- cycle ;

\draw (100,82.5) node  [font=\footnotesize]  {${\textstyle f}$};
\draw (150,82.5) node  [font=\footnotesize]  {$=$};
\draw (225,52.5) node  [font=\footnotesize]  {${\textstyle f}_{\mathsf{now}}$};
\draw (215,112.51) node  [font=\footnotesize]  {${\textstyle f_{\mathsf{later}}}$};
\draw (270,82.5) node  [font=\footnotesize]  {$;$};

\end{tikzpicture}
   \caption{Definition of feedback.}
  \label{strings:feedback}
\end{figure}

\begin{lem}\label{lemma:tighteninglong}
  The structure $(\STREAM, \fbk)$ with memories satisfies the tightening axiom (A1). Given three streams
  $u \in \STREAM(\act{A}{\stream{X}'}, \stream{X})$,
  $f \in \STREAM(
    \act{B ⊗ T}{\delay\stream{S} ⊗ \stream{X}}, 
  \stream{S} ⊗ \stream{Y})$, and
  $v \in \STREAM(\act{C}{\stream{Y}}, \stream{Y}')$;
  we claim that
  \[\fbk^{S}(u^{A} ⨾ f^{B} ⨾ v^{C}) = \act{\sigma}{u^{A} ⨾ \fbk^{S}(f^{B \tensor T}) ⨾ v^{C}}.\]
\end{lem}
\begin{proof}
  First, we note that both sides of the equation (which, from now on, we call $LHS$ and $RHS$, respectively)
  represent streams with different memories.
  \begin{itemize}
    \item $M(LHS) = A ⊗ B ⊗ C ⊗ T$,
    \item $M(RHS) = A ⊗ B ⊗ T ⊗ C$.
  \end{itemize}
  We will prove that they are related by dinaturality over the symmetry $\sigma$.
  We know that $\now(LHS) ⨾ \sigma = \now(RHS)$ by string diagrams (see \Cref{strings:tightening}).
  Then, by coinduction, we know that $\later(LHS) = \act{\sigma}{\later(RHS)}$.
\end{proof}
\begin{figure}
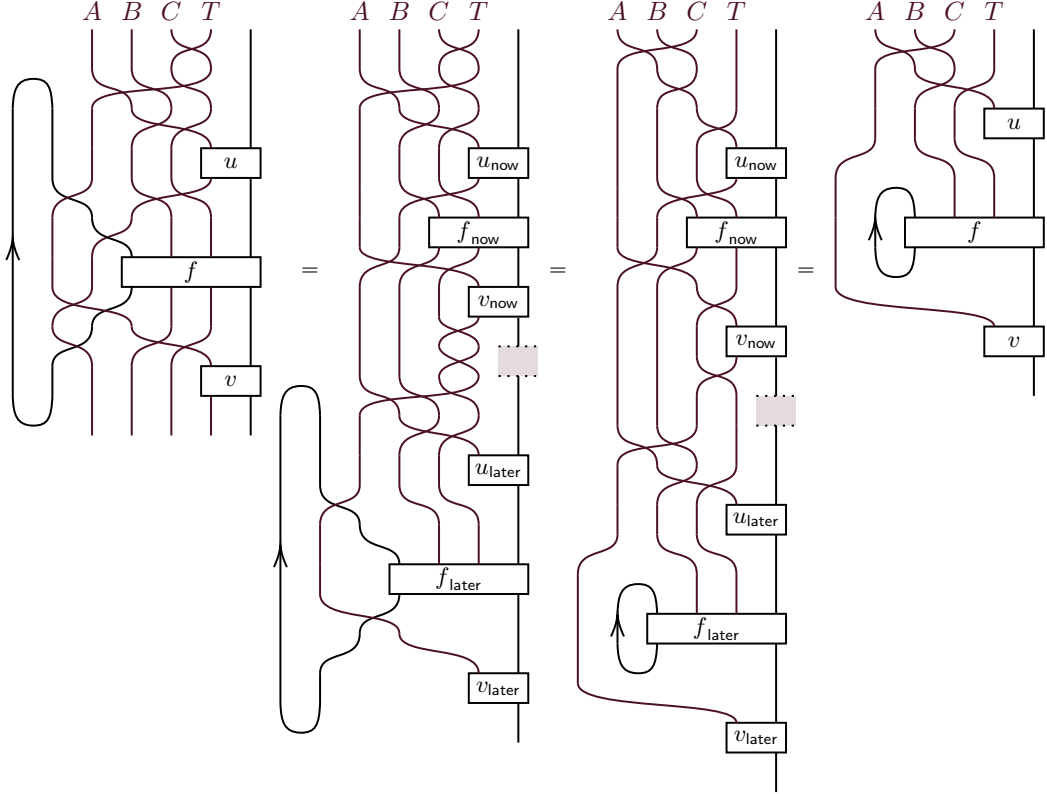


\tikzset{every picture/.style={line width=0.75pt}} %


 \caption{The tightening axiom (A1).}
\label{strings:tightening}
\end{figure}

\begin{lem}\label{lemma:tightening}
  The structure $(\STREAM, \fbk)$ satisfies the tightening axiom (A1).
  Given streams
  $u \in \STREAM(\stream{X}',\stream{X})$,
  $f \in \STREAM(\delay\stream{S}\tensor \stream{X},\stream{S}\tensor \stream{Y})$, and
  $v \in \STREAM(\stream{Y},\stream{Y}')$;
  we claim that $$\fbk^{S}(u ⨾ f ⨾ v) = u ⨾ \fbk^{S}(f) ⨾ v.$$
\end{lem}
\begin{proof}
  Consequence of \Cref{lemma:tighteninglong}, after applying the necessary coherence morphisms.
\end{proof}

\begin{lem}\label{lemma:vanishinglong}
  The structure $(\STREAM, \fbk)$ with memories satisfies the vanishing axiom (A2).
  Given a stream
    $f \in \STREAM(\act{A}{\delay\stream{S} \tensor \stream{X}}, \stream{S} \tensor \stream{Y})$,
  we claim that $$\fbk^{I}(f^{A}) = \act{\rho}{f}.$$
\end{lem}
\begin{proof}
  First, we note that both sides of the equation represent streams with different memories, $M(\fbk^{I}(f^{A})) = M(f) \tensor I$.
  We will prove that they are related by dinaturality over the right unitor $\rho$.
  We know that $\now(\fbk^{I}(f^{A})) = \now(f)$ by definition.
  Then, by coinduction, we konw that $\later(\fbk^{I}(f^{A})) = \act{\rho}{\later(f)}$.
\end{proof}

\begin{lem}\label{lemma:vanishing}
  The structure $(\STREAM, \fbk)$ satisfies the vanishing axiom (A2).
  Given a stream $f \in \STREAM(\delay\stream{S} \tensor \stream{X}, \stream{S} \tensor \stream{Y})$,
  we claim that $$\fbk^{I}(f) = f.$$
\end{lem}
\begin{proof}
  Consequence of \Cref{lemma:vanishinglong}, after applying the necessary coherence morphisms.
\end{proof}

\begin{lem}\label{lemma:joininglong}
  The structure $(\STREAM, \fbk)$ with memories satisfies the joining axiom (A3).
  Given a stream
  $f \in \STREAM(\act{(A \tensor P \tensor Q)}{\delay \stream{S} \tensor \stream{X}}, \stream{S} \tensor \stream{Y})$,
  we claim that
  \[\fbk^{S \tensor T}(f^{A \tensor (P \tensor Q)}) = 
  \act{\alpha}{\fbk^{T}(\act{\sigma}{\fbk^{S}(\act{\sigma}{f^{(A \tensor Q) \tensor P}})})}.\]
\end{lem}
\begin{proof}
  First, we note that both sides of the equation (which, from now on, we call $LHS$ and $RHS$, respectively)
  represent strams with different memories.
  \begin{itemize}
    \item $M(LHS) = M(f) \tensor (S_{0} \tensor T_{0})$,
    \item $M(RHS) = (M(f) \tensor S_{0}) \tensor T_{0}$.
  \end{itemize}
  We will prove that they are related by dinaturality over the associator $\alpha$.
  We know that $\now(LHS) \comp \alpha = \now(RHS)$ by definition.
  Then, by coinduction, we know that $\later(LHS) = \act{\alpha}{\later(RHS)}$.
\end{proof}

\begin{lem}\label{lemma:joining}
  The structure $(\STREAM, \fbk)$ satisfies the joining axiom (A3).
  Given a stream $f \in \STREAM(\delay \stream{S} \tensor \stream{X}, \stream{S} \tensor \stream{Y})$, we claim that 
  $$\fbk^{S \tensor T}(f) = \fbk^{T}(\fbk^{S}(f)).$$
\end{lem}
\begin{proof}
  Consequence of \Cref{lemma:joininglong}, after applying the necessary coherence morphisms.
\end{proof}

\begin{lem}\label{lemma:strengthlong}
  The structure $(\STREAM, \fbk)$ with memories satisfies the strength axiom (A4).
  Given two streams 
  $f \in \STREAM(\act{(A \tensor P)}{\delay \stream{S} \tensor \stream{X}}, \stream{S} \tensor \stream{Y})$, and 
  $g \in \STREAM(\act{B}{\stream{X}'},\stream{Y}')$,
  we claim that
  \[\act{\alpha}{\fbk^{S}(f^{A} ⊗ g^{B})} = \fbk^{S}(f^{A ⊗ P})^{A ⊗ P} ⊗ g^{B}.\]
\end{lem}
\begin{proof}
  First, we note that both sides of the equation (which, from now on, we call $LHS$ and $RHS$, respectively)
  represent streams with different memories.
  \begin{itemize}
     \item $M(LHS) = (M(f) ⊗ S_{0}) ⊗ M(g)$,
     \item $M(RHS) = M(f) ⊗ (S_{0} ⊗ M(g))$.
  \end{itemize}
  We will prove that they are related by dinaturality over the symmetry $\alpha$.
  We know that $\now(LHS) = \now(RHS) \comp \alpha$ by string diagrams (see \Cref{strings:strength}).
  Then, by coinduction, we know that $\act{\alpha}{\later(LHS)} = \later(RHS)$.
\end{proof}

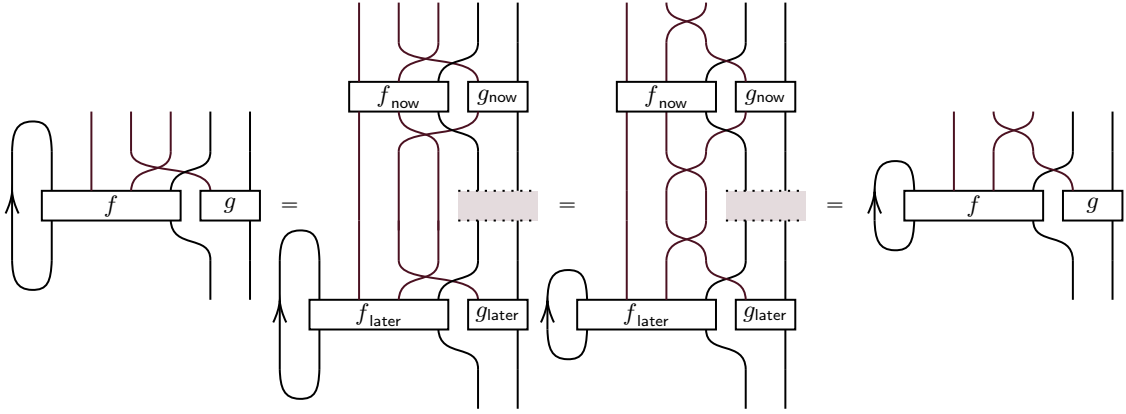
\begin{figure}

\tikzset{every picture/.style={line width=0.75pt}} %

\begin{tikzpicture}[x=0.75pt,y=0.75pt,yscale=-1,xscale=1]
\draw  [line width=0.75]  (30,115) -- (100,115) -- (100,130) -- (30,130) -- cycle ;
\draw [color={rgb, 255:red, 75; green, 17; blue, 33 }  ,draw opacity=1 ]   (95,95) .. controls (95.21,108.24) and (75.88,100.24) .. (75,115) ;
\draw [color={rgb, 255:red, 75; green, 17; blue, 33 }  ,draw opacity=1 ]   (75,95) .. controls (75.21,108.24) and (114.57,100.12) .. (115,115) ;
\draw [color={rgb, 255:red, 0; green, 0; blue, 0 }  ,draw opacity=1 ]   (115,95) .. controls (115.21,108.24) and (95.88,100.24) .. (95,115) ;
\draw  [line width=0.75]  (110,115) -- (140,115) -- (140,130) -- (110,130) -- cycle ;
\draw [color={rgb, 255:red, 0; green, 0; blue, 0 }  ,draw opacity=1 ][line width=0.75]    (135,95) -- (135,115) ;
\draw [color={rgb, 255:red, 75; green, 17; blue, 33 }  ,draw opacity=1 ][line width=0.75]    (55,75) -- (55,115) ;
\draw [color={rgb, 255:red, 0; green, 0; blue, 0 }  ,draw opacity=1 ]   (15,95) .. controls (15.25,85.45) and (17.08,80.22) .. (25,80) .. controls (32.1,79.81) and (34.17,83.27) .. (34.83,91.84) .. controls (34.9,92.82) and (34.96,93.88) .. (35,95) ;
\draw [line width=0.75]    (15,95) -- (15,150) ;
\draw [shift={(15,115.5)}, rotate = 90] [color={rgb, 255:red, 0; green, 0; blue, 0 }  ][line width=0.75]    (10.93,-3.29) .. controls (6.95,-1.4) and (3.31,-0.3) .. (0,0) .. controls (3.31,0.3) and (6.95,1.4) .. (10.93,3.29)   ;
\draw [color={rgb, 255:red, 0; green, 0; blue, 0 }  ,draw opacity=1 ][line width=0.75]    (35,95) -- (35,115) ;
\draw [color={rgb, 255:red, 75; green, 17; blue, 33 }  ,draw opacity=1 ][line width=0.75]    (75,75) -- (75,95) ;
\draw [color={rgb, 255:red, 75; green, 17; blue, 33 }  ,draw opacity=1 ][line width=0.75]    (95,75) -- (95,95) ;
\draw [color={rgb, 255:red, 0; green, 0; blue, 0 }  ,draw opacity=1 ][line width=0.75]    (115,75) -- (115,95) ;
\draw [color={rgb, 255:red, 0; green, 0; blue, 0 }  ,draw opacity=1 ][line width=0.75]    (135,75) -- (135,95) ;
\draw [color={rgb, 255:red, 0; green, 0; blue, 0 }  ,draw opacity=1 ]   (115,150) .. controls (115.21,136.76) and (95.88,144.76) .. (95,130) ;
\draw [color={rgb, 255:red, 0; green, 0; blue, 0 }  ,draw opacity=1 ][line width=0.75]    (135,150) -- (135,130) ;
\draw [color={rgb, 255:red, 0; green, 0; blue, 0 }  ,draw opacity=1 ][line width=0.75]    (115,170) -- (115,150) ;
\draw [color={rgb, 255:red, 0; green, 0; blue, 0 }  ,draw opacity=1 ][line width=0.75]    (135,170) -- (135,150) ;
\draw [color={rgb, 255:red, 0; green, 0; blue, 0 }  ,draw opacity=1 ]   (15,150) .. controls (15.25,159.55) and (17.08,164.78) .. (25,165) .. controls (32.92,165.22) and (34.58,160.88) .. (35,150) ;
\draw [color={rgb, 255:red, 0; green, 0; blue, 0 }  ,draw opacity=1 ][line width=0.75]    (35,130) -- (35,150) ;
\draw  [draw opacity=0][line width=0.75]  (140,115.5) -- (170,115.5) -- (170,130.5) -- (140,130.5) -- cycle ;
\draw  [line width=0.75]  (185,60) -- (235,60) -- (235,75) -- (185,75) -- cycle ;
\draw [color={rgb, 255:red, 75; green, 17; blue, 33 }  ,draw opacity=1 ]   (230,40) .. controls (230.21,53.24) and (210.88,45.24) .. (210,60) ;
\draw [color={rgb, 255:red, 75; green, 17; blue, 33 }  ,draw opacity=1 ]   (210,40) .. controls (210.21,53.24) and (249.57,45.12) .. (250,60) ;
\draw [color={rgb, 255:red, 0; green, 0; blue, 0 }  ,draw opacity=1 ]   (250,40) .. controls (250.21,53.24) and (230.88,45.24) .. (230,60) ;
\draw  [line width=0.75]  (245,60) -- (275,60) -- (275,75) -- (245,75) -- cycle ;
\draw [color={rgb, 255:red, 0; green, 0; blue, 0 }  ,draw opacity=1 ][line width=0.75]    (270,40) -- (270,60) ;
\draw [color={rgb, 255:red, 75; green, 17; blue, 33 }  ,draw opacity=1 ][line width=0.75]    (190,20) -- (190,60) ;
\draw [color={rgb, 255:red, 75; green, 17; blue, 33 }  ,draw opacity=1 ][line width=0.75]    (210,20) -- (210,40) ;
\draw [color={rgb, 255:red, 75; green, 17; blue, 33 }  ,draw opacity=1 ][line width=0.75]    (230,20) -- (230,40) ;
\draw [color={rgb, 255:red, 0; green, 0; blue, 0 }  ,draw opacity=1 ][line width=0.75]    (250,20) -- (250,40) ;
\draw [color={rgb, 255:red, 0; green, 0; blue, 0 }  ,draw opacity=1 ][line width=0.75]    (270,20) -- (270,40) ;
\draw [color={rgb, 255:red, 75; green, 17; blue, 33 }  ,draw opacity=1 ]   (230,95) .. controls (230.21,81.76) and (210.88,89.76) .. (210,75) ;
\draw [color={rgb, 255:red, 75; green, 17; blue, 33 }  ,draw opacity=1 ]   (210,95) .. controls (210.21,81.76) and (249.57,89.88) .. (250,75) ;
\draw [color={rgb, 255:red, 0; green, 0; blue, 0 }  ,draw opacity=1 ]   (250,95) .. controls (250.21,81.76) and (230.88,89.76) .. (230,75) ;
\draw [color={rgb, 255:red, 0; green, 0; blue, 0 }  ,draw opacity=1 ][line width=0.75]    (270,95) -- (270,75) ;
\draw [color={rgb, 255:red, 75; green, 17; blue, 33 }  ,draw opacity=1 ][line width=0.75]    (190,130) -- (190,75) ;
\draw [color={rgb, 255:red, 75; green, 17; blue, 33 }  ,draw opacity=1 ][line width=0.75]    (210,135) -- (210,95) ;
\draw [color={rgb, 255:red, 75; green, 17; blue, 33 }  ,draw opacity=1 ][line width=0.75]    (230,135) -- (230,95) ;
\draw [color={rgb, 255:red, 0; green, 0; blue, 0 }  ,draw opacity=1 ][line width=0.75]    (250,115) -- (250,95) ;
\draw [color={rgb, 255:red, 0; green, 0; blue, 0 }  ,draw opacity=1 ][line width=0.75]    (270,115) -- (270,95) ;
\draw  [draw opacity=0][fill={rgb, 255:red, 75; green, 17; blue, 33 }  ,fill opacity=0.15 ] (240,115) -- (280,115) -- (280,130) -- (240,130) -- cycle ;
\draw [color={rgb, 255:red, 0; green, 0; blue, 0 }  ,draw opacity=1 ][line width=0.75]  [dash pattern={on 0.84pt off 2.51pt}]  (240,115) -- (280,115) ;
\draw [color={rgb, 255:red, 0; green, 0; blue, 0 }  ,draw opacity=1 ][line width=0.75]  [dash pattern={on 0.84pt off 2.51pt}]  (240,130) -- (280,130) ;
\draw  [line width=0.75]  (165,170) -- (235,170) -- (235,185) -- (165,185) -- cycle ;
\draw [color={rgb, 255:red, 75; green, 17; blue, 33 }  ,draw opacity=1 ]   (230,150) .. controls (230.21,163.24) and (210.88,155.24) .. (210,170) ;
\draw [color={rgb, 255:red, 75; green, 17; blue, 33 }  ,draw opacity=1 ]   (210,150) .. controls (210.21,163.24) and (249.57,155.12) .. (250,170) ;
\draw [color={rgb, 255:red, 0; green, 0; blue, 0 }  ,draw opacity=1 ]   (250,150) .. controls (250.21,163.24) and (230.88,155.24) .. (230,170) ;
\draw  [line width=0.75]  (245,170) -- (275,170) -- (275,185) -- (245,185) -- cycle ;
\draw [color={rgb, 255:red, 0; green, 0; blue, 0 }  ,draw opacity=1 ][line width=0.75]    (270,150) -- (270,170) ;
\draw [color={rgb, 255:red, 75; green, 17; blue, 33 }  ,draw opacity=1 ][line width=0.75]    (190,130) -- (190,170) ;
\draw [color={rgb, 255:red, 0; green, 0; blue, 0 }  ,draw opacity=1 ]   (150,150) .. controls (150.25,140.45) and (152.08,135.22) .. (160,135) .. controls (167.1,134.81) and (169.17,138.27) .. (169.83,146.84) .. controls (169.9,147.82) and (169.96,148.88) .. (170,150) ;
\draw [line width=0.75]    (150,150) -- (150,205) ;
\draw [shift={(150,170.5)}, rotate = 90] [color={rgb, 255:red, 0; green, 0; blue, 0 }  ][line width=0.75]    (10.93,-3.29) .. controls (6.95,-1.4) and (3.31,-0.3) .. (0,0) .. controls (3.31,0.3) and (6.95,1.4) .. (10.93,3.29)   ;
\draw [color={rgb, 255:red, 0; green, 0; blue, 0 }  ,draw opacity=1 ][line width=0.75]    (170,150) -- (170,170) ;
\draw [color={rgb, 255:red, 75; green, 17; blue, 33 }  ,draw opacity=1 ][line width=0.75]    (210,130) -- (210,150) ;
\draw [color={rgb, 255:red, 75; green, 17; blue, 33 }  ,draw opacity=1 ][line width=0.75]    (230,130) -- (230,150) ;
\draw [color={rgb, 255:red, 0; green, 0; blue, 0 }  ,draw opacity=1 ][line width=0.75]    (250,130) -- (250,150) ;
\draw [color={rgb, 255:red, 0; green, 0; blue, 0 }  ,draw opacity=1 ][line width=0.75]    (270,130) -- (270,150) ;
\draw [color={rgb, 255:red, 0; green, 0; blue, 0 }  ,draw opacity=1 ]   (250,205) .. controls (250.21,191.76) and (230.88,199.76) .. (230,185) ;
\draw [color={rgb, 255:red, 0; green, 0; blue, 0 }  ,draw opacity=1 ][line width=0.75]    (270,205) -- (270,185) ;
\draw [color={rgb, 255:red, 0; green, 0; blue, 0 }  ,draw opacity=1 ][line width=0.75]    (250,225) -- (250,205) ;
\draw [color={rgb, 255:red, 0; green, 0; blue, 0 }  ,draw opacity=1 ][line width=0.75]    (270,225) -- (270,205) ;
\draw [color={rgb, 255:red, 0; green, 0; blue, 0 }  ,draw opacity=1 ]   (150,205) .. controls (150.25,214.55) and (152.08,219.78) .. (160,220) .. controls (167.92,220.22) and (169.58,215.88) .. (170,205) ;
\draw [color={rgb, 255:red, 0; green, 0; blue, 0 }  ,draw opacity=1 ][line width=0.75]    (170,185) -- (170,205) ;
\draw  [line width=0.75]  (465,115) -- (535,115) -- (535,130) -- (465,130) -- cycle ;
\draw [color={rgb, 255:red, 75; green, 17; blue, 33 }  ,draw opacity=1 ]   (530,75) .. controls (530.21,88.24) and (510.88,80.24) .. (510,95) ;
\draw [color={rgb, 255:red, 75; green, 17; blue, 33 }  ,draw opacity=1 ]   (510,75) .. controls (510.21,88.24) and (529.57,80.12) .. (530,95) ;
\draw [color={rgb, 255:red, 0; green, 0; blue, 0 }  ,draw opacity=1 ]   (550,95) .. controls (550.21,108.24) and (530.88,100.24) .. (530,115) ;
\draw  [line width=0.75]  (545,115) -- (575,115) -- (575,130) -- (545,130) -- cycle ;
\draw [color={rgb, 255:red, 0; green, 0; blue, 0 }  ,draw opacity=1 ][line width=0.75]    (570,95) -- (570,115) ;
\draw [color={rgb, 255:red, 75; green, 17; blue, 33 }  ,draw opacity=1 ][line width=0.75]    (490,75) -- (490,115) ;
\draw [color={rgb, 255:red, 0; green, 0; blue, 0 }  ,draw opacity=1 ]   (450,115) .. controls (450.25,105.45) and (452.08,100.22) .. (460,100) .. controls (467.1,99.81) and (469.17,103.27) .. (469.83,111.84) .. controls (469.9,112.82) and (469.96,113.88) .. (470,115) ;
\draw [line width=0.75]    (450,115) -- (450,130) ;
\draw [shift={(450,115.5)}, rotate = 90] [color={rgb, 255:red, 0; green, 0; blue, 0 }  ][line width=0.75]    (10.93,-3.29) .. controls (6.95,-1.4) and (3.31,-0.3) .. (0,0) .. controls (3.31,0.3) and (6.95,1.4) .. (10.93,3.29)   ;
\draw [color={rgb, 255:red, 75; green, 17; blue, 33 }  ,draw opacity=1 ][line width=0.75]    (510,95) -- (510,115) ;
\draw [color={rgb, 255:red, 0; green, 0; blue, 0 }  ,draw opacity=1 ][line width=0.75]    (550,75) -- (550,95) ;
\draw [color={rgb, 255:red, 0; green, 0; blue, 0 }  ,draw opacity=1 ][line width=0.75]    (570,75) -- (570,95) ;
\draw [color={rgb, 255:red, 0; green, 0; blue, 0 }  ,draw opacity=1 ][line width=0.75]    (570,150) -- (570,130) ;
\draw [color={rgb, 255:red, 0; green, 0; blue, 0 }  ,draw opacity=1 ][line width=0.75]    (550,170) -- (550,150) ;
\draw [color={rgb, 255:red, 0; green, 0; blue, 0 }  ,draw opacity=1 ][line width=0.75]    (570,170) -- (570,150) ;
\draw [color={rgb, 255:red, 0; green, 0; blue, 0 }  ,draw opacity=1 ]   (450,130) .. controls (450.25,139.55) and (452.08,144.78) .. (460,145) .. controls (467.92,145.22) and (469.58,140.88) .. (470,130) ;
\draw [color={rgb, 255:red, 75; green, 17; blue, 33 }  ,draw opacity=1 ]   (530,95) .. controls (530.21,108.24) and (549.57,100.12) .. (550,115) ;
\draw [color={rgb, 255:red, 0; green, 0; blue, 0 }  ,draw opacity=1 ]   (530,130) .. controls (530.21,143.24) and (549.57,135.12) .. (550,150) ;
\draw [color={rgb, 255:red, 75; green, 17; blue, 33 }  ,draw opacity=1 ]   (365,20) .. controls (365.21,33.24) and (345.88,25.24) .. (345,40) ;
\draw [color={rgb, 255:red, 75; green, 17; blue, 33 }  ,draw opacity=1 ]   (345,20) .. controls (345.21,33.24) and (364.57,25.12) .. (365,40) ;
\draw [color={rgb, 255:red, 0; green, 0; blue, 0 }  ,draw opacity=1 ]   (385,40) .. controls (385.21,53.24) and (365.88,45.24) .. (365,60) ;
\draw  [line width=0.75]  (380,60) -- (410,60) -- (410,75) -- (380,75) -- cycle ;
\draw [color={rgb, 255:red, 0; green, 0; blue, 0 }  ,draw opacity=1 ][line width=0.75]    (405,40) -- (405,60) ;
\draw [color={rgb, 255:red, 75; green, 17; blue, 33 }  ,draw opacity=1 ][line width=0.75]    (325,20) -- (325,60) ;
\draw [color={rgb, 255:red, 75; green, 17; blue, 33 }  ,draw opacity=1 ][line width=0.75]    (345,40) -- (345,60) ;
\draw [color={rgb, 255:red, 0; green, 0; blue, 0 }  ,draw opacity=1 ][line width=0.75]    (385,20) -- (385,40) ;
\draw [color={rgb, 255:red, 0; green, 0; blue, 0 }  ,draw opacity=1 ][line width=0.75]    (405,20) -- (405,40) ;
\draw [color={rgb, 255:red, 0; green, 0; blue, 0 }  ,draw opacity=1 ][line width=0.75]    (405,95) -- (405,75) ;
\draw [color={rgb, 255:red, 75; green, 17; blue, 33 }  ,draw opacity=1 ][line width=0.75]    (325,115) -- (325,75) ;
\draw [color={rgb, 255:red, 0; green, 0; blue, 0 }  ,draw opacity=1 ][line width=0.75]    (385,115) -- (385,95) ;
\draw [color={rgb, 255:red, 0; green, 0; blue, 0 }  ,draw opacity=1 ][line width=0.75]    (405,115) -- (405,95) ;
\draw [color={rgb, 255:red, 75; green, 17; blue, 33 }  ,draw opacity=1 ]   (365,40) .. controls (365.21,53.24) and (384.57,45.12) .. (385,60) ;
\draw [color={rgb, 255:red, 75; green, 17; blue, 33 }  ,draw opacity=1 ]   (385,75) .. controls (385.21,88.24) and (365.88,80.24) .. (365,95) ;
\draw [color={rgb, 255:red, 0; green, 0; blue, 0 }  ,draw opacity=1 ]   (365,75) .. controls (365.21,88.24) and (384.57,80.12) .. (385,95) ;
\draw [color={rgb, 255:red, 75; green, 17; blue, 33 }  ,draw opacity=1 ]   (365,95) .. controls (365.21,108.24) and (345.88,100.24) .. (345,115) ;
\draw [color={rgb, 255:red, 75; green, 17; blue, 33 }  ,draw opacity=1 ]   (345,95) .. controls (345.21,108.24) and (364.57,100.12) .. (365,115) ;
\draw [color={rgb, 255:red, 75; green, 17; blue, 33 }  ,draw opacity=1 ][line width=0.75]    (345,95) -- (345,75) ;
\draw  [draw opacity=0][fill={rgb, 255:red, 75; green, 17; blue, 33 }  ,fill opacity=0.15 ] (375,115) -- (415,115) -- (415,130) -- (375,130) -- cycle ;
\draw [color={rgb, 255:red, 0; green, 0; blue, 0 }  ,draw opacity=1 ][line width=0.75]  [dash pattern={on 0.84pt off 2.51pt}]  (375,115) -- (415,115) ;
\draw [color={rgb, 255:red, 0; green, 0; blue, 0 }  ,draw opacity=1 ][line width=0.75]  [dash pattern={on 0.84pt off 2.51pt}]  (375,130) -- (415,130) ;
\draw  [line width=0.75]  (300,170) -- (370,170) -- (370,185) -- (300,185) -- cycle ;
\draw [color={rgb, 255:red, 75; green, 17; blue, 33 }  ,draw opacity=1 ]   (365,130) .. controls (365.21,143.24) and (345.88,135.24) .. (345,150) ;
\draw [color={rgb, 255:red, 75; green, 17; blue, 33 }  ,draw opacity=1 ]   (345,130) .. controls (345.21,143.24) and (364.57,135.12) .. (365,150) ;
\draw [color={rgb, 255:red, 0; green, 0; blue, 0 }  ,draw opacity=1 ]   (385,150) .. controls (385.21,163.24) and (365.88,155.24) .. (365,170) ;
\draw  [line width=0.75]  (380,170) -- (410,170) -- (410,185) -- (380,185) -- cycle ;
\draw [color={rgb, 255:red, 0; green, 0; blue, 0 }  ,draw opacity=1 ][line width=0.75]    (405,150) -- (405,170) ;
\draw [color={rgb, 255:red, 75; green, 17; blue, 33 }  ,draw opacity=1 ][line width=0.75]    (325,130) -- (325,170) ;
\draw [color={rgb, 255:red, 75; green, 17; blue, 33 }  ,draw opacity=1 ][line width=0.75]    (345,150) -- (345,170) ;
\draw [color={rgb, 255:red, 0; green, 0; blue, 0 }  ,draw opacity=1 ][line width=0.75]    (385,130) -- (385,150) ;
\draw [color={rgb, 255:red, 0; green, 0; blue, 0 }  ,draw opacity=1 ][line width=0.75]    (405,130) -- (405,150) ;
\draw [color={rgb, 255:red, 0; green, 0; blue, 0 }  ,draw opacity=1 ][line width=0.75]    (405,205) -- (405,185) ;
\draw [color={rgb, 255:red, 0; green, 0; blue, 0 }  ,draw opacity=1 ][line width=0.75]    (385,225) -- (385,205) ;
\draw [color={rgb, 255:red, 0; green, 0; blue, 0 }  ,draw opacity=1 ][line width=0.75]    (405,225) -- (405,205) ;
\draw [color={rgb, 255:red, 75; green, 17; blue, 33 }  ,draw opacity=1 ]   (365,150) .. controls (365.21,163.24) and (384.57,155.12) .. (385,170) ;
\draw [color={rgb, 255:red, 0; green, 0; blue, 0 }  ,draw opacity=1 ]   (365,185) .. controls (365.21,198.24) and (384.57,190.12) .. (385,205) ;
\draw [color={rgb, 255:red, 75; green, 17; blue, 33 }  ,draw opacity=1 ][line width=0.75]    (325,130) -- (325,115) ;
\draw [color={rgb, 255:red, 75; green, 17; blue, 33 }  ,draw opacity=1 ][line width=0.75]    (345,130) -- (345,115) ;
\draw [color={rgb, 255:red, 75; green, 17; blue, 33 }  ,draw opacity=1 ][line width=0.75]    (365,130) -- (365,115) ;
\draw  [line width=0.75]  (320,60) -- (370,60) -- (370,75) -- (320,75) -- cycle ;
\draw  [draw opacity=0][line width=0.75]  (280,115) -- (310,115) -- (310,130) -- (280,130) -- cycle ;
\draw  [draw opacity=0][line width=0.75]  (415,115.5) -- (445,115.5) -- (445,130.5) -- (415,130.5) -- cycle ;
\draw [color={rgb, 255:red, 0; green, 0; blue, 0 }  ,draw opacity=1 ]   (285,170) .. controls (285.25,160.45) and (287.08,155.22) .. (295,155) .. controls (302.1,154.81) and (304.17,158.27) .. (304.83,166.84) .. controls (304.9,167.82) and (304.96,168.88) .. (305,170) ;
\draw [line width=0.75]    (285,170) -- (285,185) ;
\draw [shift={(285,170.5)}, rotate = 90] [color={rgb, 255:red, 0; green, 0; blue, 0 }  ][line width=0.75]    (10.93,-3.29) .. controls (6.95,-1.4) and (3.31,-0.3) .. (0,0) .. controls (3.31,0.3) and (6.95,1.4) .. (10.93,3.29)   ;
\draw [color={rgb, 255:red, 0; green, 0; blue, 0 }  ,draw opacity=1 ]   (285,185) .. controls (285.25,194.55) and (287.08,199.78) .. (295,200) .. controls (302.92,200.22) and (304.58,195.88) .. (305,185) ;

\draw (65,122.5) node  [font=\footnotesize]  {$f$};
\draw (125,122.5) node  [font=\footnotesize]  {$g$};
\draw (155,123) node  [font=\footnotesize]  {$=$};
\draw (500,122.5) node  [font=\footnotesize]  {$f$};
\draw (560,122.5) node  [font=\footnotesize]  {$g$};
\draw (295,123) node  [font=\footnotesize]  {$=$};
\draw (430,123) node  [font=\footnotesize]  {$=$};
\draw (210,67.5) node  [font=\footnotesize]  {${\textstyle f}_{\mathsf{now}}$};
\draw (260,67.5) node  [font=\footnotesize]  {$g_{\mathsf{now}}$};
\draw (200,177.5) node  [font=\footnotesize]  {${\textstyle f}_{\mathsf{later}}$};
\draw (260,177.5) node  [font=\footnotesize]  {$g_{\mathsf{later}}$};
\draw (335,177.5) node  [font=\footnotesize]  {${\textstyle f}_{\mathsf{later}}$};
\draw (395,177.5) node  [font=\footnotesize]  {$g_{\mathsf{later}}$};
\draw (345,67.5) node  [font=\footnotesize]  {${\textstyle f}_{\mathsf{now}}$};
\draw (395,67.5) node  [font=\footnotesize]  {$g_{\mathsf{now}}$};

\end{tikzpicture}
 \caption{The strength axiom (A4).}
\label{strings:strength}
\end{figure}

\begin{lem}\label{lemma:strength}
  The structure $(\STREAM, \fbk)$ with memories satisfies the strength axiom (A4).
  Given two streams
  $f \in \STREAM(\stream{S} \tensor \stream{X}, \stream{S} \tensor \stream{Y})$, and
  $g \in \STREAM(\stream{X}',\stream{Y}')$,
  we claim that
  \[\act{\alpha}{\fbk^{S}(f \tensor g)} = \fbk^{S}(f) \tensor g.\]
\end{lem}
\begin{proof}
  Consequence of \Cref{lemma:strengthlong}, after applying the necessary coherence morphisms.
\end{proof}

\begin{lem}\label{lemma:slidinglong}
  The structure $(\STREAM, \fbk)$ with memories satisfies the sliding axiom (A5).
  Given two streams
  $f \in \STREAM(\act{A}{\delay \stream{S} \tensor \stream{X}}, \stream{T} \tensor \stream{Y})$, and
  $r \in \STREAM(\act{C}{\stream{T}}, \stream{S})$
  we claim that, for each $k \colon B \tensor Q \to C \tensor P$,
  \[\act{k}{\fbk^{\stream{S}}(f^{A \tensor P} ; (r^{C} \tensor \im))} = \fbk^{\stream{T}}(\act{k}{(\delay r^{C} \tensor \im) ; f^{A \tensor P}}).\]
\end{lem}
\begin{proof}
  First, we note that both sides of the equation (which, from now on, we call $LHS$ and $RHS$, respectively)
  represent streams with different memories.
  \begin{itemize}
    \item $M(LHS) = M(f) \tensor M(r)  \tensor S_{0}$,
    \item $M(RHS) = M(r) \tensor M(f) \tensor T_{0}$.
  \end{itemize}
  We will prove that they are related by dinaturality over the symmetry and the first action of $r$, that is, $\sigma ⨾ (\now(r) \tensor \im)$.
  We know that $\now(LHS)= \now(RHS) ⨾ (\sigma ⨾ (\im \tensor \now(r)))$ by string diagrams (see \Cref{strings:sliding}).
  Using coinduction,
  \[\begin{aligned}
      & \act{(\sigma ⨾ \tid{\now(r)})}{\later(LHS)}  \\
      = & \act{(\sigma ⨾ \tid{\now(r)})}
        {\fbk^{\stream{S}}(\later(f)^{M(f) \tensor S_{0}} 
        ⨾ (\later(r)^{M(r)} \tensor \im))} \\
      = & \fbk^{\stream{T}}(\act{(\sigma 
        ⨾ \tid{\now(r)})}{(\delay \later(r)^{M(r)} \tensor \im) 
        ⨾ \later(f)^{M(f) \tensor S_{0}}}) \\
      = & \fbk^{\stream{T}}(\act{\sigma}{(\later(\delay r)^{M(r)} \tensor \im) 
        ⨾ \later(f)^{M(f) \tensor S_{0}}}) \\
      = & \later(RHS),
    \end{aligned}\]
  we show that $\act{(\sigma ⨾ \tid{\now(r)})}{\later(LHS)} = \later(RHS)$.
\end{proof}

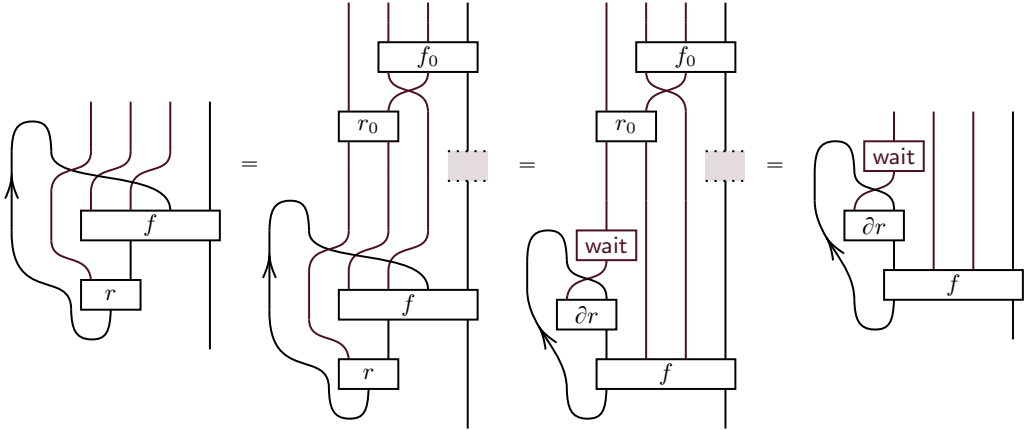
\begin{figure}

\tikzset{every picture/.style={line width=0.75pt}} %

\begin{tikzpicture}[x=0.75pt,y=0.75pt,yscale=-1,xscale=1]
\draw  [line width=0.75]  (45.01,115) -- (115.01,115) -- (115.01,130) -- (45.01,130) -- cycle ;
\draw  [line width=0.75]  (45.01,150) -- (75.01,150) -- (75.01,165) -- (45.01,165) -- cycle ;
\draw [color={rgb, 255:red, 0; green, 0; blue, 0 }  ,draw opacity=1 ]   (40.01,165) .. controls (40.26,174.55) and (42.09,179.78) .. (50.01,180) .. controls (57.92,180.22) and (59.59,175.88) .. (60.01,165) ;
\draw [color={rgb, 255:red, 0; green, 0; blue, 0 }  ,draw opacity=1 ]   (10.01,85) .. controls (10.26,75.45) and (12.09,70.22) .. (20.01,70) .. controls (27.11,69.81) and (29.18,73.27) .. (29.83,81.84) .. controls (29.91,82.82) and (29.97,83.88) .. (30.01,85) ;
\draw [line width=0.75]    (10.01,85) -- (10.01,125) ;
\draw [shift={(10.01,98)}, rotate = 90] [color={rgb, 255:red, 0; green, 0; blue, 0 }  ][line width=0.75]    (10.93,-3.29) .. controls (6.95,-1.4) and (3.31,-0.3) .. (0,0) .. controls (3.31,0.3) and (6.95,1.4) .. (10.93,3.29)   ;
\draw [color={rgb, 255:red, 0; green, 0; blue, 0 }  ,draw opacity=1 ]   (10.01,125) .. controls (9.42,162.64) and (40.22,144.24) .. (40.01,165) ;
\draw [color={rgb, 255:red, 75; green, 17; blue, 33 }  ,draw opacity=1 ][line width=0.75]    (90.01,60) -- (90.01,85) ;
\draw [color={rgb, 255:red, 0; green, 0; blue, 0 }  ,draw opacity=1 ][line width=0.75]    (70.01,130) -- (70.01,150) ;
\draw [color={rgb, 255:red, 0; green, 0; blue, 0 }  ,draw opacity=1 ][line width=0.75]    (110.01,130) -- (110.01,185) ;
\draw  [line width=0.75]  (305.01,190) -- (375.01,190) -- (375.01,205) -- (305.01,205) -- cycle ;
\draw [color={rgb, 255:red, 75; green, 17; blue, 33 }  ,draw opacity=1 ][line width=0.75]    (310.01,110) -- (310.01,125) ;
\draw  [line width=0.75]  (285.01,160) -- (315.01,160) -- (315.01,175) -- (285.01,175) -- cycle ;
\draw [color={rgb, 255:red, 0; green, 0; blue, 0 }  ,draw opacity=1 ][line width=0.75]    (370.01,205) -- (370.01,225) ;
\draw [color={rgb, 255:red, 0; green, 0; blue, 0 }  ,draw opacity=1 ]   (290.01,205) .. controls (290.26,214.55) and (292.09,219.78) .. (300.01,220) .. controls (307.92,220.22) and (309.59,215.88) .. (310.01,205) ;
\draw [color={rgb, 255:red, 0; green, 0; blue, 0 }  ,draw opacity=1 ]   (270.01,140) .. controls (270.26,130.45) and (272.09,125.22) .. (280.01,125) .. controls (287.92,124.78) and (289.59,129.12) .. (290.01,140) ;
\draw [color={rgb, 255:red, 0; green, 0; blue, 0 }  ,draw opacity=1 ][line width=0.75]    (310.01,175) -- (310.01,190) ;
\draw [color={rgb, 255:red, 75; green, 17; blue, 33 }  ,draw opacity=1 ][line width=0.75]    (350.01,110) -- (350.01,190) ;
\draw  [color={rgb, 255:red, 75; green, 17; blue, 33 }  ,draw opacity=1 ][line width=0.75]  (295.01,125) -- (325.01,125) -- (325.01,140) -- (295.01,140) -- cycle ;
\draw [color={rgb, 255:red, 0; green, 0; blue, 0 }  ,draw opacity=1 ][line width=0.75]    (370.01,100) -- (370.01,190) ;
\draw [color={rgb, 255:red, 75; green, 17; blue, 33 }  ,draw opacity=1 ][line width=0.75]    (330.01,110) -- (330.01,190) ;
\draw    (270.01,155) .. controls (270.21,174.65) and (289.87,185.65) .. (290.01,205) ;
\draw [shift={(276.25,174.12)}, rotate = 57.87] [color={rgb, 255:red, 0; green, 0; blue, 0 }  ][line width=0.75]    (10.93,-3.29) .. controls (6.95,-1.4) and (3.31,-0.3) .. (0,0) .. controls (3.31,0.3) and (6.95,1.4) .. (10.93,3.29)   ;
\draw [color={rgb, 255:red, 75; green, 17; blue, 33 }  ,draw opacity=1 ]   (70.01,85) .. controls (70.22,98.24) and (50.89,90.24) .. (50.01,105) ;
\draw [color={rgb, 255:red, 75; green, 17; blue, 33 }  ,draw opacity=1 ]   (90.01,85) .. controls (90.22,98.24) and (70.89,90.24) .. (70.01,105) ;
\draw    (30.01,85) .. controls (29.02,103.04) and (90.22,93.04) .. (90.01,115) ;
\draw [color={rgb, 255:red, 75; green, 17; blue, 33 }  ,draw opacity=1 ][line width=0.75]    (70.01,60) -- (70.01,85) ;
\draw [color={rgb, 255:red, 0; green, 0; blue, 0 }  ,draw opacity=1 ][line width=0.75]    (110.01,60) -- (110.01,81.29) -- (110.01,115) ;
\draw [color={rgb, 255:red, 75; green, 17; blue, 33 }  ,draw opacity=1 ][line width=0.75]    (30.01,105) -- (30.01,130) ;
\draw [color={rgb, 255:red, 75; green, 17; blue, 33 }  ,draw opacity=1 ]   (30.01,130) .. controls (30.22,143.24) and (48.62,135.44) .. (50.01,150) ;
\draw [color={rgb, 255:red, 75; green, 17; blue, 33 }  ,draw opacity=1 ][line width=0.75]    (70.01,105) -- (70.01,115) ;
\draw [color={rgb, 255:red, 75; green, 17; blue, 33 }  ,draw opacity=1 ][line width=0.75]    (50.01,105) -- (50.01,115) ;
\draw [color={rgb, 255:red, 75; green, 17; blue, 33 }  ,draw opacity=1 ]   (50.01,85) .. controls (50.07,88.8) and (48.52,90.85) .. (46.26,92.2) .. controls (40.64,95.53) and (30.63,94.48) .. (30.01,105) ;
\draw [color={rgb, 255:red, 75; green, 17; blue, 33 }  ,draw opacity=1 ][line width=0.75]    (50.01,60) -- (50.01,85) ;
\draw  [line width=0.75]  (175.02,155) -- (245.02,155) -- (245.02,170) -- (175.02,170) -- cycle ;
\draw  [line width=0.75]  (175.02,190) -- (205.02,190) -- (205.02,205) -- (175.02,205) -- cycle ;
\draw [color={rgb, 255:red, 0; green, 0; blue, 0 }  ,draw opacity=1 ]   (170.02,205) .. controls (170.27,214.55) and (172.1,219.78) .. (180.02,220) .. controls (187.93,220.22) and (189.6,215.88) .. (190.02,205) ;
\draw [color={rgb, 255:red, 0; green, 0; blue, 0 }  ,draw opacity=1 ]   (140.02,125) .. controls (140.27,115.45) and (142.1,110.22) .. (150.02,110) .. controls (157.12,109.81) and (159.19,113.27) .. (159.84,121.84) .. controls (159.92,122.82) and (159.97,123.88) .. (160.02,125) ;
\draw [line width=0.75]    (140.02,125) -- (140.02,165) ;
\draw [shift={(140.02,138)}, rotate = 90] [color={rgb, 255:red, 0; green, 0; blue, 0 }  ][line width=0.75]    (10.93,-3.29) .. controls (6.95,-1.4) and (3.31,-0.3) .. (0,0) .. controls (3.31,0.3) and (6.95,1.4) .. (10.93,3.29)   ;
\draw [color={rgb, 255:red, 0; green, 0; blue, 0 }  ,draw opacity=1 ]   (140.02,165) .. controls (139.43,202.64) and (170.23,184.24) .. (170.02,205) ;
\draw [color={rgb, 255:red, 75; green, 17; blue, 33 }  ,draw opacity=1 ][line width=0.75]    (220.02,100) -- (220.02,125) ;
\draw [color={rgb, 255:red, 0; green, 0; blue, 0 }  ,draw opacity=1 ][line width=0.75]    (200.02,170) -- (200.02,190) ;
\draw [color={rgb, 255:red, 0; green, 0; blue, 0 }  ,draw opacity=1 ][line width=0.75]    (240.02,170) -- (240.02,225) ;
\draw [color={rgb, 255:red, 75; green, 17; blue, 33 }  ,draw opacity=1 ]   (200.02,125) .. controls (200.23,138.24) and (180.89,130.24) .. (180.02,145) ;
\draw [color={rgb, 255:red, 75; green, 17; blue, 33 }  ,draw opacity=1 ]   (220.02,125) .. controls (220.23,138.24) and (200.89,130.24) .. (200.02,145) ;
\draw    (160.02,125) .. controls (159.03,143.04) and (220.23,133.04) .. (220.02,155) ;
\draw [color={rgb, 255:red, 75; green, 17; blue, 33 }  ,draw opacity=1 ][line width=0.75]    (200.02,100) -- (200.02,125) ;
\draw [color={rgb, 255:red, 0; green, 0; blue, 0 }  ,draw opacity=1 ][line width=0.75]    (240.02,100) -- (240.02,155) ;
\draw [color={rgb, 255:red, 75; green, 17; blue, 33 }  ,draw opacity=1 ][line width=0.75]    (160.02,145) -- (160.02,170) ;
\draw [color={rgb, 255:red, 75; green, 17; blue, 33 }  ,draw opacity=1 ]   (160.02,170) .. controls (160.23,183.24) and (178.63,175.44) .. (180.02,190) ;
\draw [color={rgb, 255:red, 75; green, 17; blue, 33 }  ,draw opacity=1 ][line width=0.75]    (200.02,145) -- (200.02,155) ;
\draw [color={rgb, 255:red, 75; green, 17; blue, 33 }  ,draw opacity=1 ][line width=0.75]    (180.02,145) -- (180.02,155) ;
\draw [color={rgb, 255:red, 75; green, 17; blue, 33 }  ,draw opacity=1 ]   (180.02,125) .. controls (180.08,128.8) and (178.53,130.85) .. (176.26,132.2) .. controls (170.65,135.53) and (160.64,134.48) .. (160.02,145) ;
\draw [color={rgb, 255:red, 75; green, 17; blue, 33 }  ,draw opacity=1 ][line width=0.75]    (180.02,100) -- (180.02,125) ;
\draw [color={rgb, 255:red, 75; green, 17; blue, 33 }  ,draw opacity=1 ]   (310.01,140) .. controls (310.22,153.24) and (290.89,145.24) .. (290.01,160) ;
\draw [color={rgb, 255:red, 0; green, 0; blue, 0 }  ,draw opacity=1 ]   (290.01,140) .. controls (290.22,153.24) and (310.89,145.24) .. (310.01,160) ;
\draw [color={rgb, 255:red, 0; green, 0; blue, 0 }  ,draw opacity=1 ][line width=0.75]    (270.01,140) -- (270.01,155) ;
\draw  [line width=0.75]  (195.01,30) -- (245.01,30) -- (245.01,45) -- (195.01,45) -- cycle ;
\draw  [line width=0.75]  (175.01,65) -- (205.01,65) -- (205.01,80) -- (175.01,80) -- cycle ;
\draw [color={rgb, 255:red, 75; green, 17; blue, 33 }  ,draw opacity=1 ][line width=0.75]    (220.01,10) -- (220.01,20) ;
\draw [color={rgb, 255:red, 75; green, 17; blue, 33 }  ,draw opacity=1 ][line width=0.75]    (220.01,65) -- (220.01,100) ;
\draw [color={rgb, 255:red, 0; green, 0; blue, 0 }  ,draw opacity=1 ][line width=0.75]    (240.01,45) -- (240.01,85) ;
\draw [color={rgb, 255:red, 75; green, 17; blue, 33 }  ,draw opacity=1 ][line width=0.75]    (200.01,10) -- (200.01,20) ;
\draw [color={rgb, 255:red, 0; green, 0; blue, 0 }  ,draw opacity=1 ][line width=0.75]    (240.01,10) -- (240.01,30) ;
\draw [color={rgb, 255:red, 75; green, 17; blue, 33 }  ,draw opacity=1 ][line width=0.75]    (180.01,10) -- (180.01,65) ;
\draw [color={rgb, 255:red, 75; green, 17; blue, 33 }  ,draw opacity=1 ][line width=0.75]    (220.01,20) -- (220.01,30) ;
\draw [color={rgb, 255:red, 75; green, 17; blue, 33 }  ,draw opacity=1 ][line width=0.75]    (200.01,20) -- (200.01,30) ;
\draw  [line width=0.75]  (450.01,145) -- (520.01,145) -- (520.01,160) -- (450.01,160) -- cycle ;
\draw [color={rgb, 255:red, 75; green, 17; blue, 33 }  ,draw opacity=1 ][line width=0.75]    (455.01,65) -- (455.01,80) ;
\draw  [line width=0.75]  (430.01,115) -- (460.01,115) -- (460.01,130) -- (430.01,130) -- cycle ;
\draw [color={rgb, 255:red, 0; green, 0; blue, 0 }  ,draw opacity=1 ][line width=0.75]    (515.01,160) -- (515.01,180) ;
\draw [color={rgb, 255:red, 0; green, 0; blue, 0 }  ,draw opacity=1 ]   (435.01,160) .. controls (435.26,169.55) and (437.09,174.78) .. (445.01,175) .. controls (452.92,175.22) and (454.59,170.88) .. (455.01,160) ;
\draw [color={rgb, 255:red, 0; green, 0; blue, 0 }  ,draw opacity=1 ]   (415.01,95) .. controls (415.26,85.45) and (417.09,80.22) .. (425.01,80) .. controls (432.92,79.78) and (434.59,84.12) .. (435.01,95) ;
\draw [color={rgb, 255:red, 0; green, 0; blue, 0 }  ,draw opacity=1 ][line width=0.75]    (455.01,130) -- (455.01,145) ;
\draw [color={rgb, 255:red, 75; green, 17; blue, 33 }  ,draw opacity=1 ][line width=0.75]    (495.01,65) -- (495.01,145) ;
\draw  [color={rgb, 255:red, 75; green, 17; blue, 33 }  ,draw opacity=1 ][line width=0.75]  (440.01,80) -- (470.01,80) -- (470.01,95) -- (440.01,95) -- cycle ;
\draw [color={rgb, 255:red, 0; green, 0; blue, 0 }  ,draw opacity=1 ][line width=0.75]    (515.01,65) -- (515.01,145) ;
\draw [color={rgb, 255:red, 75; green, 17; blue, 33 }  ,draw opacity=1 ][line width=0.75]    (475.01,65) -- (475.01,145) ;
\draw    (415.01,110) .. controls (415.21,129.65) and (434.87,140.65) .. (435.01,160) ;
\draw [shift={(421.25,129.12)}, rotate = 57.87] [color={rgb, 255:red, 0; green, 0; blue, 0 }  ][line width=0.75]    (10.93,-3.29) .. controls (6.95,-1.4) and (3.31,-0.3) .. (0,0) .. controls (3.31,0.3) and (6.95,1.4) .. (10.93,3.29)   ;
\draw [color={rgb, 255:red, 75; green, 17; blue, 33 }  ,draw opacity=1 ]   (455.01,95) .. controls (455.22,108.24) and (435.89,100.24) .. (435.01,115) ;
\draw [color={rgb, 255:red, 0; green, 0; blue, 0 }  ,draw opacity=1 ]   (435.01,95) .. controls (435.22,108.24) and (455.89,100.24) .. (455.01,115) ;
\draw [color={rgb, 255:red, 0; green, 0; blue, 0 }  ,draw opacity=1 ][line width=0.75]    (415.01,95) -- (415.01,110) ;
\draw [color={rgb, 255:red, 75; green, 17; blue, 33 }  ,draw opacity=1 ][line width=0.75]    (200.01,80) -- (200.01,90) ;
\draw [color={rgb, 255:red, 75; green, 17; blue, 33 }  ,draw opacity=1 ][line width=0.75]    (180.01,80) -- (180.01,90) ;
\draw [color={rgb, 255:red, 75; green, 17; blue, 33 }  ,draw opacity=1 ][line width=0.75]    (200.01,90) -- (200.01,100) ;
\draw [color={rgb, 255:red, 75; green, 17; blue, 33 }  ,draw opacity=1 ][line width=0.75]    (180.01,90) -- (180.01,100) ;
\draw [color={rgb, 255:red, 75; green, 17; blue, 33 }  ,draw opacity=1 ]   (220.01,45) .. controls (220.22,58.24) and (200.89,50.24) .. (200.01,65) ;
\draw [color={rgb, 255:red, 75; green, 17; blue, 33 }  ,draw opacity=1 ]   (200.01,45) .. controls (200.22,58.24) and (219.57,50.12) .. (220.01,65) ;
\draw  [line width=0.75]  (325.01,30) -- (375.01,30) -- (375.01,45) -- (325.01,45) -- cycle ;
\draw  [line width=0.75]  (305.01,65) -- (335.01,65) -- (335.01,80) -- (305.01,80) -- cycle ;
\draw [color={rgb, 255:red, 75; green, 17; blue, 33 }  ,draw opacity=1 ][line width=0.75]    (350.01,10) -- (350.01,20) ;
\draw [color={rgb, 255:red, 75; green, 17; blue, 33 }  ,draw opacity=1 ][line width=0.75]    (350.01,65) -- (350.01,110) ;
\draw [color={rgb, 255:red, 0; green, 0; blue, 0 }  ,draw opacity=1 ][line width=0.75]    (370.01,45) -- (370.01,85) ;
\draw [color={rgb, 255:red, 75; green, 17; blue, 33 }  ,draw opacity=1 ][line width=0.75]    (330.01,10) -- (330.01,20) ;
\draw [color={rgb, 255:red, 0; green, 0; blue, 0 }  ,draw opacity=1 ][line width=0.75]    (370.01,10) -- (370.01,30) ;
\draw [color={rgb, 255:red, 75; green, 17; blue, 33 }  ,draw opacity=1 ][line width=0.75]    (310.01,10) -- (310.01,65) ;
\draw [color={rgb, 255:red, 75; green, 17; blue, 33 }  ,draw opacity=1 ][line width=0.75]    (350.01,20) -- (350.01,30) ;
\draw [color={rgb, 255:red, 75; green, 17; blue, 33 }  ,draw opacity=1 ][line width=0.75]    (330.01,20) -- (330.01,30) ;
\draw [color={rgb, 255:red, 75; green, 17; blue, 33 }  ,draw opacity=1 ][line width=0.75]    (330.01,80) -- (330.01,110) ;
\draw [color={rgb, 255:red, 75; green, 17; blue, 33 }  ,draw opacity=1 ][line width=0.75]    (310.01,80) -- (310.01,110) ;
\draw [color={rgb, 255:red, 75; green, 17; blue, 33 }  ,draw opacity=1 ]   (350.01,45) .. controls (350.22,58.24) and (330.89,50.24) .. (330.01,65) ;
\draw [color={rgb, 255:red, 75; green, 17; blue, 33 }  ,draw opacity=1 ]   (330.01,45) .. controls (330.22,58.24) and (349.57,50.12) .. (350.01,65) ;
\draw  [draw opacity=0][fill={rgb, 255:red, 75; green, 17; blue, 33 }  ,fill opacity=0.15 ] (360.01,85) -- (380.01,85) -- (380.01,100) -- (360.01,100) -- cycle ;
\draw [color={rgb, 255:red, 0; green, 0; blue, 0 }  ,draw opacity=1 ][line width=0.75]  [dash pattern={on 0.84pt off 2.51pt}]  (360.01,85) -- (380.01,85) ;
\draw [color={rgb, 255:red, 0; green, 0; blue, 0 }  ,draw opacity=1 ][line width=0.75]  [dash pattern={on 0.84pt off 2.51pt}]  (360.01,100) -- (380.01,100) ;
\draw  [draw opacity=0][fill={rgb, 255:red, 75; green, 17; blue, 33 }  ,fill opacity=0.15 ] (230.01,85) -- (250.01,85) -- (250.01,100) -- (230.01,100) -- cycle ;
\draw [color={rgb, 255:red, 0; green, 0; blue, 0 }  ,draw opacity=1 ][line width=0.75]  [dash pattern={on 0.84pt off 2.51pt}]  (230.01,85) -- (250.01,85) ;
\draw [color={rgb, 255:red, 0; green, 0; blue, 0 }  ,draw opacity=1 ][line width=0.75]  [dash pattern={on 0.84pt off 2.51pt}]  (230.01,100) -- (250.01,100) ;
\draw  [draw opacity=0][line width=0.75]  (115.01,84.5) -- (145.01,84.5) -- (145.01,99.5) -- (115.01,99.5) -- cycle ;
\draw  [draw opacity=0][line width=0.75]  (255.01,85) -- (285.01,85) -- (285.01,100) -- (255.01,100) -- cycle ;
\draw  [draw opacity=0][line width=0.75]  (380.01,85) -- (410.01,85) -- (410.01,100) -- (380.01,100) -- cycle ;

\draw (80.01,122.5) node  [font=\footnotesize]  {$f$};
\draw (60.01,157.5) node  [font=\footnotesize]  {$r$};
\draw (340.01,197.5) node  [font=\footnotesize]  {$f$};
\draw (300.01,167.5) node  [font=\footnotesize]  {$\partial r$};
\draw (310.01,132.5) node  [font=\footnotesize,color={rgb, 255:red, 75; green, 17; blue, 33 }  ,opacity=1 ]  {$\mathsf{wait}$};
\draw (210.02,162.5) node  [font=\footnotesize]  {$f$};
\draw (190.02,197.5) node  [font=\footnotesize]  {$r$};
\draw (220.01,37.5) node  [font=\footnotesize]  {$f_{0}$};
\draw (190.01,72.5) node  [font=\footnotesize]  {$r_{0}$};
\draw (485.01,152.5) node  [font=\footnotesize]  {$f$};
\draw (445.01,122.5) node  [font=\footnotesize]  {$\partial r$};
\draw (455.01,87.5) node  [font=\footnotesize,color={rgb, 255:red, 75; green, 17; blue, 33 }  ,opacity=1 ]  {$\mathsf{wait}$};
\draw (350.01,37.5) node  [font=\footnotesize]  {$f_{0}$};
\draw (320.01,72.5) node  [font=\footnotesize]  {$r_{0}$};
\draw (130.01,92.5) node  [font=\footnotesize]  {$=$};
\draw (270.01,93) node  [font=\footnotesize]  {$=$};
\draw (395.01,93) node  [font=\footnotesize]  {$=$};

\end{tikzpicture}
 \caption{The sliding axiom (A5).}
\label{strings:sliding}
\end{figure}

\begin{lem}\label{lemma:sliding}
  The structure $(\STREAM, \fbk)$ satisfies the sliding axiom (A5).
  Given two streams
  $f \in \STREAM(\delay \stream{S} \tensor \stream{X}, \stream{T} \tensor \stream{Y})$, and
  $r \in \STREAM(\stream{T}, \stream{S})$
  we claim that, $\fbk^{\stream{S}}(f ⨾ (r \tensor \im)) = 
  \fbk^{\stream{T}}((\delay r \tensor \im) ⨾ f)$.
\end{lem}
\begin{proof}
  Consequence of \Cref{lemma:slidinglong}, after applying the necessary coherence morphisms.
\end{proof}

\begin{thm}[From \Cref{th:monoidalstreamsfeedback}]
  \label{th:monoidalstreamsfeedback}
  \label{th:appendix:monoidalstreamsfeedback}
  \MonoidalStreams{} over a \productive{} \symmetricMonoidalCategory{}
  $(\catC, \tensor, \monoidalunit)$ form a \feedbackMonoidalCategory{}
  $(\STREAM, \fbk)$ over the functor $\partial \colon \NcatC \to \NcatC$.
\end{thm}
\begin{proof}
  We have proven that it satisfies all the feedback axioms:
  \begin{itemize}
    \item the tightening axiom by \Cref{lemma:tightening},
    \item the vanishing axiom by \Cref{lemma:vanishing},
    \item the joining axiom by \Cref{lemma:joining},
    \item the strength axiom by \Cref{lemma:strength},
    \item and the sliding axiom by \Cref{lemma:sliding}.
  \end{itemize}
  Thus, it is a feedback structure.
\end{proof}

\begin{cor}\label{prop:functor-sequences-streams}
  There is a ``semantics'' identity-on-objects feedback monoidal functor
  \(\Semantics{} \colon \St_{\delay}{\NcatC} \to \STREAM\) 
  from the free \feedbackMonoidalCategory{} on the functor $\partial \colon \NcatC \to \NcatC$ to the category of monoidal streams.
  Every \extensionalStatefulSequence{} $\extseq{f_{n} \colon M_{n-1} \tensor X_{n} \to Y_{n} \tensor M_{n}}$ gives a monoidal stream $\Semantics{}(f)$, which is defined by $M(\Semantics{}(f)) = M_{0}$,
  $$\now(\Semantics{}(f)) = f_{0},\mbox{ and }\later(\Semantics{}(f)) = \Semantics{}(\tail{f}),$$
  and this is well-defined.
  Moreover, this functor is full when \(\catC\) is \productive{}; it is not generally faithful.
\end{cor}
\begin{proof}
  We construct $\Semantics{}$ from \Cref{th:monoidalstreamsfeedback,th:free-feedback}. Moreover, when \(\catC\) is \productive{}, by~\Cref{theorem:observationalfinalcoalgebra}, \monoidalStreams{} are \extensionalSequences{} quotiented by \observationalEquality{}, giving the fullness of the functor.
\end{proof}
\section{Cartesian Streams}
\label{section:classicalstreams}

Dataflow languages such as \Lucid{} or \Lustre{} \cite{wadge1985lucid,halbwachs1991lustre} can be thought
of as using an underlying cartesian monoidal structure: we can copy and discard variables and resources without affecting the order of operations.
These abilities correspond exactly to cartesianness thanks to Fox's theorem (\Cref{th:fox}, see \cite{fox76}).

\subsection{Causal stream functions}

In the cartesian case, there is available literature on the categorical semantics of dataflow programming languages~\cite{benveniste93,uustalu2008comonadic,cousot19,delpeuch19,oliveira84}.
Uustalu and Vene~\cite{uustalu05} provide elegant comonadic semantics to a \Lucid-like programming language using the non-empty list comonad.
In their framework, \emph{streams} with types \(\stream{X} = \streamExpr{X}\) are families of elements $\mathbf{1} \to X_{n}$.
\emph{Causal stream functions} from \(\stream{X} = \streamExpr{X}\) to \(\stream{Y} = \streamExpr{Y}\) are families of functions $f_{n} \colon X_{0} \times \dots \times X_{n} \to Y_{n}$.
Equivalently, they are, respectively, the states $(\stream{1} \to \stream{X})$ and morphisms $(\stream{X} \to \stream{Y})$ of the cokleisli category of the comonad $\neList \colon \NSet \to \NSet$ defined by
\[(\neList(\stream{X}))_{n} \coloneqq \prod^{n}_{i=0} X_{i}.\]
The functor underlying this comonad can be defined on all symmetric monoidal categories.

\begin{defi}
  Let $(\catC,\tensor,\monoidalunit)$ be a symetric monoidal category.
  There is a functor
  $\defining{linknelist}{\ensuremath{\fun{List}^{+}}} \colon \catC \to \catC$ defined on objects by
  \[\neList(X)_{n} \coloneqq \bigotimes^{n}_{i=0} X_{i}.\]
  This functor is oplax monoidal with the natural transformations $\psi^{+}_{0} \colon \neList(\monoidalunit) \to \monoidalunit$ and $\psi_{X,Y} \colon \neList(X \tensor Y) \to \neList(X) \otimes \neList(Y)$ given by symmetries, associators and unitors.
\end{defi}

The comonad structure of \(\neList\), however, can be extended to other base categories \emph{only} as long as $\catC$ is cartesian.
More precisely, the structure of \(\neList\) on cartesian categories is that of an \emph{oplax monoidal comonad}.

\begin{defi}[Oplax monoidal comonad] In a monoidal category $(\catC,\otimes,I)$,
  a comonad $(R,\varepsilon,\delta)$
  is an \emph{oplax monoidal comonad} when the endofunctor
  $R \colon \catC \to \catC$ is oplax monoidal with laxators
  $\psi_{X,Y} \colon R(X \otimes Y) \to RX \otimes RY$ and $\psi^{I} \colon RI \to I$,
  and both the counit $\varepsilon_{X} \colon RX \to X$ and the comultiplication
  $\delta_{X} \colon RX \to RRX$ are monoidal natural transformations.

  Explicitly, the comonad structure needs to satisfy: $\varepsilon_{I} = \psi_{0}$,
  $\varepsilon_{X \otimes Y} = \psi_{X,Y} \comp (\varepsilon_X \otimes \varepsilon_Y)$,
  $\delta_{I} \comp \psi_{0} \comp \psi_{0} = \psi_{0}$ and
  $\delta_{X \otimes Y} \comp \psi_{X,Y} \comp \psi_{RX,RY} = 
  \psi_{X,Y} \comp (\delta_{X} \otimes \delta_{Y})$.

  Alternatively, an \emph{oplax monoidal comonad} is a comonoid in the bicategory $\mathbf{MonOplax}$ of oplax monoidal functors with composition and monoidal natural transformations.
\end{defi}

Indeed, we can prove (\Cref{th:nelist}) that the mere existence of such a comonad implies cartesianness of the base category.
For this, we introduce a refined version of Fox's theorem (\Cref{th:refinedfoxappendix}).

\subsection{Fox's theorem}

Fox's theorem~\cite{fox76} is a classical characterisation of cartesian monoidal categories in terms of the existence of a uniform cocommutative comonoid structure on all objects of a monoidal category.

\begin{thm}[Fox's theorem~\cite{fox76}]\label{th:fox}
  A \symmetricMonoidalCategory{} $(\catC,\otimes,I)$ is cartesian monoidal if and only if every object $X \in \catC$ has a cocommutative comonoid structure $(X,\varepsilon_{X},\delta_{X})$, every morphism of the category $f \colon X \to Y$ is a comonoid homomorphism, and this structure is uniform across the monoidal category, meaning that
  $\varepsilon_{X \otimes Y} = \varepsilon_{X} \otimes \varepsilon_{Y}$, that
  $\varepsilon_{I} = \mathrm{id}$, that $\delta_{I} = \mathrm{id}$ and that
  $\delta_{X \otimes Y} = (\delta_{X} \otimes \delta_{Y}) \dcomp (\mathrm{id} \otimes \sigma_{X,Y} \otimes \mathrm{id})$.
\end{thm}

Most sources ask the comonoid structure in Fox's theorem
(\Cref{th:fox}) to be cocommutative~\cite{fox76,fong2019supplying}.
However, cocommutativity and coassociativity of the comonoid structure are implied by uniformity and naturality of the structure.
We present an original refined version of Fox's theorem, which asks for a counital comagma structure on every object.

\begin{defi}
  A \emph{comagma} in a monoidal category \(\cat{C}\) is a pair \((X, \delta)\) of an object \(X\) of \(\cat{C}\) and a morphism \(\delta \colon X \to X \tensor X\).
  A comagma is \emph{counital} if, additionally, we specify a morphism \(\varepsilon \colon X \to I\) and this makes \(\delta\) unital: \(\delta \dcomp (\id{} \tensor \varepsilon) = \id{} = \delta \dcomp (\varepsilon \tensor \id{})\) (\Cref{fig:comagma}, left).
  A \emph{homomorphism} of counital comagmas \(f \colon (X,\varepsilon_{X},\delta_{X}) \to (Y,\varepsilon_{Y},\delta_{Y})\) is a morphism \(f \colon X \to Y\) that preserves the counit and the comultiplication: \(f \dcomp \delta_{Y} = \delta_{X} \dcomp (f \tensor f)\) and \(f \dcomp \varepsilon_{Y} = \varepsilon_{X}\) (\Cref{fig:comagma}, right).
\end{defi}
\begin{figure}[h!]
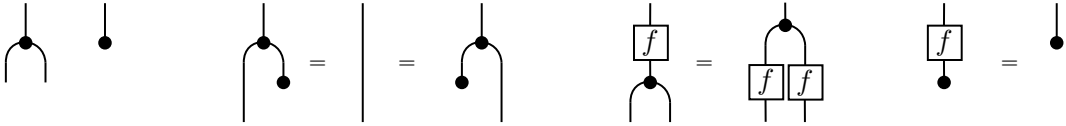

  \[\comagmaFig\]
  \caption{Structure and axioms of a counital comagma (left) and axioms of counital comagma homomorphism (right).\label{fig:comagma}}
\end{figure}

\begin{thm}[Refined Fox's theorem, c.f.~{{\cite{mellies2009categorical}}}]\label{th:refinedfoxappendix}
  A symmetric monoidal category $(\catC,\otimes,I)$ is cartesian monoidal if and
  only if every object $X \in \catC$ has a counital comagma structure
  $(X,\varepsilon_{X},\delta_{X})$, or $(X,\blackComonoidUnit_{X},\blackComonoid_{X})$, every morphism of the category
  $f \colon X \to Y$ is a comagma homomorphism, and this structure is uniform
  across the monoidal category: meaning that
  $\varepsilon_{X \otimes Y} = \varepsilon_{X} \otimes \varepsilon_{Y}$,
  $\varepsilon_{I} = \mathrm{id}$, $\delta_{I} = \mathrm{id}$ and
  $\delta_{X \otimes Y} = (\delta_{X} \otimes \delta_{Y}) ; (\im \tensor \sigma_{X,Y} \tensor \im)$.
\end{thm}
\begin{proof}
  We prove that such a comagma structure is necessarily coassociative and cocommutative.
  Note that any comagma homomorphism $f \colon A \to B$ must satisfy $\delta_{A} \comp (f \tensor f) = f \comp \delta_{B}$.
  In particular, $\delta_{X} \colon X \to X \tensor X$ must itself be a comagma homomorphism
  (see \Cref{figure:comultiplication}), meaning that
  \begin{equation}\label{eq:comultiplicationhomomorphism}
    \delta_{X} \comp (\delta_X \tensor \delta_{X}) =
    \delta_{X} \comp \delta_{X \tensor X} =
    \delta_{X} \comp (\delta_X \tensor \delta_{X}) \comp 
    (\im \tensor \sigma_{X,Y} \tensor \im),
  \end{equation}
  where the second equality follows by uniformity.
  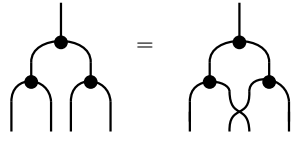
\begin{figure}[h]
\tikzset{every picture/.style={line width=0.85pt}}
\begin{tikzpicture}[x=0.75pt,y=0.75pt,yscale=-1,xscale=1]
\draw    (25,65) .. controls (24.8,50.8) and (44.2,51.6) .. (45,65) ;
\draw  [fill={rgb, 255:red, 0; green, 0; blue, 0 }  ,fill opacity=1 ] (32.1,55) .. controls (32.1,53.4) and (33.4,52.1) .. (35,52.1) .. controls (36.6,52.1) and (37.9,53.4) .. (37.9,55) .. controls (37.9,56.6) and (36.6,57.9) .. (35,57.9) .. controls (33.4,57.9) and (32.1,56.6) .. (32.1,55) -- cycle ;
\draw    (35,45) -- (35,55) ;
\draw    (35,45) .. controls (34.8,30.8) and (65.2,30.8) .. (65,45) ;
\draw  [fill={rgb, 255:red, 0; green, 0; blue, 0 }  ,fill opacity=1 ] (47.1,35) .. controls (47.1,33.4) and (48.4,32.1) .. (50,32.1) .. controls (51.6,32.1) and (52.9,33.4) .. (52.9,35) .. controls (52.9,36.6) and (51.6,37.9) .. (50,37.9) .. controls (48.4,37.9) and (47.1,36.6) .. (47.1,35) -- cycle ;
\draw    (25,65) -- (25,80) ;
\draw    (45,65) -- (45,80) ;
\draw    (50,15) -- (50,35) ;
\draw    (55,65) .. controls (54.8,50.8) and (74.2,51.6) .. (75,65) ;
\draw  [fill={rgb, 255:red, 0; green, 0; blue, 0 }  ,fill opacity=1 ] (62.1,55) .. controls (62.1,53.4) and (63.4,52.1) .. (65,52.1) .. controls (66.6,52.1) and (67.9,53.4) .. (67.9,55) .. controls (67.9,56.6) and (66.6,57.9) .. (65,57.9) .. controls (63.4,57.9) and (62.1,56.6) .. (62.1,55) -- cycle ;
\draw    (65,45) -- (65,55) ;
\draw    (55,65) -- (55,80) ;
\draw    (75,65) -- (75,80) ;
\draw    (115,65) .. controls (114.8,50.8) and (135.4,53.8) .. (135,60) ;
\draw  [fill={rgb, 255:red, 0; green, 0; blue, 0 }  ,fill opacity=1 ] (122.1,55) .. controls (122.1,53.4) and (123.4,52.1) .. (125,52.1) .. controls (126.6,52.1) and (127.9,53.4) .. (127.9,55) .. controls (127.9,56.6) and (126.6,57.9) .. (125,57.9) .. controls (123.4,57.9) and (122.1,56.6) .. (122.1,55) -- cycle ;
\draw    (125,45) -- (125,55) ;
\draw    (125,45) .. controls (124.8,30.8) and (155.2,30.8) .. (155,45) ;
\draw  [fill={rgb, 255:red, 0; green, 0; blue, 0 }  ,fill opacity=1 ] (137.1,35) .. controls (137.1,33.4) and (138.4,32.1) .. (140,32.1) .. controls (141.6,32.1) and (142.9,33.4) .. (142.9,35) .. controls (142.9,36.6) and (141.6,37.9) .. (140,37.9) .. controls (138.4,37.9) and (137.1,36.6) .. (137.1,35) -- cycle ;
\draw    (115,70) -- (115,80) ;
\draw    (140,15) -- (140,35) ;
\draw    (145,60) .. controls (144.6,50.2) and (164.2,51.6) .. (165,65) ;
\draw  [fill={rgb, 255:red, 0; green, 0; blue, 0 }  ,fill opacity=1 ] (152.1,55) .. controls (152.1,53.4) and (153.4,52.1) .. (155,52.1) .. controls (156.6,52.1) and (157.9,53.4) .. (157.9,55) .. controls (157.9,56.6) and (156.6,57.9) .. (155,57.9) .. controls (153.4,57.9) and (152.1,56.6) .. (152.1,55) -- cycle ;
\draw    (155,45) -- (155,55) ;
\draw    (165,65) -- (165,75) ;
\draw    (115,65) -- (115,75) ;
\draw    (135,60) .. controls (135.4,74.2) and (144.6,65.8) .. (145,80) ;
\draw    (165,70) -- (165,80) ;
\draw    (145,60) .. controls (145.4,74.2) and (134.6,65.8) .. (135,80) ;
\draw (86,33.4) node [anchor=north west][inner sep=0.75pt]    {$=$};
\end{tikzpicture}     \caption{Comultiplication is a comagma homomorphism.}
    \label{figure:comultiplication}
  \end{figure}

  Now, we prove cocommutativity (\Cref{string:cocommutative}): composing both sides of \Cref{eq:comultiplicationhomomorphism} with $(\epsilon_{X} \tensor \im \tensor \im \tensor \epsilon_{X})$ discards the two external outputs and gives $\delta_{X} = \delta_{X} \comp \sigma_{X}$.
  \begin{figure}[H]
\tikzset{every picture/.style={line width=0.75pt}} %
\begin{tikzpicture}[x=0.75pt,y=0.75pt,yscale=-1,xscale=1]
\draw    (72.1,65) .. controls (71.9,50.8) and (91.3,51.6) .. (92.1,65) ;
\draw  [fill={rgb, 255:red, 0; green, 0; blue, 0 }  ,fill opacity=1 ] (79.2,55) .. controls (79.2,53.4) and (80.5,52.1) .. (82.1,52.1) .. controls (83.7,52.1) and (85,53.4) .. (85,55) .. controls (85,56.6) and (83.7,57.9) .. (82.1,57.9) .. controls (80.5,57.9) and (79.2,56.6) .. (79.2,55) -- cycle ;
\draw    (82.1,45) -- (82.1,55) ;
\draw    (82.1,45) .. controls (81.9,30.8) and (112.3,30.8) .. (112.1,45) ;
\draw  [fill={rgb, 255:red, 0; green, 0; blue, 0 }  ,fill opacity=1 ] (94.2,35) .. controls (94.2,33.4) and (95.5,32.1) .. (97.1,32.1) .. controls (98.7,32.1) and (100,33.4) .. (100,35) .. controls (100,36.6) and (98.7,37.9) .. (97.1,37.9) .. controls (95.5,37.9) and (94.2,36.6) .. (94.2,35) -- cycle ;
\draw    (92.1,65) -- (92.1,80) ;
\draw    (97.1,15) -- (97.1,35) ;
\draw    (102.1,65) .. controls (101.9,50.8) and (121.3,51.6) .. (122.1,65) ;
\draw  [fill={rgb, 255:red, 0; green, 0; blue, 0 }  ,fill opacity=1 ] (109.2,55) .. controls (109.2,53.4) and (110.5,52.1) .. (112.1,52.1) .. controls (113.7,52.1) and (115,53.4) .. (115,55) .. controls (115,56.6) and (113.7,57.9) .. (112.1,57.9) .. controls (110.5,57.9) and (109.2,56.6) .. (109.2,55) -- cycle ;
\draw    (112.1,45) -- (112.1,55) ;
\draw    (102.1,65) -- (102.1,80) ;
\draw    (122.1,65) -- (122.1,70) ;
\draw    (150,65) .. controls (149.8,50.8) and (170.4,53.8) .. (170,60) ;
\draw  [fill={rgb, 255:red, 0; green, 0; blue, 0 }  ,fill opacity=1 ] (157.1,55) .. controls (157.1,53.4) and (158.4,52.1) .. (160,52.1) .. controls (161.6,52.1) and (162.9,53.4) .. (162.9,55) .. controls (162.9,56.6) and (161.6,57.9) .. (160,57.9) .. controls (158.4,57.9) and (157.1,56.6) .. (157.1,55) -- cycle ;
\draw    (160,45) -- (160,55) ;
\draw    (160,45) .. controls (159.8,30.8) and (190.2,30.8) .. (190,45) ;
\draw  [fill={rgb, 255:red, 0; green, 0; blue, 0 }  ,fill opacity=1 ] (172.1,35) .. controls (172.1,33.4) and (173.4,32.1) .. (175,32.1) .. controls (176.6,32.1) and (177.9,33.4) .. (177.9,35) .. controls (177.9,36.6) and (176.6,37.9) .. (175,37.9) .. controls (173.4,37.9) and (172.1,36.6) .. (172.1,35) -- cycle ;
\draw    (175,15) -- (175,35) ;
\draw    (180,60) .. controls (179.6,50.2) and (199.2,51.6) .. (200,65) ;
\draw  [fill={rgb, 255:red, 0; green, 0; blue, 0 }  ,fill opacity=1 ] (187.1,55) .. controls (187.1,53.4) and (188.4,52.1) .. (190,52.1) .. controls (191.6,52.1) and (192.9,53.4) .. (192.9,55) .. controls (192.9,56.6) and (191.6,57.9) .. (190,57.9) .. controls (188.4,57.9) and (187.1,56.6) .. (187.1,55) -- cycle ;
\draw    (190,45) -- (190,55) ;
\draw    (200,65) -- (200,70) ;
\draw    (150,65) -- (150,70) ;
\draw    (170,60) .. controls (170.4,74.2) and (179.6,65.8) .. (180,80) ;
\draw    (180,60) .. controls (180.4,74.2) and (169.6,65.8) .. (170,80) ;
\draw  [fill={rgb, 255:red, 0; green, 0; blue, 0 }  ,fill opacity=1 ] (119.2,70) .. controls (119.2,68.4) and (120.5,67.1) .. (122.1,67.1) .. controls (123.7,67.1) and (125,68.4) .. (125,70) .. controls (125,71.6) and (123.7,72.9) .. (122.1,72.9) .. controls (120.5,72.9) and (119.2,71.6) .. (119.2,70) -- cycle ;
\draw    (72.1,65) -- (72.1,70) ;
\draw  [fill={rgb, 255:red, 0; green, 0; blue, 0 }  ,fill opacity=1 ] (69.2,70) .. controls (69.2,68.4) and (70.5,67.1) .. (72.1,67.1) .. controls (73.7,67.1) and (75,68.4) .. (75,70) .. controls (75,71.6) and (73.7,72.9) .. (72.1,72.9) .. controls (70.5,72.9) and (69.2,71.6) .. (69.2,70) -- cycle ;
\draw    (10,45) -- (10,80) ;
\draw    (10,45) .. controls (9.8,30.8) and (40.2,30.8) .. (40,45) ;
\draw  [fill={rgb, 255:red, 0; green, 0; blue, 0 }  ,fill opacity=1 ] (22.1,35) .. controls (22.1,33.4) and (23.4,32.1) .. (25,32.1) .. controls (26.6,32.1) and (27.9,33.4) .. (27.9,35) .. controls (27.9,36.6) and (26.6,37.9) .. (25,37.9) .. controls (23.4,37.9) and (22.1,36.6) .. (22.1,35) -- cycle ;
\draw    (25,15) -- (25,35) ;
\draw    (40,45) -- (40,80) ;
\draw  [fill={rgb, 255:red, 0; green, 0; blue, 0 }  ,fill opacity=1 ] (147.1,70) .. controls (147.1,68.4) and (148.4,67.1) .. (150,67.1) .. controls (151.6,67.1) and (152.9,68.4) .. (152.9,70) .. controls (152.9,71.6) and (151.6,72.9) .. (150,72.9) .. controls (148.4,72.9) and (147.1,71.6) .. (147.1,70) -- cycle ;
\draw  [fill={rgb, 255:red, 0; green, 0; blue, 0 }  ,fill opacity=1 ] (197.1,70) .. controls (197.1,68.4) and (198.4,67.1) .. (200,67.1) .. controls (201.6,67.1) and (202.9,68.4) .. (202.9,70) .. controls (202.9,71.6) and (201.6,72.9) .. (200,72.9) .. controls (198.4,72.9) and (197.1,71.6) .. (197.1,70) -- cycle ;
\draw    (235,45) .. controls (234.8,30.8) and (265.2,30.8) .. (265,45) ;
\draw  [fill={rgb, 255:red, 0; green, 0; blue, 0 }  ,fill opacity=1 ] (247.1,35) .. controls (247.1,33.4) and (248.4,32.1) .. (250,32.1) .. controls (251.6,32.1) and (252.9,33.4) .. (252.9,35) .. controls (252.9,36.6) and (251.6,37.9) .. (250,37.9) .. controls (248.4,37.9) and (247.1,36.6) .. (247.1,35) -- cycle ;
\draw    (250,15) -- (250,35) ;
\draw    (265,45) .. controls (265.4,59.2) and (235,51.75) .. (235,80) ;
\draw    (235,45) .. controls (235.4,59.2) and (265,51.75) .. (265,80) ;
\draw (126,33.4) node [anchor=north west][inner sep=0.75pt]    {$=$};
\draw (46,33.4) node [anchor=north west][inner sep=0.75pt]    {$=$};
\draw (206,33.4) node [anchor=north west][inner sep=0.75pt]    {$=$};
\end{tikzpicture}     \caption{Cocommutativity}
    \label{string:cocommutative}
  \end{figure}
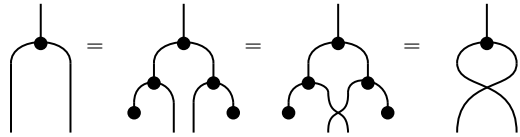

Now, we prove coassociativity (\Cref{string:coassociativity}): composing both sides of \Cref{eq:comultiplicationhomomorphism} with $(\im \tensor \epsilon_{X} \tensor \im \tensor \im)$ discards one of the middle outputs and gives $\delta_{X} \comp (\im \tensor \delta_{X}) = \delta_{X} \comp (\delta_{X} \tensor \im)$.
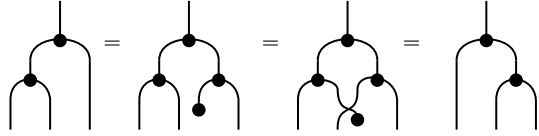
\begin{figure}
\tikzset{every picture/.style={line width=0.75pt}}
\begin{tikzpicture}[x=0.75pt,y=0.75pt,yscale=-1,xscale=1]
\draw    (80,65) .. controls (79.8,50.8) and (99.2,51.6) .. (100,65) ;
\draw  [fill={rgb, 255:red, 0; green, 0; blue, 0 }  ,fill opacity=1 ] (87.1,55) .. controls (87.1,53.4) and (88.4,52.1) .. (90,52.1) .. controls (91.6,52.1) and (92.9,53.4) .. (92.9,55) .. controls (92.9,56.6) and (91.6,57.9) .. (90,57.9) .. controls (88.4,57.9) and (87.1,56.6) .. (87.1,55) -- cycle ;
\draw    (90,45) -- (90,55) ;
\draw    (90,45) .. controls (89.8,30.8) and (120.2,30.8) .. (120,45) ;
\draw  [fill={rgb, 255:red, 0; green, 0; blue, 0 }  ,fill opacity=1 ] (102.1,35) .. controls (102.1,33.4) and (103.4,32.1) .. (105,32.1) .. controls (106.6,32.1) and (107.9,33.4) .. (107.9,35) .. controls (107.9,36.6) and (106.6,37.9) .. (105,37.9) .. controls (103.4,37.9) and (102.1,36.6) .. (102.1,35) -- cycle ;
\draw    (80,65) -- (80,80) ;
\draw    (100,65) -- (100,80) ;
\draw    (105,15) -- (105,35) ;
\draw    (110,65) .. controls (109.8,50.8) and (129.2,51.6) .. (130,65) ;
\draw  [fill={rgb, 255:red, 0; green, 0; blue, 0 }  ,fill opacity=1 ] (117.1,55) .. controls (117.1,53.4) and (118.4,52.1) .. (120,52.1) .. controls (121.6,52.1) and (122.9,53.4) .. (122.9,55) .. controls (122.9,56.6) and (121.6,57.9) .. (120,57.9) .. controls (118.4,57.9) and (117.1,56.6) .. (117.1,55) -- cycle ;
\draw    (120,45) -- (120,55) ;
\draw    (130,65) -- (130,80) ;
\draw    (110,65) -- (110,70) ;
\draw    (160,65) .. controls (159.8,50.8) and (180.4,53.8) .. (180,60) ;
\draw  [fill={rgb, 255:red, 0; green, 0; blue, 0 }  ,fill opacity=1 ] (167.1,55) .. controls (167.1,53.4) and (168.4,52.1) .. (170,52.1) .. controls (171.6,52.1) and (172.9,53.4) .. (172.9,55) .. controls (172.9,56.6) and (171.6,57.9) .. (170,57.9) .. controls (168.4,57.9) and (167.1,56.6) .. (167.1,55) -- cycle ;
\draw    (170,45) -- (170,55) ;
\draw    (170,45) .. controls (169.8,30.8) and (200.2,30.8) .. (200,45) ;
\draw  [fill={rgb, 255:red, 0; green, 0; blue, 0 }  ,fill opacity=1 ] (182.1,35) .. controls (182.1,33.4) and (183.4,32.1) .. (185,32.1) .. controls (186.6,32.1) and (187.9,33.4) .. (187.9,35) .. controls (187.9,36.6) and (186.6,37.9) .. (185,37.9) .. controls (183.4,37.9) and (182.1,36.6) .. (182.1,35) -- cycle ;
\draw    (160,70) -- (160,80) ;
\draw    (185,15) -- (185,35) ;
\draw    (190,60) .. controls (189.6,50.2) and (209.2,51.6) .. (210,65) ;
\draw  [fill={rgb, 255:red, 0; green, 0; blue, 0 }  ,fill opacity=1 ] (197.1,55) .. controls (197.1,53.4) and (198.4,52.1) .. (200,52.1) .. controls (201.6,52.1) and (202.9,53.4) .. (202.9,55) .. controls (202.9,56.6) and (201.6,57.9) .. (200,57.9) .. controls (198.4,57.9) and (197.1,56.6) .. (197.1,55) -- cycle ;
\draw    (200,45) -- (200,55) ;
\draw    (210,65) -- (210,80) ;
\draw    (160,65) -- (160,75) ;
\draw    (180,60) .. controls (180.4,74.2) and (190,66) .. (190,75) ;
\draw    (190,60) .. controls (190.4,74.2) and (179.6,65.8) .. (180,80) ;
\draw    (15,65) .. controls (14.8,50.8) and (34.2,51.6) .. (35,65) ;
\draw  [fill={rgb, 255:red, 0; green, 0; blue, 0 }  ,fill opacity=1 ] (22.1,55) .. controls (22.1,53.4) and (23.4,52.1) .. (25,52.1) .. controls (26.6,52.1) and (27.9,53.4) .. (27.9,55) .. controls (27.9,56.6) and (26.6,57.9) .. (25,57.9) .. controls (23.4,57.9) and (22.1,56.6) .. (22.1,55) -- cycle ;
\draw    (25,45) -- (25,55) ;
\draw    (25,45) .. controls (24.8,30.8) and (55.2,30.8) .. (55,45) ;
\draw  [fill={rgb, 255:red, 0; green, 0; blue, 0 }  ,fill opacity=1 ] (37.1,35) .. controls (37.1,33.4) and (38.4,32.1) .. (40,32.1) .. controls (41.6,32.1) and (42.9,33.4) .. (42.9,35) .. controls (42.9,36.6) and (41.6,37.9) .. (40,37.9) .. controls (38.4,37.9) and (37.1,36.6) .. (37.1,35) -- cycle ;
\draw    (15,65) -- (15,80) ;
\draw    (35,65) -- (35,80) ;
\draw    (40,15) -- (40,35) ;
\draw    (55,45) -- (55,80) ;
\draw  [fill={rgb, 255:red, 0; green, 0; blue, 0 }  ,fill opacity=1 ] (107.1,70) .. controls (107.1,68.4) and (108.4,67.1) .. (110,67.1) .. controls (111.6,67.1) and (112.9,68.4) .. (112.9,70) .. controls (112.9,71.6) and (111.6,72.9) .. (110,72.9) .. controls (108.4,72.9) and (107.1,71.6) .. (107.1,70) -- cycle ;
\draw  [fill={rgb, 255:red, 0; green, 0; blue, 0 }  ,fill opacity=1 ] (187.1,75) .. controls (187.1,73.4) and (188.4,72.1) .. (190,72.1) .. controls (191.6,72.1) and (192.9,73.4) .. (192.9,75) .. controls (192.9,76.6) and (191.6,77.9) .. (190,77.9) .. controls (188.4,77.9) and (187.1,76.6) .. (187.1,75) -- cycle ;
\draw    (240,45) -- (240,80) ;
\draw    (240,45) .. controls (239.8,30.8) and (270.2,30.8) .. (270,45) ;
\draw  [fill={rgb, 255:red, 0; green, 0; blue, 0 }  ,fill opacity=1 ] (252.1,35) .. controls (252.1,33.4) and (253.4,32.1) .. (255,32.1) .. controls (256.6,32.1) and (257.9,33.4) .. (257.9,35) .. controls (257.9,36.6) and (256.6,37.9) .. (255,37.9) .. controls (253.4,37.9) and (252.1,36.6) .. (252.1,35) -- cycle ;
\draw    (255,15) -- (255,35) ;
\draw    (260,65) .. controls (259.8,50.8) and (279.2,51.6) .. (280,65) ;
\draw  [fill={rgb, 255:red, 0; green, 0; blue, 0 }  ,fill opacity=1 ] (267.1,55) .. controls (267.1,53.4) and (268.4,52.1) .. (270,52.1) .. controls (271.6,52.1) and (272.9,53.4) .. (272.9,55) .. controls (272.9,56.6) and (271.6,57.9) .. (270,57.9) .. controls (268.4,57.9) and (267.1,56.6) .. (267.1,55) -- cycle ;
\draw    (270,45) -- (270,55) ;
\draw    (260,65) -- (260,80) ;
\draw    (280,65) -- (280,80) ;
\draw (140,33.4) node [anchor=north west][inner sep=0.75pt]    {$=$};
\draw (60,33.4) node [anchor=north west][inner sep=0.75pt]    {$=$};
\draw (211,33.4) node [anchor=north west][inner sep=0.75pt]    {$=$};
\end{tikzpicture}   \caption{Coassociativity}
  \label{string:coassociativity}
\end{figure}

A coassociative and cocommutative comagma is a cocommutative comonoid.
We can then apply the classical form of Fox's theorem (\Cref{th:fox}).
\end{proof}

One could hope to add effects such as probability or non-determinism to set-based streams via the bikleisli category arising from a monad-comonad distributive law $\neList \circ \stream{T} \Rightarrow \stream{T} \circ \neList$ \cite{power99,beck69}, as proposed by Uustalu and Vene~\cite{uustalu05}.
This would correspond to a lifting of the $\neList$ comonad to the kleisli category of some commutative monad $T$; the arrows $\stream{X} \to \stream{Y}$ of such a category would look as follows,
\[f_{n} \colon X_{1} \times \dots \times X_{n} \to TY_{n}.\]
However, we show that this would not result in a monoidal comonad whenever $\kleisli{T}$ is not cartesian.
To see explicitly what fails, we use the string diagrams to show how composition would not work for the case $n = 2$ (\Cref{string:klcomposition}).
This composition is not associative or unital whenever the kleisli category does not have natural comultiplications or counits, respectively (\Cref{string:klassocunital}).
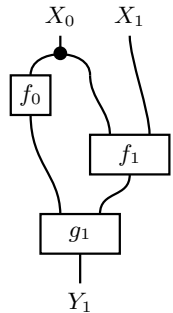
\begin{figure}[h]
\tikzset{every picture/.style={line width=0.85pt}} %
\begin{tikzpicture}[x=0.75pt,y=0.75pt,yscale=-1,xscale=1]
\draw    (325,70) -- (325,78.87) ;
\draw    (310,90) .. controls (309.8,75.8) and (340.2,75.8) .. (340,90) ;
\draw  [fill={rgb, 255:red, 0; green, 0; blue, 0 }  ,fill opacity=1 ] (322.1,78.87) .. controls (322.1,77.27) and (323.4,75.97) .. (325,75.97) .. controls (326.6,75.97) and (327.9,77.27) .. (327.9,78.87) .. controls (327.9,80.47) and (326.6,81.77) .. (325,81.77) .. controls (323.4,81.77) and (322.1,80.47) .. (322.1,78.87) -- cycle ;
\draw   (340,120) -- (380,120) -- (380,140) -- (340,140) -- cycle ;
\draw   (315,160) -- (355,160) -- (355,180) -- (315,180) -- cycle ;
\draw    (360,70) .. controls (360.14,82.14) and (370,107.83) .. (370,120) ;
\draw   (300,90) -- (320,90) -- (320,110) -- (300,110) -- cycle ;
\draw    (340,90) .. controls (340.14,102.14) and (350,107.83) .. (350,120) ;
\draw    (335,180) -- (335,195) ;
\draw    (360,140) .. controls (360.14,152.14) and (345,147.83) .. (345,160) ;
\draw    (310,110) .. controls (310.2,135) and (324.6,139.8) .. (325,160) ;
\draw (325,66.6) node [anchor=south] [inner sep=0.75pt]    {$X_{0}$};
\draw (360,130) node  [font=\footnotesize]  {$f_{1}$};
\draw (310,100) node  [font=\footnotesize]  {$f_{0}$};
\draw (335,170) node  [font=\footnotesize]  {$g_{1}$};
\draw (360,66.6) node [anchor=south] [inner sep=0.75pt]    {$X_{1}$};
\draw (335,198.4) node [anchor=north] [inner sep=0.75pt]    {$Y_{1}$};
\end{tikzpicture}
\caption{Composition in the case $n=2$.}\label{string:klcomposition}
\end{figure}
\begin{figure}[h]
  \tikzset{every picture/.style={line width=0.85pt}} %
\begin{tikzpicture}[x=0.75pt,y=0.75pt,yscale=-1,xscale=1,baseline=0cm]
\draw   (330,90) -- (370,90) -- (370,110) -- (330,110) -- cycle ;
\draw   (315,135) -- (355,135) -- (355,155) -- (315,155) -- cycle ;
\draw    (360,70) .. controls (360.14,82.14) and (360,77.83) .. (360,90) ;
\draw   (300,90) -- (320,90) -- (320,110) -- (300,110) -- cycle ;
\draw    (315,195) -- (315,215) ;
\draw    (350,110) .. controls (350.14,122.14) and (345,117.83) .. (345,130) ;
\draw    (325,70) -- (325,78.87) ;
\draw    (310,90) .. controls (309.8,75.8) and (340.2,75.8) .. (340,90) ;
\draw  [fill={rgb, 255:red, 0; green, 0; blue, 0 }  ,fill opacity=1 ] (322.1,78.87) .. controls (322.1,77.27) and (323.4,75.97) .. (325,75.97) .. controls (326.6,75.97) and (327.9,77.27) .. (327.9,78.87) .. controls (327.9,80.47) and (326.6,81.77) .. (325,81.77) .. controls (323.4,81.77) and (322.1,80.47) .. (322.1,78.87) -- cycle ;
\draw    (310,110) -- (310,118.87) ;
\draw    (295,130) .. controls (294.8,115.8) and (325.2,115.8) .. (325,130) ;
\draw  [fill={rgb, 255:red, 0; green, 0; blue, 0 }  ,fill opacity=1 ] (307.1,118.87) .. controls (307.1,117.27) and (308.4,115.97) .. (310,115.97) .. controls (311.6,115.97) and (312.9,117.27) .. (312.9,118.87) .. controls (312.9,120.47) and (311.6,121.77) .. (310,121.77) .. controls (308.4,121.77) and (307.1,120.47) .. (307.1,118.87) -- cycle ;
\draw   (285,135) -- (305,135) -- (305,155) -- (285,155) -- cycle ;
\draw   (295,175) -- (335,175) -- (335,195) -- (295,195) -- cycle ;
\draw    (295,160) .. controls (295.14,172.14) and (305,162.83) .. (305,175) ;
\draw    (335,160) .. controls (335.14,172.14) and (325,162.83) .. (325,175) ;
\draw    (435,70) -- (435,78.87) ;
\draw    (420,90) .. controls (419.8,75.8) and (450.2,75.8) .. (450,90) ;
\draw  [fill={rgb, 255:red, 0; green, 0; blue, 0 }  ,fill opacity=1 ] (432.1,78.87) .. controls (432.1,77.27) and (433.4,75.97) .. (435,75.97) .. controls (436.6,75.97) and (437.9,77.27) .. (437.9,78.87) .. controls (437.9,80.47) and (436.6,81.77) .. (435,81.77) .. controls (433.4,81.77) and (432.1,80.47) .. (432.1,78.87) -- cycle ;
\draw    (445,115) .. controls (444.8,100.8) and (475.2,100.8) .. (475,115) ;
\draw  [fill={rgb, 255:red, 0; green, 0; blue, 0 }  ,fill opacity=1 ] (457.1,103.87) .. controls (457.1,102.27) and (458.4,100.97) .. (460,100.97) .. controls (461.6,100.97) and (462.9,102.27) .. (462.9,103.87) .. controls (462.9,105.47) and (461.6,106.77) .. (460,106.77) .. controls (458.4,106.77) and (457.1,105.47) .. (457.1,103.87) -- cycle ;
\draw   (410,115) -- (430,115) -- (430,135) -- (410,135) -- cycle ;
\draw   (415,150) -- (435,150) -- (435,170) -- (415,170) -- cycle ;
\draw   (435,115) -- (455,115) -- (455,135) -- (435,135) -- cycle ;
\draw    (450,90) .. controls (450.14,102.14) and (460,91.7) .. (460,103.87) ;
\draw    (470,70) .. controls (470.14,82.14) and (495,102.83) .. (495,115) ;
\draw   (465,115) -- (505,115) -- (505,135) -- (465,135) -- cycle ;
\draw   (445,150) -- (485,150) -- (485,170) -- (445,170) -- cycle ;
\draw    (445,135) .. controls (445.14,147.14) and (455,137.83) .. (455,150) ;
\draw    (485,135) .. controls (485.14,147.14) and (475,137.83) .. (475,150) ;
\draw    (420,90) -- (420,115) ;
\draw    (420,135) .. controls (420.14,147.14) and (425,137.83) .. (425,150) ;
\draw   (425,185) -- (465,185) -- (465,205) -- (425,205) -- cycle ;
\draw    (425,170) .. controls (425.14,182.14) and (435,172.83) .. (435,185) ;
\draw    (465,170) .. controls (465.14,182.14) and (455,172.83) .. (455,185) ;
\draw    (445,205) -- (445,215) ;
\draw    (295,130) -- (295,135) ;
\draw    (325,130) -- (325,135) ;
\draw    (345,130) -- (345,135) ;
\draw    (335,155) -- (335,160) ;
\draw    (295,155) -- (295,160) ;

\draw (325,66.6) node [anchor=south] [inner sep=0.75pt]    {$X_{0}$};
\draw (350,100) node  [font=\footnotesize]  {$f_{1}$};
\draw (310,100) node  [font=\footnotesize]  {$f_{0}$};
\draw (335,145) node  [font=\footnotesize]  {$g_{1}$};
\draw (360,66.6) node [anchor=south] [inner sep=0.75pt]    {$X_{1}$};
\draw (315,218.4) node [anchor=north] [inner sep=0.75pt]    {$Y_{1}$};
\draw (295,145) node  [font=\footnotesize]  {$g_{0}$};
\draw (315,185) node  [font=\footnotesize]  {$h_{1}$};
\draw (420,125) node  [font=\footnotesize]  {$f_{0}$};
\draw (425,160) node  [font=\footnotesize]  {$g_{0}$};
\draw (445,125) node  [font=\footnotesize]  {$f_{0}$};
\draw (485,125) node  [font=\footnotesize]  {$f_{1}$};
\draw (465,160) node  [font=\footnotesize]  {$g_{1}$};
\draw (445,195) node  [font=\footnotesize]  {$h_{1}$};
\draw (445,218.4) node [anchor=north] [inner sep=0.75pt]    {$Y_{1}$};
\draw (435,66.6) node [anchor=south] [inner sep=0.75pt]    {$X_{0}$};
\draw (470,66.6) node [anchor=south] [inner sep=0.75pt]    {$X_{1}$};
\draw (376,138.4) node [anchor=north west][inner sep=0.75pt]    {$\neq $};
\end{tikzpicture}\qquad
\begin{tikzpicture}[x=0.75pt,y=0.75pt,yscale=-1,xscale=1,baseline=0.4cm]
\draw    (325,70) -- (325,78.87) ;
\draw    (310,90) .. controls (309.8,75.8) and (340.2,75.8) .. (340,90) ;
\draw  [fill={rgb, 255:red, 0; green, 0; blue, 0 }  ,fill opacity=1 ] (322.1,78.87) .. controls (322.1,77.27) and (323.4,75.97) .. (325,75.97) .. controls (326.6,75.97) and (327.9,77.27) .. (327.9,78.87) .. controls (327.9,80.47) and (326.6,81.77) .. (325,81.77) .. controls (323.4,81.77) and (322.1,80.47) .. (322.1,78.87) -- cycle ;
\draw   (340,120) -- (380,120) -- (380,140) -- (340,140) -- cycle ;
\draw    (360,70) .. controls (360.14,82.14) and (370,107.83) .. (370,120) ;
\draw   (300,90) -- (320,90) -- (320,110) -- (300,110) -- cycle ;
\draw    (340,90) .. controls (340.14,102.14) and (350,107.83) .. (350,120) ;
\draw    (335,180) -- (335,195) ;
\draw    (360,140) .. controls (360.14,152.14) and (335,167.83) .. (335,180) ;
\draw    (310,110) .. controls (310.2,135) and (324.6,139.8) .. (325,160) ;
\draw  [fill={rgb, 255:red, 0; green, 0; blue, 0 }  ,fill opacity=1 ] (322.1,160) .. controls (322.1,158.4) and (323.4,157.1) .. (325,157.1) .. controls (326.6,157.1) and (327.9,158.4) .. (327.9,160) .. controls (327.9,161.6) and (326.6,162.9) .. (325,162.9) .. controls (323.4,162.9) and (322.1,161.6) .. (322.1,160) -- cycle ;
\draw   (440,120) -- (480,120) -- (480,140) -- (440,140) -- cycle ;
\draw    (470,70) .. controls (470.14,82.14) and (470,107.83) .. (470,120) ;
\draw    (450,70) .. controls (450.14,82.14) and (450,107.83) .. (450,120) ;
\draw    (460,140) .. controls (460.14,152.14) and (460,177.83) .. (460,190) ;
\draw (325,66.6) node [anchor=south] [inner sep=0.75pt]    {$X_{0}$};
\draw (360,130) node  [font=\footnotesize]  {$f_{1}$};
\draw (310,100) node  [font=\footnotesize]  {$f_{0}$};
\draw (360,66.6) node [anchor=south] [inner sep=0.75pt]    {$X_{1}$};
\draw (335,198.4) node [anchor=north] [inner sep=0.75pt]    {$Y_{1}$};
\draw (460,130) node  [font=\footnotesize]  {$f_{1}$};
\draw (401,123.4) node [anchor=north west][inner sep=0.75pt]    {$\neq $};
\draw (450,66.6) node [anchor=south] [inner sep=0.75pt]    {$X_{0}$};
\draw (468.5,68.6) node [anchor=south] [inner sep=0.75pt]    {$X_{1}$};
\draw (460,193.4) node [anchor=north] [inner sep=0.75pt]    {$Y_{1}$};
\end{tikzpicture}
\caption{\emph{Left}: Failure of associativity if copying is not natural, as it happens, for instance, with stochastic functions. \emph{Right}: Failure of unitality if discarding is not natural, as it happens, for instance, with partial functions.}\label{string:klassocunital}
\end{figure}
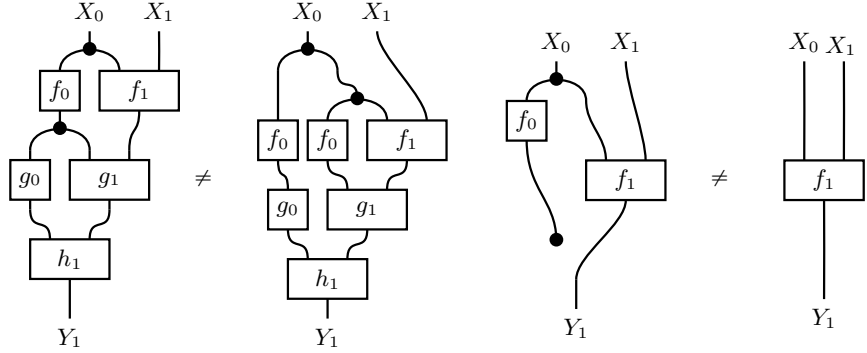

\begin{thm}\label{th:nelist}  %
  The oplax monoidal functor $\neList$ has an oplax monoidal comonad structure if and only if its base monoidal category $(\catC, \otimes, I)$ is cartesian monoidal.
\end{thm}
\begin{proof}
  When $\catC$ is cartesian, we can construct the comonad structure using projections $\prod\nolimits_{i=0}^{n} X_{i} \to X_{n}$ and copying together with symmetries $\prod\nolimits_{i=0}^{n} X_{i} \to \prod\nolimits_{i=0}^{n}\prod\nolimits_{k=0}^{i} X_{k}$.
  These are monoidal natural transformations making $\neList$ an oplax monoidal comonad.

  Suppose $(L,\varepsilon,\delta)$ is an oplax monoidal comonad structure. This means it has families of natural transformations
  \[\delta_{n} \colon \bigotimes^{n}_{i=0} X_{i} \to \bigotimes^{n}_{i=0} \bigotimes^{i}_{k=0} X_{k}
    \mbox{ and } \varepsilon_{n} \colon \bigotimes^{n}_{i=0} X_{i} \to X_{n}.\]
  We use these to construct a uniform counital comagma structure on every object of the category.
  By the refined version of Fox's theorem (\Cref{th:refinedfoxappendix}),
  this implies that $\catC$ is cartesian monoidal.

  Let $X \in \catC$ be any object.
  Choosing $n=2$, $X_{0} = X$ and $X_{1} = I$; and using coherence maps, we get $\delta_{2} \colon X \to X \tensor X$ and $\varepsilon_{2} \colon X \to I$.
  These are coassociative, counital, natural and uniform because the corresponding transformations $\delta$ and $\varepsilon$ are themselves coassociative, counital, natural and monoidal.
  This induces a uniform comagma structure in every object $(X,\delta_{2},\varepsilon_{2})$;
  with this structure, every morphism of the category is a comagma homomorphism because $\delta_{2}$ and $\varepsilon_{2}$ are natural.
\end{proof}

This result implies that we cannot directly extend Uustalu and Vene's approach to non-cartesian monoidal structures.
However, we prove in the next section that our definition of \monoidalStreams{} particularizes to their \emph{causal stream functions}~\cite{uustalu05,katsumata19}.

\subsection{Cartesian monoidal streams}

The main claim of this section is that, in a \cartesianMonoidalCategory{}, \monoidalStreams{} instantiate to \emph{causal stream functions} (\Cref{th:cartesianstreams}). Let us fix such a category, $(\catC,\times,1)$.

The first observation is that the universal property of the cartesian product simplifies the fixpoint equation that defines monoidal streams.
This is a consequence of the following chain of isomorphisms, where we apply a \hyperlink{linkcoyoneda}{Yoneda reduction} to simplify the coend.
\[\begin{aligned}\label{eq:cartesianStreams}
  & \streamProf(\stream{X}, \stream{Y}) 
  \\\cong & \quad \mbox{(By definition)} \\
  & \coend{M} \idProf(X_{0}, M \times Y_{0}) \times \streamProf(\act{M}{\tail{\stream{X}}}, \stream{Y}) \\ \cong & \quad \mbox{(Universal property of the product)} & \\
  & \coend{M} \idProf(X_{0}, M) \times \idProf(X_{0}, Y_{0}) \times  \streamProf(\act{M}{\tail{\stream{X}}}, \stream{Y}) 
  \\ \cong & \quad \mbox{(Yoneda reduction)} \\
  & \idProf(X_{0}, Y_{0}) \times  \streamProf(\act{X_{0}}{\tail{\stream{X}}}, \tail{\stream{Y}}).
\end{aligned}\]
Explicitly, the Yoneda reduction works as follows:
the first action of a stream $f \in \STREAM(\stream{X},\stream{Y})$ can be uniquely split as
$\now(f) = (f_{1},f_{2})$ for some $f_{1} \colon X_{0} \to Y_{0}$ and $f_{2} \colon X_{0} \to M(f)$.
Under the \emph{dinaturality} equivalence relation, $(\sim)$, we can always find a unique representative with $M = X_{0}$ and $f_{2} = \im_{X_{0}}$.

\begin{prop}\label{prop:cartesianproductive}
  Cartesian monoidal categories are \productive{}.
\end{prop}
\begin{proof}
  Let $\bra{\alpha} \in \Stage{1}(\stream{X},\stream{Y})$.
For some given representative $\alpha \colon X_{0} \to M \tensor Y_{0}$, we define the two projections $\alpha_{Y} = \alpha \comp \tid{\coUnit_{M}} \colon X_{0} \to Y_{0}$ and $\alpha_{M} = \alpha \comp \tid{\coUnit_{Y}}$.
The second projection $\alpha_{M}$ depends on the specific representative $\alpha$ we have chosen; however, the first projection $\alpha_{Y}$ is defined independently of the specific representative $\alpha$, as a consequence of naturality of the discarding map (see Fox's theorem for cartesian monoidal categories \Cref{th:fox}).
We define $\alpha_{0} = \delta_{X_{0}} \comp \tid{\alpha_{Y}}$.
Then, we can factor any representative as $\alpha = \alpha_{0} ; \tid{\alpha_{M}}$ (see \Cref{figure:productivity:cartesian}).
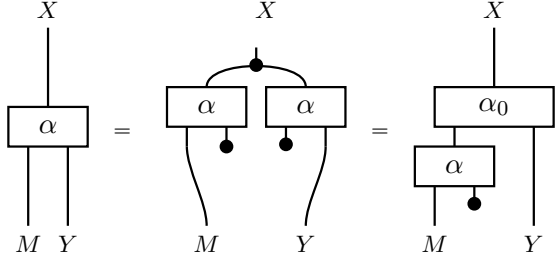
\begin{figure}
\tikzset{every picture/.style={line width=0.85pt}} %
\begin{tikzpicture}[x=0.75pt,y=0.75pt,yscale=-1,xscale=1]
\draw   (380,60) -- (420,60) -- (420,80) -- (380,80) -- cycle ;
\draw    (425,40) -- (425,48.87) ;
\draw    (400,60) .. controls (399.8,45.8) and (450.2,45.8) .. (450,60) ;
\draw  [fill={rgb, 255:red, 0; green, 0; blue, 0 }  ,fill opacity=1 ] (422.1,48.87) .. controls (422.1,47.27) and (423.4,45.97) .. (425,45.97) .. controls (426.6,45.97) and (427.9,47.27) .. (427.9,48.87) .. controls (427.9,50.47) and (426.6,51.77) .. (425,51.77) .. controls (423.4,51.77) and (422.1,50.47) .. (422.1,48.87) -- cycle ;
\draw    (410,80) -- (410,90) ;
\draw    (390,80) -- (390,90) ;
\draw    (460,90) .. controls (460.14,102.14) and (450,117.83) .. (450,130) ;
\draw   (300,70) -- (340,70) -- (340,90) -- (300,90) -- cycle ;
\draw    (310,90) -- (310,130) ;
\draw    (330,90) -- (330,130) ;
\draw    (320,30) -- (320,70) ;
\draw   (515,60) -- (575,60) -- (575,80) -- (515,80) -- cycle ;
\draw    (545,30) -- (545,60) ;
\draw    (565,80) -- (565,97.5) -- (565,130) ;
\draw    (525,80) -- (525,90) ;
\draw   (430,60) -- (470,60) -- (470,80) -- (430,80) -- cycle ;
\draw    (460,80) -- (460,90) ;
\draw    (440,80) -- (440,88.87) ;
\draw  [fill={rgb, 255:red, 0; green, 0; blue, 0 }  ,fill opacity=1 ] (437.1,88.87) .. controls (437.1,87.27) and (438.4,85.97) .. (440,85.97) .. controls (441.6,85.97) and (442.9,87.27) .. (442.9,88.87) .. controls (442.9,90.47) and (441.6,91.77) .. (440,91.77) .. controls (438.4,91.77) and (437.1,90.47) .. (437.1,88.87) -- cycle ;
\draw    (390,90) .. controls (390.14,102.14) and (400,117.83) .. (400,130) ;
\draw  [fill={rgb, 255:red, 0; green, 0; blue, 0 }  ,fill opacity=1 ] (407.1,90) .. controls (407.1,88.4) and (408.4,87.1) .. (410,87.1) .. controls (411.6,87.1) and (412.9,88.4) .. (412.9,90) .. controls (412.9,91.6) and (411.6,92.9) .. (410,92.9) .. controls (408.4,92.9) and (407.1,91.6) .. (407.1,90) -- cycle ;
\draw   (505,90) -- (545,90) -- (545,110) -- (505,110) -- cycle ;
\draw    (515,110) -- (515,130) ;
\draw    (535,110) -- (535,118.87) ;
\draw  [fill={rgb, 255:red, 0; green, 0; blue, 0 }  ,fill opacity=1 ] (532.1,118.87) .. controls (532.1,117.27) and (533.4,115.97) .. (535,115.97) .. controls (536.6,115.97) and (537.9,117.27) .. (537.9,118.87) .. controls (537.9,120.47) and (536.6,121.77) .. (535,121.77) .. controls (533.4,121.77) and (532.1,120.47) .. (532.1,118.87) -- cycle ;

\draw (430,26.6) node [anchor=south] [inner sep=0.75pt]    {$X_0$};
\draw (400,70) node  [font=\normalsize]  {$\alpha $};
\draw (450,133.4) node [anchor=north] [inner sep=0.75pt]    {$Y_0$};
\draw (320,80) node  [font=\normalsize]  {$\alpha $};
\draw (320,26.6) node [anchor=south] [inner sep=0.75pt]    {$X_0$};
\draw (330,133.4) node [anchor=north] [inner sep=0.75pt]    {$Y_0$};
\draw (310,133.4) node [anchor=north] [inner sep=0.75pt]    {$M$};
\draw (357.5,86.6) node [anchor=south] [inner sep=0.75pt]    {$=$};
\draw (545,70) node  [font=\normalsize]  {$\alpha _{0}$};
\draw (545,26.6) node [anchor=south] [inner sep=0.75pt]    {$X_0$};
\draw (565,133.4) node [anchor=north] [inner sep=0.75pt]    {$Y_0$};
\draw (400,133.4) node [anchor=north] [inner sep=0.75pt]    {$M$};
\draw (450,70) node  [font=\normalsize]  {$\alpha $};
\draw (487.5,86.6) node [anchor=south] [inner sep=0.75pt]    {$=$};
\draw (525,100) node  [font=\normalsize]  {$\alpha $};
\draw (515,133.4) node [anchor=north] [inner sep=0.75pt]    {$M$};
\end{tikzpicture}   \caption{Productivity for cartesian categories.}
  \label{figure:productivity:cartesian}
\end{figure}
Now, assume that we have two representatives $\bra{\alpha_{i}} = \bra{\alpha_{j}}$ for which $\bra{\alpha_{i} \comp u} = \bra{\alpha_{j} \comp v}$.
By naturality of the discarding map, $\alpha_{i} \comp \tid{\coUnit} = \alpha_{j} \comp \tid{\coUnit}$, and we call this map $\alpha_{Y}$.
Again by naturality of the discarding map, $\alpha_{i} \comp u \comp \tid{\coUnit} = \alpha_{j} \comp v \comp \tid{\coUnit}$, and discarding the output in $Y$, we get that $\alpha_{M,i} \comp u \comp \tid{\coUnit} = \alpha_{M,j} \comp v \comp \tid{\coUnit}$, which implies $\bra{\alpha_{M,i} \comp u} = \bra{\alpha_{M,j} \comp v}$.
\end{proof}

The definition of monoidal streams in the cartesian case is thus simplified (\Cref{prop:fixpoint-cartesian-streams}).
From there, the explicit construction of cartesian monoidal streams is straightforward.

\begin{defi}[Cartesian monoidal streams]\label{prop:fixpoint-cartesian-streams}
The set of \emph{cartesian monoidal streams}, given inputs $\stream{X}$ and outputs $\stream{Y}$, is the terminal fixpoint of the equation
  \[\streamProf(\stream{X}, \stream{Y}) \cong \idProf(X_{0}, Y_{0}) \times  \streamProf(\act{X_{0}}{\tail{\stream{X}}}, \tail{\stream{Y}}).\]
  In other words, a cartesian monoidal stream $f \in \STREAM(\stream{X},\stream{Y})$ is a pair consisting of
    \begin{itemize}
      \item $\fst(f) \in \hom{}(X_{0}, Y_{0})$, the \emph{first action}, 
      and
      \item $\snd(f) \in \STREAM(\act{X_{0}}{\stream{\tail{X}}}, \stream{\tail{Y}})$, the \emph{rest of the action}.
    \end{itemize}
\end{defi}

\begin{thm}\label{th:cartesianstreams}
  In the cartesian case, the final fixpoint of the equation in \Cref{eq:observationalstreamshort} is given by
  the set of causal functions,
  \begin{equation*}
    \streamProf(\stream{X}, \stream{Y}) = \prod_{n \in \naturals}^{\infty} \idProf(X_{0} \times \cdots \times X_{n}, Y_{n}).
  \end{equation*}
  That is, the category $\STREAM$ of \monoidalStreams{} coincides with the cokleisli monoidal category of the non-empty list monoidal comonad \(\neList \colon \NcatC \to \NcatC\).

\end{thm}
\begin{proof}
  By Adamek's theorem~(\Cref{th:adamek}).
\end{proof}

\begin{cor}
  Let $(\catC, \times, \mathbf{1})$ be a cartesian monoidal category. The category $\STREAM$ is cartesian monoidal.
\end{cor}

\begin{rem}
  \label{rem:pointed-delayed-feedback}
  The feedback operation in cartesian streams is similar to the delayed trace of stateful morphism sequences described by Sprunger and Katsumata~\cite{katsumata19}.  In fact, extensional equivalence classes of stateful morphism sequences over a cartesian category \(\cat{C}\) are in bijection with cartesian streams over the same category \(\cat{C}\).
  Explicitly, for a stateful morphism sequence, \((i,\stream{s}) \colon \stream{X} \to \stream{Y}\), the corresponding cartesian stream is \(\act{i}{\stream{s}}\), of the same type.
  Vice versa, a cartesian stream \(\stream{f} \colon \stream{X} \to \stream{Y}\) gives an equivalence class of stateful morphism sequences \((\id{1},\stream{f})\) of the same type.
  These mappings give isomorphisms when stateful morphism sequences are quotiented by extensional equivalence as defined by Sprunger and Katsumata~\cite{katsumata19}.

  Moreover, delayed trace and feedback are interderivable. Following Sprunger and Katsumata~\cite{katsumata19}, consider the operation \(\bigcirc\) that removes the first component of the stream --- that yields \(\bigcirc \stream{T} = (T_{1}, T_{2}, \dots)\) for a stream \(\stream{T} = (T_{0}, T_{1}, \dots)\).
  The delay \(\delay\) is its left inverse, \(\bigcirc \delay \stream{T} = \stream{T}\).
  For a stateful morphism sequence \((i,\stream{s}) \colon \stream{T} \times \stream{X} \to \bigcirc \stream{T} \times \stream{Y}\), its delayed trace along \(\stream{T}\) with initial point \(p\) can be expressed as a feedback along \(\tail{\stream{T}}\).
  Vice versa, the feedback of \(\stream{f} \colon \delay \stream{U} \times \stream{X} \to \stream{U} \times \stream{Y}\) along \(\stream{U}\) can be expressed as a delayed trace along \(\delay\stream{U}\).
  \begin{align*}
    \mathsf{tr}_{p}^{\stream{T}}(i,\stream{s}) &= \fbk^{\tail{\stream{T}}}(\act{(i \times p)}{\stream{s}}), &
    \fbk^{\stream{U}}(\stream{f}) &= \mathsf{tr}_{\id{1}}^{\delay\stream{U}}(\stream{f}).
  \end{align*}
  We restrict our discussion to the cartesian case studied by Sprunger and Katsumata~\cite{katsumata19}, but a similar correspondence would hold generalising stateful morphism sequences to the monoidal case (cf.~\cite{carette21});
  in general, the equivalence relation of monoidal streams does not coincide with the one of monoidal stateful morphism sequences.
\end{rem}

\subsection{Example: the Fibonacci sequence}\label{example:fibonacci}
Consider $(\Set,\times,\mathbf{1})$, the \cartesianMonoidalCategory{} of small sets and functions.
And let us go back to the morphism $\mathsf{fib} \in \STREAM(\mathbf{1},\mathbb{N})$ that we presented in the Introduction (\Cref{figure:fibonacci}).
By \Cref{th:cartesianstreams}, a morphism of this type is, equivalently, a sequence of natural numbers. Using the previous definitions in \Cref{sec:monoidal-streams,section:classicalstreams}, \hyperlink{linkexamplefibonacci}{we can explicitly compute} this sequence to be $\mathsf{fib} = [0,1,1,2,3,5,8,\dots]$ (see the implementation \cite{Roman_Arrow_Streams_for_2022}).

\section{Stochastic Streams}
\label{section:stochastic-streams}

\MonoidalCategories{} are well suited for reasoning about probabilistic processes. Several different categories of probabilistic channels have been proposed in the literature \cite{panangaden1999, baez2016, cho2019}. They were largely unified by Fritz \cite{fritz2020} under the name of \emph{\MarkovCategories{}}. For simplicity, we work in the discrete stochastic setting, i.e. in the Kleisli category of the finite distribution monad, $\Stoch$, but we will be careful to isolate the relevant structure of \MarkovCategories{} that we use.

The main result of this section is that \emph{\controlledStochasticProcesses{}} \cite{fleming1975,ross1996stochastic} are precisely monoidal streams over $\Stoch$.
That is, controlled stochastic processes are the canonical solution over $\Stoch$ of the fixpoint equation in \Cref{eq:observationalstreamshort}.

\subsection{Preliminaries: Markov categories}
\begin{defi}
  [Markov category, {{\cite[Definition 2.1]{fritz2020}}}]
  \defining{linkmarkovcategory}{} A \emph{Markov category} $\catC$ is a \symmetricMonoidalCategory{} in which
  each object $X \in \catC$ has a cocommutative comonoid structure
  $(X, \varepsilon = \coUnit_{X} \colon X \to I, \delta = \coMult_{X} \colon X \to X \tensor X)$ with
  \begin{itemize}[label=$\triangleright$]
    \item \emph{uniform} comultiplications, $\blackComonoid_{X \tensor Y} = (\blackComonoid_{X} \tensor \blackComonoid_{Y}) \comp \tid{\sigma_{X,Y}}$;
    \item \emph{uniform} counits, $\coUnit_{X \tensor Y} = \coUnit_{X} \tensor \coUnit_{Y}$; and
    \item \emph{natural} counits, $f \comp \coUnit_{Y} = \coUnit_{X}$ for each $f \colon X \to Y$.
  \end{itemize}
  Crucially, comultiplications do not need to be natural.
\end{defi}

\begin{defi}\label{def:distributionmonad}\defining{linkfinitedistribution}
  The finite distribution commutative monad $\defining{linkdistr}{\ensuremath{\fun{D}}} \colon \Set \to \Set$ associates to each set the set of finite-support probability distributions over it.
  \[\distr(X) \defn
    \left\{ p \colon X \to [0,1] \ \\
      \middle|%
      \sum_{p(x) > 0}^{|\{ x \mid p(x) > 0\}| < \infty}p(x) = 1 \right\}.\]
  We call $\defining{linkstoch}{\ensuremath{\cat{Stoch}}}$ the symmetric monoidal Kleisli category of the finite distribution monad, $\kleisli{\distr}$.
  We write $f(y \vert x)$ for the probability $f(x)(y) \in [0, 1]$. Composition, $(f \comp g)$, is defined by
    $$(f \dcomp g)(z \vert x) = \sum_{y \in Y} g(z \vert y) f(y \vert x).$$
\end{defi}

The cartesian product $(\times)$ in $\Set$ induces a monoidal (non-cartesian) product in $\Stoch$.
That is, $\Stoch$ has comonoids $(\blackComonoid)_X: X \to X \times X$ on every object, with $(\blackComonoidUnit)_X \colon X \to 1$ as counit.
However, contrary to what happens in $\Set$, these comultiplications are not natural: \emph{sampling and copying the result} is different from \emph{taking two independent samples}. The \MarkovCategory{} $\Stoch$ also has \emph{conditionals}~\cite{fritz2020}, a property which we will use to prove the main result regarding stochastic processes.

\begin{remC}[{{\cite[Remark 2.4]{fritz2020}}}]
  Any cartesian category is a \MarkovCategory{}. However, not any \MarkovCategory{} is cartesian, and the most interesting examples are those that fail to be cartesian, such as $\Stoch$. The failure of comultiplication being natural makes it impossible to apply Fox's theorem (\Cref{th:fox}).
\end{remC}

\begin{rem}[Triangle notation]
  In a \MarkovCategory{}, given any $f \colon X_{0} \to Y_{0}$ and any $g \colon Y_{0} \tensor X_{0} \tensor X_{1} \to Y_{1}$, we write $(f \triangleleft g) \colon X_{0} \tensor X_{1} \to Y_{0} \tensor Y_{1}$ for the morphism defined by
  \[
    (f \triangleleft g) = 
      \tid{(\coMult_{X_0}) \comp f \comp (\coMult_{Y_0})} \comp \tid{g},
  \]
  which is the string diagram in \Cref{figure:notationtriangle}.
  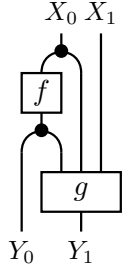
\begin{figure}[h!]
\tikzset{every picture/.style={line width=0.85pt}} %
\begin{tikzpicture}[x=0.75pt,y=0.75pt,yscale=-1,xscale=1]
\draw   (190,90) -- (210,90) -- (210,110) -- (190,110) -- cycle ;
\draw    (210,70) -- (210,78.87) ;
\draw    (190,130) -- (190,170) ;
\draw   (200,140) -- (240,140) -- (240,160) -- (200,160) -- cycle ;
\draw    (220,160) -- (220,170) ;
\draw    (230,70) -- (230,140) ;
\draw    (220,90) -- (220,140) ;
\draw    (210,130) -- (210,140) ;
\draw    (200,90) .. controls (199.8,75.8) and (220.2,75.8) .. (220,90) ;
\draw  [fill={rgb, 255:red, 0; green, 0; blue, 0 }  ,fill opacity=1 ] (207.1,78.87) .. controls (207.1,77.27) and (208.4,75.97) .. (210,75.97) .. controls (211.6,75.97) and (212.9,77.27) .. (212.9,78.87) .. controls (212.9,80.47) and (211.6,81.77) .. (210,81.77) .. controls (208.4,81.77) and (207.1,80.47) .. (207.1,78.87) -- cycle ;
\draw    (200,110) -- (200,118.87) ;
\draw    (190,130) .. controls (189.8,115.8) and (210.2,115.8) .. (210,130) ;
\draw  [fill={rgb, 255:red, 0; green, 0; blue, 0 }  ,fill opacity=1 ] (197.1,118.87) .. controls (197.1,117.27) and (198.4,115.97) .. (200,115.97) .. controls (201.6,115.97) and (202.9,117.27) .. (202.9,118.87) .. controls (202.9,120.47) and (201.6,121.77) .. (200,121.77) .. controls (198.4,121.77) and (197.1,120.47) .. (197.1,118.87) -- cycle ;
\draw (200,100) node  [font=\normalsize]  {$f$};
\draw (220,150) node  [font=\normalsize]  {$g$};
\draw (210,66.6) node [anchor=south] [inner sep=0.75pt]  [font=\small]  {$X_{0}$};
\draw (190,173.4) node [anchor=north] [inner sep=0.75pt]  [font=\small]  {$Y_{0}$};
\draw (220,173.4) node [anchor=north] [inner sep=0.75pt]  [font=\small]  {$Y_{1}$};
\draw (230,66.6) node [anchor=south] [inner sep=0.75pt]  [font=\small]  {$X_{1}$};
\end{tikzpicture}   \caption{The morphism $(f \triangleleft g)$.}
  \label{figure:notationtriangle}
  \end{figure}
\end{rem}

\begin{prop}[Triangle notation is associative]
  Up to symmetries,
  \[
    (f \triangleleft g) \triangleleft h = f \triangleleft (g \triangleleft h).
  \]
  We may simply write $(f \triangleleft g \triangleleft h)$ for any of the two, omitting the symmetry.
\end{prop}
\begin{proof}
  Using string diagrams (\Cref{figure:assoctriangle}). Note that $(\coMult)$ is coassociative and cocommutative.
  \begin{figure}[ht]
\tikzset{every picture/.style={line width=0.85pt}} %
\begin{tikzpicture}[x=0.75pt,y=0.75pt,yscale=-1,xscale=1]
\draw   (320,60) -- (340,60) -- (340,80) -- (320,80) -- cycle ;
\draw    (340,40) -- (340,48.87) ;
\draw    (320,100) -- (320,140) ;
\draw   (330,110) -- (370,110) -- (370,130) -- (330,130) -- cycle ;
\draw    (350,130) -- (350,140) ;
\draw    (360,40) -- (360,110) ;
\draw    (350,60) -- (350,110) ;
\draw    (340,100) -- (340,110) ;
\draw    (330,60) .. controls (329.8,45.8) and (350.2,45.8) .. (350,60) ;
\draw  [fill={rgb, 255:red, 0; green, 0; blue, 0 }  ,fill opacity=1 ] (337.1,48.87) .. controls (337.1,47.27) and (338.4,45.97) .. (340,45.97) .. controls (341.6,45.97) and (342.9,47.27) .. (342.9,48.87) .. controls (342.9,50.47) and (341.6,51.77) .. (340,51.77) .. controls (338.4,51.77) and (337.1,50.47) .. (337.1,48.87) -- cycle ;
\draw    (330,80) -- (330,88.87) ;
\draw    (320,100) .. controls (319.8,85.8) and (340.2,85.8) .. (340,100) ;
\draw  [fill={rgb, 255:red, 0; green, 0; blue, 0 }  ,fill opacity=1 ] (327.1,88.87) .. controls (327.1,87.27) and (328.4,85.97) .. (330,85.97) .. controls (331.6,85.97) and (332.9,87.27) .. (332.9,88.87) .. controls (332.9,90.47) and (331.6,91.77) .. (330,91.77) .. controls (328.4,91.77) and (327.1,90.47) .. (327.1,88.87) -- cycle ;
\draw    (360,20) -- (360,28.87) ;
\draw    (340,40) .. controls (339.8,25.8) and (380.2,25.8) .. (380,40) ;
\draw  [fill={rgb, 255:red, 0; green, 0; blue, 0 }  ,fill opacity=1 ] (357.1,28.87) .. controls (357.1,27.27) and (358.4,25.97) .. (360,25.97) .. controls (361.6,25.97) and (362.9,27.27) .. (362.9,28.87) .. controls (362.9,30.47) and (361.6,31.77) .. (360,31.77) .. controls (358.4,31.77) and (357.1,30.47) .. (357.1,28.87) -- cycle ;
\draw    (380,20) -- (380,28.87) ;
\draw    (360,40) .. controls (359.8,25.8) and (400.2,25.8) .. (400,40) ;
\draw  [fill={rgb, 255:red, 0; green, 0; blue, 0 }  ,fill opacity=1 ] (377.1,28.87) .. controls (377.1,27.27) and (378.4,25.97) .. (380,25.97) .. controls (381.6,25.97) and (382.9,27.27) .. (382.9,28.87) .. controls (382.9,30.47) and (381.6,31.77) .. (380,31.77) .. controls (378.4,31.77) and (377.1,30.47) .. (377.1,28.87) -- cycle ;
\draw    (380,40) -- (380,180) ;
\draw    (320,140) -- (320,148.87) ;
\draw    (310,160) .. controls (309.8,145.8) and (330.2,145.8) .. (330,160) ;
\draw  [fill={rgb, 255:red, 0; green, 0; blue, 0 }  ,fill opacity=1 ] (317.1,148.87) .. controls (317.1,147.27) and (318.4,145.97) .. (320,145.97) .. controls (321.6,145.97) and (322.9,147.27) .. (322.9,148.87) .. controls (322.9,150.47) and (321.6,151.77) .. (320,151.77) .. controls (318.4,151.77) and (317.1,150.47) .. (317.1,148.87) -- cycle ;
\draw    (350,140) -- (350,148.87) ;
\draw    (340,160) .. controls (339.8,145.8) and (360.2,145.8) .. (360,160) ;
\draw  [fill={rgb, 255:red, 0; green, 0; blue, 0 }  ,fill opacity=1 ] (347.1,148.87) .. controls (347.1,147.27) and (348.4,145.97) .. (350,145.97) .. controls (351.6,145.97) and (352.9,147.27) .. (352.9,148.87) .. controls (352.9,150.47) and (351.6,151.77) .. (350,151.77) .. controls (348.4,151.77) and (347.1,150.47) .. (347.1,148.87) -- cycle ;
\draw    (400,40) .. controls (401,90.17) and (389.67,121.5) .. (390,180) ;
\draw    (410,20) -- (410,40) ;
\draw    (410,40) .. controls (411,90.17) and (399,119.5) .. (400,180) ;
\draw    (330,160) .. controls (330,179.5) and (360,161.17) .. (360,180) ;
\draw    (360,160) .. controls (360.67,171.5) and (370,169.5) .. (370,180) ;
\draw   (350,180) -- (410,180) -- (410,200) -- (350,200) -- cycle ;
\draw    (310,160) -- (310,220) ;
\draw    (340,160) -- (340,220) ;
\draw    (380,200) -- (380,220) ;
\draw   (470,40) -- (490,40) -- (490,60) -- (470,60) -- cycle ;
\draw    (490,20) -- (490,28.87) ;
\draw    (480,40) .. controls (479.8,25.8) and (500.2,25.8) .. (500,40) ;
\draw  [fill={rgb, 255:red, 0; green, 0; blue, 0 }  ,fill opacity=1 ] (487.1,28.87) .. controls (487.1,27.27) and (488.4,25.97) .. (490,25.97) .. controls (491.6,25.97) and (492.9,27.27) .. (492.9,28.87) .. controls (492.9,30.47) and (491.6,31.77) .. (490,31.77) .. controls (488.4,31.77) and (487.1,30.47) .. (487.1,28.87) -- cycle ;
\draw    (480,60) -- (480,68.87) ;
\draw    (470,80) .. controls (469.8,65.8) and (490.2,65.8) .. (490,80) ;
\draw  [fill={rgb, 255:red, 0; green, 0; blue, 0 }  ,fill opacity=1 ] (477.1,68.87) .. controls (477.1,67.27) and (478.4,65.97) .. (480,65.97) .. controls (481.6,65.97) and (482.9,67.27) .. (482.9,68.87) .. controls (482.9,70.47) and (481.6,71.77) .. (480,71.77) .. controls (478.4,71.77) and (477.1,70.47) .. (477.1,68.87) -- cycle ;
\draw    (515,20) -- (515,85) ;
\draw    (540,20) -- (540,180) ;
\draw    (500,40) -- (500,90) ;
\draw    (500,80) -- (500,100) ;
\draw    (490,110) .. controls (489.8,95.8) and (510.2,95.8) .. (510,110) ;
\draw  [fill={rgb, 255:red, 0; green, 0; blue, 0 }  ,fill opacity=1 ] (497.1,100) .. controls (497.1,98.4) and (498.4,97.1) .. (500,97.1) .. controls (501.6,97.1) and (502.9,98.4) .. (502.9,100) .. controls (502.9,101.6) and (501.6,102.9) .. (500,102.9) .. controls (498.4,102.9) and (497.1,101.6) .. (497.1,100) -- cycle ;
\draw    (515,80) -- (515,100) ;
\draw    (505,110) .. controls (504.8,95.8) and (525.2,95.8) .. (525,110) ;
\draw  [fill={rgb, 255:red, 0; green, 0; blue, 0 }  ,fill opacity=1 ] (512.1,100) .. controls (512.1,98.4) and (513.4,97.1) .. (515,97.1) .. controls (516.6,97.1) and (517.9,98.4) .. (517.9,100) .. controls (517.9,101.6) and (516.6,102.9) .. (515,102.9) .. controls (513.4,102.9) and (512.1,101.6) .. (512.1,100) -- cycle ;
\draw    (475,110) .. controls (474.8,95.8) and (495.2,95.8) .. (495,110) ;
\draw  [fill={rgb, 255:red, 0; green, 0; blue, 0 }  ,fill opacity=1 ] (482.1,100) .. controls (482.1,98.4) and (483.4,97.1) .. (485,97.1) .. controls (486.6,97.1) and (487.9,98.4) .. (487.9,100) .. controls (487.9,101.6) and (486.6,102.9) .. (485,102.9) .. controls (483.4,102.9) and (482.1,101.6) .. (482.1,100) -- cycle ;
\draw    (470,80) .. controls (470.17,102.5) and (459.5,111.17) .. (460,130) ;
\draw    (490,77.1) .. controls (490.67,88.6) and (485,89.5) .. (485,100) ;
\draw   (465,130) -- (505,130) -- (505,150) -- (465,150) -- cycle ;
\draw    (505,110) .. controls (504.83,119.17) and (495.17,119.5) .. (495,130) ;
\draw    (490,110) .. controls (489.83,119.17) and (485.17,119.5) .. (485,130) ;
\draw    (495,110) .. controls (494.83,119.17) and (510.17,119.5) .. (510,130) ;
\draw    (510,110) .. controls (509.83,119.17) and (520.17,119.5) .. (520,130) ;
\draw    (525,110) .. controls (524.83,119.17) and (530.17,119.5) .. (530,130) ;
\draw    (475,110) -- (475,130) ;
\draw    (460,130) -- (460,170) ;
\draw    (520,130) -- (520,180) ;
\draw    (485,150) -- (485,158.87) ;
\draw    (475,170) .. controls (474.8,155.8) and (495.2,155.8) .. (495,170) ;
\draw  [fill={rgb, 255:red, 0; green, 0; blue, 0 }  ,fill opacity=1 ] (482.1,158.87) .. controls (482.1,157.27) and (483.4,155.97) .. (485,155.97) .. controls (486.6,155.97) and (487.9,157.27) .. (487.9,158.87) .. controls (487.9,160.47) and (486.6,161.77) .. (485,161.77) .. controls (483.4,161.77) and (482.1,160.47) .. (482.1,158.87) -- cycle ;
\draw    (530,130) -- (530,180) ;
\draw    (495,170) .. controls (494.5,176.5) and (510.17,173.17) .. (510,180) ;
\draw   (490,180) -- (550,180) -- (550,200) -- (490,200) -- cycle ;
\draw    (460,170) -- (460,215) -- (460,220) ;
\draw    (520,200) -- (520,220) ;
\draw    (475,170) .. controls (474.83,179.17) and (480.17,209.5) .. (480,220) ;
\draw    (510,130) -- (510,170) ;
\draw    (510,170) .. controls (509.5,176.5) and (495.17,173.17) .. (495,180) ;
\draw (330,70) node  [font=\normalsize]  {$f$};
\draw (350,120) node  [font=\normalsize]  {$g$};
\draw (380,190) node  [font=\normalsize]  {$h$};
\draw (480,50) node  [font=\normalsize]  {$f$};
\draw (485,140) node  [font=\normalsize]  {$g$};
\draw (520,190) node  [font=\normalsize]  {$h$};
\draw (360,16.6) node [anchor=south] [inner sep=0.75pt]  [font=\small]  {$X_{0}$};
\draw (380,16.6) node [anchor=south] [inner sep=0.75pt]  [font=\small]  {$X_{1}$};
\draw (410,16.6) node [anchor=south] [inner sep=0.75pt]  [font=\small]  {$X_{2}$};
\draw (310,223.4) node [anchor=north] [inner sep=0.75pt]  [font=\small]  {$Y_{0}$};
\draw (340,223.4) node [anchor=north] [inner sep=0.75pt]  [font=\small]  {$Y_{1}$};
\draw (380,223.4) node [anchor=north] [inner sep=0.75pt]  [font=\small]  {$Y_{2}$};
\draw (490,16.6) node [anchor=south] [inner sep=0.75pt]  [font=\small]  {$X_{0}$};
\draw (515,16.6) node [anchor=south] [inner sep=0.75pt]  [font=\small]  {$X_{1}$};
\draw (540,16.6) node [anchor=south] [inner sep=0.75pt]  [font=\small]  {$X_{2}$};
\draw (460,223.4) node [anchor=north] [inner sep=0.75pt]  [font=\small]  {$Y_{0}$};
\draw (480,223.4) node [anchor=north] [inner sep=0.75pt]  [font=\small]  {$Y_{1}$};
\draw (520,223.4) node [anchor=north] [inner sep=0.75pt]  [font=\small]  {$Y_{2}$};
\draw (425,117.4) node [anchor=north west][inner sep=0.75pt]    {$=$};
\end{tikzpicture}     \caption{Associativity, up to symmetries, of the triangle operation.}
    \label{figure:assoctriangle}
  \end{figure}
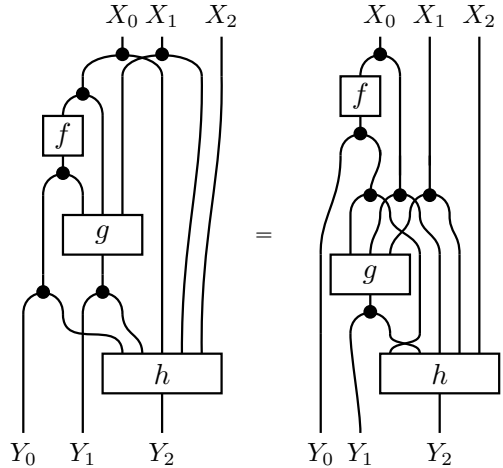
\end{proof}

The structure of a \MarkovCategory{} is very basic. In most cases, we do need extra structure to reason about probabilities: this is the role of conditionals and ranges.

\begin{defi}[Conditionals, {{\cite[Definition 11.5]{fritz2020}}}]
  Let $\catC$ be a \MarkovCategory{}. We say that $\catC$ has
  \emph{conditionals} if for every morphism $f \colon A \to X \tensor Y$,
  writing $f_{Y} \colon A \to X$ for its first projection, there
  exists $c_{f} \colon X \tensor A \to Y$ such that
  $f = f_{Y} \triangleleft c_{f}$ (\Cref{figure:conditionalsmarkov}).
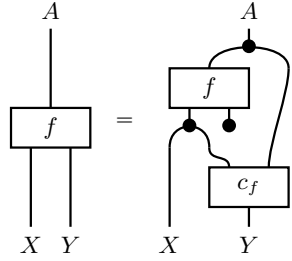
\begin{figure}[!ht]
\tikzset{every picture/.style={line width=0.85pt}} %
\begin{tikzpicture}[x=0.75pt,y=0.75pt,yscale=-1,xscale=1]
\draw   (100,110) -- (140,110) -- (140,130) -- (100,130) -- cycle ;
\draw    (120,70) -- (120,110) ;
\draw    (110,130) -- (110,170) ;
\draw    (130,130) -- (130,170) ;
\draw    (220,70) -- (220,78.87) ;
\draw    (200,90) .. controls (199.8,75.8) and (240.2,75.8) .. (240,90) ;
\draw  [fill={rgb, 255:red, 0; green, 0; blue, 0 }  ,fill opacity=1 ] (217.1,78.87) .. controls (217.1,77.27) and (218.4,75.97) .. (220,75.97) .. controls (221.6,75.97) and (222.9,77.27) .. (222.9,78.87) .. controls (222.9,80.47) and (221.6,81.77) .. (220,81.77) .. controls (218.4,81.77) and (217.1,80.47) .. (217.1,78.87) -- cycle ;
\draw   (180,90) -- (220,90) -- (220,110) -- (180,110) -- cycle ;
\draw    (210,110) -- (210,118.87) ;
\draw  [fill={rgb, 255:red, 0; green, 0; blue, 0 }  ,fill opacity=1 ] (207.1,118.87) .. controls (207.1,117.27) and (208.4,115.97) .. (210,115.97) .. controls (211.6,115.97) and (212.9,117.27) .. (212.9,118.87) .. controls (212.9,120.47) and (211.6,121.77) .. (210,121.77) .. controls (208.4,121.77) and (207.1,120.47) .. (207.1,118.87) -- cycle ;
\draw    (190,110) -- (190,118.87) ;
\draw    (180,130) .. controls (179.8,115.8) and (200.2,115.8) .. (200,130) ;
\draw  [fill={rgb, 255:red, 0; green, 0; blue, 0 }  ,fill opacity=1 ] (187.1,118.87) .. controls (187.1,117.27) and (188.4,115.97) .. (190,115.97) .. controls (191.6,115.97) and (192.9,117.27) .. (192.9,118.87) .. controls (192.9,120.47) and (191.6,121.77) .. (190,121.77) .. controls (188.4,121.77) and (187.1,120.47) .. (187.1,118.87) -- cycle ;
\draw    (180,130) -- (180,170) ;
\draw   (200,140) -- (240,140) -- (240,160) -- (200,160) -- cycle ;
\draw    (200,130) .. controls (199.8,137.4) and (210.2,133) .. (210,140) ;
\draw    (240,90) .. controls (240.14,102.14) and (230,127.83) .. (230,140) ;
\draw    (220,160) -- (220,170) ;
\draw (120,120) node  []  {$f$};
\draw (200,100) node  []  {$f$};
\draw (220,150) node  []  {$c_{f}$};
\draw (151,113.4) node [anchor=north west][inner sep=0.75pt]    {$=$};
\draw (120,66.6) node [anchor=south] [inner sep=0.75pt]    {$A$};
\draw (220,66.6) node [anchor=south] [inner sep=0.75pt]    {$A$};
\draw (110,173.4) node [anchor=north] [inner sep=0.75pt]    {$X$};
\draw (130,173.4) node [anchor=north] [inner sep=0.75pt]    {$Y$};
\draw (180,173.4) node [anchor=north] [inner sep=0.75pt]    {$X$};
\draw (220,173.4) node [anchor=north] [inner sep=0.75pt]    {$Y$};
\end{tikzpicture}   \caption{Condititionals in a Markov category.}
  \label{figure:conditionalsmarkov}
\end{figure}
\end{defi}

\begin{prop}
  The Markov category $\Stoch$ has conditionals \cite[Example 11.6]{fritz2020}.
\end{prop}
\begin{proof}
  Let $f \colon A \to X \tensor Y$. If $Y$ is empty, we are automatically done. If not, pick some $y_{0} \in Y$, and define
  \[c_{f}(y|x,a) = \left\{
	\begin{array}{ll}
		f(x,y|a) / \sum_{x \in X} f(x,y|a)  \\ \qquad \mbox{if } f(x,y|a) > 0 \mbox{ for some } x \in X, \\
		(y = y_{0}) \quad \mbox{otherwise}.
	\end{array}
    \right. \]
  It is straightforward to check that this does indeed define a distribution, and that it factors the original $f$ as expected.
\end{proof}

\begin{defi}[Ranges]
  In a Markov category, a \emph{range} for a morphism $f \colon A \to B$ is a morphism $r_{f} \colon A \tensor B \to A \tensor B$
  that
  \begin{enumerate}
    \item does not change its output $f \triangleleft \im_{A \tensor B} = f \triangleleft r_{f}$,
    \item is \emph{deterministic}, meaning $r_{f} \comp \coMult_{A \tensor B} = \coMult_{A \tensor B} \comp (r_{f} \tensor r_{f})$,
    \item and has the \emph{range} property, $f \triangleleft g = f \triangleleft h$ must imply
    \[(r_{f} \tensor \im) \comp g = (r_{f} \tensor \im) \comp h\]
    for any suitably typed $g$ and $h$.
  \end{enumerate}
  We say that a Markov category \emph{has ranges} if there exists a range for each morphism of the category.
\end{defi}

\begin{prop}
  The \MarkovCategory{} $\Stoch$ has ranges.
\end{prop}
\begin{proof}
  Given $f \colon A \to B$, we know that for each $a \in A$ there exists some $b_{a} \in B$ such that $f(b_{a}| a) > 0$.
  We fix such $b_{a} \in B$, and we define $r_{f} \colon A \tensor B \to A \tensor B$ as
  \[r_{f}(a,b) = \left\{
	\begin{array}{ll}
		(a,b)  & \mbox{if } f(b|a) > 0, \\
		(a,b_{a}) & \mbox{if } f(b|a) = 0.
	\end{array}
  \right.\]
  It is straightforward to check that it satisfies all the properties of ranges.
\end{proof}

\begin{rem}
  There already exists a notion of categorical \emph{range} in the literature, due to Cockett, Guo and Hofstra \cite{cockett2012range}.
  It arises in parallel to the notion of \emph{support} in \emph{restriction categories} \cite{cockett02}.
  The definition better suited for our purposes is different, even if it seems inspired by the same idea.
  The main difference is that we are using a \emph{controlled range}; that is, the range of a morphism depends on the input to the original morphism.
  We keep the name hoping that it will not cause any confusion, as we do not deal explicitly with restriction categories in this text.
\end{rem}

\begin{thm}\label{appendix:productivemarkov}
  Any \MarkovCategory{} with conditionals and ranges is \productive{}.
\end{thm}
\begin{proof}
  Given any $\bra{\alpha} \in \Stage{1}(\stream{X},\stream{Y})$, we can define
  \[\alpha_{0} = \coMult_{X_0} \comp \tid{\alpha \comp (\coMult_{Y_0} \tensor \coUnit_{M})}.\]
  This is indeed well-defined because of naturality of the discarding map $(\coUnit)_{M} \colon M \to I$ in any \MarkovCategory{}.
  Let $c_{\alpha} \colon Y_0 \tensor X_0 \to M$ be a conditional of $\alpha$. This representative can then be factored as $\alpha = \alpha_{0} \comp \tid{c_{\alpha}}$ (\Cref{figure:productivemarkov}).
  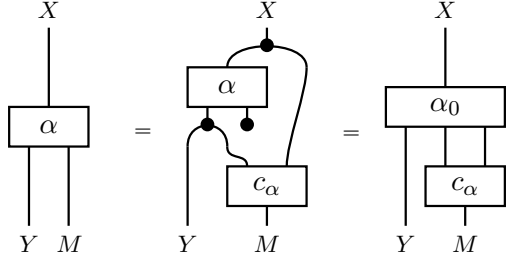
\begin{figure}
\tikzset{every picture/.style={line width=0.85pt}} %

\begin{tikzpicture}[x=0.75pt,y=0.75pt,yscale=-1,xscale=1]
\draw   (390,50) -- (430,50) -- (430,70) -- (390,70) -- cycle ;
\draw    (430,30) -- (430,38.87) ;
\draw    (410,50) .. controls (409.8,35.8) and (450.2,35.8) .. (450,50) ;
\draw  [fill={rgb, 255:red, 0; green, 0; blue, 0 }  ,fill opacity=1 ] (427.1,38.87) .. controls (427.1,37.27) and (428.4,35.97) .. (430,35.97) .. controls (431.6,35.97) and (432.9,37.27) .. (432.9,38.87) .. controls (432.9,40.47) and (431.6,41.77) .. (430,41.77) .. controls (428.4,41.77) and (427.1,40.47) .. (427.1,38.87) -- cycle ;
\draw    (420,70) -- (420,78.87) ;
\draw  [fill={rgb, 255:red, 0; green, 0; blue, 0 }  ,fill opacity=1 ] (417.1,78.87) .. controls (417.1,77.27) and (418.4,75.97) .. (420,75.97) .. controls (421.6,75.97) and (422.9,77.27) .. (422.9,78.87) .. controls (422.9,80.47) and (421.6,81.77) .. (420,81.77) .. controls (418.4,81.77) and (417.1,80.47) .. (417.1,78.87) -- cycle ;
\draw    (400,70) -- (400,78.87) ;
\draw    (390,90) .. controls (389.8,75.8) and (410.2,75.8) .. (410,90) ;
\draw  [fill={rgb, 255:red, 0; green, 0; blue, 0 }  ,fill opacity=1 ] (397.1,78.87) .. controls (397.1,77.27) and (398.4,75.97) .. (400,75.97) .. controls (401.6,75.97) and (402.9,77.27) .. (402.9,78.87) .. controls (402.9,80.47) and (401.6,81.77) .. (400,81.77) .. controls (398.4,81.77) and (397.1,80.47) .. (397.1,78.87) -- cycle ;
\draw    (390,90) -- (390,130) ;
\draw   (410,100) -- (450,100) -- (450,120) -- (410,120) -- cycle ;
\draw    (410,90) .. controls (409.8,97.4) and (420.2,93) .. (420,100) ;
\draw    (450,50) .. controls (450.14,62.14) and (440,87.83) .. (440,100) ;
\draw    (430,120) -- (430,130) ;
\draw   (300,70) -- (340,70) -- (340,90) -- (300,90) -- cycle ;
\draw    (310,90) -- (310,130) ;
\draw    (330,90) -- (330,130) ;
\draw    (320,30) -- (320,70) ;
\draw   (490,60) -- (550,60) -- (550,80) -- (490,80) -- cycle ;
\draw    (520,30) -- (520,60) ;
\draw    (520,80) -- (520,100) ;
\draw    (500,80) -- (500,97.5) -- (500,130) ;
\draw   (510,100) -- (550,100) -- (550,120) -- (510,120) -- cycle ;
\draw    (530,120) -- (530,130) ;
\draw    (540,80) -- (540,100) ;
\draw (430,26.6) node [anchor=south] [inner sep=0.75pt]    {$X_0$};
\draw (430,133.4) node [anchor=north] [inner sep=0.75pt]    {$M$};
\draw (410,60) node  [font=\normalsize]  {$\alpha $};
\draw (430,110) node  [font=\normalsize]  {$c_{\alpha }$};
\draw (390,133.4) node [anchor=north] [inner sep=0.75pt]    {$Y_0$};
\draw (320,80) node  [font=\normalsize]  {$\alpha $};
\draw (320,26.6) node [anchor=south] [inner sep=0.75pt]    {$X_0$};
\draw (310,133.4) node [anchor=north] [inner sep=0.75pt]    {$Y_0$};
\draw (330.5,133.4) node [anchor=north] [inner sep=0.75pt]    {$M$};
\draw (367.5,86.6) node [anchor=south] [inner sep=0.75pt]    {$=$};
\draw (520,70) node  [font=\normalsize]  {$\alpha _{0}$};
\draw (530,110) node  [font=\normalsize]  {$c_{\alpha }$};
\draw (472.5,87.6) node [anchor=south] [inner sep=0.75pt]    {$=$};
\draw (520,26.6) node [anchor=south] [inner sep=0.75pt]    {$X_0$};
\draw (500,133.4) node [anchor=north] [inner sep=0.75pt]    {$Y_0$};
\draw (530,133.4) node [anchor=north] [inner sep=0.75pt]    {$M$};
\end{tikzpicture}     \caption{Productivity for Markov categories.}
    \label{figure:productivemarkov}
  \end{figure}

Now assume that for two representatives $\bra{\alpha_{i}} = \bra{\alpha_{j}}$ we have that $\bra{\alpha_{i} \comp u} = \bra{\alpha_{j} \comp v}$.
By naturality of the discarding, $\alpha_{i} \comp \tid{\varepsilon} = \alpha_{j} \comp \tid{\varepsilon}$, and let $r$ be a range of this map.
Again by naturality of discarding, we have $\tid{\alpha_{i}} \comp \tid{u} \comp \tid{\coUnit} = \tid{\alpha_{j}} \comp \tid{v} \comp \tid{\coUnit}$.
Let then $c_{i}$ and $c_{j}$ be conditionals of $\alpha_{i}$ and $\alpha_{j}$:
we have that $(\alpha_{0}\triangleleft c_{i}) \comp u \comp \tid{\coUnit} = (\alpha_{0} \triangleleft c_{j}) \comp v \comp \tid{\coUnit}$.
By the properties of ranges (\Cref{figure:propertiesrange}), $(\alpha_{0}\triangleleft r \comp c_{i}) \comp u \comp \coUnit_{M(u)} = (\alpha_{0}\triangleleft r \comp c_{j}) \comp v \comp \coUnit_{M(v)}$,
and thus, $r \comp c_{i} \comp u \comp \coUnit_{M(u)} = r \comp c_{j} \comp v \comp \coUnit_{M(v)}$.
We pick $s_{i} = r \comp c_{i}$ and we have proven that $\bra{\tid{s_{i}} \comp u} = \bra{\tid{s_{j}} \comp v}$.
\end{proof}

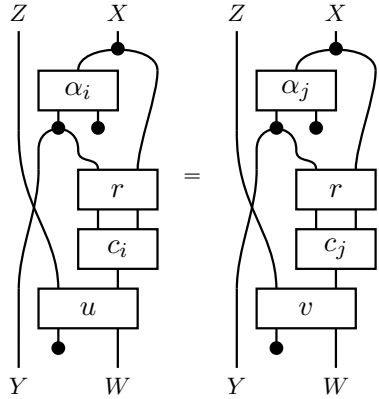
\begin{figure}
\tikzset{every picture/.style={line width=0.85pt}} %
\begin{tikzpicture}[x=0.75pt,y=0.75pt,yscale=-1,xscale=1]
\draw   (160.01,50) -- (200.01,50) -- (200.01,70) -- (160.01,70) -- cycle ;
\draw    (200.01,30) -- (200.01,38.87) ;
\draw    (180.01,50) .. controls (179.81,35.8) and (220.21,35.8) .. (220.01,50) ;
\draw  [fill={rgb, 255:red, 0; green, 0; blue, 0 }  ,fill opacity=1 ] (197.11,38.87) .. controls (197.11,37.27) and (198.41,35.97) .. (200.01,35.97) .. controls (201.61,35.97) and (202.91,37.27) .. (202.91,38.87) .. controls (202.91,40.47) and (201.61,41.77) .. (200.01,41.77) .. controls (198.41,41.77) and (197.11,40.47) .. (197.11,38.87) -- cycle ;
\draw    (190.01,70) -- (190.01,78.87) ;
\draw  [fill={rgb, 255:red, 0; green, 0; blue, 0 }  ,fill opacity=1 ] (187.11,78.87) .. controls (187.11,77.27) and (188.41,75.97) .. (190.01,75.97) .. controls (191.61,75.97) and (192.91,77.27) .. (192.91,78.87) .. controls (192.91,80.47) and (191.61,81.77) .. (190.01,81.77) .. controls (188.41,81.77) and (187.11,80.47) .. (187.11,78.87) -- cycle ;
\draw    (170.01,70) -- (170.01,78.87) ;
\draw    (160.01,90) .. controls (159.81,75.8) and (180.21,75.8) .. (180.01,90) ;
\draw  [fill={rgb, 255:red, 0; green, 0; blue, 0 }  ,fill opacity=1 ] (167.11,78.87) .. controls (167.11,77.27) and (168.41,75.97) .. (170.01,75.97) .. controls (171.61,75.97) and (172.91,77.27) .. (172.91,78.87) .. controls (172.91,80.47) and (171.61,81.77) .. (170.01,81.77) .. controls (168.41,81.77) and (167.11,80.47) .. (167.11,78.87) -- cycle ;
\draw    (160.01,90) -- (160.01,100) ;
\draw   (180.01,130) -- (220.01,130) -- (220.01,150) -- (180.01,150) -- cycle ;
\draw    (180.01,90) .. controls (179.81,97.4) and (190.21,93) .. (190.01,100) ;
\draw    (220.01,50) .. controls (220.15,62.14) and (210.01,87.83) .. (210.01,100) ;
\draw    (200.01,150) -- (200.01,160) ;
\draw   (180.01,100) -- (220.01,100) -- (220.01,120) -- (180.01,120) -- cycle ;
\draw    (190.01,120) -- (190.01,130) ;
\draw    (210.01,120) -- (210.01,130) ;
\draw   (160.01,160) -- (210.01,160) -- (210.01,180) -- (160.01,180) -- cycle ;
\draw    (200,180) -- (200,200) ;
\draw   (270.01,50) -- (310.01,50) -- (310.01,70) -- (270.01,70) -- cycle ;
\draw    (310.01,30) -- (310.01,38.87) ;
\draw    (290.01,50) .. controls (289.81,35.8) and (330.21,35.8) .. (330.01,50) ;
\draw  [fill={rgb, 255:red, 0; green, 0; blue, 0 }  ,fill opacity=1 ] (307.11,38.87) .. controls (307.11,37.27) and (308.41,35.97) .. (310.01,35.97) .. controls (311.61,35.97) and (312.91,37.27) .. (312.91,38.87) .. controls (312.91,40.47) and (311.61,41.77) .. (310.01,41.77) .. controls (308.41,41.77) and (307.11,40.47) .. (307.11,38.87) -- cycle ;
\draw    (300.01,70) -- (300.01,78.87) ;
\draw  [fill={rgb, 255:red, 0; green, 0; blue, 0 }  ,fill opacity=1 ] (297.11,78.87) .. controls (297.11,77.27) and (298.41,75.97) .. (300.01,75.97) .. controls (301.61,75.97) and (302.91,77.27) .. (302.91,78.87) .. controls (302.91,80.47) and (301.61,81.77) .. (300.01,81.77) .. controls (298.41,81.77) and (297.11,80.47) .. (297.11,78.87) -- cycle ;
\draw    (280.01,70) -- (280.01,78.87) ;
\draw    (270.01,90) .. controls (269.81,75.8) and (290.21,75.8) .. (290.01,90) ;
\draw  [fill={rgb, 255:red, 0; green, 0; blue, 0 }  ,fill opacity=1 ] (277.11,78.87) .. controls (277.11,77.27) and (278.41,75.97) .. (280.01,75.97) .. controls (281.61,75.97) and (282.91,77.27) .. (282.91,78.87) .. controls (282.91,80.47) and (281.61,81.77) .. (280.01,81.77) .. controls (278.41,81.77) and (277.11,80.47) .. (277.11,78.87) -- cycle ;
\draw    (270.01,90) -- (270.01,100) ;
\draw   (290.01,130) -- (330.01,130) -- (330.01,150) -- (290.01,150) -- cycle ;
\draw    (290.01,90) .. controls (289.81,97.4) and (300.21,93) .. (300.01,100) ;
\draw    (330.01,50) .. controls (330.15,62.14) and (320.01,87.83) .. (320.01,100) ;
\draw    (310.01,150) -- (310.01,160) ;
\draw   (290.01,100) -- (330.01,100) -- (330.01,120) -- (290.01,120) -- cycle ;
\draw    (300.01,120) -- (300.01,130) ;
\draw    (320.01,120) -- (320.01,130) ;
\draw   (270.01,160) -- (320.01,160) -- (320.01,180) -- (270.01,180) -- cycle ;
\draw    (310,180) -- (310,200) ;
\draw    (270.01,100) .. controls (270.18,119.83) and (260.18,139.5) .. (260.01,160) ;
\draw    (260.01,80) .. controls (259.51,119.17) and (279.51,130.17) .. (280.01,160) ;
\draw    (260.01,30) -- (260.01,80) ;
\draw    (260.01,160) -- (260,200) ;
\draw    (160.02,100) .. controls (160.19,119.83) and (150.19,139.5) .. (150.02,160) ;
\draw    (150.01,80) .. controls (149.51,119.17) and (169.51,130.17) .. (170.01,160) ;
\draw    (150.01,30) -- (150.01,80) ;
\draw    (150.02,160) -- (150,200) ;
\draw    (280,180) -- (280,190) ;
\draw    (170,180) -- (170,190) ;
\draw  [fill={rgb, 255:red, 0; green, 0; blue, 0 }  ,fill opacity=1 ] (167.1,190) .. controls (167.1,188.4) and (168.4,187.1) .. (170,187.1) .. controls (171.6,187.1) and (172.9,188.4) .. (172.9,190) .. controls (172.9,191.6) and (171.6,192.9) .. (170,192.9) .. controls (168.4,192.9) and (167.1,191.6) .. (167.1,190) -- cycle ;
\draw  [fill={rgb, 255:red, 0; green, 0; blue, 0 }  ,fill opacity=1 ] (277.1,190) .. controls (277.1,188.4) and (278.4,187.1) .. (280,187.1) .. controls (281.6,187.1) and (282.9,188.4) .. (282.9,190) .. controls (282.9,191.6) and (281.6,192.9) .. (280,192.9) .. controls (278.4,192.9) and (277.1,191.6) .. (277.1,190) -- cycle ;
\draw (200.01,26.6) node [anchor=south] [inner sep=0.75pt]    {$X$};
\draw (180.01,60) node  [font=\normalsize]  {$\alpha _{i}$};
\draw (200.01,140) node  [font=\normalsize]  {$c_{i}$};
\draw (200.01,110.5) node  [font=\normalsize]  {$r$};
\draw (185.01,170) node  [font=\normalsize]  {$u$};
\draw (310.01,26.6) node [anchor=south] [inner sep=0.75pt]    {$X$};
\draw (290.01,60) node  [font=\normalsize]  {$\alpha _{j}$};
\draw (310.01,140) node  [font=\normalsize]  {$c_{j}$};
\draw (310.01,110.5) node  [font=\normalsize]  {$r$};
\draw (295.01,170) node  [font=\normalsize]  {$v$};
\draw (237.51,107.6) node [anchor=south] [inner sep=0.75pt]    {$=$};
\draw (150.01,26.6) node [anchor=south] [inner sep=0.75pt]    {$Z$};
\draw (150,203.4) node [anchor=north] [inner sep=0.75pt]    {$Y$};
\draw (200,203.4) node [anchor=north] [inner sep=0.75pt]    {$W$};
\draw (260.01,26.6) node [anchor=south] [inner sep=0.75pt]    {$Z$};
\draw (260,203.4) node [anchor=north] [inner sep=0.75pt]    {$Y$};
\draw (310,203.4) node [anchor=north] [inner sep=0.75pt]    {$W$};
\end{tikzpicture}   \caption{Applying the properties of range.}
  \label{figure:propertiesrange}
\end{figure}

\subsection{Stochastic processes}

We start by recalling the notion of \emph{stochastic process} from probability theory and its ``controlled'' version.
The latter is used in the context of \emph{stochastic control} \cite{fleming1975,ross1996stochastic}, where the user has access to the parameters or optimization variables of a probabilistic model.

A \emph{discrete stochastic process} is defined as a collection of random variables $Y_1, \dots Y_n$ indexed by discrete time.
At any time step $n$, these random variables are distributed according to some $p_n \in \distr(Y_1 \times \dots \times Y_n)$.
Since the future cannot influence the past, the marginal of $p_{n+1}$ over $Y_{n+1}$ must equal $p_{n}$.
When this occurs, we say that the family of distributions $(p_n)_{n \in \naturals}$ is \emph{causal}.

More generally, there may be some additional variables $X_1, \dots, X_n$ which we have control over.
In this case, a \emph{\controlledStochasticProcess{}} is defined as a collection of \emph{controlled random variables}
distributing according to $f_n \colon  X_1 \times \dots \times X_n \to \distr(Y_1, \dots, Y_n)$.
Causality ensures that the marginal of $f_{n+1}$ over $Y_{n+1}$ must equal $f_{n}$.

\begin{defi}\label{def:stoch-rep-process}\defining{linkstochasticprocess}{}
Let \(\stream{X}\) and \(\stream{Y}\) be sequences of sets.
A \emph{controlled stochastic process} \(f \colon \stream{X} \to \stream{Y}\) is a sequence of functions \[f_{n} \colon X_{0} \times \dots \times X_{n} \to \distr(Y_{0} \times \dots \times Y_{n})\] satisfying \emph{causality} (the \emph{marginalisation property}).
That is, such that \(f_{n}\) coincides with the marginal distribution of \(f_{n+1}\) on the first \(n\) variables,  \(f_{n+1} \dcomp \distr\proj[Y_{0},\dots,Y_{n}] = \proj[X_{0},\dots,X_{n}] \dcomp f_{n}\), making the diagram in \Cref{figure:commute} commute.
\begin{figure}[h]
  \centering
  \begin{minipage}{0.45\textwidth}
    \begin{tikzcd}
      X_0 \times \dots \times X_{n+1} \rar{f_{n+1}} \dar[swap]{\pi_{0,\dots,n}} & D(Y_0 \times \dots \times Y_{n+1}) \dar{D\pi_{0,\dots,n}} \\
      X_0 \times \dots \times X_{n}  \rar{f_{n}} & D(Y_0 \times \dots \times Y_{n})
    \end{tikzcd}
  \end{minipage}\qquad
  \begin{minipage}{0.45\textwidth}

\tikzset{every picture/.style={line width=0.75pt}} %

\begin{tikzpicture}[x=0.75pt,y=0.75pt,yscale=-1,xscale=1]
\draw   (80,115) -- (150,115) -- (150,130) -- (80,130) -- cycle ;
\draw    (90,100) -- (90,115) ;
\draw    (120,130) -- (120,145) ;
\draw    (140,130) -- (140,140) ;
\draw    (90,130) -- (90,145) ;
\draw    (120,100) -- (120,115) ;
\draw    (140,100) -- (140,115) ;
\draw  [fill={rgb, 255:red, 0; green, 0; blue, 0 }  ,fill opacity=1 ] (137.1,140) .. controls (137.1,138.4) and (138.4,137.1) .. (140,137.1) .. controls (141.6,137.1) and (142.9,138.4) .. (142.9,140) .. controls (142.9,141.6) and (141.6,142.9) .. (140,142.9) .. controls (138.4,142.9) and (137.1,141.6) .. (137.1,140) -- cycle ;
\draw  [draw opacity=0] (90,95) -- (120,95) -- (120,115) -- (90,115) -- cycle ;
\draw  [draw opacity=0] (90,130) -- (120,130) -- (120,150) -- (90,150) -- cycle ;
\draw  [draw opacity=0] (150,110) -- (180,110) -- (180,135) -- (150,135) -- cycle ;
\draw   (180,115) -- (230,115) -- (230,130) -- (180,130) -- cycle ;
\draw    (190,100) -- (190,115) ;
\draw    (220,130) -- (220,145) ;
\draw    (190,130) -- (190,145) ;
\draw    (220,100) -- (220,115) ;
\draw    (240,100) -- (240,110) ;
\draw  [fill={rgb, 255:red, 0; green, 0; blue, 0 }  ,fill opacity=1 ] (237.1,110) .. controls (237.1,108.4) and (238.4,107.1) .. (240,107.1) .. controls (241.6,107.1) and (242.9,108.4) .. (242.9,110) .. controls (242.9,111.6) and (241.6,112.9) .. (240,112.9) .. controls (238.4,112.9) and (237.1,111.6) .. (237.1,110) -- cycle ;
\draw  [draw opacity=0] (190,95) -- (220,95) -- (220,115) -- (190,115) -- cycle ;
\draw  [draw opacity=0] (190,130) -- (220,130) -- (220,150) -- (190,150) -- cycle ;

\draw (115,122.5) node  [font=\footnotesize]  {$f_{n+1}$};
\draw (165,122.5) node    {$=$};
\draw (90.5,97.6) node [anchor=south] [inner sep=0.75pt]  [font=\footnotesize]  {$X_{0}$};
\draw (105,105) node  [font=\footnotesize]  {$\dotsc $};
\draw (105,140) node  [font=\footnotesize]  {$\dotsc $};
\draw (90,148.4) node [anchor=north] [inner sep=0.75pt]  [font=\footnotesize]  {$Y_{0}$};
\draw (120,96.6) node [anchor=south] [inner sep=0.75pt]  [font=\footnotesize]  {$X_{n}$};
\draw (140,96.6) node [anchor=south] [inner sep=0.75pt]  [font=\footnotesize]  {$X_{n+1}$};
\draw (120,148.4) node [anchor=north] [inner sep=0.75pt]  [font=\footnotesize]  {$Y_{n}$};
\draw (205,122.5) node  [font=\footnotesize]  {$f_{n}$};
\draw (190.5,97.6) node [anchor=south] [inner sep=0.75pt]  [font=\footnotesize]  {$X_{0}$};
\draw (205,105) node  [font=\footnotesize]  {$\dotsc $};
\draw (205,140) node  [font=\footnotesize]  {$\dotsc $};
\draw (190,148.4) node [anchor=north] [inner sep=0.75pt]  [font=\footnotesize]  {$Y_{0}$};
\draw (220,96.6) node [anchor=south] [inner sep=0.75pt]  [font=\footnotesize]  {$X_{n}$};
\draw (240,96.6) node [anchor=south] [inner sep=0.75pt]  [font=\footnotesize]  {$X_{n+1}$};
\draw (220,148.4) node [anchor=north] [inner sep=0.75pt]  [font=\footnotesize]  {$Y_{n}$};

\end{tikzpicture}
   \end{minipage}
  \caption{Marginalisation for stochastic processes.}
  \label{figure:commute}
\end{figure}

Controlled stochastic processes with componentwise composition, identities and tensoring, are the morphisms of a symmetric monoidal category \(\StochProc\).
\end{defi}

\begin{prop}[Factoring as conditionals]
  A stochastic process $f \colon \stream{X} \to \stream{Y}$ can be always written as
  \[f_{n} = c_{0} \triangleleft c_{1} \triangleleft \dots \triangleleft c_{n},\]
  for some family of functions
  \[c_{n} \colon Y_{0} \times \dots \times Y_{n-1} \times X_{0} \times \dots \times X_{n} \to Y_{n},\]
  called the \emph{conditionals} of the stochastic process.
\end{prop}
\begin{proof}
  We proceed by induction, noting first that $c_{0} = f_{0}$.
  In the general case, we apply conditionals to rewrite $f_{n+1} = (f_{n+1} ; (\coUnit)_{Y_{n+1}}) \triangleleft c_{n+1}$.
  Because of the marginalization property, we know that $f_{n+1} ; (\coUnit)_{Y_{n+1}} = f_{n}$.
  So finally, $f_{n+1} = f_{n} \triangleleft c_{n+1}$, which by the induction hypothesis gives the desired result.
\end{proof}

\begin{prop}\label{prop:condtoproc}
  If two families of conditionals give rise to the same stochastic process,
  \[c_{0} \triangleleft c_{1} \triangleleft \dots \triangleleft c_{n} =
  c_{0}' \triangleleft c_{1}' \triangleleft \dots \triangleleft c_{n}',\]
then, they also give rise to the same n-stage processes in $\Stoch$,
\[\bra{c_{0} \triangleleft \im |c_{1} \triangleleft \im |\dots|c_{n} \triangleleft \im} = \bra{c_{0}' \triangleleft \im |c_{1}' \triangleleft \im|\dots|c_{n}' \triangleleft \im}.\]
\end{prop}
\begin{proof}
  We start by defining a family of morphisms $r_{n}$ by induction.
  We take $r_{0} = \im$ and $r_{n+1}$ to be a \emph{range} of $r_{n}; \coMult; c_{n}$.

  Let us prove now that for any $n \in \naturals$ and $i \leq n$,
  \[r_{i}\comp\coMult\comp c_{i} \triangleleft \dots \triangleleft c_{n} = r_{i}\comp \coMult \comp c_{i}' \triangleleft \dots \triangleleft c_{n}'.\]
  We proceed by induction. Observing that $c_{0} = c_{0}'$, we prove it for $n = 0$ and also for the case $i = 0$ for any $n \in \naturals$.
Assume we have it proven for $n$, so in particular we know that $r_{i}\comp  \coMult \comp  c_{i} = r_{i}\comp  \coMult \comp  c_{i}'$ for any $i \leq n$. Now, by induction on $i$, we can use the properties of ranges to show that
  \[\begin{gathered}
      r_{i}\comp  \coMult \comp  c_{i} \triangleleft \dots \triangleleft c_{n} = r_{i}\comp  \coMult \comp  c_{i}' \triangleleft \dots \triangleleft c_{n}' \\
      (r_{i}\comp  \coMult \comp  c_{i} \triangleleft \im) \comp  c_{i+1} \triangleleft  \dots \triangleleft c_{n} =
      (r_{i}\comp  \coMult \comp  c_{i}' \triangleleft \im) \comp  c_{i+1}' \triangleleft \dots \triangleleft c_{n}' \\
      (r_{i}\comp  \coMult \comp  c_{i} \triangleleft r_{i+1}) \comp  c_{i+1} \triangleleft  \dots \triangleleft c_{n} =
      (r_{i}\comp  \coMult \comp  c_{i}' \triangleleft r_{i+1}) \comp  c_{i+1}' \triangleleft \dots \triangleleft c_{n}' \\
      r_{i+1} \comp  \coMult \comp   c_{i+1} \triangleleft  \dots \triangleleft c_{n} =
      r_{i+1} \comp  \coMult \comp   c_{i+1}' \triangleleft \dots \triangleleft c_{n}'. \\
    \end{gathered}\]
  In particular, $r_{n}\comp  \coMult \comp  c_{n} = r_{n} \comp  \coMult \comp  c_{n}'$.

  Now, we claim the following for each $n \in \naturals$ and each $i \leq n$,
  \[\begin{gathered}\bra{c_{0} \triangleleft \im | \dots | c_{n} \triangleleft \im} = \\
  \bra{r_{0} (c_{0} \triangleleft \im) | \dots | r_{i} (c_{i} \triangleleft \im) | c_{i+1} | \dots | c_{n} \triangleleft \im}.\end{gathered}\]
  It is clear for $n = 0$ and for $i = 0$. In the inductive case for $i$,
  \[\begin{aligned}
      \bra{r_{0} (c_{0} \triangleleft \im) | \dots | r_{i} (c_{i} \triangleleft \im) | c_{i+1}\triangleleft \im | \dots | c_{n} \triangleleft \im} = \\
      \bra{r_{0} (c_{0} \triangleleft \im) | \dots | r_{i} \comp  \coMult \comp  c_{i} \triangleleft \im | c_{i+1}\triangleleft \im | \dots | c_{n} \triangleleft \im} = \\
      \bra{r_{0} (c_{0} \triangleleft \im) | \dots | r_{i} \comp  \coMult \comp  c_{i} \triangleleft r_{i+1} | c_{i+1} \triangleleft \im | \dots | c_{n} \triangleleft \im} = \\
      \bra{r_{0} (c_{0} \triangleleft \im) | \dots | r_{i} \comp  \coMult \comp  c_{i} \triangleleft \im | r_{i+1} (c_{i+1}\triangleleft \im) | \dots | c_{n} \triangleleft \im} = \\
      \bra{r_{0} (c_{0} \triangleleft \im) | \dots | r_{i} (c_{i} \triangleleft \im) | r_{i+1} (c_{i+1}\triangleleft \im) | \dots | c_{n} \triangleleft \im}.
   \end{aligned}\]
 A particular case of this claim is then that
 \[\begin{aligned}
     \bra{c_{0} \triangleleft \im | \dots |c_{n} \triangleleft \im} = \\
     \bra{r_{0} (c_{0} \triangleleft \im) | \dots| r_{n} (c_{n} \triangleleft \im)} = \\
     \bra{r_{0} (c_{0}' \triangleleft \im) | \dots| r_{n} (c_{n}' \triangleleft \im)} = \\
     \bra{c_{0}' \triangleleft \im | \dots |c_{n}' \triangleleft \im}.
   \end{aligned}\]
 This can be then proven for any $n \in \naturals$.
\end{proof}

\begin{cor}
  Any stochastic process $f \in \StochProc(\stream{X},\stream{Y})$ with
  a family of conditionals $c_{n}$ gives rise to the observational sequence
  \[\begin{gathered}
      \obser(f) = [\left\langle (c_{n} \triangleleft \im) \colon (X_{0} \times Y_{0} \times \dots \times X_{n-1} \times Y_{n-1}) \times X_{n} \to \right. \\ \left. (X_{0} \times Y_{0} \times \dots \times X_{n} \times Y_{n}) \times Y_{n}  \right\rangle ]_{\approx},
    \end{gathered}\]
  which is independent of the chosen family of conditionals.
\end{cor}
\begin{proof}
  Any two families of conditionals for $f$ give rise to the same n-stage processes in $\Stoch{}$ (by \Cref{prop:condtoproc}).
  Being a \productive{} category, observational sequences are determined by their n-stage procesess.
\end{proof}

\subsection{Stochastic streams are stochastic processes}
Stochastic monoidal streams and stochastic processes not only are the same thing but they compose in the same way: they are \emph{isomorphic} as categories.

\begin{prop}
  An observational sequence in $\Stoch$,
  \[[\braket{g_{n} \colon M_{n-1} \tensor X_{n} \to M_{n} \tensor Y_{n}}]_{\approx} \in \oSeq(\stream{X},\stream{Y})\]
  gives rise to a stochastic process $\mathrm{proc}(g) \in \StochProc(\stream{X},\stream{Y})$
  defined by
  $\mathrm{proc}(g)_{n} = \tid{g_{0}} \comp  \tid{g_{1}} \comp  \dots \comp  \tid{g_{n}} \comp  \tid{\varepsilon_{M_{n}}}$.
\end{prop}
\begin{proof}
  The symmetric monoidal category $\Stoch$ is \productive{}: by \Cref{lemma:observationalstatefulisterminalsequence}, observational sequences are determined by their n-stage truncations
  \[\bra{g_{0}|\dots|g_{n}} \in \Stage{n}(\stream{X},\stream{Y}).\]
  Each n-stage truncation gives rise to the n-th component of the stochastic process,
  $\mathrm{proc}(g)_{n} = \tid{g_{0}} \comp  \tid{g_{1}} \comp  \dots \comp  \tid{g_{n}} \comp  \tid{\varepsilon_{M_{n}}}$,
  and this is well-defined: composing the morphisms is invariant to \emph{sliding equivalence}, and the last discarding map is natural.

  It only remains to show that they satisfy the marginalisation property. Indeed,
  \[\begin{aligned}
      \mathrm{proc}(g)_{n+1} \comp  \tid{\varepsilon_{n+1}}
      &= \tid{g_{0}} \comp  \tid{g_{1}} \comp  \dots \comp  \tid{g_{n+1}} \comp  \tid{\varepsilon_{M_{n+1}}} \comp  \tid{\varepsilon_{Y_{n+1}}} \\
      &= \tid{g_{0}} \comp  \tid{g_{1}} \comp  \dots \comp  \tid{g_{n}} \comp  \tid{\varepsilon_{M_{n}}} \\
      &= \mathrm{proc}(g)_n.
    \end{aligned}\]
  Thus, $\mathrm{proc}(g)$ is a stochastic process in $\StochProc(\stream{X},\stream{Y})$.
\end{proof}

\begin{prop}\label{prop:procobs}
  Let $f \in \StochProc(\stream{X},\stream{Y})$, we have that $\proc(\obser(f)) = f$.
\end{prop}
\begin{proof}
  Indeed, for $c_{n}$ some family of conditionals,
  \[f_{n} = c_{0} \triangleleft \dots \triangleleft c_{n}
    = (c_{0} \triangleleft \im) \comp  \dots \comp  (c_{n} \triangleleft \im) \comp  \varepsilon_{M_{n}}.\qedhere\]
\end{proof}

\begin{thm}\label{theorem:obstoch}
  Observational sequences in $\Stoch$ are in bijection with stochastic processes.
\end{thm}
\begin{proof}
  The function $\obser$ is injective by \Cref{prop:procobs}.
  We only need to show it is also surjective.

  We will prove that any n-stage process $\bra{g_{0}|\dots|g_{n}}$
  can be equivalently written in the form $\bra{(c_{0} \triangleleft \im)|\dots|(c_{n} \triangleleft \im)}$.
  We proceed by induction.
  Given any $\bra{g_{0}}$ we use conditionals and dinaturality to rewrite it as
  \[\bra{g_{0}} = \bra{c_{0} \triangleleft c_{M}} = \bra{c_{0} \triangleleft \im}.\]
  Given any $\bra{c_{0} \triangleleft \im|\dots|c_{n} \triangleleft \im |g_{n+1}}$,
  we use again conditionals and dinaturality to rewrite it as
  \[\begin{aligned}
    \bra{c_{0} \triangleleft \im|\dots|c_{n} \triangleleft \im|g_{n+1}} = \\
    \bra{c_{0} \triangleleft \im|\dots|c_{n} \triangleleft \im|c_{n+1}\triangleleft c_{M}} = \\
    \bra{c_{0} \triangleleft \im|\dots|c_{n} \triangleleft \im|c_{n+1}\triangleleft \im}.
  \end{aligned}\]
   We have shown that $\obser$ is both injective and surjective.
\end{proof}

\begin{thm}[From \Cref{th:stochasticprocesses}]
  \label{theorem:obstochiso}
  \label{th:stochasticprocesses}
  The category $\StochProc$ of stochastic processes is monoidally isomorphic to the category $\STREAM$ over $\Stoch$.
\end{thm}
\begin{proof}
  We have shown in \Cref{theorem:obstoch} that $\proc$ is a bijection.
  Let us show that it preserves compositions. Indeed,
  \[\begin{aligned}
      \proc(g \comp  h)_{n}
      & = \tid{g_0} \comp  \tid{h_{0}} \comp  \dots \comp  \tid{g_{n}}\comp  \tid{h_{n}}\comp  \tid{\varepsilon_{M_{n} \tensor N_{n}}} \\
      & = \tid{g_0} \comp  \dots \comp  \tid{g_{n}}\comp  \tid{h_{0}} \comp  \dots \comp  \tid{h_{n}}\comp  \tid{(\varepsilon_{M_{n}} \tensor \varepsilon_{N_{n}})} \\
      & = \tid{g_0} \dots \tid{g_{n}} \comp \tid{\varepsilon_{M_{n}}} \comp  \tid{h_{0}} \dots \tid{h_{n}} \comp  \tid{\varepsilon_{N_{n}}} \\
      & = \proc(g)_{n} \comp  \proc(h)_{n}.
  \end{aligned}\]
It also trivially preserves the identity. It induces thus an identity-on-objects functor which is moreover
an equivalence of categories. Let us finally show that it preserves tensoring of morphisms.
  \[\begin{aligned}
      \proc(g \tensor h)_{n}
      & = \tid{(g_0 \tensor h_{0})} \comp  \dots  \comp \tid{(g_{n} \tensor h_{n})} \comp  \tid{\varepsilon_{M_{n} \tensor N_{n}}} \\
      & = \tid{(g_0 \tensor h_{0})} \comp  \dots \tid{((g_{n}  \varepsilon_{M_{n}}) \tensor (h_{n} \varepsilon_{N_{n}}))}  \\
      & = \tid{(g_0 \dots g_{n}  \varepsilon_{M_{n}})} \tensor \tid{(h_{0} \dots h_{n}  \varepsilon_{N_{n}})} \\
      & = \proc(g)_{n} \tensor \proc(h)_{n}.
  \end{aligned}\]
It is thus also a monoidal equivalence.
\end{proof}

We expect that the theorem above can be generalised to interesting categories
of probabilistic channels over measurable spaces (such as the ones covered
in \cite{panangaden1999, cho2019, fritz2020}).

\begin{cor}
  \(\StochProc\) is a \feedbackMonoidalCategory{}.
\end{cor}

\subsection{Examples}

We have characterized in two equivalent ways the notion of controlled stochastic process.
This yields a categorical semantics for probabilistic dataflow programming:
we may use the syntax of feedback monoidal categories to specify simple stochastic programs and evaluate their semantics in $\StochProc$.

\begin{exa}[Random Walk] \label{example:walk}

Recall the morphism $\mathsf{walk} \in \STREAM(\mathbf{1},\mathbb{Z})$ that we
depicted back in \Cref{string:walk}.

Here, $\mathsf{unif} \in \STREAM(\mathbf{1},\{-1,1\})$, is a uniform random generator
that, at each step, outputs either $1$ or $(-1)$. The output of this
uniform random generator is then added to the current position, and we declare
the starting position to be $0$.
Our implementation of this morphism \cite{Roman_Arrow_Streams_for_2022}, following the definitions from \Cref{sec:monoidal-streams} is, by \Cref{th:stochasticprocesses}, a discrete stochastic process, and it produces samples like the following ones.
\[\begin{aligned}
& [0,1,0,-1,-2,-1,-2,-3,-2,-3,\dots] \\
& [0,1,2,1,2,1,2,3,4,5,\dots] \\
& [0,-1,-2,-1,-2,-1,0,-1,0,-1,\dots]
\end{aligned}\]
 \end{exa}

\begin{exa}[Ehrenfest model]\label{example:ehrenfest}
The Ehrenfest model \cite[\S 1.4]{kelly11reversibility} is a simplified model of particle diffusion.

\begin{figure}[h!]
\begin{minipage}{0.35\linewidth}
\tikzset{every picture/.style={line width=0.85pt}} %
\begin{tikzpicture}[x=0.75pt,y=0.75pt,yscale=-1,xscale=1]
\draw    (240,130) .. controls (240.4,111.8) and (300.67,118.5) .. (300,100) ;
\draw   (210,130) -- (250,130) -- (250,150) -- (210,150) -- cycle ;
\draw   (270,130) -- (310,130) -- (310,150) -- (270,150) -- cycle ;
\draw    (230.08,62.47) .. controls (231.97,102.56) and (260.33,130.96) .. (260,170) ;
\draw [shift={(230,60)}, rotate = 88.89] [color={rgb, 255:red, 0; green, 0; blue, 0 }  ][line width=0.75]    (10.93,-3.29) .. controls (6.95,-1.4) and (3.31,-0.3) .. (0,0) .. controls (3.31,0.3) and (6.95,1.4) .. (10.93,3.29)   ;
\draw    (230,150) -- (230,158.87) ;
\draw    (220,170) .. controls (219.8,155.8) and (240.2,155.8) .. (240,170) ;
\draw  [fill={rgb, 255:red, 0; green, 0; blue, 0 }  ,fill opacity=1 ] (227.1,158.87) .. controls (227.1,157.27) and (228.4,155.97) .. (230,155.97) .. controls (231.6,155.97) and (232.9,157.27) .. (232.9,158.87) .. controls (232.9,160.47) and (231.6,161.77) .. (230,161.77) .. controls (228.4,161.77) and (227.1,160.47) .. (227.1,158.87) -- cycle ;
\draw    (300,170) -- (300,201.13) ;
\draw    (290,150) -- (290,160) ;
\draw  [fill={rgb, 255:red, 0; green, 0; blue, 0 }  ,fill opacity=1 ] (287.1,158.87) .. controls (287.1,157.27) and (288.4,155.97) .. (290,155.97) .. controls (291.6,155.97) and (292.9,157.27) .. (292.9,158.87) .. controls (292.9,160.47) and (291.6,161.77) .. (290,161.77) .. controls (288.4,161.77) and (287.1,160.47) .. (287.1,158.87) -- cycle ;
\draw    (280,170) .. controls (279.8,155.8) and (300.2,155.8) .. (300,170) ;
\draw    (260,170) .. controls (260.67,184.17) and (279.67,185.5) .. (280,170) ;
\draw    (230,60) .. controls (229.8,45.8) and (250.2,45.8) .. (250,60) ;
\draw    (260,80) .. controls (260.14,92.14) and (280,117.83) .. (280,130) ;
\draw   (290,60) -- (330,60) -- (330,80) -- (290,80) -- cycle ;
\draw    (310,80) -- (310,90) ;
\draw  [fill={rgb, 255:red, 0; green, 0; blue, 0 }  ,fill opacity=1 ] (307.1,90) .. controls (307.1,88.4) and (308.4,87.1) .. (310,87.1) .. controls (311.6,87.1) and (312.9,88.4) .. (312.9,90) .. controls (312.9,91.6) and (311.6,92.9) .. (310,92.9) .. controls (308.4,92.9) and (307.1,91.6) .. (307.1,90) -- cycle ;
\draw    (300,100) .. controls (299.8,85.8) and (320.2,85.8) .. (320,100) ;
\draw    (200,170) .. controls (200.67,184.17) and (219.67,185.5) .. (220,170) ;
\draw    (240,170) -- (240,200) ;
\draw    (170,60) .. controls (169.8,45.8) and (190.2,45.8) .. (190,60) ;
\draw    (300,130) .. controls (300.33,111.17) and (320.67,126.83) .. (320,100) ;
\draw    (170.08,62.47) .. controls (171.96,102.57) and (200,131.29) .. (200,170) ;
\draw [shift={(170,60)}, rotate = 88.89] [color={rgb, 255:red, 0; green, 0; blue, 0 }  ][line width=0.75]    (10.93,-3.29) .. controls (6.95,-1.4) and (3.31,-0.3) .. (0,0) .. controls (3.31,0.3) and (6.95,1.4) .. (10.93,3.29)   ;
\draw   (180,60) -- (220,60) -- (220,80) -- (180,80) -- cycle ;
\draw   (190,20) -- (230,20) -- (230,40) -- (190,40) -- cycle ;
\draw  [fill={rgb, 255:red, 255; green, 255; blue, 255 }  ,fill opacity=1 ] (240,60) -- (280,60) -- (280,80) -- (240,80) -- cycle ;
\draw  [fill={rgb, 255:red, 255; green, 255; blue, 255 }  ,fill opacity=1 ] (250,20) -- (290,20) -- (290,40) -- (250,40) -- cycle ;
\draw    (210,40) -- (210,60) ;
\draw    (270,40) -- (270,60) ;
\draw    (200,80) .. controls (200.14,92.14) and (220,117.83) .. (220,130) ;
\draw (230,140) node  [font=\normalsize]  {$move$};
\draw (290,140) node  [font=\normalsize]  {$move$};
\draw (310,70) node  [font=\normalsize]  {$unif$};
\draw (200,70) node  [font=\normalsize]  {$fby$};
\draw (210,30) node  [font=\footnotesize]  {$( 1..4)$};
\draw (260,70) node  [font=\normalsize]  {$fby$};
\draw (270,30) node  [font=\footnotesize]  {$()$};
\end{tikzpicture}
\end{minipage}
\begin{minipage}{0.45\linewidth}
  \[\begin{gathered}
     \mathsf{ehr} \defn \\
     (1...4) \tensor () ;
     \fbk( \sigma ; \\
    \fby \tensor  \fby \tensor \mathsf{unif}; \\
    \im \tensor \im \tensor \COPY {;} \sigma ; \\
     \mathsf{move} \tensor \mathsf{move} {;} \\
     \COPY)
  \end{gathered}\]
\end{minipage}
\caption{Ehrenfest model: sig. flow graph and morphism.}
\label{fig:example:ehrenfest}
\end{figure}
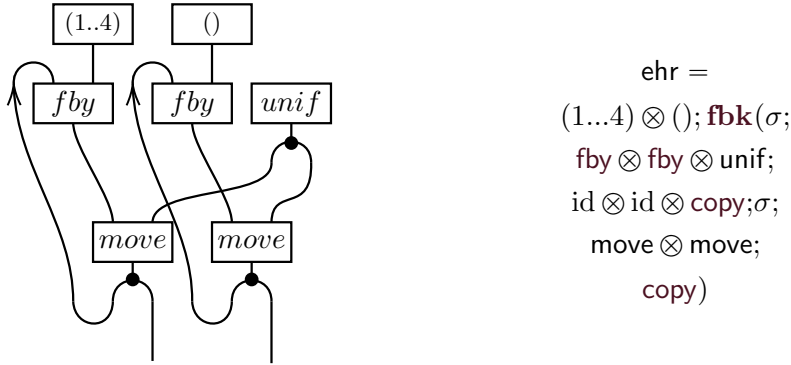

Assume we have two urns with 4 balls, labelled from 1 to 4.
Initially, the balls are all in the first urn.
We randomly (and uniformly) pick an integer from 1 to 4, and the ball labelled by that
number is removed from its box and placed in the other box.
We iterate the procedure, with independent uniform selections each time.

Our implementation of this morphism \cite{Roman_Arrow_Streams_for_2022}, following the definitions from \Cref{sec:monoidal-streams} yields samples such as the following.

$\begin{aligned}
& [([2,3,4],[1]),\ & ([3,4],[1,2]),\ && ([1,3,4],[2]), \\
& ([1,4],[2,3]),\ & ([1],[2,3,4]),\ && ([],[1,2,3,4]), \\
& ([2],[1,3,4]),\ & \dots]
\end{aligned}$
 \end{exa}

\section{Conclusions}
Monoidal streams are a common generalization of streams, causal functions and stochastic processes.
In the same way that streams give semantics to dataflow programming~\cite{wadge1985lucid,halbwachs1991lustre} with plain functions, monoidal streams give semantics to dataflow programming with monoidal theories of processes.
Signal flow graphs are a common tool to describe control flow in dataflow programming.
Signal flow graphs are also the natural string diagrams of feedback monoidal categories.
Monoidal streams form a feedback monoidal category, and signal flow graphs are a formal syntax to describe and reason about monoidal streams.
The second syntax we present comes from the type theory of monoidal categories, and it is inspired by the original syntax of dataflow programming.
We have specifically studied stochastic dataflow programming, but the same framework allows for \emph{linear}, \emph{quantum} and \emph{effectful} theories of resources.

The literature on dataflow and feedback is rich enough to provide multiple diverging definitions and approaches.
What we can bring to this discussion are universal constructions.
Universal constructions justify some mathematical object as \emph{the canonical object} satisfying some properties.
In our case, these exact properties are extracted from three, arguably under-appreciated, but standard category-theoretic tools: \emph{dinaturality}, \emph{feedback}, and \emph{coalgebra}.
\emph{Dinaturality}, profunctors and coends, sometimes regarded as highly theoretical developments, are the natural language to describe how processes communicate and compose.
\emph{Feedback}, sometimes eclipsed by trace in the mathematical literature, keeps appearing in multiple variants across computer science.
\emph{Coalgebra} is the established tool to specify and reason about stateful systems.
\pagebreak[5]
\subsection{Further work}

\paragraph{Other theories.}
Many interesting examples of theories of processes are not monoidal but just \emph{premonoidal categories}~\cite{jeffrey97,power02}. For instance, the kleisli categories of arbitrary monads, where effects (e.g. reading and writing to a global state) do not need to commute. %
Premonoidal streams can be constructed by restricting dinaturality to their \emph{centres}, and they seem to be easily implementable in existing software for category theory \cite{defelice2020discopy}.
Another important source of theories of processes that we have not covered is that of \emph{linearly distributive} and \emph{*-autonomous categories} \cite{seely87,cockett1997,blute93,blute96}.

Within monoidal categories, we would like to make monoidal streams explicit in the cases of partial maps~\cite{cockett02} for dataflow programming with different clocks \cite{uustalu2008comonadic}, non-deterministic maps~\cite{broy2001algebra,lee09} and quantum processes~\cite{carette21}. Generalizing this, a notion of process similar to the one we employ is described by Krstic, Launchbury, and Pavlovic \cite{krsticLP01}.

\paragraph{The 2-categorical view.}

We describe the  morphisms of a category as a final coalgebra.
However, it is also straightforward to describe the 2-endofunctor that should give rise to this category as a final coalgebra itself.

\paragraph{Implementation of the type theory.}

Justifying that the output of monoidal streams is the expected one requires some computations, which we have already implemented separately in the Haskell programming language (see \cite{Roman_Arrow_Streams_for_2022}).

Agda has similar foundations and supports the coinductive definitions of this text (\Cref{sec:monoidal-streams}).
It is possible to implement a whole interpreter for a \Lucid{}-like stochastic programming language with a dedicated parser, but that requires some software engineering effort that we postpone for further work.

\paragraph{Expressivity.}
A final question we do not pursue here is \emph{expressivity}: the class of functions a monoidal stream can define.
Analogous to Bucciarelli and Leperchey's finitary-PCF-interdefinability-classes~\cite{bucciarelli2004hypergraphs} there would be a notion of interdefinability relative to a PROP of generators, even if it remains unclear to us if their hypergraph approach could be transferred to our case.

We are reasonably confident that other questions about expressivity of monoidal streams can be answered from our \Cref{th:extensionalfreefeedback} and the previous literature on the subject.
There might be analogous results to those of Kahn, Plotkin and Park~\cite{kahn1976coroutines,kahn1993concrete,park1981concurrency} on recursive stream equations and to those of Panangaden and Stark~\cite{panangaden1988computations} on total-monotone-continuous relations.
Our streams are synchronous \cite{milner1980calculus,milner1983calculi} and thus it should be of interest to compare monoidal streams over Scott-continuous functions with the ``synchronous Kahn networks'' of Caspi and Pouzet~\cite{caspi1996synchronous}.
 \vfill

\subsection*{Acknowledgements}

We thank Pawe{\l} Soboci{\'n}ski, Edward
Morehouse, Niels Voorneveld, George Kaye, Chad Nester, Nathanael Arkor, Prakash
Panangaden and Ichiro Hasuo for discussion. We thank the anonymous reviewers at
LMCS and LiCS'22, whose comments have helped improve this manuscript multiple
times; we thank the reviewers for the earlier versions of this work at NWPT'21
and SYCO8.

\newpage
\bibliographystyle{alphaurl}
\bibliography{bibliography}

\newcommand{\etalchar}[1]{$^{#1}$}
\begin{thebibliography}{DLGR{\etalchar{+}}23}

\bibitem[AC09]{abramsky2009categorical}
Samson Abramsky and Bob Coecke.
\newblock Categorical quantum mechanics.
\newblock {\em Handbook of quantum logic and quantum structures}, 2:261--325, 2009.
\newblock \href {https://arxiv.org/abs/0808.1023} {\path{arXiv:0808.1023}}.

\bibitem[Ack79]{ackerman1979data}
William~B Ackerman.
\newblock Data flow languages.
\newblock In {\em 1979 International Workshop on Managing Requirements Knowledge (MARK)}, pages 1087--1095. IEEE, 1979.

\bibitem[Ada74]{adamek74}
Jiří Adamek.
\newblock Free algebras and automata realizations in the language of categories.
\newblock {\em Commentationes Mathematicae Universitatis Carolinae}, 015(4):589--602, 1974.
\newblock URL: \url{http://eudml.org/doc/16649}.

\bibitem[Ad{\'a}05]{adamek2005introduction}
Jiří Ad{\'a}mek.
\newblock Introduction to coalgebra.
\newblock {\em Theory and Applications of Categories}, 14(8):157--199, 2005.

\bibitem[Ada19]{adamek19}
Jiří Adamek.
\newblock On {Terminal} {Coalgebras} {Derived} from {Initial} {Algebras}.
\newblock In Markus Roggenbach and Ana Sokolova, editors, {\em 8th Conference on Algebra and Coalgebra in Computer Science, {CALCO} 2019, June 3-6, 2019, London, United Kingdom}, volume 139 of {\em LIPIcs}, pages 12:1--12:21. Schloss Dagstuhl - Leibniz-Zentrum f{\"{u}}r Informatik, 2019.
\newblock \href {https://doi.org/10.4230/LIPIcs.CALCO.2019.12} {\path{doi:10.4230/LIPIcs.CALCO.2019.12}}.

\bibitem[AI15]{abadi15}
Mart{\'{\i}}n Abadi and Michael Isard.
\newblock Timely {Dataflow}: {A} {Model}.
\newblock In Susanne Graf and Mahesh Viswanathan, editors, {\em Formal Techniques for Distributed Objects, Components, and Systems - 35th {IFIP} {WG} 6.1 International Conference, {FORTE} 2015, Held as Part of the 10th International Federated Conference on Distributed Computing Techniques, DisCoTec 2015, Grenoble, France, June 2-4, 2015, Proceedings}, volume 9039 of {\em Lecture Notes in Computer Science}, pages 131--145. Springer, 2015.
\newblock \href {https://doi.org/10.1007/978-3-319-19195-9\_9} {\path{doi:10.1007/978-3-319-19195-9\_9}}.

\bibitem[AM89]{aczel1989final}
Peter Aczel and Nax Mendler.
\newblock A final coalgebra theorem.
\newblock In {\em Category theory and computer science}, pages 357--365. Springer, 1989.

\bibitem[AW77]{ashcroft1977lucid}
Edward~A. Ashcroft and William~W. Wadge.
\newblock Lucid, a nonprocedural language with iteration.
\newblock {\em Communications of the ACM}, 20(7):519--526, 1977.

\bibitem[BCLH93]{benveniste93}
Albert Benveniste, Paul Caspi, Paul {Le Guernic}, and Nicolas Halbwachs.
\newblock Data-flow synchronous languages.
\newblock In J.~W. de~Bakker, Willem~P. de~Roever, and Grzegorz Rozenberg, editors, {\em A Decade of Concurrency, Reflections and Perspectives, {REX} School/Symposium, Noordwijkerhout, The Netherlands, June 1-4, 1993, Proceedings}, volume 803 of {\em Lecture Notes in Computer Science}, pages 1--45. Springer, 1993.
\newblock \href {https://doi.org/10.1007/3-540-58043-3\_16} {\path{doi:10.1007/3-540-58043-3\_16}}.

\bibitem[BCST96]{blute96}
Richard~F Blute, J~Robin~B Cockett, Robert~AG Seely, and Todd~H Trimble.
\newblock Natural deduction and coherence for weakly distributive categories.
\newblock {\em Journal of Pure and Applied Algebra}, 113(3):229--296, 1996.

\bibitem[B{\'{E}}93]{bloom93}
Stephen~L. Bloom and Zolt{\'{a}}n {\'{E}}sik.
\newblock {\em Iteration Theories - The Equational Logic of Iterative Processes}.
\newblock {EATCS} Monographs on Theoretical Computer Science. Springer, 1993.
\newblock \href {https://doi.org/10.1007/978-3-642-78034-9} {\path{doi:10.1007/978-3-642-78034-9}}.

\bibitem[Bec69]{beck69}
Jon Beck.
\newblock Distributive laws.
\newblock In {\em Seminar on triples and categorical homology theory}, pages 119--140. Springer, 1969.

\bibitem[BFP16]{baez2016}
John~C. Baez, Brendan Fong, and Blake~S. Pollard.
\newblock A {{Compositional Framework}} for {{Markov Processes}}.
\newblock {\em Journal of Mathematical Physics}, 57(3):033301, March 2016.
\newblock \href {https://arxiv.org/abs/1508.06448} {\path{arXiv:1508.06448}}, \href {https://doi.org/10.1063/1.4941578} {\path{doi:10.1063/1.4941578}}.

\bibitem[BHP{\etalchar{+}}19]{bonchi19}
Filippo Bonchi, Joshua Holland, Robin Piedeleu, Paweł Sobociński, and Fabio Zanasi.
\newblock Diagrammatic algebra: from linear to concurrent systems.
\newblock {\em Proc. {ACM} Program. Lang.}, 3({POPL}):25:1--25:28, 2019.
\newblock \href {https://doi.org/10.1145/3290338} {\path{doi:10.1145/3290338}}.

\bibitem[BL04]{bucciarelli2004hypergraphs}
Antonio Bucciarelli and Benjamin Leperchey.
\newblock Hypergraphs and degrees of parallelism: A completeness result.
\newblock In {\em International Conference on Foundations of Software Science and Computation Structures}, pages 58--71. Springer, 2004.

\bibitem[Blu93]{blute93}
Richard Blute.
\newblock Linear logic, coherence, and dinaturality.
\newblock {\em Theor. Comput. Sci.}, 115(1):3--41, 1993.
\newblock \href {https://doi.org/10.1016/0304-3975(93)90053-V} {\path{doi:10.1016/0304-3975(93)90053-V}}.

\bibitem[BM89]{bloom1989remark}
Bard Bloom and Albert~R Meyer.
\newblock A remark on bisimulation between probabilistic processes.
\newblock In {\em International Symposium on Logical Foundations of Computer Science}, pages 26--40. Springer, 1989.

\bibitem[BMSS11]{birkedal11}
Lars Birkedal, Rasmus~Ejlers M{\o}gelberg, Jan Schwinghammer, and Kristian St{\o}vring.
\newblock First steps in synthetic guarded domain theory: Step-indexing in the topos of trees.
\newblock In {\em Proceedings of the 26th Annual {IEEE} Symposium on Logic in Computer Science, {LICS} 2011, June 21-24, 2011, Toronto, Ontario, Canada}, pages 55--64. {IEEE} Computer Society, 2011.
\newblock \href {https://doi.org/10.1109/LICS.2011.16} {\path{doi:10.1109/LICS.2011.16}}.

\bibitem[B{\c{S}}01]{broy2001algebra}
Manfred Broy and Gheorghe {\c{S}}tef{\u{a}}nescu.
\newblock The algebra of stream processing functions.
\newblock {\em Theoretical Computer Science}, 258(1-2):99--129, 2001.

\bibitem[BSS18]{Bonchi18}
Filippo Bonchi, Jens Seeber, and Paweł Sobociński.
\newblock Graphical conjunctive queries.
\newblock In Dan~R. Ghica and Achim Jung, editors, {\em 27th {EACSL} Annual Conference on Computer Science Logic, {CSL} 2018, September 4-7, 2018, Birmingham, {UK}}, volume 119 of {\em LIPIcs}, pages 13:1--13:23. Schloss Dagstuhl - Leibniz-Zentrum f{\"{u}}r Informatik, 2018.
\newblock \href {https://doi.org/10.4230/LIPIcs.CSL.2018.13} {\path{doi:10.4230/LIPIcs.CSL.2018.13}}.

\bibitem[BSZ14]{bonchi14}
Filippo Bonchi, Paweł Sobociński, and Fabio Zanasi.
\newblock A categorical semantics of signal flow graphs.
\newblock In {\em International Conference on Concurrency Theory}, pages 435--450. Springer, 2014.

\bibitem[BSZ15]{bonchi15}
Filippo Bonchi, Paweł Sobociński, and Fabio Zanasi.
\newblock Full abstraction for signal flow graphs.
\newblock {\em ACM SIGPLAN Notices}, 50(1):515--526, 2015.

\bibitem[CdVP21]{carette21}
Titouan Carette, Marc de~Visme, and Simon Perdrix.
\newblock Graphical language with delayed trace: Picturing quantum computing with finite memory.
\newblock In {\em 36th Annual {ACM/IEEE} Symposium on Logic in Computer Science, {LICS} 2021, Rome, Italy, June 29 - July 2, 2021}, pages 1--13. {IEEE}, 2021.
\newblock \href {https://doi.org/10.1109/LICS52264.2021.9470553} {\path{doi:10.1109/LICS52264.2021.9470553}}.

\bibitem[CFS16]{coeckeFS16}
Bob Coecke, Tobias Fritz, and Robert~W. Spekkens.
\newblock A mathematical theory of resources.
\newblock {\em Inf. Comput.}, 250:59--86, 2016.
\newblock \href {https://doi.org/10.1016/j.ic.2016.02.008} {\path{doi:10.1016/j.ic.2016.02.008}}.

\bibitem[CGH12]{cockett2012range}
J.~Robin~B. Cockett, Xiuzhan Guo, and Pieter Hofstra.
\newblock Range {Categories} {I}: {General} theory.
\newblock {\em Theory and Applications of Categories}, 26(17):412--452, 2012.

\bibitem[CJ19]{cho2019}
Kenta Cho and Bart Jacobs.
\newblock Disintegration and {{Bayesian Inversion}} via {{String Diagrams}}.
\newblock {\em Mathematical Structures in Computer Science}, pages 1--34, March 2019.
\newblock \href {https://arxiv.org/abs/1709.00322} {\path{arXiv:1709.00322}}, \href {https://doi.org/10.1017/S0960129518000488} {\path{doi:10.1017/S0960129518000488}}.

\bibitem[CL02]{cockett02}
J.~Robin~B. Cockett and Stephen Lack.
\newblock Restriction categories {I:} categories of partial maps.
\newblock {\em Theoretical Computer Science}, 270(1-2):223--259, 2002.
\newblock \href {https://doi.org/10.1016/S0304-3975(00)00382-0} {\path{doi:10.1016/S0304-3975(00)00382-0}}.

\bibitem[Cou19]{cousot19}
Patrick Cousot.
\newblock Syntactic and semantic soundness of structural dataflow analysis.
\newblock In Bor{-}Yuh~Evan Chang, editor, {\em Static Analysis - 26th International Symposium, {SAS} 2019, Porto, Portugal, October 8-11, 2019, Proceedings}, volume 11822 of {\em Lecture Notes in Computer Science}, pages 96--117. Springer, 2019.
\newblock \href {https://doi.org/10.1007/978-3-030-32304-2\_6} {\path{doi:10.1007/978-3-030-32304-2\_6}}.

\bibitem[CP96]{caspi1996synchronous}
Paul Caspi and Marc Pouzet.
\newblock Synchronous {Kahn} networks.
\newblock {\em ACM SIGPLAN Notices}, 31(6):226--238, 1996.

\bibitem[CS97]{cockett1997}
J~Robin~B Cockett and Robert~AG Seely.
\newblock Weakly distributive categories.
\newblock {\em Journal of Pure and Applied Algebra}, 114(2):133--173, 1997.

\bibitem[Del19]{delpeuch19}
Antonin Delpeuch.
\newblock A complete language for faceted dataflow programs.
\newblock In John Baez and Bob Coecke, editors, {\em Proceedings Applied Category Theory 2019, {ACT} 2019, University of Oxford, UK, 15-19 July 2019}, volume 323 of {\em {EPTCS}}, pages 1--14, 2019.
\newblock \href {https://doi.org/10.4204/EPTCS.323.1} {\path{doi:10.4204/EPTCS.323.1}}.

\bibitem[Den74]{dennis1974dataflow}
Jack~B. Dennis.
\newblock First version of a data flow procedure language.
\newblock In B.~Robinet, editor, {\em Programming Symposium}, pages 362--376. Springer Berlin Heidelberg, 1974.
\newblock \href {https://doi.org/10.1007/3-540-06859-7_145} {\path{doi:10.1007/3-540-06859-7_145}}.

\bibitem[DFL72]{dennis1972schemas}
J.~B. Dennis, J.~B. Fosseen, and J.~P. Linderman.
\newblock Data flow schemas.
\newblock In Andrei Ershov and Valery~A. Nepomniaschy, editors, {\em International Symposium on Theoretical Programming}, pages 187--216, Berlin, Heidelberg, 1972. Springer Berlin Heidelberg.
\newblock \href {https://doi.org/10.1007/3-540-06720-5_15} {\path{doi:10.1007/3-540-06720-5_15}}.

\bibitem[dFTC20]{defelice2020discopy}
Giovanni de~Felice, Alexis Toumi, and Bob Coecke.
\newblock {DisCoPy}: Monoidal categories in {Python}.
\newblock In {\em Proceedings of Applied Category Theory, {ACT} 2019, University of Oxford, UK, 15-19 July 2019}, volume 323 of {\em {EPTCS}}, pages 183--197, 2020.
\newblock \href {https://doi.org/10.4204/EPTCS.333.13} {\path{doi:10.4204/EPTCS.333.13}}.

\bibitem[DLGR{\etalchar{+}}21]{feedbackspans2020}
Elena Di~Lavore, Alessandro Gianola, Mario Rom{\'a}n, Nicoletta Sabadini, and Paweł Sobociński.
\newblock A canonical algebra of open transition systems.
\newblock In Gwen Sala{\"u}n and Anton Wijs, editors, {\em Formal Aspects of Component Software}, pages 63--81, Cham, 2021. Springer International Publishing.

\bibitem[DLGR{\etalchar{+}}23]{2023canonicalalgebra}
Elena Di~Lavore, Alessandro Gianola, Mario Rom{\'a}n, Nicoletta Sabadini, and Pawe{\l} Soboci{\'{n}}ski.
\newblock {Span(Graph): a Canonical Feedback Algebra of Open Transition Systems}.
\newblock {\em Software and Systems Modeling}, 22:495--520, 2023.
\newblock \href {https://arxiv.org/abs/2010.10069} {\path{arXiv:2010.10069}}, \href {https://doi.org/10.1007/s10270-023-01092-7} {\path{doi:10.1007/s10270-023-01092-7}}.

\bibitem[Fos72]{fosseen1972representation}
John~Blake Fosseen.
\newblock Representation of algorithms by maximally parallel schemata.
\newblock Master's thesis, Massachusetts Institute of Technology, 1972.

\bibitem[Fox76]{fox76}
Thomas Fox.
\newblock Coalgebras and cartesian categories.
\newblock {\em Communications in Algebra}, 4(7):665--667, 1976.

\bibitem[FR75]{fleming1975}
Wendell~Helms Fleming and Raymond~W. Rishel.
\newblock {\em Deterministic and Stochastic Optimal Control}.
\newblock Number vol 1 in Applications of {{Mathematics}}. {Springer-Verlag}, {Berlin ; New York}, 1975.

\bibitem[Fri20]{fritz2020}
Tobias Fritz.
\newblock A synthetic approach to {{Markov}} kernels, conditional independence and theorems on sufficient statistics.
\newblock {\em Advances in Mathematics}, 370:107239, 2020.
\newblock URL: \url{http://arxiv.org/abs/1908.07021}, \href {https://arxiv.org/abs/1908.07021} {\path{arXiv:1908.07021}}.

\bibitem[FS19]{fong2019supplying}
Brendan Fong and David~I Spivak.
\newblock Supplying bells and whistles in symmetric monoidal categories.
\newblock {\em arXiv preprint arXiv:1908.02633}, 2019.

\bibitem[Gar21]{garner2021stream}
Richard Garner.
\newblock Stream processors and comodels.
\newblock {\em arXiv preprint arXiv:2106.05473}, 2021.

\bibitem[GKS22]{kaye22}
Dan~R. Ghica, George Kaye, and David Sprunger.
\newblock Full abstraction for digital circuits, 2022.
\newblock \href {https://arxiv.org/abs/2201.10456} {\path{arXiv:2201.10456}}.

\bibitem[GN03]{gay03}
Simon~J. Gay and Rajagopal Nagarajan.
\newblock Intensional and extensional semantics of dataflow programs.
\newblock {\em Formal Aspects Comput.}, 15(4):299--318, 2003.
\newblock \href {https://doi.org/10.1007/s00165-003-0018-1} {\path{doi:10.1007/s00165-003-0018-1}}.

\bibitem[HJ98]{hermida98}
Claudio Hermida and Bart Jacobs.
\newblock Structural induction and coinduction in a fibrational setting.
\newblock {\em Inf. Comput.}, 145(2):107--152, 1998.
\newblock \href {https://doi.org/10.1006/inco.1998.2725} {\path{doi:10.1006/inco.1998.2725}}.

\bibitem[HLR92]{halbwachs1991lustre}
Nicolas Halbwachs, Fabienne Lagnier, and Christophe Ratel.
\newblock Programming and verifying real-time systems by means of the synchronous data-flow language {LUSTRE}.
\newblock {\em {IEEE} Trans. Software Eng.}, 18(9):785--793, 1992.
\newblock \href {https://doi.org/10.1109/32.159839} {\path{doi:10.1109/32.159839}}.

\bibitem[HMH14]{hoshino14}
Naohiko Hoshino, Koko Muroya, and Ichiro Hasuo.
\newblock Memoryful geometry of interaction: from coalgebraic components to algebraic effects.
\newblock In Thomas~A. Henzinger and Dale Miller, editors, {\em Joint Meeting of the Twenty-Third {EACSL} Annual Conference on Computer Science Logic {(CSL)} and the Twenty-Ninth Annual {ACM/IEEE} Symposium on Logic in Computer Science (LICS), {CSL-LICS} '14, Vienna, Austria, July 14 - 18, 2014}, pages 52:1--52:10. {ACM}, 2014.
\newblock \href {https://doi.org/10.1145/2603088.2603124} {\path{doi:10.1145/2603088.2603124}}.

\bibitem[Hon81]{honsell1981modelli}
Furio Honsell.
\newblock {\em Modelli della teoria degli insiemi, principi di regolarità e di libera costruzione}.
\newblock PhD thesis, Tesi di Laurea, Università di Pisa, 1981.

\bibitem[HPW98]{hildebrandt1998relational}
Thomas Hildebrandt, Prakash Panangaden, and Glynn Winskel.
\newblock A relational model of non-deterministic dataflow.
\newblock In {\em International Conference on Concurrency Theory}, pages 613--628. Springer, 1998.

\bibitem[Jac16]{jacobs2005coalgebras}
Bart Jacobs.
\newblock {\em Introduction to Coalgebra: Towards Mathematics of States and Observation}, volume~59 of {\em Cambridge Tracts in Theoretical Computer Science}.
\newblock Cambridge University Press, 2016.
\newblock \href {https://doi.org/10.1017/CBO9781316823187} {\path{doi:10.1017/CBO9781316823187}}.

\bibitem[Jef97]{jeffrey97}
Alan Jeffrey.
\newblock Premonoidal categories and flow graphs.
\newblock {\em Electron. Notes Theor. Comput. Sci.}, 10:51, 1997.
\newblock \href {https://doi.org/10.1016/S1571-0661(05)80688-7} {\path{doi:10.1016/S1571-0661(05)80688-7}}.

\bibitem[JS91]{joyal1991geometry}
Andr{\'e} Joyal and Ross Street.
\newblock The geometry of tensor calculus, {I}.
\newblock {\em Advances in mathematics}, 88(1):55--112, 1991.

\bibitem[JSV96]{joyal96}
Andr\'e Joyal, Ross Street, and Dominic Verity.
\newblock Traced monoidal categories.
\newblock {\em Mathematical Proceedings of the Cambridge Philosophical Society}, 119:447 -- 468, 04 1996.
\newblock \href {https://doi.org/10.1017/S0305004100074338} {\path{doi:10.1017/S0305004100074338}}.

\bibitem[Kah74]{kahn1974semantics}
Gilles Kahn.
\newblock The semantics of a simple language for parallel programming.
\newblock {\em Information processing}, 74:471--475, 1974.

\bibitem[Kel11]{kelly11reversibility}
Frank~P. Kelly.
\newblock {\em Reversibility and stochastic networks}.
\newblock Cambridge University Press, 2011.

\bibitem[KLP01]{krsticLP01}
Sava Krstic, John Launchbury, and Dusko Pavlovic.
\newblock Categories of processes enriched in final coalgebras.
\newblock In Furio Honsell and Marino Miculan, editors, {\em Foundations of Software Science and Computation Structures, 4th International Conference, {FOSSACS} 2001 Held as Part of the Joint European Conferences on Theory and Practice of Software, {ETAPS} 2001 Genova, Italy, April 2-6, 2001, Proceedings}, volume 2030 of {\em Lecture Notes in Computer Science}, pages 303--317. Springer, 2001.
\newblock \href {https://doi.org/10.1007/3-540-45315-6\_20} {\path{doi:10.1007/3-540-45315-6\_20}}.

\bibitem[KM77]{kahn1976coroutines}
Gilles Kahn and David~B. MacQueen.
\newblock Coroutines and networks of parallel processes.
\newblock In Bruce Gilchrist, editor, {\em Information Processing, Proceedings of the 7th {IFIP} Congress 1977, Toronto, Canada, August 8-12, 1977}, pages 993--998. North-Holland, 1977.

\bibitem[Kos73]{kosinski1973data}
Paul~R. Kosinski.
\newblock A data flow language for operating systems programming.
\newblock In {\em Proceeding of ACM SIGPLAN-SIGOPS interface meeting on Programming languages-operating systems}, pages 89--94, 1973.

\bibitem[KP93]{kahn1993concrete}
Gilles Kahn and Gordon~D. Plotkin.
\newblock Concrete domains.
\newblock {\em Theoretical Computer Science}, 121(1-2):187--277, 1993.

\bibitem[KS17]{kozen17}
Dexter Kozen and Alexandra Silva.
\newblock Practical coinduction.
\newblock {\em Mathematical Structures in Computer Science}, 27(7):1132--1152, 2017.
\newblock \href {https://doi.org/10.1017/S0960129515000493} {\path{doi:10.1017/S0960129515000493}}.

\bibitem[KSW97]{sabadini95}
Piergiulio Katis, Nicoletta Sabadini, and Robert F.~C. Walters.
\newblock Bicategories of processes.
\newblock {\em Journal of Pure and Applied Algebra}, 115(2):141--178, 1997.

\bibitem[KSW99]{katis99}
Piergiulio Katis, Nicoletta Sabadini, and Robert F.~C. Walters.
\newblock On the algebra of feedback and systems with boundary.
\newblock In {\em Rendiconti del Seminario Matematico di Palermo}, 1999.

\bibitem[KSW02]{katis02}
Piergiulio Katis, Nicoletta Sabadini, and Robert F.~C. Walters.
\newblock Feedback, trace and fixed-point semantics.
\newblock {\em {RAIRO-Theor. Informatics Appl.}}, 36(2):181--194, 2002.
\newblock \href {https://doi.org/10.1051/ita:2002009} {\path{doi:10.1051/ita:2002009}}.

\bibitem[Lam68]{lambek68}
Joachim Lambek.
\newblock A fixpoint theorem for complete categories.
\newblock {\em Mathematische Zeitschrift}, 103(2):151--161, 1968.

\bibitem[Lam86]{lambek1986a}
J.~Lambek.
\newblock Cartesian closed categories and typed {$\lambda$}-calculi.
\newblock In Guy Cousineau, Pierre-Louis Curien, and Bernard Robinet, editors, {\em Combinators and {{Functional Programming Languages}}}, Lecture {{Notes}} in {{Computer Science}}, pages 136--175, {Berlin, Heidelberg}, 1986. {Springer}.
\newblock \href {https://doi.org/10.1007/3-540-17184-3_44} {\path{doi:10.1007/3-540-17184-3_44}}.

\bibitem[LM09]{lee09}
Edward~A. Lee and Eleftherios Matsikoudis.
\newblock The semantics of dataflow with firing.
\newblock {\em From {Semantics} to {Computer} {Science}: {Essays} in {Honour} of {Gilles} {Kahn}}, pages 71--94, 2009.

\bibitem[Lor21]{loregian2021}
Fosco Loregian.
\newblock {\em (Co)end Calculus}.
\newblock London Mathematical Society Lecture Note Series. Cambridge University Press, 2021.
\newblock \href {https://doi.org/10.1017/9781108778657} {\path{doi:10.1017/9781108778657}}.

\bibitem[LS89]{lynch1989proof}
Nancy~A. Lynch and Eugene~W. Stark.
\newblock A proof of the {Kahn} principle for input/output automata.
\newblock {\em Information and Computation}, 82(1):81--92, 1989.

\bibitem[{Mac}78]{maclane78}
Saunders {Mac Lane}.
\newblock {\em Categories for the Working Mathematician}.
\newblock Graduate Texts in Mathematics. Springer New York, 1978.
\newblock \href {https://doi.org/10.1007/978-1-4757-4721-8} {\path{doi:10.1007/978-1-4757-4721-8}}.

\bibitem[Mam20]{mamouras20}
Konstantinos Mamouras.
\newblock Semantic foundations for deterministic dataflow and stream processing.
\newblock In Peter M{\"{u}}ller, editor, {\em Programming Languages and Systems - 29th European Symposium on Programming, {ESOP} 2020, Held as Part of the European Joint Conferences on Theory and Practice of Software, {ETAPS} 2020, Dublin, Ireland, April 25-30, 2020, Proceedings}, volume 12075 of {\em Lecture Notes in Computer Science}, pages 394--427. Springer, 2020.
\newblock \href {https://doi.org/10.1007/978-3-030-44914-8\_15} {\path{doi:10.1007/978-3-030-44914-8\_15}}.

\bibitem[{Mas}53]{mason53}
S.~J. {Mason}.
\newblock {Feedback} {Theory} - {Some} properties of signal flow graphs.
\newblock {\em Proceedings of the Institute of Radio Engineers}, 41(9):1144--1156, 1953.
\newblock \href {https://doi.org/10.1109/JRPROC.1953.274449} {\path{doi:10.1109/JRPROC.1953.274449}}.

\bibitem[Mel09]{mellies2009categorical}
Paul-Andr{\'e} Mellies.
\newblock Categorical semantics of linear logic.
\newblock {\em Panoramas et syntheses}, 27:15--215, 2009.

\bibitem[Mil80]{milner1980calculus}
Robin Milner.
\newblock {\em A calculus of communicating systems}.
\newblock Springer, 1980.

\bibitem[Mil83]{milner1983calculi}
Robin Milner.
\newblock Calculi for synchrony and asynchrony.
\newblock {\em Theoretical computer science}, 25(3):267--310, 1983.

\bibitem[ML69]{maclane06:universe}
Saunders Mac~Lane.
\newblock One universe as a foundation for category theory.
\newblock In {\em Reports of the Midwest Category Seminar III}, pages 192--200, Berlin, Heidelberg, 1969. Springer Berlin Heidelberg.

\bibitem[Oli84]{oliveira84}
Jos{\'{e}}~Nuno Oliveira.
\newblock {\em The formal semantics of deterministic dataflow programs}.
\newblock PhD thesis, University of Manchester, {UK}, 1984.
\newblock URL: \url{http://ethos.bl.uk/OrderDetails.do?uin=uk.bl.ethos.376586}.

\bibitem[Pan99]{panangaden1999}
Prakash Panangaden.
\newblock The category of {Markov} kernels.
\newblock {\em Electronic Notes in Theoretical Computer Science}, 22:171--187, January 1999.
\newblock \href {https://doi.org/10.1016/S1571-0661(05)80602-4} {\path{doi:10.1016/S1571-0661(05)80602-4}}.

\bibitem[Par81]{park1981concurrency}
David Park.
\newblock Concurrency and automata on infinite sequences.
\newblock In {\em Theoretical computer science}, pages 167--183. Springer, 1981.

\bibitem[Pow02]{power02}
John Power.
\newblock Premonoidal categories as categories with algebraic structure.
\newblock {\em Theor. Comput. Sci.}, 278(1-2):303--321, 2002.
\newblock \href {https://doi.org/10.1016/S0304-3975(00)00340-6} {\path{doi:10.1016/S0304-3975(00)00340-6}}.

\bibitem[PS88]{panangaden1988computations}
Prakash Panangaden and Eugene~W. Stark.
\newblock Computations, residuals, and the power of indeterminacy.
\newblock In {\em International Colloquium on Automata, Languages, and Programming}, pages 439--454. Springer, 1988.

\bibitem[PW99]{power99}
John Power and Hiroshi Watanabe.
\newblock Distributivity for a monad and a comonad.
\newblock In Bart Jacobs and Jan J. M.~M. Rutten, editors, {\em Coalgebraic Methods in Computer Science, {CMCS} 1999, Amsterdam, The Netherlands, March 20-21, 1999}, volume~19 of {\em Electronic Notes in Theoretical Computer Science}, page 102. Elsevier, 1999.
\newblock \href {https://doi.org/10.1016/S1571-0661(05)80271-3} {\path{doi:10.1016/S1571-0661(05)80271-3}}.

\bibitem[Rom20]{roman2020}
Mario Rom{\'{a}}n.
\newblock Comb diagrams for discrete-time feedback.
\newblock {\em CoRR}, abs/2003.06214, 2020.
\newblock \href {https://arxiv.org/abs/2003.06214} {\path{arXiv:2003.06214}}.

\bibitem[Rom22]{Roman_Arrow_Streams_for_2022}
Mario Román.
\newblock {Arrow Streams for Dataflow Programming}.
\newblock \url{https://github.com/mroman42/arrow-streams}, 12 2022.
\newblock \href {https://doi.org/10.5281/zenodo.15978541} {\path{doi:10.5281/zenodo.15978541}}.

\bibitem[Ros96]{ross1996stochastic}
Sheldon~M. Ross.
\newblock {\em Stochastic processes}, volume~2.
\newblock John Wiley \& Sons, 1996.

\bibitem[RR88]{robinson88}
Edmund Robinson and Giuseppe Rosolini.
\newblock Categories of partial maps.
\newblock {\em Inf. Comput.}, 79(2):95--130, 1988.
\newblock \href {https://doi.org/10.1016/0890-5401(88)90034-X} {\path{doi:10.1016/0890-5401(88)90034-X}}.

\bibitem[Rut00]{rutten00}
Jan J. M.~M. Rutten.
\newblock Universal coalgebra: a theory of systems.
\newblock {\em Theoretical Computer Science}, 249(1):3--80, 2000.
\newblock \href {https://doi.org/10.1016/S0304-3975(00)00056-6} {\path{doi:10.1016/S0304-3975(00)00056-6}}.

\bibitem[San09]{sangiorgi2009origins}
Davide Sangiorgi.
\newblock On the origins of bisimulation and coinduction.
\newblock {\em ACM Transactions on Programming Languages and Systems (TOPLAS)}, 31(4):1--41, 2009.

\bibitem[See87]{seely87}
Robert~A.G. Seely.
\newblock {\em Linear logic, *-autonomous categories and cofree coalgebras}.
\newblock Ste. Anne de Bellevue, Quebec: CEGEP John Abbott College, 1987.

\bibitem[Sha42]{shannon42}
Claude~E. Shannon.
\newblock {\em The Theory and Design of Linear Differential Equation Machines}.
\newblock Bell Telephone Laboratories, 1942.

\bibitem[SJ19]{jacobs:causalfunctions}
David Sprunger and Bart Jacobs.
\newblock The differential calculus of causal functions.
\newblock {\em CoRR}, abs/1904.10611, 2019.
\newblock URL: \url{http://arxiv.org/abs/1904.10611}, \href {https://arxiv.org/abs/1904.10611} {\path{arXiv:1904.10611}}.

\bibitem[SK19]{katsumata19}
David Sprunger and Shin{-}ya Katsumata.
\newblock Differentiable causal computations via delayed trace.
\newblock In {\em 34th Annual {ACM/IEEE} Symposium on Logic in Computer Science, {LICS} 2019, Vancouver, BC, Canada, June 24-27, 2019}, pages 1--12. {IEEE}, 2019.
\newblock \href {https://doi.org/10.1109/LICS.2019.8785670} {\path{doi:10.1109/LICS.2019.8785670}}.

\bibitem[SLG94]{stoltenberg}
Viggo Stoltenberg{-}Hansen, Ingrid Lindstr{\"{o}}m, and Edward~R. Griffor.
\newblock {\em Mathematical theory of domains}, volume~22 of {\em Cambridge tracts in theoretical computer science}.
\newblock Cambridge University Press, 1994.

\bibitem[TR98]{turi1998foundations}
Daniele Turi and Jan Rutten.
\newblock On the foundations of final coalgebra semantics: non-well-founded sets, partial orders, metric spaces.
\newblock {\em Mathematical Structures in Computer Science}, 8(5):481--540, 1998.

\bibitem[UV05]{uustalu05}
Tarmo Uustalu and Varmo Vene.
\newblock The essence of dataflow programming.
\newblock In Kwangkeun Yi, editor, {\em Programming Languages and Systems, Third Asian Symposium, {APLAS} 2005, Tsukuba, Japan, November 2-5, 2005, Proceedings}, volume 3780 of {\em Lecture Notes in Computer Science}, pages 2--18. Springer, 2005.
\newblock \href {https://doi.org/10.1007/11575467\_2} {\path{doi:10.1007/11575467\_2}}.

\bibitem[UV08]{uustalu2008comonadic}
Tarmo Uustalu and Varmo Vene.
\newblock Comonadic notions of computation.
\newblock In Jiří Ad{\'{a}}mek and Clemens Kupke, editors, {\em Proceedings of the Ninth Workshop on Coalgebraic Methods in Computer Science, {CMCS} 2008, Budapest, Hungary, April 4-6, 2008}, volume 203 of {\em Electronic Notes in Theoretical Computer Science}, pages 263--284. Elsevier, 2008.
\newblock \href {https://doi.org/10.1016/j.entcs.2008.05.029} {\path{doi:10.1016/j.entcs.2008.05.029}}.

\bibitem[WA{\etalchar{+}}85]{wadge1985lucid}
William~W Wadge, Edward~A Ashcroft, et~al.
\newblock {\em Lucid, the dataflow programming language}, volume 303.
\newblock Academic Press London, 1985.

\bibitem[Wen75]{weng1975steam}
Kung-Song Weng.
\newblock Steam-oriented computation in recursive data flow schemas.
\newblock Master's thesis, Massachusetts Institute of Technology, 1975.

\end{thebibliography}

\end{document}